%% file: ms.tex
\documentclass[a4paper,12pt,openany]{book}

\usepackage[table]{xcolor}
\usepackage[final]{pdfpages} 

\usepackage{times} 
\usepackage[hmarginratio=2:2]{geometry}

\usepackage{amsmath, amssymb, latexsym}
\usepackage{amsfonts}
\usepackage{amsthm}
\usepackage{mathtools}
\usepackage[utf8]{inputenc}
\usepackage{epsfig}
\usepackage{pdfpages}
\usepackage{booktabs}
\usepackage{float}
\usepackage{wrapfig}
\usepackage{colonequals}
\usepackage{tikz}
\usepackage{xspace}
\usepackage{paralist}
\usetikzlibrary{arrows.meta,decorations.pathmorphing,positioning,fit,trees,shapes,shadows,automata,calc,decorations.markings,patterns} 
\tikzset{>=latex}

\usepackage{protocolj,arydshln}

\usepackage[ruled,vlined,linesnumbered]{algorithm2e}

\usepackage[hyphens]{url}
\usepackage{multirow}
\usepackage[font={small},skip=3pt,belowskip=-7pt]{caption}

\usepackage{subcaption}
\usepackage[bookmarks=true]{hyperref}

\usepackage{array}
\usepackage{rotating}   
\usepackage{makecell}
\usepackage{colortbl}

\usepackage{emptypage}
\usepackage{diagbox}

\usepackage{subfloat}

\usepackage{pifont}

\usepackage{hyperref} 
\usepackage[nameinlink,capitalize]{cleveref}

\usepackage{algpseudocode}

\SetKwInput{KwInput}{Input}                
\SetKwInput{KwOutput}{Output}              

\newcommand{\diff}[1]{\textcolor{red}{#1}}

\newcommand{\myparagraph}[1]{\vspace{0.1cm}\noindent{\bf #1.}}

\newcolumntype{P}[1]{>{\centering\arraybackslash}p{#1}}

\newcommand{\Hzk}{\ensuremath{\mathsf{H}_{\mathsf{zk}}}}

\newcommand{\gls}[1]{#1}

\newcommand{\phon}{P_{hon}}
\newcommand{\pmal}{P_{grd}}
\newcommand{\halfhalf}{(\frac{1}{2}, \frac{1}{2})}
\newcommand{\bracks}[2]{\left( #1, #2 \right)}

\newtheorem{claim}{Claim}

\newtheorem{corollary}{Corollary}

\definecolor{ao(english)}{rgb}{0.0, 0.5, 0.0}

\newcommand{\varDash}[1]{{\operatorname{\mathit{#1}}}}

\newenvironment{proof1} {\begin{proof}[Justification]} {\end{proof}}

\DeclareUnicodeCharacter{00A0}{ }

\newcommand{\bazka}[4]{
	\begin{tabular}[t]{c|c|c}
		$\phon$/$\pmal$ & \gls{honest} & \gls{greedy} \\
		\hline
		\gls{honest} & (#1,#1) & (#2,#3) \\
		\hline
		\gls{greedy} & (#3,#2) & (#4,#4) \\
\end{tabular}}

\newtheorem{theorem}{Theorem}
\newtheorem{hypothesis}{Hypothesis}

\newcommand{\specialcell}[2][c]{%
	\begin{tabular}[#1]{@{}l@{}}#2\end{tabular}}

\usepackage[normalem]{ulem}
\newcommand{\cmark}{\textcolor{black}{\ding{51}}}
\newcommand{\xmark}{\textcolor{black}{\ding{55}}}
\newcommand{\trot}[1]{\multicolumn{1}{l}{\rlap{\rotatebox{25}{#1}~}}}

\algrenewcommand\algorithmiccomment[1]{\hfill \textcolor{gray}{$\triangleright$ \textit{#1}}}

\algnewcommand{\IIf}[1]{\State\algorithmicif\ #1\ \algorithmicthen}
\algnewcommand{\EElse}[1]{\State\algorithmicelse\ #1\ }
\algnewcommand{\EndIIf}{\unskip}
\algrenewcommand\algorithmicindent{1.2em}%

\newcommand*{\dittoclosing}{---''---}

\usepackage{listings}

\include{prismstyle}

\usepackage{nicefrac}

\newcommand{\name}{CHANGE-ME\xspace}%

\begin{document} 

\input{title.tex}

\pagenumbering{roman}

\input{abstract.tex}

\clearpage
\input{ack.tex}

\newtheorem{mydef}{Definition}
\newtheorem{example}{Example}
\newtheorem{myproof}{Proof}
\newtheorem{myproposition}{Proposition}
\newtheorem{mytheorem}{Theorem}
\newtheorem{myexample}{Example}
\newtheorem{myproblem}{Problem}
\newtheorem{mycorollary}{Corollary}
\newtheorem{mymethod}{Approach}
\newtheorem{remark}{Remark}

\renewcommand{\tableautorefname}{Table}
\renewcommand{\algorithmautorefname}{Algorithm}
\renewcommand{\figureautorefname}{Figure}
\renewcommand{\equationautorefname}{Equation}
\renewcommand{\sectionautorefname}{Appendix}

\renewcommand{\chapterautorefname}{Chapter}
\renewcommand{\sectionautorefname}{Section}
\renewcommand{\subsectionautorefname}{Section}
\renewcommand{\subsubsectionautorefname}{Section}
\renewcommand{\paragraphautorefname}{ Section}
\renewcommand{\tablename}{Table}
\renewcommand{\figurename}{Figure}

\newcommand{\mysubsubsection}[1]{\smallskip\subsubsection{#1}}

\tableofcontents

\newcommand{\squishlist}{
   \begin{list}{$\bullet$}
    { \setlength{\itemsep}{5pt}    \setlength{\parsep}{0pt}
      \setlength{\topsep}{5pt}     \setlength{\partopsep}{0pt}
      \setlength{\leftmargin}{1.35em} \setlength{\labelwidth}{1em}
      \setlength{\labelsep}{0.5em} } }

\newcommand{\squishend}{
    \end{list}  }

\clearpage
\pagenumbering{arabic}

\part{COMMENTED RESEARCH}

\chapter{Introduction}\label{chapter:intro}
\input{./sec/introduction.tex}

\chapter{Background}\label{chapter:background}
\input{./sec/background.tex}

\chapter{Standardization in Threat Modeling}\label{chapter:sra}
\input{./sec/sra-short.tex}


\chapter{Consensus Protocols}\label{chapter:consensus-protocols}
\input{./sec/consensus-protocols.tex}


\chapter{Cryptocurrency Wallets}\label{chapter:wallets}
\input{./sec/wallets.tex}


\chapter{Electronic Voting}\label{chapter:evoting}
\input{./sec/evoting.tex}

\chapter{Secure Logging}\label{chapter:logging}
\input{./sec/logging.tex}

\chapter{Conclusion}\label{chapter:conclusion}
\input{./sec/conclusion.tex}



\clearpage

\addcontentsline{toc}{chapter}{\\Bibliography}
\bibliographystyle{alpha}
\bibliography{unique_identifiers}

\end{document}

%% file: prismstyle.tex
\definecolor{prismgreen}{HTML}{009900}
\definecolor{prismred}{HTML}{cc0000}
\definecolor{prismblue}{HTML}{0000cc}

\lstdefinelanguage{Prism}{
        basicstyle=\color{prismred}\scriptsize\ttfamily,
        literate=*	{0}{{\textcolor{prismblue}{0}}}{1}
			{1}{{\textcolor{prismblue}{1}}}{1}
			{2}{{\textcolor{prismblue}{2}}}{1}
			{3}{{\textcolor{prismblue}{3}}}{1}
			{4}{{\textcolor{prismblue}{4}}}{1}
			{5}{{\textcolor{prismblue}{5}}}{1}
			{6}{{\textcolor{prismblue}{6}}}{1}
			{7}{{\textcolor{prismblue}{7}}}{1}
			{8}{{\textcolor{prismblue}{8}}}{1}
			{9}{{\textcolor{prismblue}{9}}}{1}
			{.0}{{\textcolor{prismblue}{.0}}}{2}
			{.1}{{\textcolor{prismblue}{.1}}}{2}
			{.2}{{\textcolor{prismblue}{.2}}}{2}
			{.3}{{\textcolor{prismblue}{.3}}}{2}
			{.4}{{\textcolor{prismblue}{.4}}}{2}
			{.5}{{\textcolor{prismblue}{.5}}}{2}
			{.6}{{\textcolor{prismblue}{.6}}}{2}
			{.7}{{\textcolor{prismblue}{.7}}}{2}
			{.8}{{\textcolor{prismblue}{.8}}}{2}
			{.9}{{\textcolor{prismblue}{.9}}}{2}
			{[}{{\textcolor{black}{[}}}{1}
			{]}{{\textcolor{black}{]}}}{1}
			{+}{{\textcolor{black}{+}}}{1}
			{-}{{\textcolor{black}{-}}}{1}
			{=}{{\textcolor{black}{=}}}{1}
			{<}{{\textcolor{black}{<}}}{1}
			{>}{{\textcolor{black}{>}}}{1}
			{\&}{{\textcolor{black}{\&}}}{1}
			{|}{{\textcolor{black}{|}}}{1}
			{:}{{\textcolor{black}{:}}}{1}
			{;}{{\textcolor{black}{;}}}{1}
			{(}{{\textcolor{black}{(}}}{1}
			{)}{{\textcolor{black}{)}}}{1}
			{..}{{\textcolor{black}{..}}}{2},
        keywords= {bool,ceil,const,ctmc,double,dtmc,endinit,endmodule,endrewards, endsystem,F,false,floor,formula,G,global,I,init,int,label,max,mdp,min,module,nondeterministic,P,Pmin,Pmax,prob,probabilistic,rate,rewards,Rmin,Rmax,S,stochastic,system,true,U, option, either, assignment, relation, operation, hole, variable},
        keywordstyle={\bfseries\color{black}},
        numberstyle=\footnotesize\color{black},
        comment=[l] {//}, morecomment=[s]{/*}{*/},
        commentstyle= \color{prismgreen},
        tabsize=4,
        captionpos=b,
        escapechar=^,
        moredelim=[is][\color{orange}]{@}{@},
}

%% file: title.tex
\thispagestyle{empty}

\includegraphics[width=150mm]{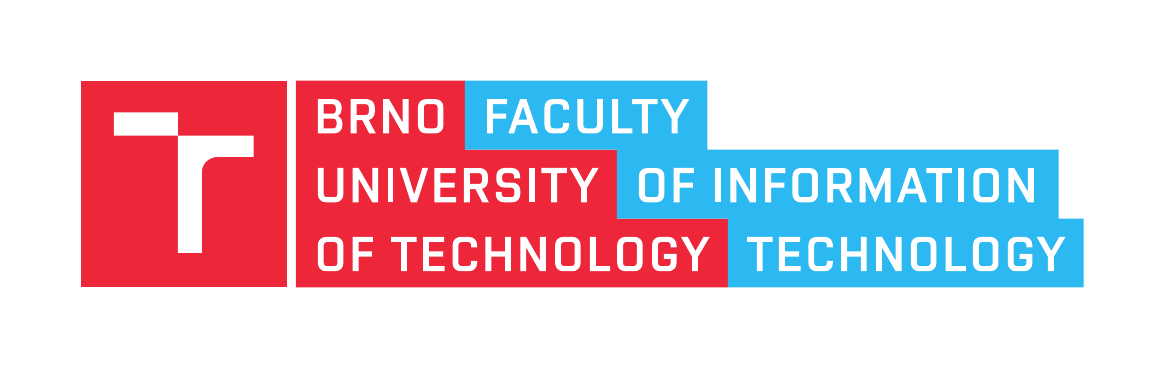}
{\fontfamily{ptm}\selectfont


\vspace{5cm}

\begin{center}
  {\LARGE \bf Towards Secure Decentralized Applications and Consensus Protocols in Blockchains}    \\[10mm]
  {\large (on Selfish Mining, Undercutting Attacks, DAG-Based Blockchains, E-Voting, Cryptocurrency Wallets, Secure-Logging, and CBDC)}  \\[10mm]    
  {\large Habilitation Thesis}      
\end{center}                                            
 
 \vfill
                 
\large Ing. Ivan Homoliak  Ph.D.                                   
  \hfill  Brno 2024

}

\thispagestyle{empty}

\eject

\thispagestyle{empty}

\ \\

%% file: abstract.tex
\section*{Abstract}
With the rise of cryptocurrencies, many new applications and approaches leveraging the decentralization of blockchains have emerged. 
Blockchains are full-stack distributed systems in which multiple sub-systems interact together.
Although most of the deployed blockchains and decentralized applications running on them need better scalability and performance, their security is undoubtedly another critical factor for their success.
Due to the complexity of blockchains and many decentralized applications, their security assessment and analysis require a more holistic view than in the case of traditional distributed or centralized systems.

In this thesis, we summarize our contributions to the security of blockchains and a few types of decentralized applications.
In detail, we contribute to the standardization of vulnerability/threat analysis by proposing a security reference architecture for blockchains.
Then, we contribute to the security of consensus protocols in single-chain Proof-of-Work blockchains and their resistance to selfish mining attacks, undercutting attacks as well as greedy transaction selection attacks on blockchains with Direct Acyclic Graphs.   
Next, we contribute to cryptocurrency wallets by proposing a new classification of authentication schemes as well as a novel approach to two-factor authentication based on One-Time Passwords.
Next, we contribute to the area of e-voting by proposing a practical boardroom voting protocol that we later extend to its scalable version supporting millions of participants, while maintaining its security and privacy properties.
In the area of e-voting, we also propose a novel repetitive voting framework,  enabling vote changes in between elections while avoiding peak-end effects.
Finally, we contribute to secure logging with blockchains and trusted computing by proposing a new approach to a centralized ledger that guarantees non-equivocation, integrity, censorship evidence, and other features.
In the follow-up contribution to secure logging, we built on top of our centralized ledger and propose an interoperability protocol for central bank digital currencies, which provides atomicity of transfer operations.


\section*{Keywords}
Security, privacy, blockchains, distributed ledgers, threat modeling, standardization, decentralized applications, consensus protocols, proof-of-work, selfish mining attacks, undercutting attacks, incentive attacks, DAG-based blockchains, cryptocurrency wallets, two-factor authentication, 2FA, electronic voting, e-voting, public bulletin board, secure logging, data provenance, interoperability, Central Bank Digital Currency, CBDC, Trusted Execution Environment, TEE, cross-chain protocol.

%% file: ack.tex
\section*{Acknowledgment}
First, I would like to thank my supervisor and colleague from SUTD -- Pawel Szalachowski -- for collaborating and igniting many interesting research ideas as well as explaining to me that ideas are important but cheap in contrast to realization. 
Second, I would like to thank all my co-authors, especially Sarad Venugopalan, Pieter Hartel, Daniel Reijsbergen, Federico~Matteo Ben{\v{c}}i{\'c}, Fran Casino, and Martin Hrub{\'y}. 
Third, I would like to thank Tomas Vojnar, Pavel Zemcik, Pavel Smrz, and Ales Smrcka for their support in funding my research at FIT BUT.
Next, I would like to thank all Ph.D./MSc/BSc students I have had the chance to work with, especially Martin Peresini, Ivana Stan{\v{c}}{\'\i}kov{\'a}, and Dominik Breitenbacher / Tomas Hladky, Jakub Handzus, Jakub Kubik, and Martin Ersek / and Rastislav Budinsk{\'y}. 
Also, I would like to thank Milan Ceska (and Tomas Vojnar) for inspiring me in the habilitation process and helping me to maximize cross-domain transfer learning in this process.  
Finally, I would like to thank around 20\% of anonymous reviewers (mostly from security venues) for providing useful and constructive feedback. 


\vfill \emph{Over the time, I have been supported by number of projects, in particular, by EU ECSEL projects, EU HORIZON EUROPE projects, National Research Foundation -- Prime Minister's Office in Singapore, Technology Agency of the Czech Republic, and the Czech IT4Innovations Centre of Excellence project.
}

%% file: sec/introduction.tex

The popularity of blockchain systems has rapidly increased in recent years, mainly due to the decentralization of control that they aim to provide.
Blockchains are full-stack distributed systems in which multiple layers, (sub)systems, and dynamics interact together.
Hence, they should leverage a secure and resilient networking architecture, a robust consensus protocol, and a safe environment for building higher-level applications.
Although most of the deployed blockchains need better scalability and well-aligned incentives to unleash their full potential, their security is undoubtedly a critical factor for their success.
As these systems are actively being developed and deployed, it is often challenging to understand how secure they are, or what security implications are introduced by some specific components they consist of.
Moreover, due to their complexity and novelty (e.g., built-in protocol incentives), their security assessment and analysis require a more holistic view than in the case of traditional distributed systems.

In this work, we first present our contributions to the standardization of vulnerability/threat analysis and modeling in blockchains, and then we present our contributions to particular areas in blockchains' consensus protocols, cryptocurrency wallets,  electronic voting, and secure logging with the focus on security and/or privacy aspects.   
In the following, we introduce these areas and outline our contributions.

\section{Standardization in Threat Modeling}
Although some standardization efforts have already been undertaken in the field of block\-chains and distributed ledgers, they are either specific to a particular platform \cite{EEA-standards} or still under development \cite{iso-security-threats,iso-reference-architecture}.
Hence, there is a lack of platform-agnostic standards in blockchain implementation, interoperability, services, and applications, as well as the analysis of its security threats \cite{gartner-lack-of-standards,barry-medium}.
All of these areas are challenging, and it might take years until they are standardized and agreed upon across a diverse spectrum of stakeholders. 

We believe that it is critical to provide blockchain stakeholders (developers, users, standardization bodies, regulators, etc.) with a comprehensive systematization of knowledge about the security and privacy aspects of today's blockchain systems.
We aim to achieve this goal, with a particular focus on system design and architectural aspects.
We do not limit our work to an enumeration of security issues, but additionally, discuss the origins of those issues while listing possible countermeasures and mitigation techniques together with their potential implications.
In sum, we propose the security reference architecture (SRA) for blockchains, which is based on models that demonstrate the stacked hierarchy of different threat categories (similar to the ISO/OSI hierarchy \cite{zimmermann1980osi}) and is inspired by security modeling performed in the cloud computing \cite{liu2011nist,xiao2013cloud}.
As our next contribution in this direction, we enrich the threat-risk assessment standard ISO/IEC 15408 \cite{cc2017} to fit the blockchain infrastructure.
We achieve this by embedding the stacked model into this standard.
More details in this direction of our research are elaborated in \autoref{chapter:sra}.

\section{Consensus Protocols}
While the previous area of the thesis was theoretical and analytical, in the current area of consensus protocols we aim to investigate practical security aspects of blockchains, and their consensus protocols in particular. 
Consensus protocols represent a means to provide naturally incentivized decentralization, immutability,  and other features of blockchains (see \autoref{sec:background-features}). 
Therefore, modeling and simulation of consensus protocols in terms of security and incentives is an important research direction.
There exist several principally different categories of consensus protocols such as Proof-of-Resource (PoR),  Proof-of-Stake, and Byzantine-Fault-Tolerant protocols (see \autoref{sec:consensus}), each of them potentially vulnerable to different types of threats.
Nevertheless, in this research area, we focus on PoR protocols and Proof-of-Work (PoW) protocols in particular.

As our first contribution, we design StrongChain \cite{strongchain} consensus protocol that improves the resistance of  Nakamoto consensus \cite{nakamoto2008bitcoin} to selfish mining by rewarding partial partial PoW puzzle solutions and incorporating them to the total ``weight'' of the chain.
While the idea of rewarding partial puzzle solutions is not novel \cite{zamyatinflux,pass2017fruitchains,rizun2016subchains}, StrongChain achieves resistance to selfish mining in a space-efficient manner that does not create a new vulnerability (such as selfish mining on a subchain in \cite{zamyatinflux}).
At the same time, StrongChain improves on accuracy of distributed time and decreases the reward variance of miners, and thus it creates better conditions for more decentralized mining.

Our second contribution is in the area of consensus protocols that utilize Directed Acyclic Graphs (DAGs) to solve the limited processing throughput of traditional single-chain Proof-of-Work (PoW) blockchains.
Many such protocols (e.g., Inclusive \cite{lewenberg2015inclusive}, GHOSTDAG \cite{sompolinsky2020phantom}, PHANTOM \cite{sompolinsky2020phantom}, SPECTRE \cite{sompolinsky2016spectre}, Prism \cite{bagaria2019prism}) utilize a random transaction selection (\gls{RTS}) strategy to avoid transaction duplicates across parallel blocks in DAG and thus maximize the network throughput.
However, these works have not rigorously examined incentive-oriented greedy behaviors when transaction selection deviates from the protocol, which motivated our research.
Therefore,  we first perform a generic game-theoretic analysis abstracting several DAG-based blockchain protocols that use the \gls{RTS} strategy \cite{perevsini2023incentive}, and we prove that such a strategy does not constitute a Nash equilibrium. 
Then, we design a simulator \cite{perevsini2023dag} and perform experiments confirming that greedy actors who do not follow the \gls{RTS} strategy can profit more than honest miners and harm the processing throughput of the protocol \cite{perevsini2023incentive}.
We show that this effect is indirectly proportional to the network propagation delay.
Finally, we show that greedy miners are incentivized to form a shared mining pool to increase their profits, which undermines decentralization and degrades the design of the protocols in question.
Finally, we elaborate on a few techniques to mitigate such incentive attacks.

In our last contribution, we mainly focus on the undercutting attacks in the transaction-fee-based regime (i.e., without block rewards) of PoW blockchains with the longest chain fork-choice rule. 
Note that such a regime is expected to occur in Bitcoin's consensus protocol around the year 2140.
Additionally, we focus on two closely related problems: (1) fluctuations in mining revenue and (2) the mining gap -- i.e., a situation, in which the immediate reward from transaction fees does not cover miners' expenditures.
To mitigate these issues, we propose a solution \cite{budinsky2023fee} that splits transaction fees from a mined block into two parts -- (1) an instant reward for the miner of a block and (2) a deposit sent to one or more fee-redistribution smart contracts ($\mathcal{FRSC}$s) that are part of the consensus protocol.
At the same time, these $\mathcal{FRSC}$s reward the miner of a block with a certain fraction of the accumulated funds over a predefined time.
This setting enables us to achieve several interesting properties that improve the incentive stability and security of the protocol, which is beneficial for honest miners.
With our solution, the fraction of \textsc{Default-Compliant} miners who strictly do not execute undercutting attacks is lowered from the state-of-the-art \cite{carlsten2016instability} result of 66\% to 30\%.
More details in this direction of our research are presented in \autoref{chapter:consensus-protocols}.

\section{Cryptocurrency Wallets}
With the recent rise in the popularity of cryptocurrencies, the security and management of crypto-tokens have become critical.
We have witnessed many attacks on users and wallet pro\-viders, which have resulted in significant financial losses. 
To remedy these issues, several wallet solutions have been proposed. 
%
%
According to the previous work~\cite{eskandari2018first,2015-Bitcoin-SOK}, there are a few categories of common (single-factor) key management approaches, such as password-protected/password-derived wallets, hardware wallets, and server-side/client-side hosted wallets.
Each category has its respective drawbacks and vulnerabilities.

To increase the security of former wallet categories, multi-factor authentication (MFA) is often used, which enables spending crypto-tokens only when several secrets are used together.
However, we emphasize that different security implications stem from the multi-factor authentication executed \textit{against a centralized party} (e.g., username/password or Google Authenticator) and \textit{against the blockchain} itself.
In the former, the authentication factor is only as secure as the centralized party, while the latter provides stronger security that depends on the assumption of an honest majority of decentralized consensus nodes (i.e., miners) and the security of cryptographic primitives used. 

In our first contribution in this direction, we propose a classification scheme \cite{homoliak2020air-extend} for cryptocurrency wallets that distinguishes between the authentication factors validated against the blockchain and a centralized party (or a device). 
We apply this classification to several existing wallets that we also compare in terms of various security features.

In our second contribution, we focus on the security vs. usability of wallets using MFA against the blockchain, provided by the wallets from a split control category~\cite{eskandari2018first}.
MFA in these wallets can be constructed by threshold cryptography wallets \cite{goldfeder2015securing,mycelium-entropy}, multi-signatures \cite{Armory-SW-Wallet,Electrum-SW-Wallet,TrustedCoin-cosign,copay-wallet}, and state-aware smart-contracts \cite{TrezorMultisig2of3,parity-wallet,ConsenSys-gnosis}.
The last class of wallets is of our concern, as spending rules and security features can be encoded in a smart contract. 
Although there are several smart-contract wallets using MFA against the blockchain~\cite{TrezorMultisig2of3,ConsenSys-gnosis}, to the best of our knowledge, none of them provides an air-gapped authentication in the form of short OTPs similar to Google Authenticator. 
Therefore, we propose SmartOTPs \cite{homoliak2020smartotps}, a framework for smart-contract  cryptocurrency wallets,
which provides 2FA against data stored on the blockchain.
The first factor is represented by the user's private key and the second factor by OTPs.
To produce OTPs, the authenticator device of Smart\-OTPs utilizes hash-based cryptographic constructs, namely a pseudo-random function, a Merkle tree, and hash chains.
We propose a novel combination of these elements that minimizes the amount of data transferred from the authenticator to the client, which enables us to implement the authenticator in a fully air-gapped setting. 
SmartOTPs provide protection against three exclusively occurring attackers: the attacker who possesses the user's private key \textit{or} the attacker who possesses the user's authenticator \textit{or} the attacker that tampers with the client. 
More details in this direction of our research are presented in \autoref{chapter:wallets}.

\section{Electronic Voting}
Voting is an integral part of democratic governance, where eligible participants can cast a vote for their representative choice (e.g., candidate or policy) through a secret ballot.
Electronic voting (e-voting) is usually centralized and suffers from a single point of failure that can be manifested in censorship, tampering, and issues with the availability of a service.
To improve some features of e-voting, decentralized blockchain-based solutions can be employed, where the blockchain represents a public bulletin board that in contrast to a centralized bulletin board provides extremely high availability, censorship resistance, and correct code execution.
A blockchain ensures that all entities in the voting system have the same view of the actions made by others due to its immutability and append-only features. 
A few blockchain-based e-voting solutions have been proposed in recent years, mostly focusing on boardroom voting~\cite{McCorrySH17,EPRINT:PanRoy18,Li2020,yu2018platform} or small-scale voting~\cite{DBLP:conf/fc/SeifelnasrGY20,icissp:DMMM18,Li2020}. 

Decentralization was a desired property of e-voting even before the invention of blockchains.
For example, (partially) decentralized e-voting that uses the homomorphic properties of El-Gamal encryption was introduced by Cramer et al.~\cite{cgs97}. 
It assumes a threshold number of honest election authorities to provide the privacy of vote.
However, when this threshold is adversarial, it does not protect from computing partial tallies, making statistical inferences about it, or even worse  -- revealing the vote choices of participants.
A solution that removed trust in tallying authorities was for the first time proposed by Kiayias and Yung~\cite{Kiayias2002} in their privacy-preserving self-tallying boardroom voting protocol.
A similar protocol was later proposed by Hao et al.~\cite{HaoRZ10}, which was later extended to a blockchain environment by McCorry et al.~\cite{McCorrySH17} in their Open Vote Network (OVN).
%
An interesting property of OVN is that it requires only a single honest voting participant to maintain the privacy of the votes.
%
However, OVN  supports only two vote choices (based on~\cite{HaoRZ10}), assumes no stalling participants, requires expensive on-chain tally computation, and does not scale in the number of participants.
The scalability of OVN was partially improved by Seifelnasr et al.~\cite{DBLP:conf/fc/SeifelnasrGY20}, but retaining the limitation of 2 choices and missing robustness.

In our first contribution within blockchain-based electronic voting, we introduce BBB-Voting \cite{homoliak2023bbb}, a similar blockchain-based approach for decentralized voting such as OVN, but in contrast to OVN, BBB-Voting supports 1-out-of-$k$ choices and provides robustness that enables recovery from stalling participants.
We make a  cost-optimized implementation using an Ethereum-based environment, which we compare with OVN and show that our work decreases the costs for voters by $13.5\%$ in normalized gas consumption.
Finally, we show how BBB-Voting can be extended to support the number of participants limited only by the expenses paid by the semi-trusted\footnote{The authority is only trusted to do identity management of participants honestly, which is an equivalent trust model as in OVN.} authority and the computing power to obtain the tally.

In our second contribution, we introduce SBvote \cite{stanvcikova2023sbvote} (as an extension of BBB-Voting),
a blockchain-based self-tallying voting protocol that is scalable in the number of voters, and therefore suitable for large-scale elections.
The evaluation of our proof-of-concept implementation shows that the protocol's scalability is limited only by the underlying blockchain platform.
Despite the limitations imposed by the throughput of the blockchain platforms, SBvote can accommodate elections with millions of voters.
We evaluated the scalability of SBvote on two public smart contract platforms -- Gnosis and Harmony.

In our last contribution, we propose Always on Voting (AoV) \cite{venugopalan2023always} -- a repetitive blockchain-based voting framework that allows participants to continuously vote and change elected candidates or policies without waiting for the next elections.
Participants are permitted to privately change their vote at any point in time, while the effect of their change is manifested at the end of each epoch, whose duration is shorter than the time between two main elections. 
To thwart the problem of peak-end effect in epochs, the ends of epochs are randomized and made unpredictable, while preserved within soft bounds. 
In AoV, we make the synergy between a Bitcoin puzzle oracle, verifiable delay function, and smart contract properties to achieve these goals.
AoV can be integrated with various existing blockchain-based e-voting solutions.
More details in this direction of our research are presented in \autoref{chapter:evoting}.

\section{Secure Logging}

Centralized ledger systems designed for secure logging are append-only databases providing immutability (i.e., tamper resistance) as a core property. 
To facilitate their append-only feature, cryptographic constructions, such as hash chains or hash trees, are usually deployed.
Traditionally, public ledger systems are centralized, and controlled by a single entity that acts as a trusted party.  
In such a setting, ledgers are being deployed in various applications, including payments, logging, timestamping services, repositories, or public logs of various artifacts (e.g.,
keys \cite{melara2015coniks,chase2019seemless}, certificates issued by authorities \cite{laurie2013certificate}, and binaries \cite{fahl2014hey}).
Unfortunately, centralized ledgers have also several drawbacks, like a lack of efficient verifiability or a higher risk of censorship and equivocation.

In our first contribution to secure logging, we propose Aquareum \cite{homoliak2020aquareum}, a framework for centralized ledgers mitigating their main limitations.  
Aquareum employs a trusted execution environment (TEE) and a
public smart contract platform to provide verifiability, non-equivocation, and mitigation of censorship.  
In Aquareum, a ledger operator deploys a pre-defined TEE enclave code, which verifies the consistency and correctness of the ledger for every ledger update. 
Then, proof produced by the enclave is published at an existing public smart contract platform, guaranteeing that the given snapshot of the ledger is verified and no alternative snapshot of this ledger exists.  
Furthermore, whenever a client suspects that her query (or transaction)
is censored, she can (confidentially) request a resolution of the query via the smart contract platform.
The ledger operator noticing the query is obligated to handle it by passing the query to the enclave that creates a public proof of query resolution and publishes it using the smart contract platform.  
With such a censorship-evident design, an operator is publicly visible when misbehaving, thus the clients can take appropriate actions (e.g., sue the operator) or encode some automated service-level agreements into their smart contracts. 
Since Aquareum is integrated with a Turing-complete virtual machine,
it allows arbitrary transaction processing logic, including tokens or
client-specified smart contracts.

In our second contribution, we present CBDC-AquaSphere, a protocol that uses a combination of a trusted execution environment (TEE) and a public blockchain to enable interoperability over independent centralized CBDC ledgers (based on Aquareum). 
Our interoperability protocol uses a custom adaptation of atomic swap protocol and is executed by any pair of CBDC instances to realize a one-way transfer.
It ensures features such as atomicity, verifiability, correctness, censorship resistance, and privacy while offering high scalability in terms of the number of CBDC instances.
Our approach enables two possible deployment scenarios that can be combined: 
(1) CBDC instances represent central banks of multiple countries, and 
(2) CBDC instances represent the set of retail banks and a paramount central bank of a single country.
More details in this direction of our research are presented in \autoref{chapter:logging}.

\section{Author's Contribution} 
In \autoref{tab:contribution}, we describe the author's contributions to the papers contained in this thesis.\footnote{Note that the table contains alphabetic ordering of papers by the author names.}
Since there is no standard metric assessing the qualitative and quantitative contributions, the table describes the author's contribution to common parts in the process of creating a paper in computer science.  

\newcommand{\cross}[0]{\diagbox[height=\line]{}{}  \diagbox[dir=NE, height=\line]{}{}}

\newcommand{\mybox}[0]{\rule{0.3cm}{0.3cm}}
\newcommand{\boxGray}[0]{\color{lightgray} \mybox}
\newcommand{\boxBlack}[0]{\color{black} \mybox}
\newcommand{\rectG}[0]{\boxGray\boxGray}
\newcommand{\rectB}[0]{\boxBlack\boxBlack}

\renewcommand{\arraystretch}{1.8}

\begin{table}[t]
	\centering 
	\scriptsize
	\begin{tabular} { c c c c c c c c c }
		\toprule
		\textbf{Paper} & \textbf{Topic} & \textbf{Approach}  & \textbf{Proofs} & \textbf{Implementation} & \textbf{Experiments} & \textbf{\specialcell{Security\\Analysis}} & \textbf{\specialcell{Literature\\~~Review}} & \textbf{Writing}  \\ 
		\toprule
		
		\textbf{\cite{budinsky2023fee}}  & \rectG & \rectG & \cross &   & \rectG & \rectB & \rectB & \rectB \\
		
		\cite{drga2023detecting} & \rectB & \rectB & \cross & \rectB & \rectB & \rectB & \rectG & \rectB \\
		
		\cite{hartel2019empirical} &  & \rectG & \cross & \cross & \rectB & \cross & \rectG & \rectB \\
		
		\cite{hellebrandt2019increasing} & & \rectG & \cross & \cross & \cross & \rectG & \rectG & \rectB \\
		
		\textbf{\cite{homoliak2020smartotps}} &  & \rectB & \rectG & \rectB & \rectB & \rectB & \rectB  & \rectB \\
		
		\textbf{\cite{homoliak2023cbdc}} & \rectB  & \rectB & \rectG & \rectG & \rectG & \rectB & \rectB & \rectB \\
		
		\textbf{\cite{homoliak2020aquareum}} &  & \rectB & \rectG & \rectB & \rectB & \rectB & \rectB & \rectB \\	
		
		\cite{homoliak2019security} &  &  \rectB & \cross & \cross & \cross & \rectB & \rectB & \rectB \\
		
		\textbf{\cite{homoliak2020security}} & & \rectB & \cross & \cross & \cross & \rectB &  \rectB & \rectB \\
		
		\cite{hum2020coinwatch} &  & & \cross & & \rectB & \rectB & \rectG & \rectB \\
		
		\cite{perevsini2021dag} & \rectB & \rectB & \cross &  & \rectB & \rectB & \rectG & \rectB \\
		
		\cite{perevsini2023sword} &  & \rectG & \cross &  & \rectG & \rectB & \rectG & \rectB \\
		
		\textbf{\cite{perevsini2023incentive}}& \rectB & \rectB & \cross &  & \rectB & \rectB & \rectB & \rectB \\
		
		\textbf{\cite{stanvcikova2023sbvote}} & \rectB & \rectB & \cross & & & \rectG & \rectG & \rectG \\
		
		\textbf{\cite{szalachowski2019strongchain}} &  & \rectG &  & \rectB & \rectB & \rectB & \rectB & \\
		
		\textbf{\cite{homoliak2023bbb}} &  & \rectG & \rectG & \rectB & \rectB & \rectB & \rectB & \rectB \\
		
		\textbf{\cite{venugopalan2023always}} & & \rectG & & & & \rectB & \rectB & \rectB \\
		
		\bottomrule
		
	\end{tabular}
	\caption{The author's contributions to the selected papers related to this thesis. 
		Essential contribution is depicted in black, partial (still important) contribution is depicted in gray, minor or no contribution is depicted in white, and non-applicable field is depicted by a cross. 
		The papers highlighted in bold are attached to this thesis.}
	\label{tab:contribution} 
\end{table}
\renewcommand{\arraystretch}{1}

\section{Organization of the Thesis}
The rest of the thesis is organized as follows.
In \autoref{chapter:background}, we describe preliminaries and background related to this thesis.
Next, in \autoref{chapter:sra}, we describe our contributions to the standardization of threat modeling for blockchains, and we introduce the security reference architecture as a layered model.
\autoref{chapter:consensus-protocols} summarizes our contributions to the security of the consensus protocol in blockchains -- in particular, we focus on Proof-of-Work (PoW) protocols: we describe StrongChain, transaction selection (incentive) attacks on DAG-based blockchains, and undercutting attacks on PoW blockchains.
In \autoref{chapter:wallets}, we deal with cryptocurrency wallets, where we describe our proposed classification of authentication schemes for such wallets as well as SmartOTPs, our contribution to the two-factor authentication on blockchains, and its security. 
Then, in \autoref{chapter:evoting}, we focus on electronic voting using blockchains as an instance of a public bulletin board, and we describe our proposals BBB-Voting and SBvote as well as the always-on-voting framework for repetitive voting, accompanied with the analysis of their security and privacy aspects.
\autoref{chapter:logging} focuses on secure logging, where we present Aquareum, a centralized ledger based on blockchain and trusted computing; later in this chapter, we build on Aquareum and propose an interoperability protocol for central bank digital currencies called CBDC-AquaSphere.
\autoref{chapter:conclusion} concludes the paper and outlines our future research directions.

%% file: sec/background.tex
In this chapter, we summarize the background and preliminaries of the thesis.
The reader familiar with the topics of blockchain, trusted computing, and integrity-preserving data structures can skip this chapter and proceed to \autoref{chapter:sra}.
This chapter is based on the papers \cite{homoliak2019security,homoliak2020security,strongchain,venugopalan2023always,homoliak2020aquareum,homoliak2023cbdc}.

The blockchain is a data structure representing an append-only distributed ledger
that consists of entries (a.k.a., transactions) aggregated within ordered blocks.
The order of the blocks is agreed upon by mutually untrusting participants running a consensus protocol -- these participants are also referred to as nodes.
The blockchain is resistant against modifications by design since blocks are linked using a cryptographic hash function, and each new block has to be agreed upon by nodes running a consensus protocol.

A transaction is an elementary data entry that may contain arbitrary data, e.g., an order to transfer native cryptocurrency (i.e., crypto-tokens), a piece of application code (i.e., smart contract), the execution orders of such application code, etc. 
Transactions sent to a blockchain are validated by all nodes that maintain a
replicated state of the blockchain.

\section{Features of Blockchains}\label{sec:features-of-blockchain}
Blockchains were initially introduced as a means of coping with the centralization of monetary assets management, resulting in their most popular application -- a decentralized cryptocurrency with a native crypto-token.
Nevertheless, other blockchain applications have emerged, benefiting from features other than decentralization, e.g.,  privacy, energy efficiency, throughput, etc.
We split the features of blockchains into inherent and non-inherent ones, where the former involves ``traditional'' features that were aimed to provide by all blockchains while the latter involves features specific to particular blockchain types.  
These features are summarized in the following.

\subsection{Inherent Features}\label{sec:background-features}
\begin{compactitem}
	\item [\textbf{Decentralization:}] is achieved by a distributed consensus protocol -- the protocol ensures that each modification of the ledger is a result of interaction among participants. 
	In the consensus protocol, participants are equal, i.e., no single entity is designed as an authority.
	An important result of decentralization is resilience to node failures.
	
	\item [\textbf{Censorship Resistance:}] is achieved  due to decentralization, and it ensures that each valid transaction is processed and included in the blockchain.

	\item [\textbf{Immutability:}] means that the history of the ledger cannot be easily modified -- it requires a significant quorum of colluding nodes.
	The immutability of history is achieved by a cryptographic one-way function (i.e., a hash function) that creates integrity-preserving links between the previous record (i.e., block) and the current one. 
	In this way, integrity-preserving chains (e.g., blockchains) or graphs (e.g., direct acyclic graphs~\cite{sompolinsky2016spectre,rocket2018snowflake,popov2016tangle} or trees~\cite{sompolinsky2013accelerating}) are built in an append-only fashion.
	However, the immutability of new blocks is not immediate and depends on the time to the finality of a particular consensus protocol (see \autoref{sec:backgroun-design-goal}).
	
	\item [\textbf{Availability:}] although distributed ledgers are highly redundant in terms of data storage (i.e., full nodes store replicated data), the main advantage of such redundancy is paid off by the extremely high availability of the system.
	This feature may be of special interest to applications that cannot tolerate outages. 
	
	\item [\textbf{Auditability:}] correctness of each transaction and block recorded in the blockchain can be validated by any participating node, which is possible due to the publicly-known rules of a consensus protocol.

	\item [\textbf{Transparency:}] the transactions stored in the blockchain as well as the actions of protocol participants are visible to other participants and in most cases even to the public.	
	
\end{compactitem}


\subsection{Non-Inherent Features}
Additionally to the inherent features, blockchains may be equipped with other features that aim to achieve extra goals.
Below we list a few examples of such non-inherent features.
\begin{compactitem}
	\item[\textbf{Energy Efficiency:}] running an open distributed ledger often means that scarce resources are wasted (e.g., Proof-of-Work).
	However, there are available consensus protocols that do not waste scarce resources, but instead emulate the consumption of scarce resources (i.e., Proof-of-Burn), or the interest rate on an investment (i.e., Proof-of-Stake). 
	See examples of these protocols in \autoref{sec:consensus}.
	
	\item[\textbf{Scalability:}] describes how the consensus protocol scales when the number of participants increases. 
	Protocols whose behavior is not negatively affected by an increasing number of participants have high scalability.
	
	\item[\textbf{Throughput:}] represents the number of transactions that can be processed per unit of time. 
	Some consensus protocols have only a small throughput (e.g., Proof-of-Work), while others are designed with the intention to maximize throughput (e.g., Byzantine Fault Tolerant (BFT) protocols with a small number of participants).

	\item[\textbf{Privacy \& Anonymity:}] 
	by design, data recorded on a public blockchain is visible to all nodes or public, which may lead to privacy and anonymity issues. 
	Therefore, multiple solutions increasing anonymity (e.g., ring signatures~\cite{rivest2001leak} in Monero) and privacy (e.g., zk-SNARKs~\cite{ben2014succinct} in Zcash) were proposed in the context of cryptocurrencies, while other efforts have been made in privacy-preserving smart contract platforms~\cite{kosba2016hawk,cheng2018ekiden}.
	
	\item [\textbf{Accountability and Non-Repudiation:}] if blockchains or  applications running on top of them are designed in such a way that identities of nodes (or application users) are known and verified, accountability and non-repudiation of actions performed can be provided too.
	
\end{compactitem}

\begin{figure}[t]
	\centering        
	\includegraphics[width=0.73\columnwidth]{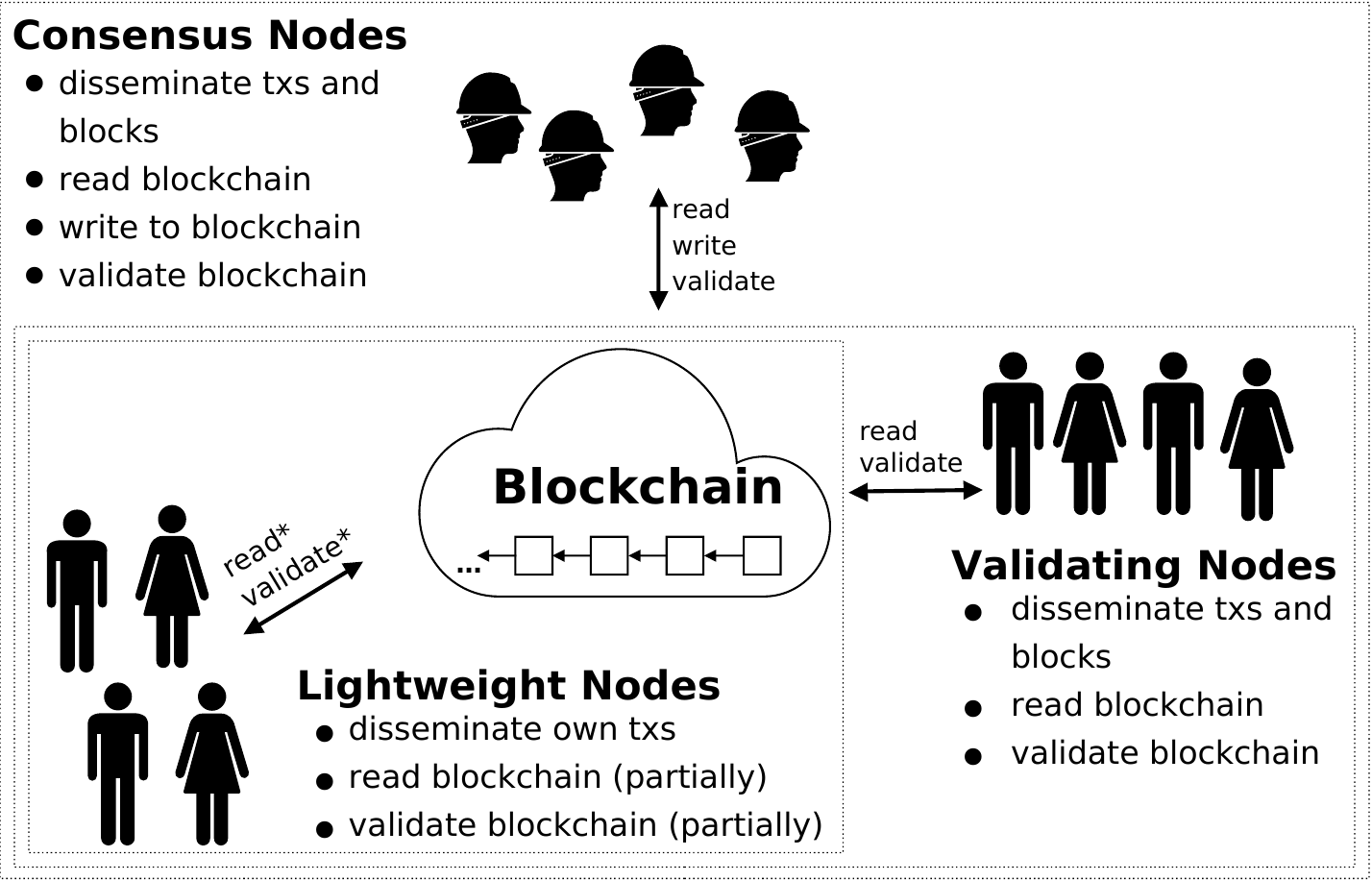}         
	\caption{Involved parties with their interactions and hierarchy.}
	\label{fig:node-types}
\end{figure}

\section{Involved Parties}\label{sec:involved-parties}
Blockchains usually involve three native types of parties that can be organized into a hierarchy, according to the actions that they perform (see~\autoref{fig:node-types}):

\begin{compactdesc}
	
	\item[\textbf{(1) Consensus nodes}] (a.k.a., \textit{miners} in Proof-of-Resource protocols) actively participate in the underlying consensus protocol. 
	These nodes can read the blockchain and write to it by appending new transactions.
	Additionally, they can validate the blockchain and thus check whether writes of other consensus nodes are correct. 
	Consensus nodes can prevent malicious behaviors (e.g., by not appending invalid transactions, or ignoring an incorrect chain).
	
	\item[\textbf{(2) Validating nodes}] read the entire blockchain, validate it, and disseminate transactions. 
	Unlike consensus nodes, validating nodes cannot write to the blockchain, and thus they cannot prevent malicious behaviors. 
	On the other hand, they can detect malicious behavior since they possess copies of the entire blockchain.
	
	\item[\textbf{(3) Lightweight nodes}] (a.k.a., clients or Simplified Payment Verification (SPV) clients) benefit from most of the blockchain functionalities, but they are equipped only with limited information about the blockchain.  
	These nodes can read only fragments of the blockchain (usually block headers) and validate only a small number of transactions that concern them, while they rely on consensus and validating nodes. 
	Therefore, they can detect only a limited set of attacks, pertaining to their own transactions. 
\end{compactdesc}

\subsubsection{\textit{Additional Involved Parties}}
Note that besides native types of involved parties, many applications using or running on the blockchain introduce their own (centralized) components.

\section{Types of Blockchains}\label{sec:types-of-blockchains}
Based on how a new node enters a consensus protocol, we distinguish the following blockchain types:
\begin{compactitem}
	\item[\textbf{Permissionless}] blockchains allow anyone to join the consensus  protocol without permission.
	To prevent Sybil attacks, this type of blockchains usually requires consensus nodes to establish their identities by running a Proof-of-Resource protocol,  where the consensus power of a node is proportional to its resources allocated. 
	\item[\textbf{Permissioned}] blockchains require a consensus node to obtain permission to join the consensus protocol from a centralized or federated authority(ies), while nodes usually have equal consensus power (i.e., one vote per node).
	%
	\item[\textbf{Semi-Permissionless}] blockchains require a consensus node to obtain some form of permission (i.e., stake) before joining the protocol; however, such permission can be given by any consensus node.
	The consensus power of a node is proportional to the stake that it has. 
	%
\end{compactitem}

\section{Design Goals of Consensus Protocols}\label{sec:backgroun-design-goal}

\subsection{Standard Design Goals -- Liveness and Safety}\label{sec:standard-design-goals}
The standard design goals of consensus protocols are \textit{liveness} and \textit{safety}. 
To meet these goals, an \textit{eventual-synchrony} network model~\cite{dwork1988consensus} is usually assumed due to its simplicity. 
In this model, upper bounds are put on an asynchronous delivery of each message, hence each message is eventually/synchronously delivered.
\textbf{Liveness} ensures that all valid transactions are eventually processed -- i.e.,  if a transaction is received by a single honest node, it will eventually be delivered to all honest nodes.
\textbf{Safety} ensures that if an honest node accepts (or rejects) a transaction, then all other honest nodes make the same decision.
Usually, consensus protocols satisfy safety and liveness only under certain assumptions: the minimal fraction of honest consensus power or the maximal fraction of adversarial consensus power.
With regard to safety, literature often uses the term \textit{finality} and \textit{time to finality}. 
\textbf{Finality} represents the sequence of the blocks from the genesis block up to the block $B$, where it can be assumed that this sequence of blocks is infeasible to overturn.
To reach finality up to the block $B$, several successive blocks need to be appended after $B$ -- the number of such blocks is referred to as \textit{the number of confirmations}. 
%

\subsection{Specific Design Goals}
As a result of this study, we learned that standard design goals of the consensus protocol should be amended by specific design goals related to the type of the blockchain.
In permissionless type, \textit{elimination of Sybil entities}, \textit{a fresh and fair leader/committee election}, and \textit{non-interactive verification of the consensus result} is required to meet. 
In contrast, the (semi)-permissionless types do not require the elimination of Sybil entities.

\subsection{Means to Achieve Design Goals}
\paragraph{Simulation of the Verifiable Random Function (VRF).}\label{sec:background-VRF}
To ensure a fresh and fair leader/committee election, all consensus nodes should contribute to the pseudo-randomness generation that determines the fresh result of the election. 
This can be captured by the concept of the VRF~\cite{micali1999verifiable}, which ensures the unpredictability and fairness of the election process.
Therefore, the leader/committee election process can be viewed as a simulation of VRF~\cite{wang2018survey}.
Due to the properties of VRF, the correctness of the election result can be verified non-interactively after the election took place.

\paragraph{Incentive and Rewarding Schemes.}\label{sec:incentives}
An important aspect for protocol designers is to include a rewarding/incentive scheme that motivates consensus nodes to participate honestly in the protocol.
In the context of public (permissionless) blockchains that introduce their native crypto-tokens, this is achieved by block creation rewards as well as transaction fees, and optionally penalties for misbehavior.
Transaction fees and block creation rewards are attributed to the consensus node(s) that create a valid block (e.g., \cite{nakamoto2008bitcoin}), although alternative incentive schemes rewarding more consensus nodes at the same time are also possible (e.g., \cite{strongchain}).
While transaction fees are included in a particular transaction, the block reward is usually part of the first transaction in the block (a.k.a., \textit{coinbase} transaction).

\section{Basis of Consensus Protocols}\label{sec:background-basis}
\textit{Lottery} and \textit{voting} are two marginal techniques that deal with the establishment of a consensus~\cite{hyperledger1}.
However, in addition to them, their combinations have become popular.

\paragraph{Lottery-Based Protocols.}
These protocols provide consensus by running a lottery that elects a leader/committee, who produces the block.
The advantages of lottery-based approaches are a small network traffic overheads and high scalability (in the number of consensus nodes) since the process is usually non-interactive (e.g.,~\cite{nakamoto2008bitcoin}, \cite{chen2017security}, \cite{kiayias2017ouroboros}). 
However, a disadvantage of this approach is the possibility of multiple ``winners'' being elected, who propose conflicting blocks, which naturally leads to inconsistencies called \textit{forks}.
Forks are resolved by fork-choice rules, which compute the difficulty of each branch and select the one. 
For the \textit{longest chain rule}, the chain with the largest number of blocks is selected in the case of a conflict, while for the \textit{strongest chain rule}, the selection criteria involve the quality of each block in the chain (e.g., \cite{sompolinsky2016spectre,sompolinsky2013accelerating,ethereum-classic-GHOST,zamyatinflux,strongchain}).
Note that the possibility of forks in this category of protocols causes an increase of the time to finality, which in turn might enable some attacks such as double-spending. 

\paragraph{Voting-Based Protocols.} 
In this group of protocols, the agreement on transactions is reached through the votes of all participants. 
Examples include Byzantine Fault Tolerant (BFT) protocols -- which require the consensus of a majority quorum (usually $\frac{2}{3}$) of all consensus nodes
(e.g.,  \cite{castro1999practical,aublin2013rbft,buchman2018tendermint,kiayias2018ouroboros,duan2014bchain}).
The advantage of this category is a low-latency finality due to a negligible likelihood of forks. 
The protocols from this group suffer from low scalability, and thus their throughput forms a trade-off with scalability (i.e., the higher the number of nodes, the lower the throughput).

\paragraph{Combinations.}
To improve the scalability of voting-based protocols, it is desirable to shrink the number of consensus nodes participating in the voting by a lottery, so that only nodes of such a committee vote for a block (e.g.,~\cite{gilad2017algorand}, \cite{daian2017snow}, \cite{Hanke2018}, \cite{zilliqa2017zilliqa}, \cite{kiayias2018ouroboros}). 
Another option to reduce active voting nodes is to split them into several groups (a.k.a., \textit{shards}) that run a consensus protocol in parallel (e.g., \cite{Kokoris-KogiasJ18-omniledger, ZamaniM018-rapidchain}). 
Such a setting further increases the throughput in contrast to the single-group option, but on the other hand, it requires a mechanism that accomplishes inter-shard transactions.

\section{Failure Models in Distributed Consensus Protocols}
The relevant literature mentions two main failure models for consensus protocols~\cite{schneider1990implementing}: 

\begin{compactitem}
	\item[\textbf{Fail-Stop Failures:}]
	A node either stops its operation or continues to operate, while obviously exposing its faulty behavior to other nodes.
	Hence, all other nodes are aware of the faulty state of that node (e.g., tolerated in Paxos~\cite{lamport1998paxos}, Raft~\cite{ongaro2014search}, Viewstamped Replication~\cite{oki1988viewstamped}).
	
	\item[\textbf{Byzantine Failures:}]
	In this model, the failed nodes (a.k.a., Byzantine nodes) may perform arbitrary actions, including malicious behavior targeting the consensus protocol and collusions with other Byzantine nodes.
	Hence, the Byzantine failure model is of particular interest to security-critical applications, such as blockchains (e.g., Nakamoto's consensus~\cite{nakamoto2008bitcoin}, pure BFT protocols~\cite{castro1999practical}, \cite{aublin2013rbft}, \cite{buchman2018tendermint}, \cite{cachin2002sintra}, and hybrid protocols~\cite{gilad2017algorand,Kokoris-KogiasJ18-omniledger,ZamaniM018-rapidchain}). 
\end{compactitem}

\section{Nakamoto Consensus and Bitcoin}
\label{sec:pre:bitcoin}
The Nakamoto consensus protocol allows decentralized and distributed network
comprised of mutually distrusting participants to reach an agreement on the
state of the global distributed ledger (i.e., blockchain)~\cite{nakamoto2008bitcoin}.  
To resolve any \textit{forks} of the blockchain the protocol specifies to always
accept the longest chain as the current one.  
Bitcoin is a peer-to-peer cryptocurrency that deploys Nakamoto consensus as
its core mechanism to avoid double-spending. 
Transactions spending bitcoins are announced to
the Bitcoin network, where miners validate, serialize all non-included
transactions, and try to create (mine) a block of transactions with a PoW
embedded into the block header.  
A valid block must fulfill the condition that for a cryptographic hash function
$H$, the hash value of the block header is less than the target $T$.

\subsection{Incentive Scheme}
Brute-forcing the nonce (together with some other fields) is
the only way to produce the PoW, which costs computational resources
of the miners.
To incentivize miners, the Bitcoin protocol allows the miner who finds a block
to insert a coinbase transaction minting a specified amount of new
bitcoins and collecting transaction fees offered by the included transactions.
Currently, every block mints 6.25 new bitcoins. This amount is halved 
every four years, upper-bounding the number of bitcoins that will be created to
a fixed total of 21 million coins. It implies that after around the year 2140, 
no new coins will be created, and the transaction fees will be the only source of
reward for miners. Because of its design, Bitcoin is a deflationary currency. 

In the original white paper, Nakamoto heuristically argues that the consensus
protocol remains secure as long as a majority ($>50\%$) of the participants'
computing power honestly follow the rules specified by the protocol, which is
compatible with their own economic incentives.

\subsection{Difficulty and Fork-Choice Rule}
The overall hash rate of the Bitcoin network  and the difficulty of the PoW
determine how long it takes to generate a new block for the whole network (the
block interval).  To stabilize the block interval at about 10 minutes for
the constantly changing total mining power, the Bitcoin network adjusts the target $T$ every
$2016$ blocks (about two weeks, i.e., a \textit{difficulty window}) according to
the following formula
\begin{equation}\label{eqn:adjust_Ts}
	T_{new} = T_{old} 
	\cdot \frac{\textit{Time~of~the~last~2016~blocks}}{\textit{2016}\cdot\textit{10~minutes}}.
\end{equation}
In simple terms, the difficulty increases if the network is finding blocks
faster than every 10 minutes, and decrease otherwise.  With dynamic difficulty, Nakamoto's longest chain fork-choice rule was considered as a
bug,\footnote{\url{https://goo.gl/thhusi}} as it is trivial to produce long chains
that have low difficulty. The rule was replaced by the strongest-PoW chain rule
where competing chains are measured in terms of PoW they aggregated.  As long as
there is one chain with the highest PoW, this chain is chosen as the current
one.

\subsection{UTXO Model}
Bitcoin introduced and uses the \textit{unspent transaction output} (UTXO) model.  The
validity of a Bitcoin transaction is verified by executing a script proving that
the transaction sender is authorized to redeem unspent coins.
Also, the Bitcoin scripting language offers a mechanism (\texttt{OP\_RETURN}) for
recording data on the blockchain, which facilitates third-party applications
built-on Bitcoin.

\subsection{Light Clients and Simple Payment Verification (SPV)}
Bitcoin proposes the simplified payment verification (SPV) protocol, that allows
re\-source-limited clients to verify that a transaction is indeed included in a
block provided only with the block header and a short transaction's inclusion
proof. The key advantage of the protocol is that SPV clients can verify the
existence of a transaction without downloading or storing the whole block.  SPV
clients are provided only with block headers and on-demand request from the
network inclusion proofs of the transactions they are interested~in.
%
%
%
%

\section{Integrity Preserving Data Structures}\label{sec:background:integrity-structures}

\subsection{Merkle Tree}\label{sec:MT-background}
A Merkle tree~\cite{merkle1989certified} is a  data structure based on the binary tree in which each leaf node contains a hash of a single data block, while each non-leaf node contains a hash of its concatenated children.
At the top of a Merkle tree is the root hash, which provides a tamper-evident summary of the contents.
A Merkle tree enables efficient verification as to whether some data are associated with a leaf node by comparing the expected root hash of a tree with the one computed from a hash of the data in the query and the remaining nodes required to reconstruct the root hash (i.e., \textit{proof} or \textit{authentication path}).
The reconstruction of the root hash has the logarithmic time and space complexity, which makes the Merkle tree an efficient scheme for membership verification.
The Merkle tree is utilized for example in  Bitcoin (and other blockchains) for aggregation of transactions within a block into the root hash, providing the integrity snapshot and at the same time enabling SPV clients to download only data related to authentication path of the transaction whose inclusion in a block is to be verified.

To provide a membership verification of element $x_i$ in the list of elements $X = \{x_i\}, i \geq 1$, the Merkle tree supports the following operations:
\begin{itemize}
	\item{$\mathbf{MkRoot(X) \rightarrow  Root}$:} an aggregation of all elements of the list $X$ by a Merkle tree, providing a single value $Root$. 
	
	\item{$\mathbf{MkProof(x_i, X) \rightarrow  \pi^{mk}}$:} a Merkle proof generation for the $i$th element $x_i$ present in the list of all elements $X$. 
	
	\item{$\mathbf{\pi^{mk}.Verify(x_i, Root) \rightarrow  \{True, False\}}$:}  verification of the Merkle proof $\pi^{mk}$, witnessing that $x_i$ is included in the list $X$ that is aggregated by the Merkle tree with the root hash  $Root$.
\end{itemize}

\subsection{History Tree}\label{sec:background-historyT} 
A Merkle tree has been primarily used for proving membership.
However, Crosby and Wallach~\cite{crosby2009efficient} extended its application for an append-only tamper-evident log, denoted as a \textit{history tree}.
A history tree is the Merkle tree, in which leaf nodes are added in an append-only fashion, and which allows to produce logarithmic proofs witnessing that arbitrary two versions of the tree are consistent (i.e., one version of the tree is an extension of another).
Therefore, once added, a leaf node cannot be modified or removed. 

A history tree brings a versioned computation of hashes over the Merkle tree, enabling to prove that different versions (i.e., commitments) of a log, with distinct root hashes, make consistent claims about the past.
To provide a tamper-evident history system~\cite{crosby2009efficient}, the log represented by the history tree $L$ supports the following operations:
\begin{itemize}
	\item $\mathbf{L.add(x) \rightarrow C_j}$: appending of the record $x$ to $L$, returning a new commitment $C_j$ that represents the most recent value of the root hash of the history tree.
	
	\item $\mathbf{L.IncProof(C_i, C_j) \rightarrow \pi^{inc}}$: an incremental proof generation between two commitments $C_i$ and $C_j$, where $i \leq j$.
	
	\item $\mathbf{L.MemProof(i, C_j) \rightarrow \pi^{mem} }$: a membership proof generation for the record $x_i$ from the commitment $C_j$, where $i \leq j$. 
	
	\item $\mathbf{\pi^{inc}.Verify(C_i, C_j) \rightarrow \{True, False\}}$:  verification of the incremental proof $\pi^{inc}$, witnessing that the commitment $C_j$ contains the same history of records $x_k, k \in \{0,\ldots,i\}$ as the commitment $C_i$, where $i \leq j$.  
	
	\item $\mathbf{\pi^{mem}.Verify(i, x_i, C_j) \rightarrow \{True, False\}}$:  verification of the membership proof $\pi^{mem}$, witnessing that $x_i$ is the $i$th record in the $j$th version of $L$, fixed by the commitment $C_j$,  $i \leq j$. 
	
	\item $\mathbf{\pi^{inc}.DeriveNewRoot() \rightarrow C_j}$: a reconstruction of the commitment $C_j$ from the incremental proof $\pi^{inc}$ that was generated by $L.IncProof(C_i, C_j)$.
	
	\item $\mathbf{\pi^{inc}.DeriveOldRoot() \rightarrow C_i}$: a reconstruction of the commitment $C_i$ from the incremental proof $\pi^{inc}$ that was generated by $L.IncProof(C_i, C_j)$.
	
\end{itemize}

\subsection{Radix and Merkle-Patricia Tries}
Radix trie serves as a key-value storage.
In the Radix trie, every node at the $l$-th layer of the trie has the form of $\langle (p_0, p_1, \ldots, p_n), v\rangle$, where $v$ is a stored value and all $p_i, ~i\in \{0,1, \ldots, n\}$ represent the pointers on the nodes in the next (lower) layer $l+1$ of the trie, which is selected by following the $(l+1)$-th item of the key.
Note that key consists of an arbitrary number of items that belong to an alphabet with $n$ symbols (e.g., hex symbols). 
Hence, each node of the Radix trie has $n$ children and to access a leaf node (i.e., data $v$), one must descend the trie starting from the root node while following the items of the key one-by-one.
Note that Radix trie requires underlying database of key-value storage that maps pointers to nodes.
However, Radix trie does not contain integrity protection, and when its key is too long (e.g., hash value), the Radix trie will be sparse, thus imposing a high overhead for storage of all the nodes on the path from the root to values.

Merkle Patricia Trie (MPT)~\cite{wood2014ethereum,Merkle-Patricia-Trie-eth} is a combination of the Merkle tree (see \autoref{sec:MT-background}) and Radix trie data structures, and similar the Radix Trie, it serves as a key-value data storage.
However, in contrast to Radix trie, the pointers are replaced by a cryptographically secure hash of the data in nodes, providing integrity protection.
In detail, MPT guarantees integrity by using a cryptographically secure hash of the value for the MPT key as well as for the realization of keys in the underlying database that maps the hashes of nodes to their content; therefore, the hash of the root node of the MPT represents an integrity snapshot of the whole MPT trie.
Next, Merkle-Patricia trie introduces the \textit{extension nodes}, due to which, there is no need to keep a dedicated node for each item of the path in the key. 
The MPT trie $T$ supports the following operations:
\begin{itemize}
	\item[$\mathbf{T.root \rightarrow Root}$:] accessing the hash of the root node of MPT, which is stored as a key in the underlying database.
	
	\item $\mathbf{T.add(k, x) \rightarrow Root}$: adding the value $x$ with the key $k$ to $T$ while obtaining the new hash value of the root node. 
	
	\item $\mathbf{T.get(k) \rightarrow \{x, \perp\}}$: fetching a value $x$ that corresponds to key $k$; return $\perp$ if no such value exists.  
	
	\item $\mathbf{T.delete(k) \rightarrow \{True, False\}}$: deleting the entry with key equal to $k$, returning $True$ upon success, $False$ otherwise.
	
	\item $\mathbf{T.MptProof(k) \rightarrow  \{\pi^{mpt}, \pi^{\overline{mpt}}\}}$: a MPT (inclusion / exclusion) proof generation for the entry with key $k$.

	\item $\mathbf{\pi^{mpt}.Verify(k, Root) \rightarrow \{True, False\}}$:  verification of the MPT proof $\pi^{mpt}$, witnessing that entry with the key $k$ is in the MPT whose hash of the root node is equal to $Root$.			
	
	\item $\mathbf{\pi^{\overline{mpt}}.VerifyNeg(k, Root) \rightarrow \{True, False\}}$:  verification of the negative MPT proof, witnessing that entry with the key $k$ is not in the MPT with the root hash equal to $Root$.			

\end{itemize}

\section{Verifiable Delay Function}\label{sec:background:vdf}
The functionality of Verifiable Delay Function (VDF) \cite{Boneh2018} is similar to a time lock,\footnote{Time locks are computational problems that can only be solved by running a continuous computation for a given amount of time.} but in addition to it, by providing a short proof, a verifier may easily check if the prover knows the output of the VDF.
The function is effectively serialized, and parallel processing does not help to speed up VDF computation. 
A moderate amount of sequential computation is required to compute VDF.
Given a time delay $t$, a VDF  must satisfy the following conditions: 
for any input $x$, anyone equipped with commercial hardware can find $y$ = VDF($x, t$) in $t$ sequential steps, but an adversary with $p$ parallel processing units must not distinguish $y$ from a random number  in significantly fewer steps.
Further, given output $y$ of VDF, the prover can supply a proof $\pi$ to a verifier, who may check the output $y ~=~ \text{VDF}(x, t)$ using $\pi$ in logarithmic time w.r.t. time delay $t$ (i.e., $VDF\_Verify (y,\pi) \stackrel{?}{=} True$).

Finally, the safety factor $A_{max}$ is defined as the time ratio that the adversary is estimated to run VDF computation faster on proprietary hardware as opposed to a benign VDF computation using commercial hardware (see Drake~\cite{Drake2018}). 
CPU over-clocking records~\cite{SAMUEL2020} indicate that $A_{max}=10$ is a reasonable estimate.

\section{Atomic Swap}\label{sec:atomicswap}
A basic atomic swap assumes two parties $\mathbb{A}$ and $\mathbb{B}$ owning crypto-tokens in two different blockchains.
$\mathbb{A}$ and $\mathbb{B}$ wish to execute cross-chain exchange atomically and thus achieve a \textit{fairness} property, i.e., either both of the parties receive the agreed amount of crypto-tokens or neither of them.
First, this process involves an agreement on the amount and exchange rate, and second, the execution of the exchange itself.

In a centralized scenario~\cite{micali2003simple}, the approach is to utilize a trusted third party for the execution of the exchange.
In contrast to the centralized scenario, blockchains allow us to execute such an exchange without a requirement of the trusted party.
%
The atomic swap protocol~\cite{atomic-swap} enables conditional redemption of the funds in the first blockchain to $\mathbb{B}$ upon revealing of the hash pre-image (i.e., secret) that redeems the funds on the second blockchain to $\mathbb{A}$.
The atomic swap protocol is based on two Hashed Time-Lock Contracts (HTLC) that are deployed by both parties in both blockchains.

Although HTLCs can be implemented by Turing-incomplete smart contracts with support for hash-locks and time-locks, for clarity, we provide a description assuming Turing-complete smart contracts, requiring four transactions:
\begin{enumerate}
	\item $\mathbb{A}$ chooses a random string $x$ (i.e., a secret) and computes its hash $h(x)$.
	Using $h(x)$, $\mathbb{A}$ deploys $HTLC_\mathbb{A}$ on the first blockchain and sends the agreed amount to it, which later enables anybody to do a conditional transfer of that amount to $\mathbb{B}$ upon calling a particular method of $HTLC_\mathbb{A}$ with $x = h(x)$ as an argument (i.e., hash-lock). 
	Moreover, $\mathbb{A}$ defines a time-lock, which, when expired, allows $\mathbb{A}$ to recover funds into her address by calling a dedicated method: this is to prevent aborting of the protocol by another party.
	
	\item When $\mathbb{B}$ notices that $HTLC_\mathbb{A}$ has been already deployed, she deploys $HTLC_\mathbb{B}$ on the second blockchain and sends the agreed amount there, enabling a conditional transfer of that amount to $\mathbb{A}$ upon revealing the correct pre-image of $h(x)$ ($h(x)$ is visible from already deployed $HTLC_\mathbb{A}$).
	$\mathbb{B}$ also defines a time-lock in $HTLC_\mathbb{B}$ to handle abortion by $\mathbb{A}$.
	
	\item Once $\mathbb{A}$ notices deployed $HTLC_\mathbb{B}$, she calls a method of $HTLC_\mathbb{B}$ with revealed $x$, and in turn, she obtains the funds on the second blockchain.
	
	\item Once $\mathbb{B}$ notices that $x$ was revealed by $\mathbb{A}$ on the second blockchain, she calls a method of $HTLC_\mathbb{A}$ with $x$ as an argument, and in turn, she obtains the funds on the first blockchain.
\end{enumerate}
If any of the parties aborts, the counter-party waits until the time-lock expires and redeems the funds.

\section{Trusted Execution Environment}\label{sec:background:tee}
Trusted Execution Environment (TEE) is a hardware-based component that can securely execute arbitrary code in an isolated environment. 
TEE uses cryptography primitives and hardware-embedded secrets that protect data confidentiality and the integrity of computations.
In particular, the adversary model of TEE usually includes privileged applications and an operating system, which may compromise unprivileged user-space applications.
There are several practical instances of TEE, such as Intel Software Guard Extensions (SGX)~\cite{anati2013innovative,mckeen2013innovative,hoekstra2013using} available at Intel's CPUs or based on RISC-V architecture such as Keystone-enclave~\cite{Keystone-enclave} and Sanctum~\cite{costan2016sanctum}.
In the context of this work (i.e., \autoref{chapter:logging}), we built on top of Intel SGX, therefore we adopt the terminology introduced by it.

\paragraph{Intel SGX.}
Intel SGX is a set of instructions that ensures hardware-level isolation of protected user-space codes called \textit{enclaves}. 
An enclave process cannot execute system calls but can read and write memory outside the enclave. 
Thus isolated execution in SGX may be viewed as an ideal model in which a process is guaranteed to be executed correctly with ideal confidentiality, while it might run on a potentially malicious operating system.

Intel SGX allows a local process or a remote system to securely communicate with the enclave as well as execute verification of the integrity of the enclave's code. 
When an enclave is created, the CPU outputs a report of its initial state, also referred to as a \textit{measurement}, which is signed by the private key of TEE and encrypted by a public key of Intel Attestation Service (IAS). 
The hardware-protected signature serves as the proof that the measured code is running in an SGX-protected enclave, while the encryption by IAS public key ensures that the SGX-equipped CPU is genuine and was manufactured by Intel.
This proof is also known as a \textit{quote} or \textit{attestation}, and it can be verified by a local process or by a remote system. 
The enclave process-provided public key can be used by a verifier to establish a secure channel with the enclave or to verify the signature during the attestation.
We assume that a trustworthy measurement of the enclave's code is available for any client that wishes to verify an attestation.

\section{Central Bank Digital Currency (CBDC)}
CBDC is often defined as a digital liability backed and issued by a central bank that is widely available to the general public. 
CBDC encompasses many potential benefits such as efficiency and resiliency, flexible monetary policies, and enables enhanced control of tax evasion and money laundering~\cite{kiff2020survey}. 
However, regulations, privacy and identity management issues, as well as design vulnerabilities are potential risks that are shared with cryptocurrencies. 
Many blockchain-based CBDC projects rely on using some sort of stable coins adapting permissioned blockchains due to their scalability and the capability to establish specific privacy policies, as compared to public blockchains~\cite{sethaput2021blockchain,zhang2021blockchain}. 
Therefore, the level of decentralization and coin volatility are two main differences between blockchain-based CBDCs and common cryptocurrencies.
These CBDCs are often based on permissioned blockchain projects such as Corda~\cite{brown2016corda}, variants of Hyperledger~\cite{hyperledger-github}, and Quorum~\cite{espel2017proposal}.

CDBC solutions are often designed as multi-layer projects~\cite{2022cbdctypes}. 
Wholesale CBDC targets communication of financial institutions and inter-bank settlements. 
Retail CBDC includes accessibility to the general public or their customers. 

%% file: sec/sra-short.tex

In this chapter, we present our contribution to standardization for threat modeling and it is based on the papers \cite{homoliak2019security} \cite{homoliak2020security} (see also \autoref{sec:sra-contributing-papers}).
In particular, we introduce the security reference architecture (SRA) for blockchains, which adopts a stacked model (similar to the ISO/OSI) describing the nature and hierarchy of various security and privacy aspects.
The SRA contains four layers: (1) the network layer, (2) the consensus layer, (3) the replicated state machine layer, and (4) the application layer.
At each of these layers, we identify known security threats, their origin, and countermeasures, while we also analyze several cross-layer dependencies. 
%
%
Next, to enable better reasoning about the security aspects of blockchains by the practitioners, we propose a blockchain-specific version of the threat-risk assessment standard ISO/IEC 15408 by embedding the stacked model into this standard. 
Finally, we provide designers of blockchain platforms and applications with a design methodology following the model of SRA and its hierarchy.

\section{Methodology and Scope}

We aim to consolidate the literature, categorize found vulnerabilities and threats according to their origin, and as a result, we create four main categories (also referred to as layers).
At the level of particular main categories, we apply sub-categorization that is based on the existing knowledge and operation principles specific to such subcategories, especially concerning the security implications.
If some subcategories impose equivalent security implications, we merge them into a single subcategory.
See the road-map of all the categories in \autoref{fig:overview}.
Our next aim is to indicate and explain the co-occurrences or relations of multiple threats, either at the same main category or across more categories.

\section{Security Reference Architecture}
We present two models of the security reference architecture, which facilitate systematic studying of vulnerabilities and threats related to the blockchains and applications running on top of them.
First, we introduce the stacked model, which we then project into the threat-risk assessment model.

\begin{figure}[t]
	\centering
	\includegraphics[width=0.45\textwidth]{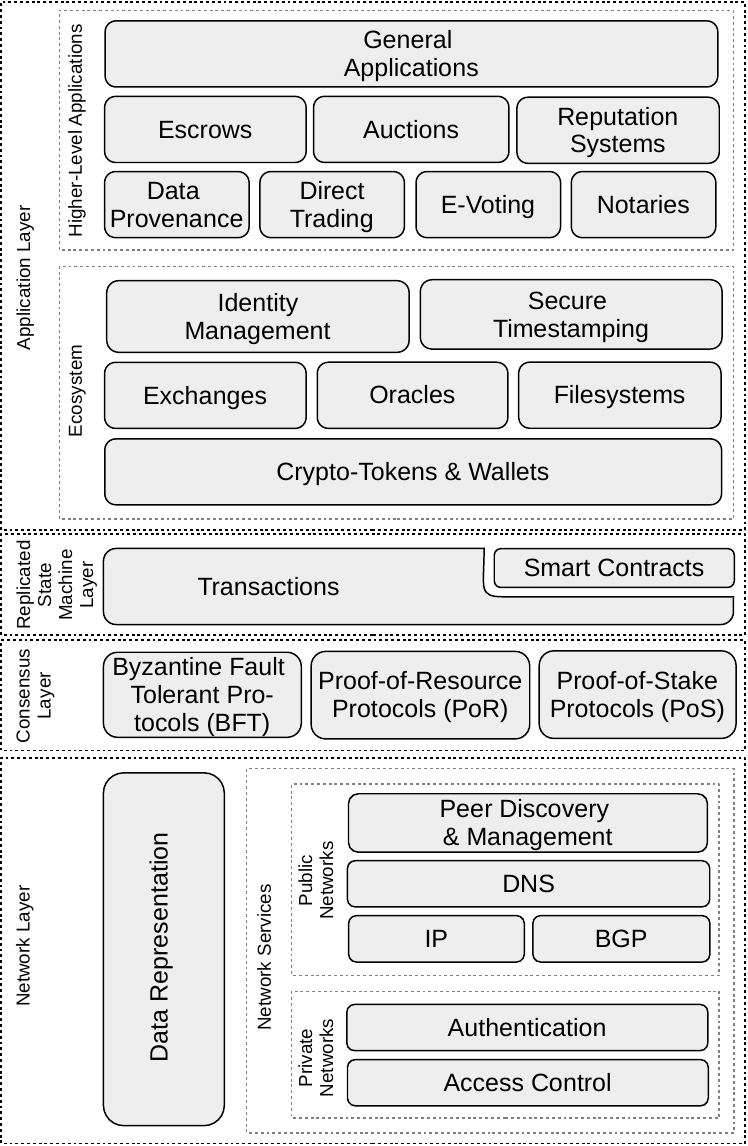}         
	\caption{Stacked model of the security reference architecture.}
	\label{fig:overview}
\end{figure}

\subsection{Stacked Model}\label{sec:stacked-model}
\label{sec:LayeredModel}
To classify the security aspects of blockchains, we
utilize a stacked model consisting of four layers (see \autoref{fig:overview}).  
A similar stacked model was already proposed in the literature~\cite{wang2018survey}, but in contrast to it, we preserve only such a granularity level that enables us to isolate security threats and their nature, which is the key focus of our work.
In the following, we briefly describe each layer.

\begin{compactenum}
	\item[\textbf{(1) The  network layer}] consists of the data
	representation and network services planes. 
	The data representation plane deals with the storage, encoding, and protection of data, while the
	network service plane contains the discovery and communication with protocol peers, addressing, routing, and naming services. 	
	
	\item[\textbf{(2) The consensus layer}] deals with the ordering of transactions, and we divide it into three main categories according to the protocol type: Byzantine Fault Tolerant, 
	Proof-of-Resource,
	and Proof-of-Stake protocols.
	
	
	\item[\textbf{(3) The  replicated state machine (RSM) layer}] deals with the interpretation of transactions, according to which the state of the blockchain is updated.
	In this layer, transactions are categorized into two parts, where the first part deals with the privacy of data in transactions as well as the privacy of the users who created them, and the second part -- smart contracts -- deals with the security and safety aspects of decentralized code execution in this environment.
	%
	
	\item[\textbf{(4) The application layer}] contains the most common end-user functionalities and services.
	We divide this layer into two groups.
	The first group represents the applications that provide common functionalities for most of the higher-level blockchain applications, and it contains the following categories: wallets, exchanges, oracles, filesystems, identity management, and secure timestamping.
	We refer to this group as applications of the blockchain ecosystem.
	The next group of application types resides at a higher level and focuses on providing certain end-user functionality. 
	This group contains categories such as e-voting, notaries, identity management, auctions, escrows, etc.

\end{compactenum}

\subsection{Threat-Risk Assessment Model}
To better capture the security-related aspects of blockchain systems, we introduce a threat-risk model (see \autoref{fig:iso15408}) that is
based on the template of ISO/IEC 15408~\cite{cc2017} and projection of our stacked model (see \autoref{fig:overview}). 
This model includes the following components and actors: 
\begin{compactitem}
	\item[\textbf{Owners}] are blockchain users who run any type of node and they exist at the application layer and the consensus layer.
	Owners possess crypto-tokens, and they might use or provide blockchain-based applications and services.
	Additionally, owners involve consensus nodes that earn crypto-tokens from running the consensus protocol.
	\item[\textbf{Assets}] are present at the application layer, and they consist of monetary value (i.e., crypto-tokens or other tokens) as well as the availability of application-layer services and functionalities built on top of blockchains (e.g., notaries, escrows, data provenance, auctions).
	The authenticity of users, the privacy of users, and the privacy of data might also be considered as application-specific assets.
	Furthermore, we include here the reputation of service providers using the blockchain services.
	%
	\item[\textbf{Threat agents}] are spread across all the layers of the stacked model, and they mostly involve malicious users whose intention is to steal assets, break functionalities, or disrupt services.
	However, threat agents might also be inadvertent entities, such as developers of smart contracts who unintentionally create bugs and designers of blockchain applications who make mistakes in the design or  ignore some issues.
	\item[\textbf{Threats}] facilitate various attacks on assets, and they exist at all layers of the stacked model.
	Threats arise from vulnerabilities in the network, smart
	contracts, applications, from consensus protocol deviations, violations of consensus protocol assumptions.
	\item[\textbf{Countermeasures}] protect owners from threats by minimizing the risk of compromising/losing the assets.   
	Alike the threats and threat agents, countermeasures can be applied at each of the layers of our stacked model, and they involve various security/privacy/safety solutions, incentive schemes, reputation techniques, best practices, etc. 
	Nevertheless, we emphasize that their utilization usually imposes some limitations such as higher complexity and additional performance overheads (e.g., resulting in decreased throughput). 

	\item[\textbf{Risks}] are related to the application layer, and they are caused by threats and their agents. 
	Risks may lead to a loss of monetary assets, a loss of privacy, a loss of reputation, service malfunctions, and disruptions of services and applications (i.e., availability issues).
	
\end{compactitem}

\begin{figure}[t]
	\begin{center}
		\includegraphics[width=0.53\textwidth]{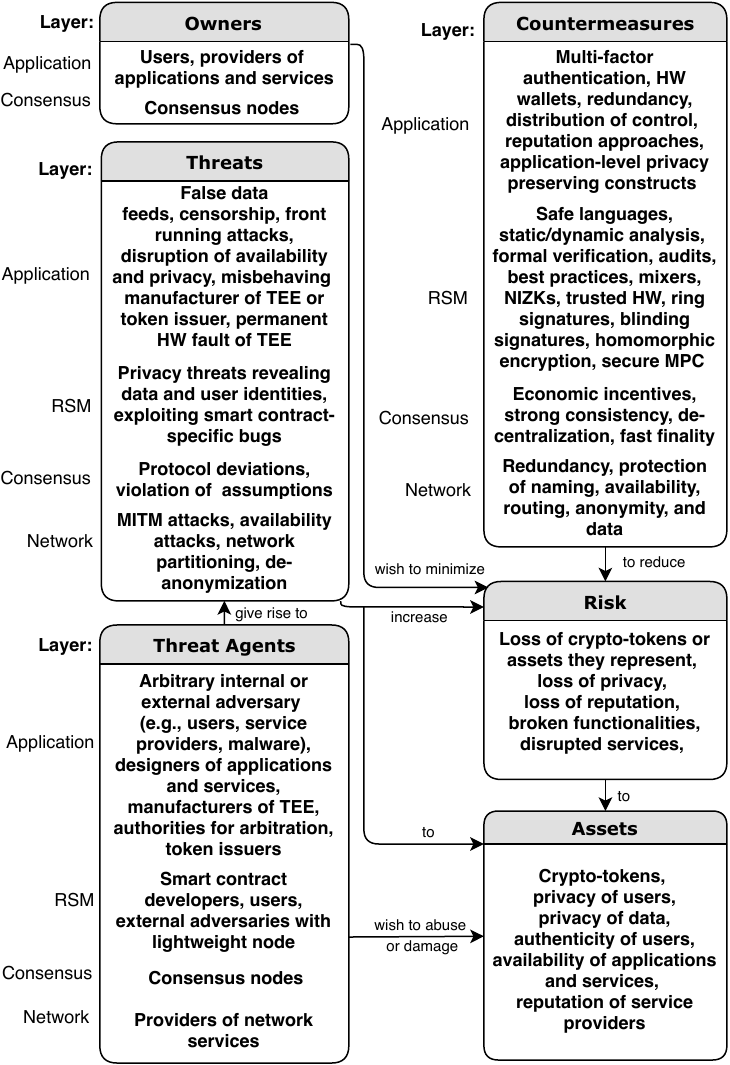}
		\caption{Threat-risk assessment model of the security reference architecture.}
		\label{fig:iso15408}
	\end{center}        
\end{figure}

\smallskip\noindent
The owners wish to minimize the risk caused by threats that arise
from threat agents.
Within our stacked model, different threat agents appear at each layer.  
At the \textbf{network layer}, there are service providers including parties
managing IP addresses and DNS names.
The threats at this layer arise from man-in-the-middle (MITM) attacks, 
network partitioning, de-anonymization, and availability attacks. 
Countermeasures contain protection of availability, naming, routing,
anonymity, and data.
At the \textbf{consensus layer}, consensus nodes may be malicious and wish to alter the outcome of the consensus protocol by deviating from it.
Moreover, if they are powerful enough, malicious nodes might violate assumptions of consensus protocols to take over the execution of the protocol or cause its disruption.
The countermeasures include well-designed economic incentives, strong consistency, decentralization, and fast finality solutions.
At the \textbf{RSM layer}, the threat agents may stand for developers who (un)intentionally introduce semantic bugs in smart contracts (intentional bugs represent backdoors) as well as users and external adversaries running lightweight nodes who pose threats due to the exploitation of such bugs. 
Countermeasures include safe languages, static/dynamic analysis, formal verification, audits, best practices, and design patterns.
Other threats of the RSM layer are related to compromising the privacy of data and user identities with mitigation techniques involving mixers, privacy-preserving cryptography constructs (e.g., non-interactive zero-knowledge proofs (NIZKs), ring signatures, blinding signatures, homomorphic encryption) as well as usage of trusted hardware (respecting its assumptions and attacker models declared).
%
At the \textbf{application layer}, threat agents are broad and involve arbitrary internal or external adversaries such as users, service providers, malware, designers of applications and services, manufactures of trusted execution environments (TEE) for concerned applications (e.g., oracles, auctions), authorities in the case of applications that require them for arbitration (e.g., escrows, auctions) or filtering of users (e.g., e-voting, auctions), token issuers.
The threats on this layer might arise from false data feeds, censorship by application-specific authorities (e.g., auctions, e-voting), front running attacks, disruption of the availability of centralized components, compromising application-level privacy, misbehaving of the token issuer, misbehaving of manufacturer of TEE or permanent hardware (HW) faults in TEE.
Examples of mitigation techniques are multi-factor authentication, HW wallets with displays for signing transactions, redundancy/distributions of some centralized components,  reputation systems, and privacy preserving-constructs as part of the applications themselves. 
We elaborate closer on vulnerabilities, threats, and countermeasures (or mitigation techniques) related to each layer of the stacked model in the following sections. 

\paragraph{Involved Parties \& Blockchain's Life-Cycle.}
In \autoref{chapter:background}, we presented several types of involved parties in the blockchain infrastructure (see \autoref{fig:node-types}). 
We emphasize that these parties are involved in the operational stage of the blockchain's life-cycle.
However, in the design and development stages of the blockchain's life-cycle, programmers and designers should also be considered as potential threat agents who influence the security aspects of the whole blockchain infrastructure (regardless of whether their intention is malicious or not).
This is of great concern especially for applications built on top of blockchains (i.e., at the application layer) since these applications are usually not thoroughly reviewed by the community or public, as it is typical for other (lower) layers.

\begin{figure}[b]
	\begin{center}		
		\includegraphics[width=0.45\columnwidth]{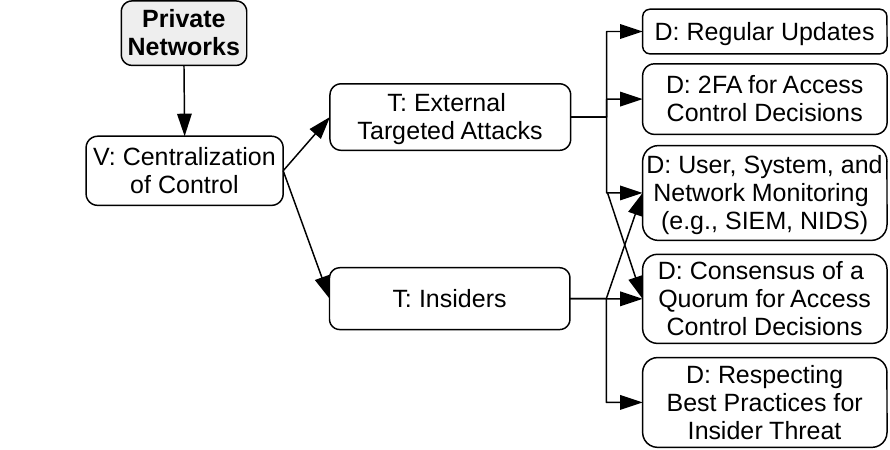} 		
		\caption{Vulnerabilities, threats, and defenses in private networks (network layer).}
		\label{fig:attacks-private-networks}
	\end{center}	
\end{figure}

\section{Network Layer}\label{sec:network}
Blockchains usually introduce peer-to-peer overlay networks built on top of other networks. 
Hence, blockchains inherit security and privacy issues from their underlying networks. 
%
%
In our model (see \autoref{fig:overview}), we divide the network layer into  \textit{data representation} and \textit{network services} sub-planes. 
The data representation plane is protected by cryptographic primitives that ensure data integrity, user authentication, and optionally confidentiality, privacy, anonymity, non-repudiation, and accountability. 
The main services provided by the network layer are peer management and discovery, which rely on the internals of the underlying network, such as domain name resolution (i.e., DNS) or network routing protocols. 
Based on permission to join the blockchain system, the networks are either private or public. 
We model security threats and mitigation techniques for both private and public networks as vulnerability/threat/defense (VTD) graphs in \autoref{fig:attacks-private-networks} and \autoref{fig:attacks-public-networks}, and we refer the interested reader to our paper \cite{homoliak2020security} for more details.

\begin{figure}[!h]
	\begin{center}		
		\includegraphics[width=0.6\columnwidth]{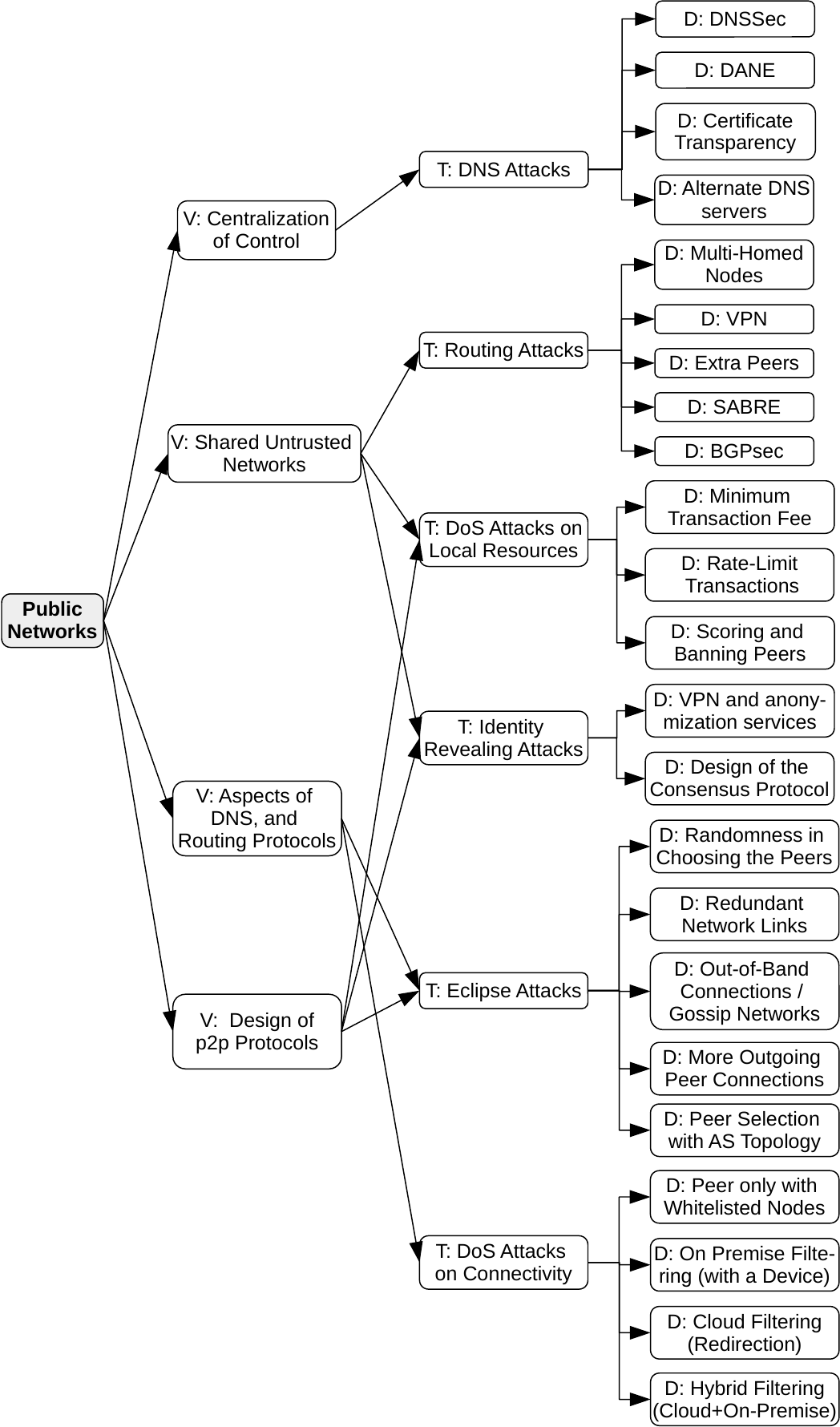} 		
		\caption{Vulnerabilities, threats, and defenses in public networks (network layer).}
		\label{fig:attacks-public-networks}
	\end{center}	
\end{figure}

\clearpage

\section{Consensus Layer}
\label{sec:consensus}

The consensus layer of the stacked model deals with the ordering of transactions, while the interpretation of them is left for the RSM layer (see \autoref{sec:smart_contracts}).
The consensus layer includes three main categories of consensus protocols concerning different principles of operation and thus their security aspects -- Proof-of-Resource Protocols (PoR), Byzantine Fault Tolerant (BFT) Protocols, and Proof-of-Stake Protocols (PoS).
Nevertheless, we can identify vulnerabilities and threats that are generic to all categories. 
Next, we outline modeling of security threats and mitigation techniques generic to all consensus protocols as VTD graphs in \autoref{fig:attacks-consensus-generic}, while particular categories of protocols are modeled in  \autoref{fig:attacks-consensus-PoR}, \autoref{fig:attacks-consensus-PoS}, and \autoref{fig:attacks-consensus-BFT}.
For details about these categories and their threats, we refer the interested reader to our paper \cite{homoliak2020security}. 

		\begin{figure}[t]
				\begin{center}		
						\includegraphics[width=0.6\columnwidth]{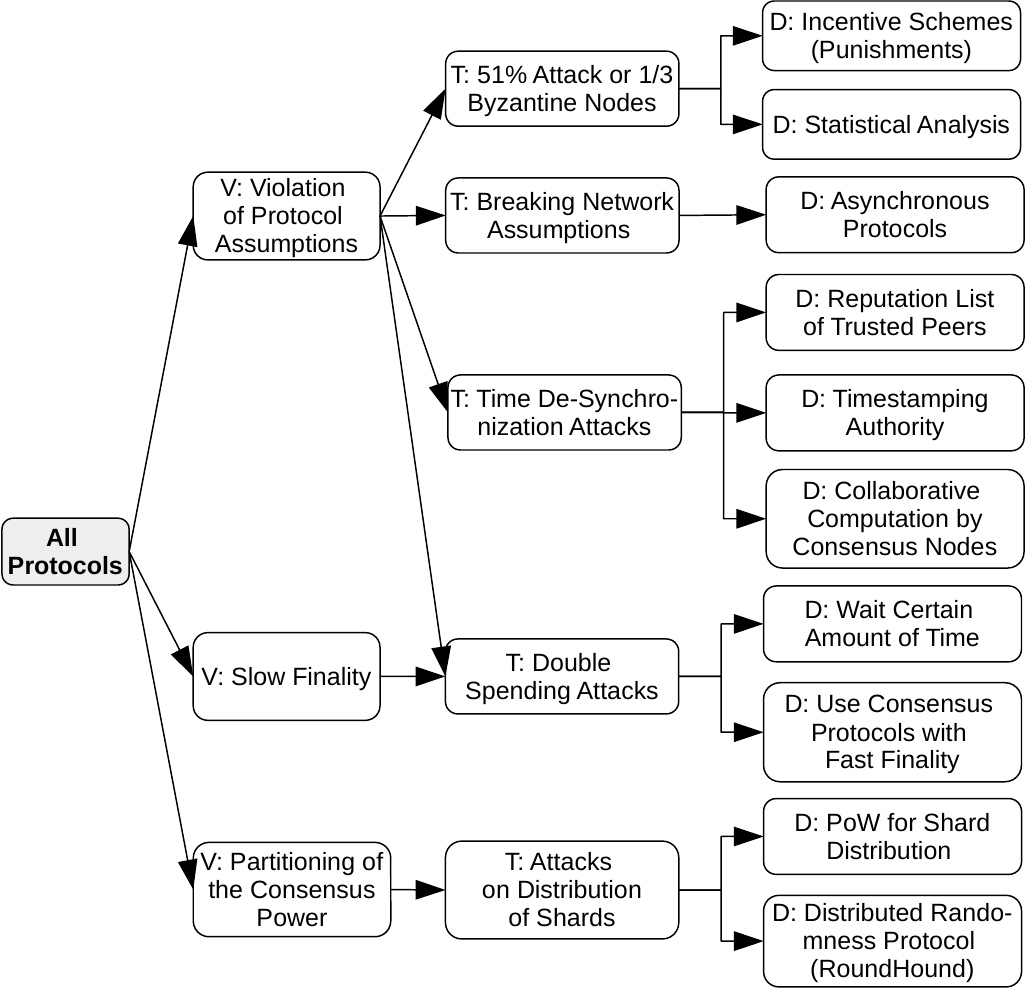} 		
						\caption{Generic threats and defenses of the consensus layer.}
						\label{fig:attacks-consensus-generic}
					\end{center}	
		\end{figure}
	
		\begin{figure}[!bh]
		\begin{center}		
			\includegraphics[width=0.7\columnwidth]{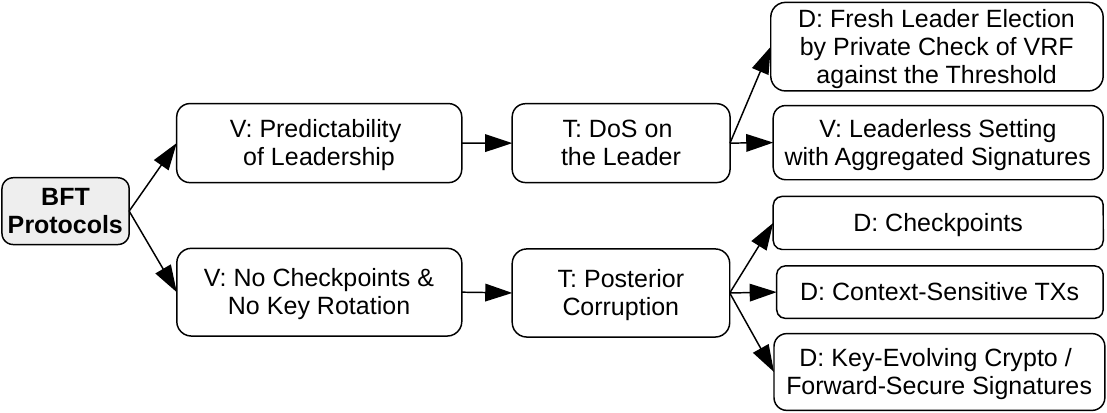} 		
			\caption{Vulnerabilities, threats, and defenses of BFT protocols (consensus layer).}
			\label{fig:attacks-consensus-BFT}
		\end{center}	
		\end{figure}
	
	\clearpage
	
		\begin{figure}[!h]
			\begin{center}		
					\includegraphics[width=0.6\textwidth]{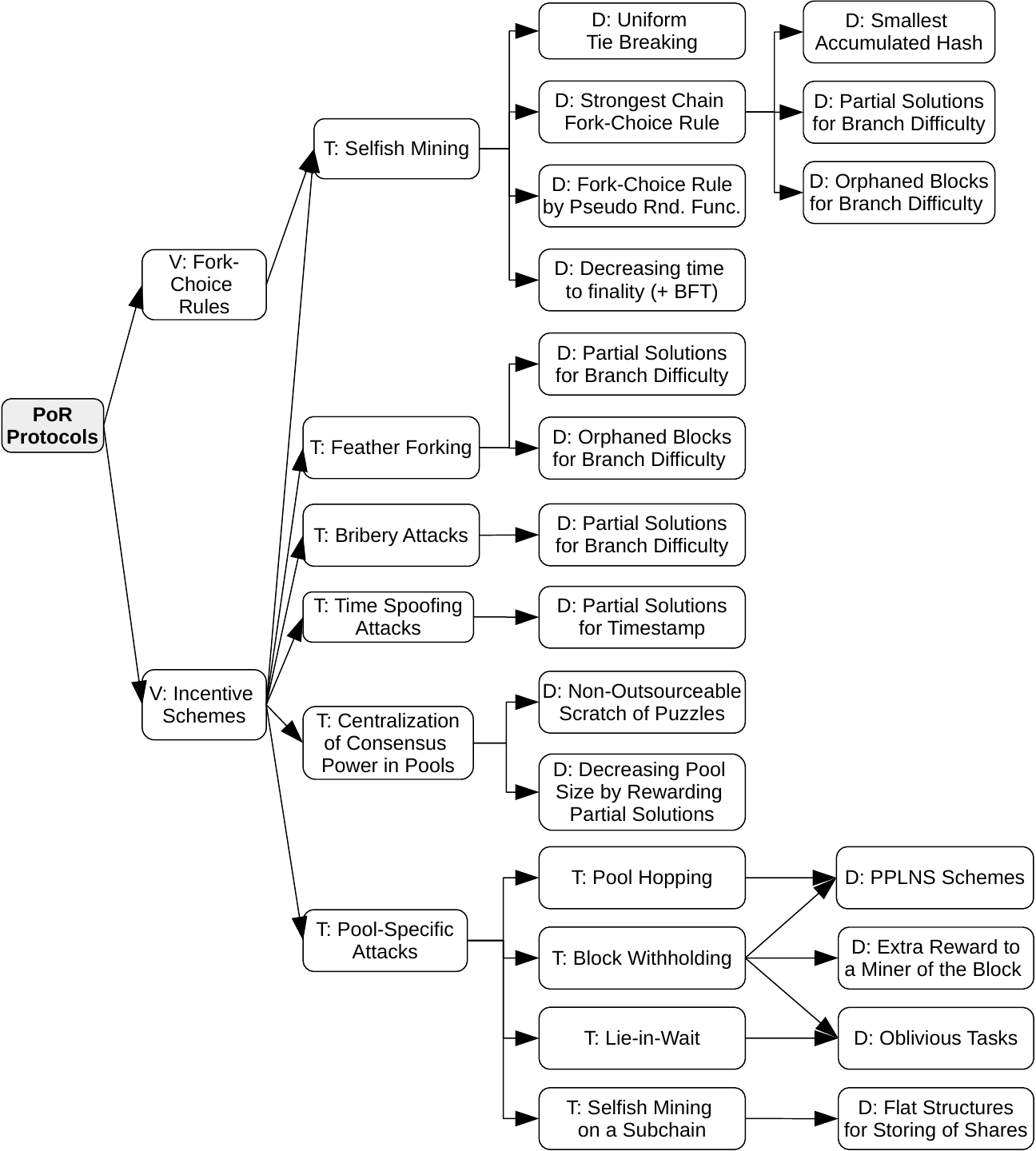} 		
					\caption{Vulnerabilities, threats, and defenses of PoR protocols (consensus layer).}
					\label{fig:attacks-consensus-PoR}
					\vspace{-0.3cm}
				\end{center}	
		\end{figure}
	
		\begin{figure}[!h]
			\begin{center}		
					\includegraphics[width=0.55\columnwidth]{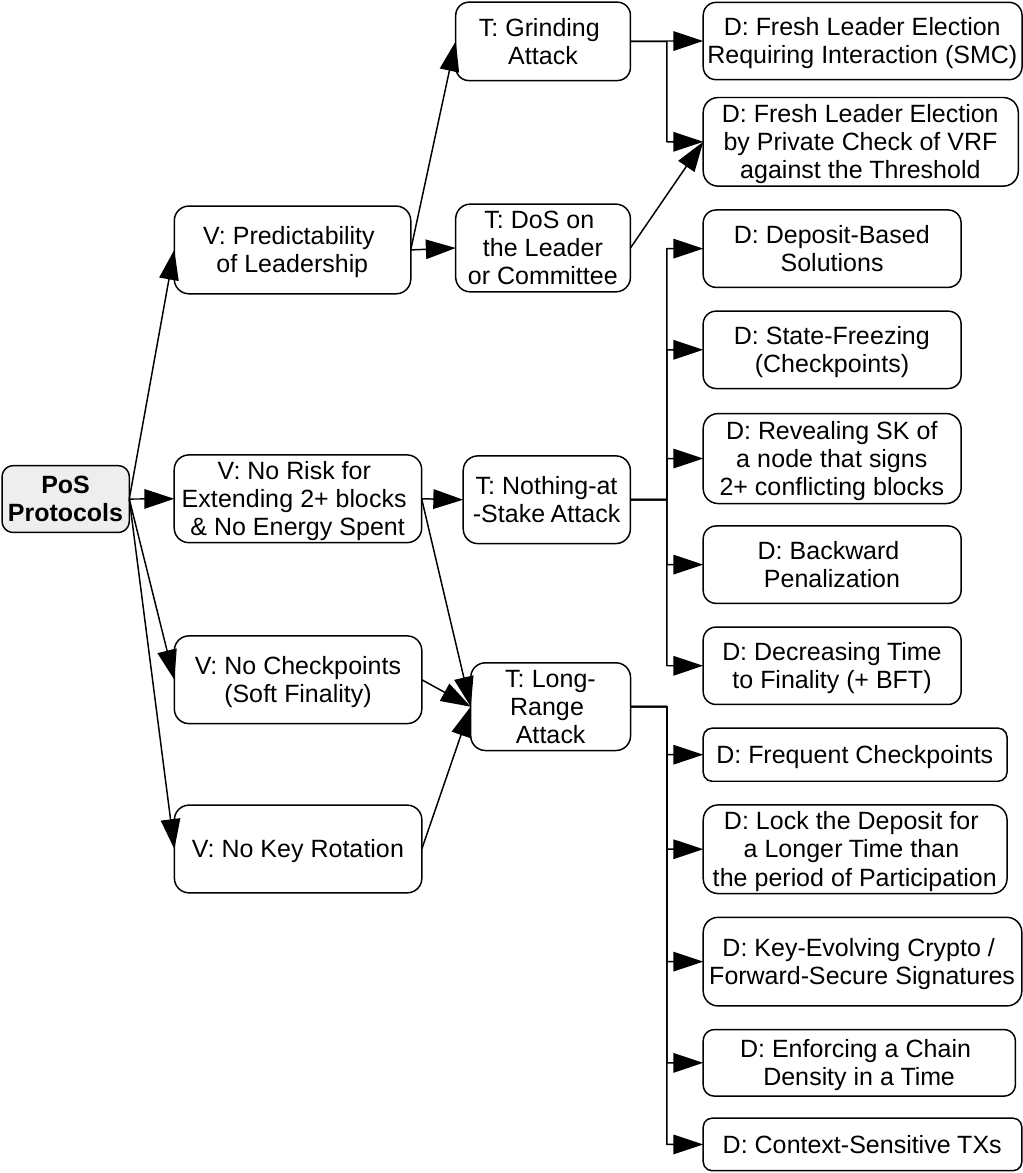} 		
					\caption{Vulnerabilities, threats, and defenses of PoS protocols (consensus layer).}
					\label{fig:attacks-consensus-PoS}
				\end{center}	
		\end{figure}

\clearpage

\section{Replicated State Machine Layer}
\label{sec:smart_contracts}
The Replicated State Machine (RSM) layer is responsible for the interpretation and execution of transactions that are already ordered by the consensus layer.
Concerning security threats for this layer are related to the privacy of users, privacy and confidentiality of data, and smart contract-specific bugs. 
%
We split the security threats of the RSM layer into two parts: standard transactions and smart contracts. 
%

\subsection{Transaction Protection}
Transactions containing plain-text data are digitally signed by private keys of users, enabling anybody to verify the validity of transactions with the corresponding public keys.     
However, such an approach provides only pseudonymous identities that can be traced to real IP addresses (and sometimes to identities) by a network-eavesdropping adversary, and moreover, it does not ensure the confidentiality of data~\cite{feng2019}.
Therefore, several blockchain-embedded mechanisms for the privacy of data and user identities were proposed in the literature, which we review in \cite{homoliak2020security}.
Note that some privacy-preserving techniques can be applied also on the application layer of our stacked model but imposing higher programming overheads and costs, 
which is common in the case of blockchain platforms that do not support them natively.
We outline modeling of security threats and mitigation techniques related to transactions and their privacy as VTD graphs in \autoref{fig:attacks-RSM-TXs}.
For details of particular vulnerabilities and threats, we refer the interested reader to our paper \cite{homoliak2020security}. 
\begin{figure}[t]
		\vspace{-0.2cm}
	\begin{center}		
			\includegraphics[width=0.6\columnwidth]{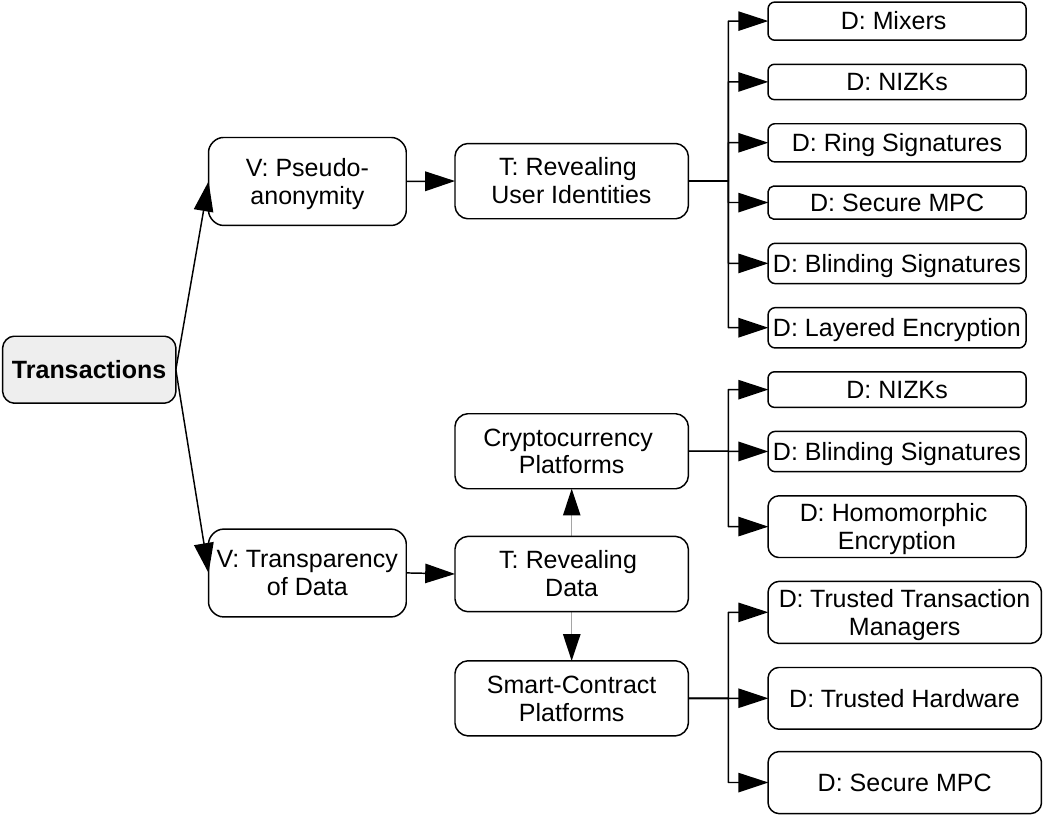} 		
			\caption{Vulnerabilities, threats, and defenses of privacy threats (RSM layer).}
			\label{fig:attacks-RSM-TXs}
		\end{center}	
\end{figure}

\begin{figure}[t]
	\begin{center}		
		\includegraphics[width=0.6\columnwidth]{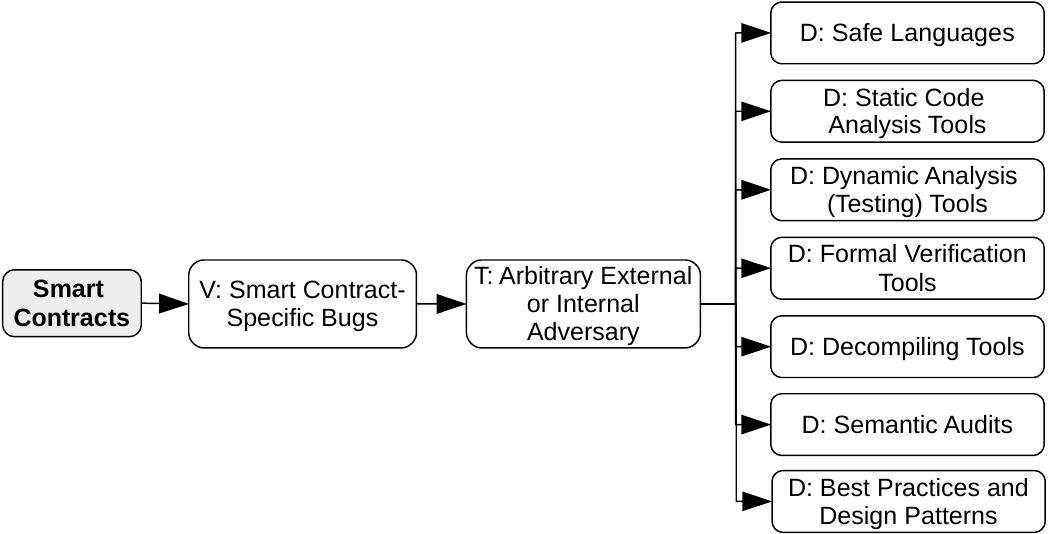} 		
		\caption{Vulnerabilities, threats, and defenses of smart contract platforms (RSM layer).}			
		\label{fig:attacks-RSM-sc}
		\vspace{-0.5cm}
	\end{center}	
\end{figure}

\subsection{Smart Contracts}\label{sec:smart-contracts}
Smart contracts introduced to automate legal contracts, 
now serve as a method for building decentralized applications on blockchains. 
They are usually written in a blockchain-specific programming language that may be Turing-complete (i.e., contain arbitrary programming logic) or only serve for limited purposes.
We outline modeling of security threats and mitigation techniques related to smart contracts as VTD graphs in \autoref{fig:attacks-RSM-TXs}.
For details of particular vulnerabilities and threats, we refer the interested reader to our paper \cite{homoliak2020security}.

\begin{figure}[!b]
	\centering
	\includegraphics[width=0.51\textwidth]{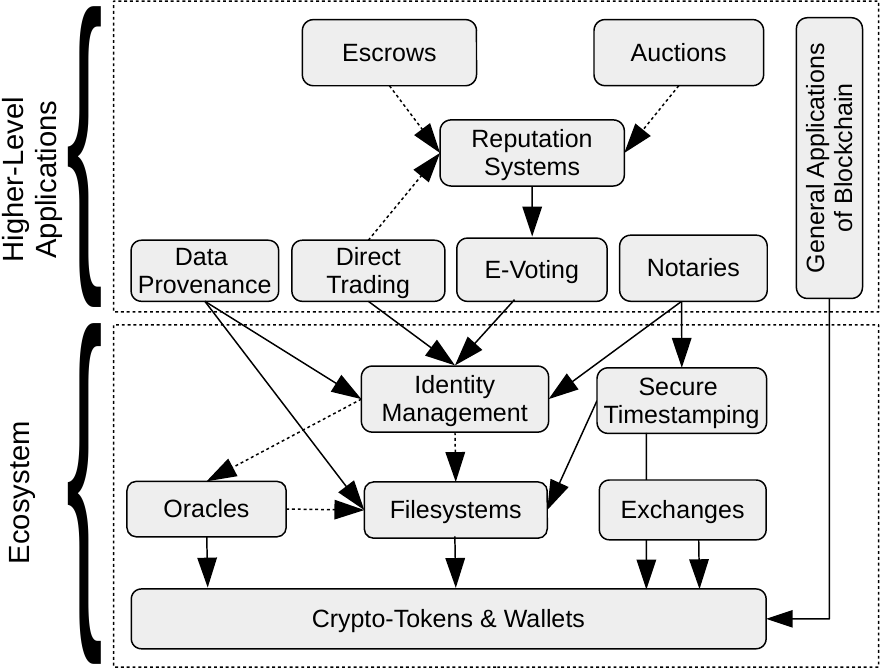}
	\caption{Hierarchy in inheritance of security aspects across categories of the application layer. Dotted arrows represent application-specific and optional dependencies.}\label{fig:dependencies}   
\end{figure} 

\section{Application Layer: Ecosystem Applications}
\label{sec:apps}

We present a functionality-oriented categorization of the applications running on or utilizing the blockchain in \autoref{fig:dependencies}, where we depict hierarchy in the inheritance of security aspects among particular categories.
In this categorization, we divide the applications into categories according to the main functionality/goal that is to be achieved by using the blockchain.
%
Security threats of this layer are mostly specific to particular types of applications.
Nevertheless, there are a few application-level categories that are often utilized by other higher-level applications. 
In the current section, we isolate such categories into a dedicated application-level group denoted as an \textit{ecosystem}, while we cover the rest of the applications in \autoref{sec:apps-applications}.
The group of ecosystem applications contains five categories, and we outline their security threats and mitigation techniques in VTD graphs as follows: 
(1) \textbf{crypto-tokens and wallets} (see \autoref{fig:attacks-APP-wallets}), (2) \textbf{exchanges} (see \autoref{fig:attacks-APP-exchanges}), (3) \textbf{oracles} (see \autoref{fig:attacks-APP-oracles}), (4) \textbf{filesystems} (see \autoref{fig:attacks-APP-DFs}), (5) \textbf{identity management} (see \autoref{fig:attacks-APP-identity}), and (6) \textbf{secure-timestamping} (see \autoref{fig:attacks-APP-timestamps}).
For details of these categories of applications and their security threats and mitigation techniques, we refer the interested reader to our paper \cite{homoliak2020security}.

\begin{figure*}[!t]
	\begin{center}		
		\vspace{-0.2cm}
		\includegraphics[width=0.8\textwidth]{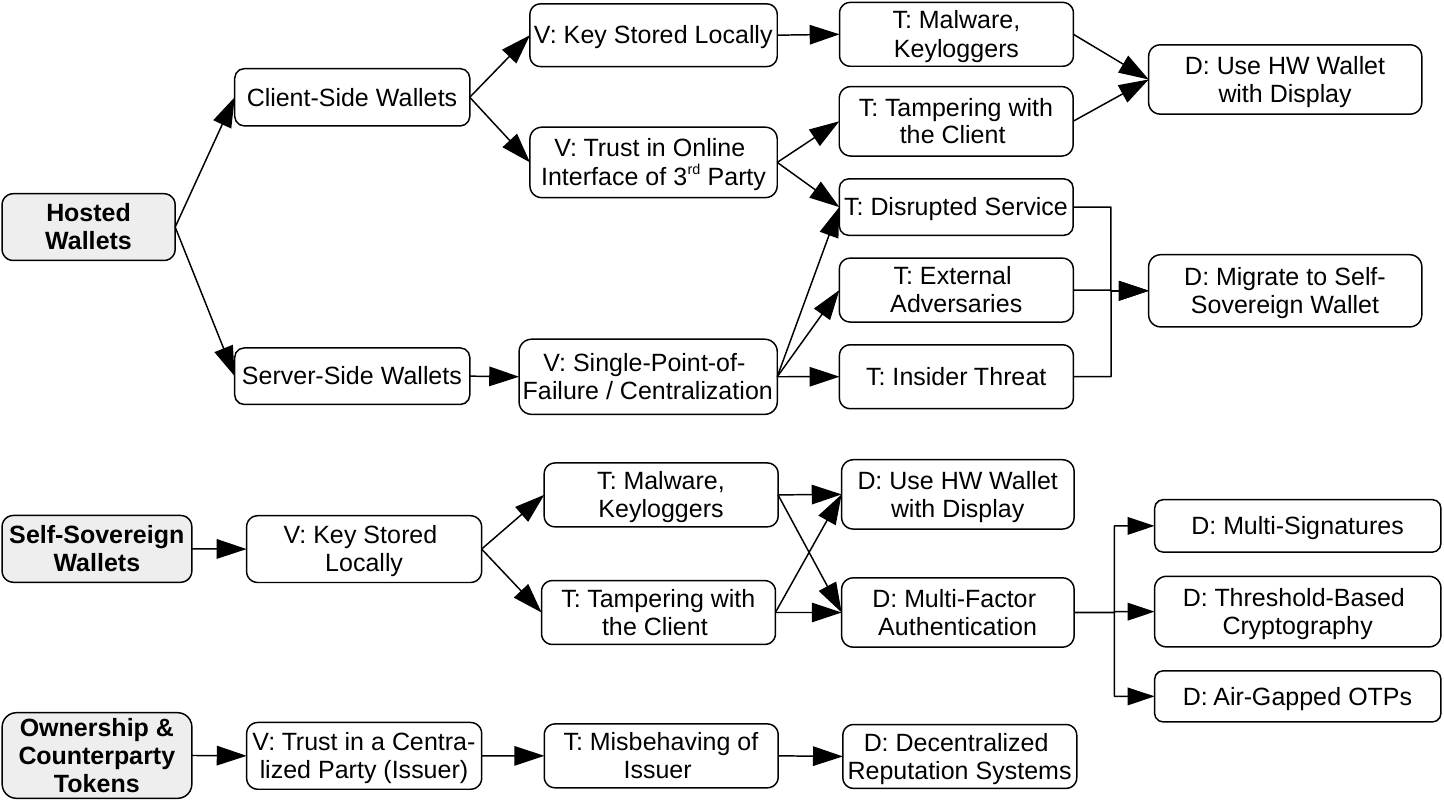} 		
		\vspace{0.1cm}
		\caption{Vulnerabilities, threats, and defenses of the crypto-token $\&$ wallets category.}
		\label{fig:attacks-APP-wallets}
		\vspace{-0.4cm}
	\end{center}	
\end{figure*}

\begin{figure*}[ht]
	\begin{center}
		\vspace{-0.1cm}			
		\includegraphics[width=0.80\textwidth]{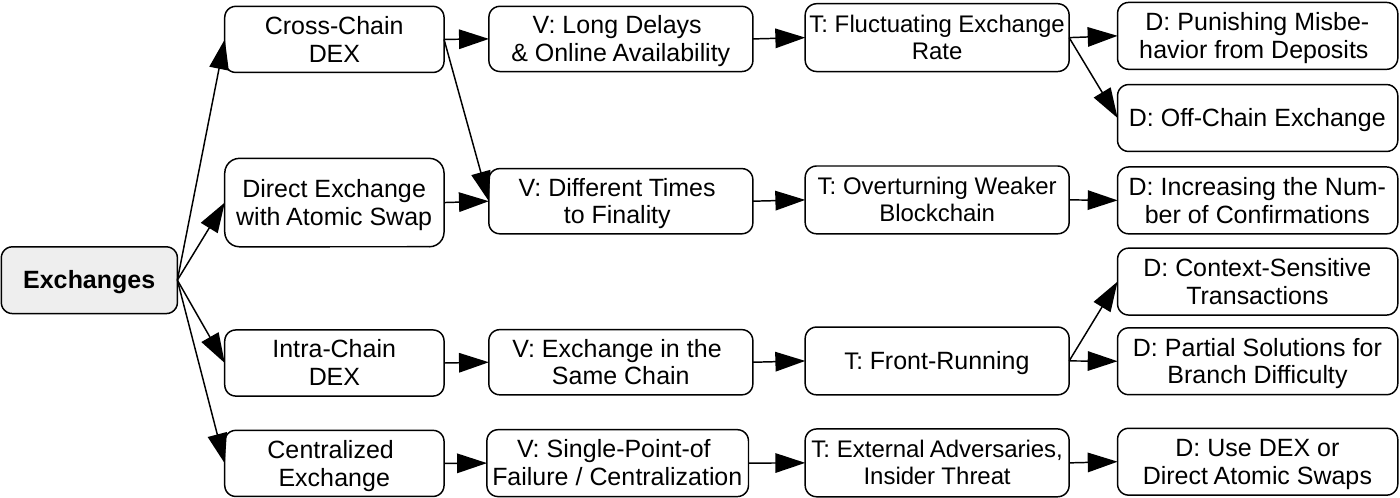} 		
		\vspace{0.1cm}
		\caption{Vulnerabilities, threats, and defenses of the exchanges category.}
		\label{fig:attacks-APP-exchanges}
	\end{center}	
	\vspace{-0.4cm}
\end{figure*}

\begin{figure*}[!h]
	\begin{center}		
		\includegraphics[width=0.82\textwidth]{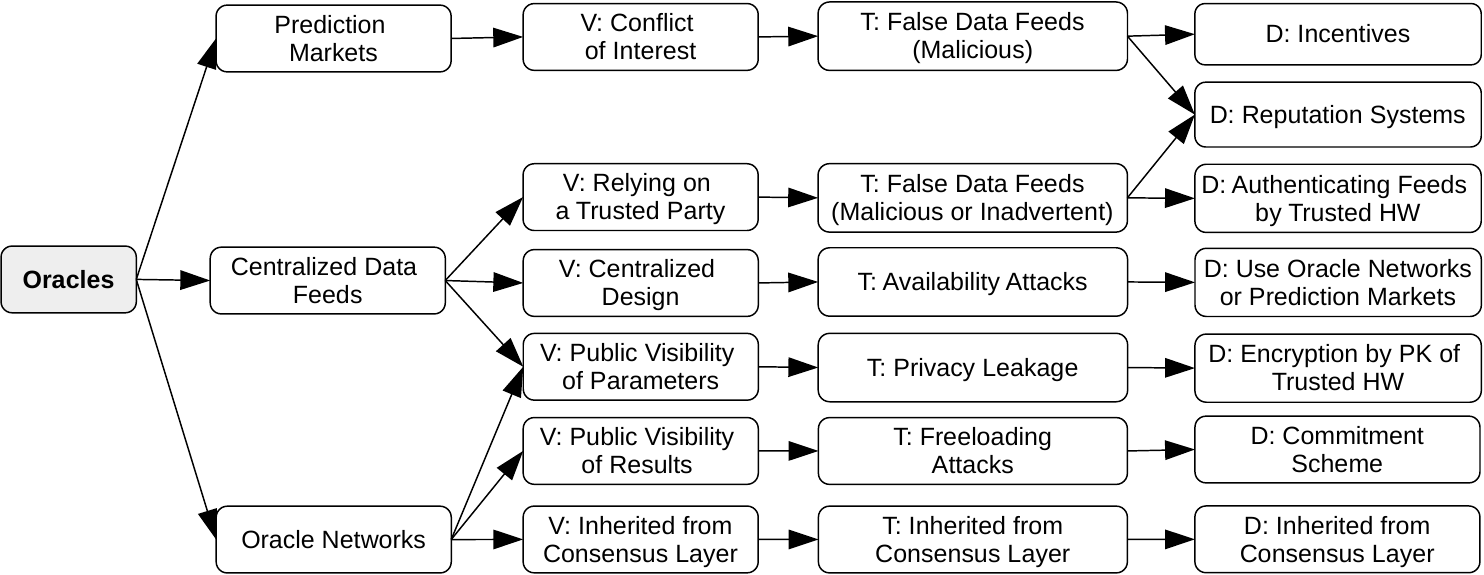} 		
		\caption{Vulnerabilities, threats, and defenses of the oracles category.}
		\vspace{0.1cm}
		\label{fig:attacks-APP-oracles}
		\vspace{-0.4cm}
	\end{center}	
\end{figure*}

\begin{figure*}[th]
	\begin{center}		
		\vspace{-0.4cm}
		\includegraphics[width=0.82\textwidth]{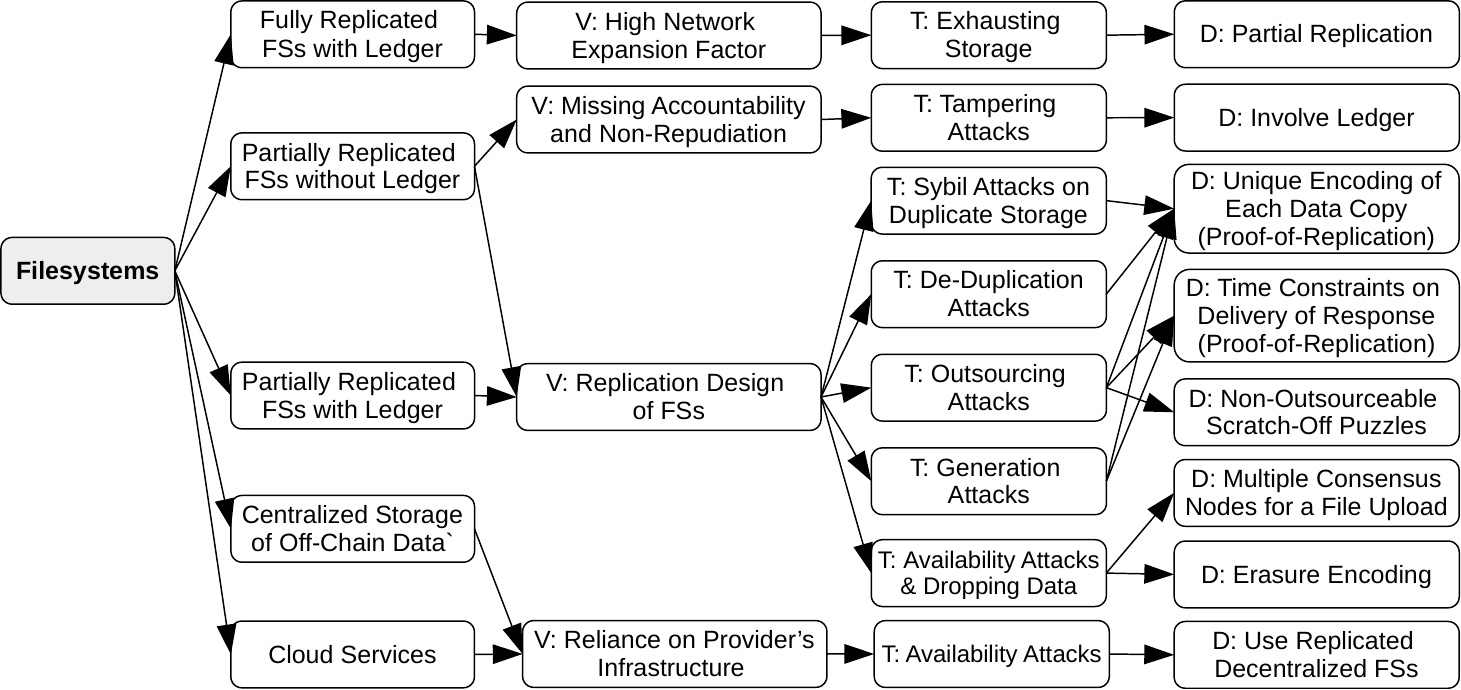} 		
		\caption{Vulnerabilities, threats, and defenses of the filesystems category.}
		\label{fig:attacks-APP-DFs}
	\end{center}	
\end{figure*}

\begin{figure}[t]
	\begin{center}			
		\includegraphics[width=0.71\textwidth]{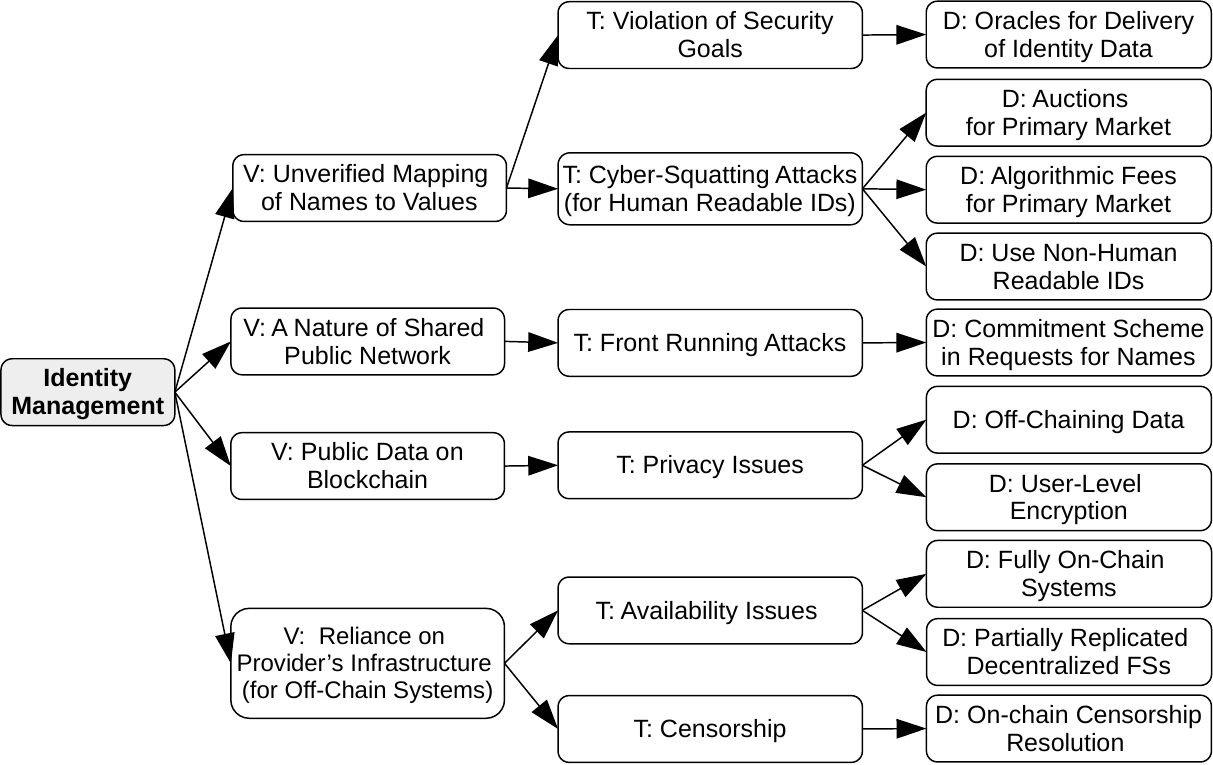} 		
		\caption{Vulnerabilities, threats, and defenses of the identity management category.}
		\label{fig:attacks-APP-identity}
	\end{center}	
\end{figure}

\begin{figure}[t]
	\begin{center}			
		\includegraphics[width=0.72\textwidth]{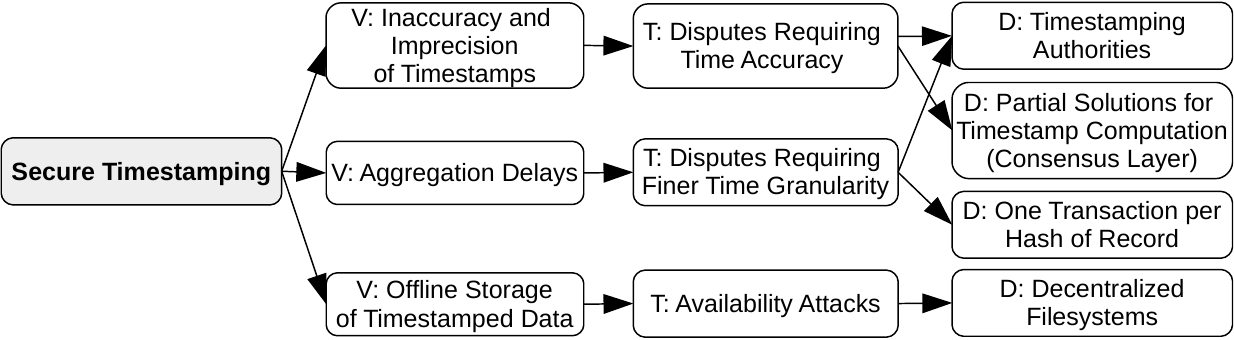} 		
		\caption{Vulnerabilities, threats, and defenses of the secure timestamping category.}
		\label{fig:attacks-APP-timestamps}
	\end{center}	
\end{figure}

\clearpage
\section{Application Layer: Higher-Level Applications}
\label{sec:apps-applications}

\begin{figure}[t]
	\begin{center}	
		\vspace{-0.2cm}		
		\includegraphics[width=0.7\textwidth]{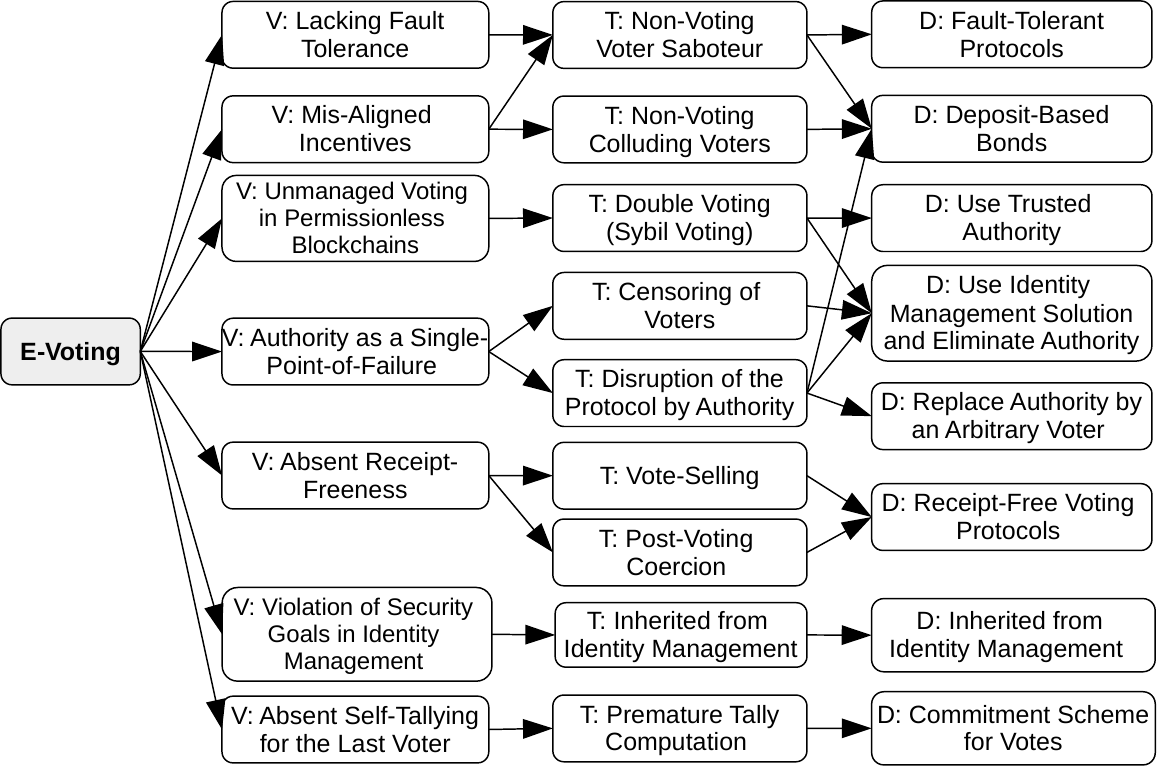} 		
		\caption{Vulnerabilities, threats, and defenses of the e-voting category.}
		\label{fig:attacks-APP-evoting}
	\end{center}	
\end{figure}

\begin{figure}[t]
	\begin{center}			
		\includegraphics[width=0.7\textwidth]{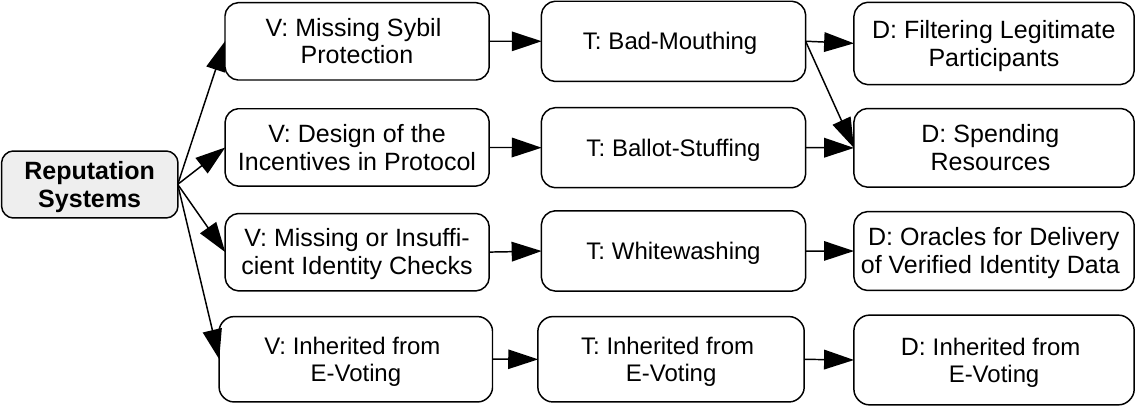} 		
		\caption{Vulnerabilities, threats, and defenses of the reputation systems category.}
		\label{fig:attacks-APP-reputation}		
	\end{center}	
\end{figure}

In this section, we focus on more specific higher-level applications as opposed to ecosystem applications.
In detail, we deal with eight categories of applications, and we outline their security threats and mitigation techniques in VTD graphs as follows: 
(1) \textbf{e-voting} (see \autoref{fig:attacks-APP-evoting}), (2) \textbf{reputation systems} (see \autoref{fig:attacks-APP-reputation}), (3) \textbf{data provenance} (see \autoref{fig:attacks-APP-provenance}), (4) \textbf{notaries} (see \autoref{fig:attacks-APP-notary}), (5) \textbf{direct trading} (see \autoref{fig:attacks-APP-trading}), (6) \textbf{escrows} (see \autoref{fig:attacks-APP-escrow}), (7) \textbf{auctions} (see \autoref{fig:attacks-APP-auctions}), and (8) \textbf{general application of blockchains}.
For details of these categories of applications and their security threats and mitigation techniques, we refer the interested reader to our paper \cite{homoliak2020security}. 

\begin{figure}[b]
	\begin{center}			
		\includegraphics[width=0.7\textwidth]{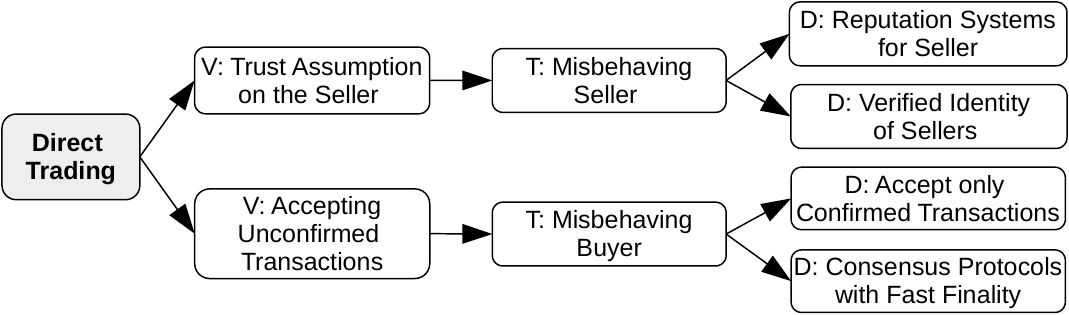} 		
		\caption{Vulnerabilities, threats, and defenses of the direct trading category.}
		\label{fig:attacks-APP-trading}
		\vspace{-0.4cm}
	\end{center}	
\end{figure}

\begin{figure}[t]
	\begin{center}			
		\includegraphics[width=0.7\textwidth]{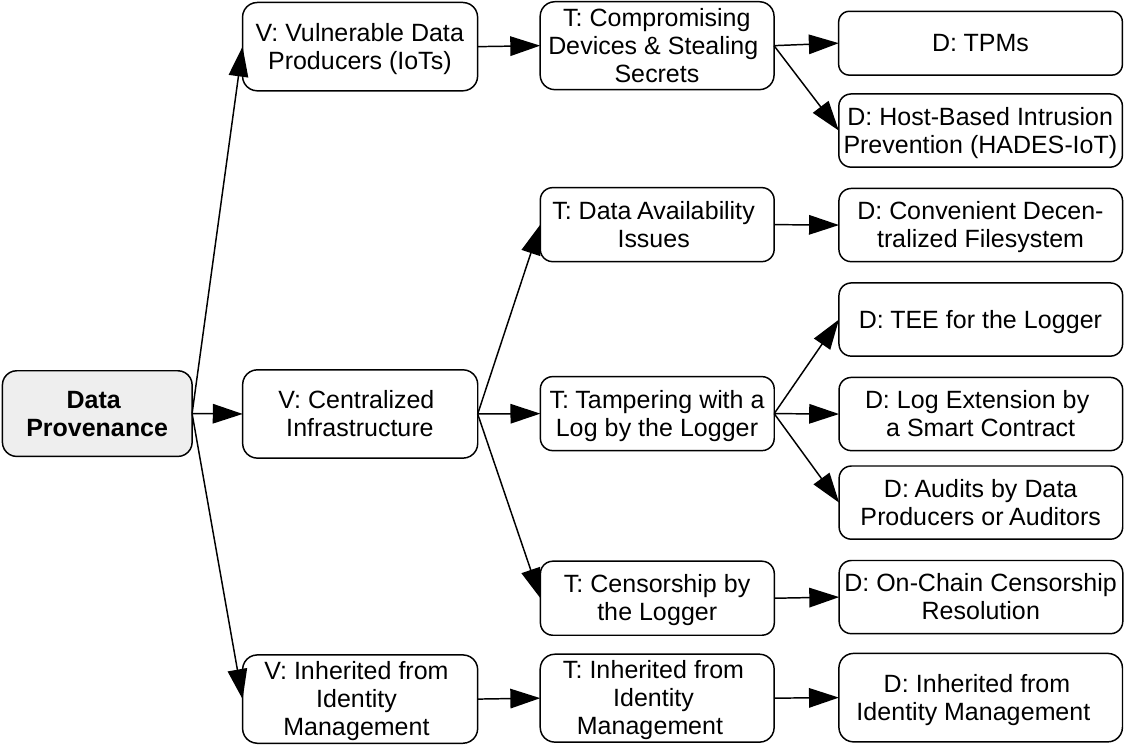} 		
		\caption{Vulnerabilities, threats, and defenses of the data provenance category.}
		\label{fig:attacks-APP-provenance}
	\end{center}	
\end{figure}

\begin{figure}[t]
	\begin{center}			
		\includegraphics[width=0.75\textwidth]{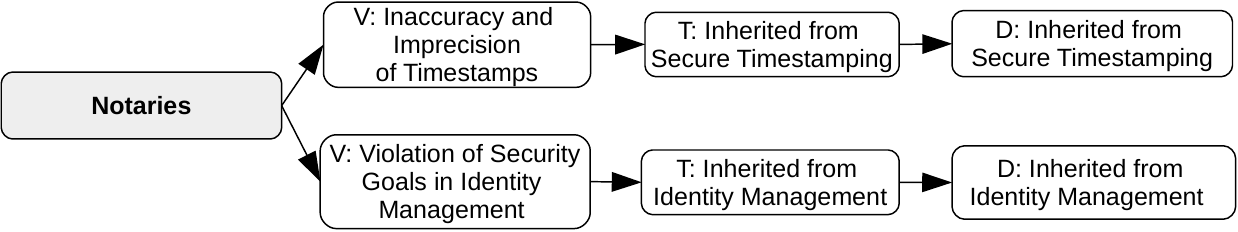} 		
		\caption{Vulnerabilities, threats, and defenses of the notaries category.}
		\label{fig:attacks-APP-notary}
	\end{center}	
\end{figure}

\begin{figure}[t]
	\begin{center}			
		\includegraphics[width=0.75\textwidth]{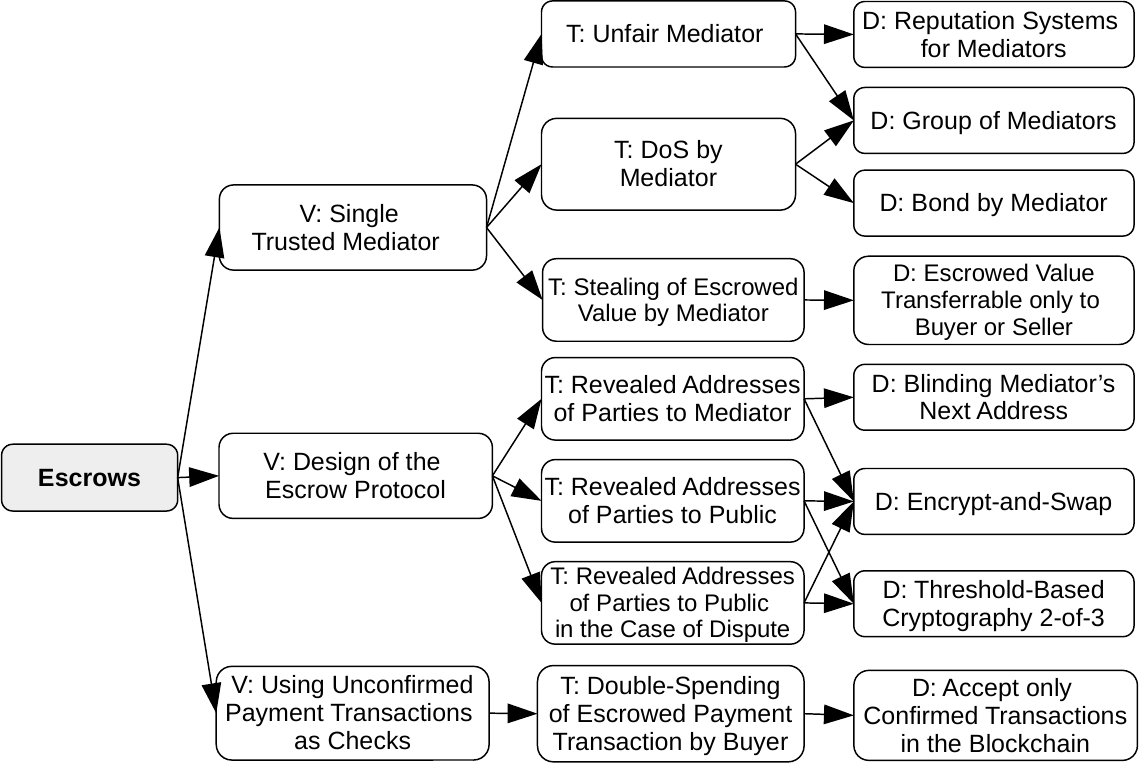} 		
		\caption{Vulnerabilities, threats, and defenses of the escrows category.}
		\label{fig:attacks-APP-escrow}
	\end{center}	
\end{figure}

\clearpage
\begin{figure}[t]
	\begin{center}			
		\includegraphics[width=0.7\textwidth]{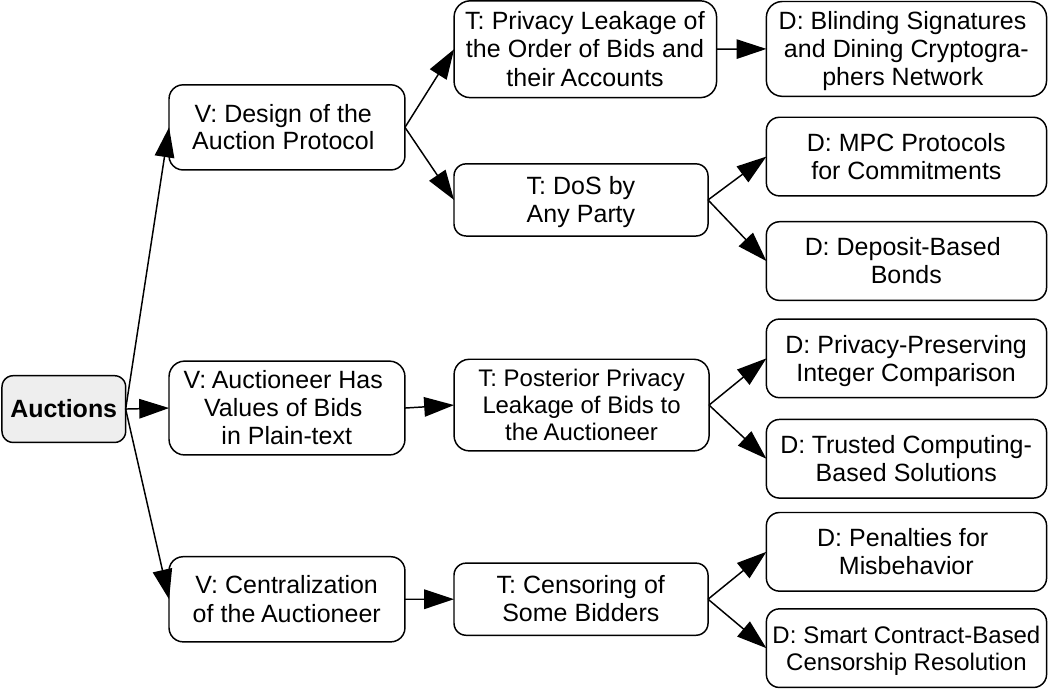} 		
		\caption{Vulnerabilities, threats, and defenses of the auctions category.}
		\label{fig:attacks-APP-auctions}
	\end{center}	
\end{figure}

\section{Lessons Learned}
\label{sec:lessons}

In this section, we summarize lessons learned concerning the security reference architecture (SRA) and its practical utilization.
First, we describe the hierarchy of security dependencies among particular layers of the SRA.
Second, assuming such a hierarchy, we describe a security-oriented methodology for designers of blockchain platforms and applications.
Finally, we summarize the design goals of particular blockchain types and discuss the security-specific features of the blockchains.

%
%

\subsection{Hierarchy of Dependencies in the SRA}
In the proposed model of the SRA, we observe that consequences of
vulnerabilities presented at lower layers of the SRA are manifested in the same layers and/or at higher layers, especially at the application layer.
Therefore, we refer to \textit{security dependencies} of these layers on lower layers or the same layers, i.e., \textit{reflexive} and \textit{bottom-up} dependencies. 
We describe these two types of dependencies in the following.

\paragraph{Reflexive Dependencies.}
If a layer of the SRA contains some assets, it also contains a reflexive security dependency on the countermeasures presented in the same layer.
It means that a countermeasure at a particular layer protects the assets presented in the same layer.
For example, in the case of the consensus layer whose protocols reward consensus nodes for participation,  the countermeasures against selfish mining attacks protect rewards (i.e., crypto-tokens) of consensus nodes.
In the case of the RSM layer, the privacy of user identities and data is protected by various countermeasures of this layer (e.g., blinding signatures, secure multiparty computations).
Another group of reflexive security dependencies is presented at the application layer. 
Although the application layer contains some bottom-up security dependencies (see \autoref{fig:dependencies}), we argue that with regard to the overall stacked model of the SRA they can be viewed as reflexive security dependencies of the application layer.

\paragraph{Bottom-Up Dependencies.}
If a layer of the SRA contains some assets, besides reflexive security dependencies, it also contains bottom-up security dependencies on the countermeasures of all lower layers.
Hence, the consequences of vulnerabilities presented at lower layers of SRA might be manifested at the same layers (i.e., reflexive dependencies) but more importantly, they are manifested at higher layers, especially at the application layer.
For example, context-sensitive transactions and partial solutions as countermeasures of the consensus layer can protect against front-running attacks of intra-chain DEXes, which occur at the application layer.
Another example represents programming bugs in the RSM layer, which influence the correct functionality at the application layer.
The eclipse attack is an example that impacts the consensus layer from the network layer -- a victim consensus node operates over the attacker-controlled chain, and thus causes a loss of crypto-tokens by a consensus node and at the same time it decreases honest consensus power of the network.
In turn, this might simplify selfish mining attacks at the consensus layer, which in turn might impact the correct functionality of a blockchain-based application at the application layer.
Bottom-up security dependencies are also presented in the context of the application layer, as we have already mentioned in \autoref{sec:apps}.

\setlength{\tabcolsep}{2.0pt}	
\begin{table*}[!h]
	\scriptsize
	\begin{tabular}{@{}ccclllll@{}}
		\toprule
		\textbf{Layer}                                                                                         & \multicolumn{3}{c}{\textbf{\begin{tabular}[c]{@{}c@{}}Category \\ in a Layer\end{tabular}}}                                                                                                                                                                                                                                   & \multicolumn{1}{c}{\textbf{Pros}}                                                                                                                                                                                                         & \multicolumn{1}{c}{\textbf{Cons}}                                                                                                                                                                                                           &  &  \\ 
		
		\Xhline{2\arrayrulewidth}\noalign{\smallskip}  
		
		\multirow{2}{*}{\textbf{\begin{tabular}[c]{@{}c@{}}\rotatebox[origin=l]{90}{Network Layer}\end{tabular}}}                      & \multicolumn{3}{c}{\textbf{\begin{tabular}[c]{@{}c@{}}Private\\ Networks\end{tabular}}}                                                                                                                                                                                                                                       & \begin{tabular}[c]{@{}l@{}}$\bullet$ low latency, high throughput\\ $\bullet$ centralized administration,\\ $\bullet$ ease of access control \\ $\bullet$ the privacy of data and identities\\ $\bullet$ meeting regulatory obligations\\ $\bullet$ resilience to external attacks\end{tabular} & \begin{tabular}[c]{@{}l@{}}$\bullet$ VPN is required for geographically spread\\ participants\\ $\bullet$ suitable only for permissioned blockchains\\ $\bullet$ insider threat and external attacks at nodes\\with administrative privileges\end{tabular}                                     &  &  \\
		
		\cmidrule(l){5-8} 
		
		& \multicolumn{3}{c}{\textbf{\begin{tabular}[c]{@{}c@{}}Public\\ Networks\end{tabular}}}                                                                                                                                                                                                                                        & \begin{tabular}[c]{@{}l@{}}$\bullet$ high decentralization\\ $\bullet$ high availability\\ $\bullet$ openness \& low entry barrier\end{tabular}                                                     & \begin{tabular}[c]{@{}l@{}}$\bullet$ high and non-uniform latency\\ $\bullet$ single-point-of-failure (DNS, IP, and ASes\\are managed by centralized parties)\\ $\bullet$ external adversaries (botnets, compromised\\BGP/DNS servers)\\ $\bullet$ stolen identities\end{tabular} &  &  \\
		
		\Xhline{2\arrayrulewidth}\noalign{\smallskip}  
		
		\multirow{6}{*}{\textbf{\begin{tabular}[c]{@{}c@{}}\rotatebox[origin=l]{90}{Consensus Layer~~~~~~~~~~~~~~~~~~}\end{tabular}}}                    & \multicolumn{3}{c}{\textbf{PoR}}                                                                                                                                                                                                                                                                                              & \begin{tabular}[c]{@{}l@{}}$\bullet$ high cost of overriding\\the history of blockchain\\ $\bullet$ high scalability\end{tabular}                                                                                                                            & \begin{tabular}[c]{@{}l@{}}$\bullet$ high operational costs\\ $\bullet$ low throughput\\ $\bullet$ low finality\end{tabular}                                                                                                                                           &  &  \\
		
		\cmidrule(l){5-8} 
		
		& \multicolumn{3}{c}{\textbf{BFT}}                                                                                                                                                                                                                                                                                              & \begin{tabular}[c]{@{}l@{}}$\bullet$ high throughput\\(with a small number of nodes)\\ $\bullet$ fast finality\end{tabular}                                                                                                                                                                 & \begin{tabular}[c]{@{}l@{}}$\bullet$ low scalability \\ $\bullet$ high communication complexity\\ $\bullet$ limited number of nodes (efficient use only\\in permissioned blockchains)\end{tabular}                                                                      &  &  \\
		
		\cmidrule(l){5-8} 
		
		& \multicolumn{3}{c}{\textbf{PoS}}                                                                                                                                                                                                                                                                                              & $\bullet$ energy efficiency                                                                                                                                                                                                                        & \begin{tabular}[c]{@{}l@{}}$\bullet$ PoS specific attacks and issues\\ $\bullet$ supports only semi-permissionless setting\\ $\bullet$ slow finality\end{tabular}                                                                                                      &  &  \\
		
		\cmidrule(l){5-8} 
		
		& \multicolumn{3}{c}{\textbf{PoS+BFT}}                                                                                                                                                                                                                                                                                          & \begin{tabular}[c]{@{}l@{}}$\bullet$ energy efficiency\\ $\bullet$ high scalability\\ $\bullet$ probabilistic security guarantees\\ $\bullet$ lower communication overheads than BFT\end{tabular}                                                                            & \begin{tabular}[c]{@{}l@{}}$\bullet$ some PoS specific attacks\\ $\bullet$ supports only semi-permissionless setting\end{tabular}                                                                                                                             &  &  \\
		
		\cmidrule(l){5-8} 		
		
		& \multicolumn{3}{c}{\textbf{PoR+BFT}}                                                                                                                                                                                                                                                                                          & \begin{tabular}[c]{@{}l@{}}$\bullet$ high scalability\\ $\bullet$ fast finality\end{tabular}                                                                                                                                                     & $\bullet$ spending some scarce resources                                                                                                                                                                                                             &  &  \\
		
		\cmidrule(l){5-8} 		
		
		& \multicolumn{3}{c}{\textbf{\begin{tabular}[c]{@{}c@{}}PoR+PoS \\ (i.e., PoA)\end{tabular}}}                                                                                                                                                                                                                                   & $\bullet$ high scalability                                                                                                                                                                                                                         & \begin{tabular}[c]{@{}l@{}}$\bullet$ spending some scarce resources\\ $\bullet$ some PoS specific attacks\\ $\bullet$ slow finality\end{tabular}                                                                                                                       &  &  \\

		\Xhline{2\arrayrulewidth}\noalign{\smallskip}  
		
		\multirow{8}{*}{\textbf{\begin{tabular}[c]{@{}c@{}}\rotatebox[origin=l]{90}{Replicated State Machine Layer~~~~~~~~~~~~~}\end{tabular}}} & \multirow{6}{*}{\rotatebox[origin=l]{90}{\textbf{Transactions~~~~~~~~~~~~~~~~~~~~~}}} & \multirow{4}{*}{\textbf{\begin{tabular}[c]{@{}c@{}}\rotatebox[origin=l]{90}{Protecting Identities~~~~}\end{tabular}}} & \textbf{Standard Approach}                                                                                                                 & \begin{tabular}[c]{@{}l@{}}$\bullet$ fast processing\\ $\bullet$ ease of verification\end{tabular}                                                                                                                                                          & \begin{tabular}[c]{@{}l@{}}$\bullet$ identities are only pseoudonymous and can\\be traced to IPs\\ $\bullet$ all data of transactions are publicly visible\end{tabular}                                                                                        &  &  \\
		
		\cmidrule(l){5-8} 
		
		&                                                                                     &                                                                                            & \textbf{\begin{tabular}[c]{@{}l@{}}Standard Approach\\ + Mixers\end{tabular}}                                                              & \begin{tabular}[c]{@{}l@{}}$\bullet$ privacy identity protection\\ of users in  a group\\ $\bullet$ ease of verification\end{tabular}                                                                                                                         & \begin{tabular}[c]{@{}l@{}}$\bullet$ additional complexity, in some cases\\ unlinkability by the mixer or involved\\parties in a group\\ $\bullet$ all data of transactions are\\publicly visible\end{tabular}                                            &  &  \\
		
		\cmidrule(l){5-8} 
		
		&                                                                                     &                                                                                            & \textbf{\begin{tabular}[c]{@{}l@{}}NIZKs and \\ Ring-Signatures\end{tabular}}                                                              & \begin{tabular}[c]{@{}l@{}}$\bullet$ identities are anonymized to the extend\\ of the group\end{tabular}                                                                                                                                          & \specialcell{$\bullet$ additional computation overheads\\for running the schemes}                                                                                                                                                                                     &  &  \\
		
		\cmidrule(l){5-8} 
		
		&                                                                                     &                                                                                            & \textbf{\begin{tabular}[c]{@{}l@{}}MPC\\ Blinding Signatures, \\ Layered-Encryption\end{tabular}}                        & $\bullet$ unlinkability for all involved parties                                                                                                                                                                                                   & \specialcell{$\bullet$ additional computation overheads\\for running the schemes}                                                                                                                                                                                                                                                                                                                               &  &  \\
		
		\cmidrule(l){3-8} 
		&                                                                                     & \multirow{2}{*}{\textbf{\begin{tabular}[c]{@{}c@{}}\rotatebox[origin=l]{90}{Protecting Data}\end{tabular}}}       & \textbf{\begin{tabular}[c]{@{}l@{}}NIZKs, Blinding \\ Signatures, Homomor-\\phic Encryption\end{tabular}}                                 & \begin{tabular}[c]{@{}l@{}}$\bullet$ privacy of data in cryptocurrency \\ platforms\end{tabular}                                                                                                                                                   & \specialcell{$\bullet$ additional computation overheads\\for running the schemes}                                                                                                                                                                                                                                                                                                                                                                        &  &  \\
		
		\cmidrule(l){5-8} 
		
		&                                                                                     &                                                                                            & \textbf{\begin{tabular}[c]{@{}l@{}}Trusted Transaction \\ Managers,\\ Trusted Hardware,\\ MPC\end{tabular}} & \begin{tabular}[c]{@{}l@{}}$\bullet$ privacy of data in transactions\\of smart contract platforms\end{tabular}                                                                                                                                    & \specialcell{$\bullet$ additional computation overheads\\for running the schemes}                                                                                                                                                                                                                                                                                                                                                                   &  &  \\
		
		\cmidrule(l){2-8} 
		
		& \multirow{2}{*}{\textbf{\rotatebox[origin=l]{90}{\begin{tabular}[c]{@{}c@{}}Smart\\ Contracts\end{tabular}}}} & \multicolumn{2}{c}{\textbf{\begin{tabular}[c]{@{}c@{}}Turing-Complete\\ Languages\end{tabular}}}                                                                                                                                     & \begin{tabular}[c]{@{}l@{}}$\bullet$ smart contracts may contain an arbitrary\\ programming logic\end{tabular}                                                                                                                                     & \specialcell{$\bullet$ wide surface for making the\\programming bugs that often results\\in vulnerabilities}                                                                                                                                                          &  &  \\
		
		\cmidrule(l){5-8} 
		
		&                                                                                     & \multicolumn{2}{c}{\textbf{\begin{tabular}[c]{@{}c@{}}Turing-Incomplete \\ Languages\end{tabular}}}                                                                                                                                  & $\bullet$ small attack surface and emphasis on safety                                                                                                                                                                                              & \specialcell{$\bullet$ the programming logic serves only\\for limited purposes}                                                                                                                                                                                    &  &  \\ 
		\bottomrule
		
	\end{tabular}
	\caption{Pros and cons of various categories within the first three layers of the stacked model.}\label{tab:pros-cons}
\end{table*}

\setlength{\tabcolsep}{2.4pt}	
\begin{table*}[!h]
	\scriptsize
	\centering
	\begin{tabular}{@{}ccclllll@{}}
		\toprule
		\textbf{\specialcell{Application\\~~Category}}                                                                                         & \multicolumn{3}{c}{\textbf{\begin{tabular}[c]{@{}c@{}}Subcategory\end{tabular}}}                                                                                                                                                                                                                                   & \multicolumn{1}{c}{\textbf{Pros}}                                                                                                                                                                                                         & \multicolumn{1}{c}{\textbf{Cons}}                                                                                                                                                                                                           &  &  \\ 
		
		\Xhline{2\arrayrulewidth}\noalign{\smallskip}

		\multirow{3}{*}{\textbf{\begin{tabular}[c]{@{}c@{}}\rotatebox[origin=l]{90}{Wallets~~~~~~~}\end{tabular}}}                      & \multicolumn{3}{c}{\textbf{\begin{tabular}[c]{@{}c@{}}Server-Side\\Hosted\\ Wallets\end{tabular}}}                                                                                                                                                                                                                                       & \begin{tabular}[c]{@{}l@{}}$\bullet$ simplicity of control for end-users \\ $\bullet$ no storage requirements for end-users\\ \end{tabular} & \begin{tabular}[c]{@{}l@{}}$\bullet$ keys stored at the server, susceptibility to\\the theft of keys by external or internal attacks\\ $\bullet$ single-point-of-failure, availability attacks  \end{tabular}                                     &  &  \\
		
		\cmidrule(l){5-8} 
		
		& \multicolumn{3}{c}{\textbf{\begin{tabular}[c]{@{}c@{}}Client-Side\\Hosted\\Wallets\end{tabular}}}                                                                                                                                                                                                                                      & \begin{tabular}[c]{@{}l@{}}$\bullet$ simplicity of control for end-users \\ $\bullet$ no storage requirements for end-users\\ $\bullet$ keys stored locally \end{tabular}                                                     & \begin{tabular}[c]{@{}l@{}} $\bullet$ single-point-of-failure, availability attacks \\ $\bullet$ possibility of key theft by malware\\ $\bullet$ possibility of tampering attacks \end{tabular} &  &  \\
		
		\cmidrule(l){5-8} 
		
		& \multicolumn{3}{c}{\textbf{\begin{tabular}[c]{@{}c@{}}Self-Sovereign\\Wallets\end{tabular}}}                                                                                                                                                                                                                                      & \begin{tabular}[c]{@{}l@{}}$\bullet$ keys stored locally\\or in a dedicated hardware device \\ \end{tabular}                                                     & \begin{tabular}[c]{@{}l@{}} $\bullet$ moderate storage requirements for end-users \\ $\bullet$ more difficult control for end-users \\ $\bullet$ extra device to carry in the case of hardware wallet \end{tabular} &  &  \\
		
		\Xhline{2\arrayrulewidth}\noalign{\smallskip}

		\multirow{3}{*}{\textbf{\begin{tabular}[c]{@{}c@{}}\rotatebox[origin=l]{90}{Exchanges~~~~~~~~~~~~}\end{tabular}}}                      & \multicolumn{3}{c}{\textbf{\begin{tabular}[c]{@{}c@{}}Centralized\\ Exchange\end{tabular}}}                                                                                                                                                                                                                                       & \begin{tabular}[c]{@{}l@{}}$\bullet$ a high throughput and speed of operations \\ $\bullet$ the simplicity of control for end-users\\$\bullet$ low costs for exchange transactions\\ $\bullet$ trading of obscure crypto-tokens  \end{tabular} & \begin{tabular}[c]{@{}l@{}}$\bullet$ risk of insider threat due to centralization\\ $\bullet$ external threats to exchange infrastructure\\ $\bullet$ overheads for secure storage of secrets\\ $\bullet$ a fee specified by the operator \end{tabular}                                     &  &  \\
		
		\cmidrule(l){5-8} 
		
		& \multicolumn{3}{c}{\textbf{\begin{tabular}[c]{@{}c@{}}Direct\\Cross-Chain\\ Exchange\end{tabular}}}                                                                                                                                                                                                                                      & \begin{tabular}[c]{@{}l@{}}$\bullet$ fairness of the exchange \\ $\bullet$ no fee to any operator\\  \end{tabular}                                                     & \begin{tabular}[c]{@{}l@{}}$\bullet$ costs for 4 transactions of the atomic swap\\ $\bullet$ user has to find the counter-order on her own\\ $\bullet$ counter-orders might not exist\\ $\bullet$ a lower throughput than in a centralized exchange\\ $\bullet$ a higher complexity for end-users  \end{tabular} &  &  \\
		
		\cmidrule(l){5-8} 
		
		& \multicolumn{3}{c}{\textbf{\begin{tabular}[c]{@{}c@{}}Cross-Chain\\DEX\end{tabular}}}                                                                                                                                                                                                                                      & \begin{tabular}[c]{@{}l@{}}$\bullet$ fairness of the exchange \\ $\bullet$ order matching made by DEX \\ $\bullet$ trading of obscure crypto-tokens \end{tabular}                                                     & \begin{tabular}[c]{@{}l@{}}$\bullet$ costs for 4 or 6 transactions of the atomic swap\\ $\bullet$ a lower throughput than in a centralized exchange\\ $\bullet$ a fee specified by the operator\\ \end{tabular} &  &  \\
		
		\cmidrule(l){5-8} 
		
		& \multicolumn{3}{c}{\textbf{\begin{tabular}[c]{@{}c@{}}Intra-Chain\\ DEX\end{tabular}}}                                                                                                                                                                                                                                      & \begin{tabular}[c]{@{}l@{}} $\bullet$ fairness of the exchange\\ $\bullet$ uniform finality for every pair \\ $\bullet$ a high speed of operations \\  \end{tabular}                                                     & \begin{tabular}[c]{@{}l@{}}$\bullet$ a limited number of pairs specific to the target platform\\ $\bullet$ a fee specified by the operator\\ $\bullet$ costs for smart contract execution \end{tabular} &  &  \\
		
		\Xhline{2\arrayrulewidth}\noalign{\smallskip}  
		
		\multirow{3}{*}{\textbf{\begin{tabular}[c]{@{}c@{}}\rotatebox[origin=l]{90}{Oracles~~~~~~~}\end{tabular}}}                      & \multicolumn{3}{c}{\textbf{\begin{tabular}[c]{@{}c@{}}Prediction\\ Markets\end{tabular}}}                                                                                                                                                                                                                                       & \begin{tabular}[c]{@{}l@{}}$\bullet$ early (close to accurate) estimation\\~~of the future event's result \\ $\bullet$ decentralization \\ \end{tabular} & \begin{tabular}[c]{@{}l@{}}$\bullet$ possible conflict of interest\\ $\bullet$ a limited set of data specific to a few events \\ $\bullet$ a long time to obtain a final result, especially in\\~~the case of disputes \\ \end{tabular}                                     &  &  \\
		
		\cmidrule(l){5-8} 
		
		& \multicolumn{3}{c}{\textbf{\begin{tabular}[c]{@{}c@{}}Centralized\\ Data Feeds\end{tabular}}}                                                                                                                                                                                                                                      & \begin{tabular}[c]{@{}l@{}}$\bullet$ wide range of data \\ $\bullet$ fast provisioning time \\ $\bullet$ handling of private parameters of requests \\ $\bullet$ censorship evidence \\ \end{tabular}                                                     & \begin{tabular}[c]{@{}l@{}} $\bullet$ centralization (accidentally or intentionally wrong data) \\ $\bullet$ availability issues \\ \end{tabular} &  &  \\
		
		\cmidrule(l){5-8} 
		
		& \multicolumn{3}{c}{\textbf{\begin{tabular}[c]{@{}c@{}}Oracle\\ Networks\end{tabular}}}                                                                                                                                                                                                                                      & \begin{tabular}[c]{@{}l@{}}$\bullet$ decentralization \\ $\bullet$ wide range of data \\ $\bullet$ fast provisioning time \end{tabular}                                                     & \begin{tabular}[c]{@{}l@{}}$\bullet$ unsupported private parameters of requests\\ $\bullet$ publicly visible data and requests\\ \end{tabular} &  &  \\
		
		\Xhline{2\arrayrulewidth}\noalign{\smallskip}  
		
		\multirow{3}{*}{\textbf{\begin{tabular}[c]{@{}c@{}}\rotatebox[origin=l]{90}{Filesystems~~~~~~~~~~}\end{tabular}}}                      & \multicolumn{3}{c}{\textbf{\begin{tabular}[c]{@{}c@{}}Fully\\Replicated FSs\\with Ledger\end{tabular}}}                                                                                                                                                                                                                                       & \begin{tabular}[c]{@{}l@{}}$\bullet$ a high availability\\$\bullet$ accountability and auditability \end{tabular} & \begin{tabular}[c]{@{}l@{}}$\bullet$ a high storage overheads and operational costs\\ $\bullet$ a high price  \\ \end{tabular}                                     &  &  \\		
		\cmidrule(l){5-8} 
		
		& \multicolumn{3}{c}{\textbf{\begin{tabular}[c]{@{}c@{}}Partially\\Replicated FSs\\with Ledger\end{tabular}}}                                                                                                                                                                                                                                      & \begin{tabular}[c]{@{}l@{}}$\bullet$ reasonably high availability\\ $\bullet$ accountability and auditability\\ $\bullet$ a lower price than in a full replication \end{tabular}                                                     & \begin{tabular}[c]{@{}l@{}}$\bullet$ attack vectors specific to partial replication \end{tabular} &  &  \\
		
		\cmidrule(l){5-8} 
		
		& \multicolumn{3}{c}{\textbf{\begin{tabular}[c]{@{}c@{}}Partially\\Replicated FSs\\without Ledger\end{tabular}}}                                                                                                                                                                                                                                      & \begin{tabular}[c]{@{}l@{}} $\bullet$  reasonably high availability\\$\bullet$ a lower price than in a full replication\\  \end{tabular}                                                     & \begin{tabular}[c]{@{}l@{}}$\bullet$ a lack of native accountability and auditability\\$\bullet$ low durability due to a lack of incentives for storage \end{tabular} &  &  \\
		
		\cmidrule(l){5-8} 
		
		& \multicolumn{3}{c}{\textbf{\begin{tabular}[c]{@{}c@{}}Centralized\\Storage of\\Off-Chain Data\end{tabular}}}                                                                                                                                                                                                                                      & \begin{tabular}[c]{@{}l@{}}$\bullet$ a low price \\ $\bullet$ accountability and auditability   \end{tabular}                                                     & \begin{tabular}[c]{@{}l@{}}$\bullet$  a low availability \end{tabular} &  &  \\	
		
		\bottomrule
		
	\end{tabular}
	\caption{Pros and cons of some categories from the application layer.}\label{tab:pros-cons-2}
\end{table*}

\subsection{Methodology for Designers}\label{sec:designing-solutions}
A hierarchy of security dependencies in the SRA can be utilized during the design of new blockchain-based solutions.
When designing a new \textbf{blockchain platform} or a new \textbf{blockchain application}, we recommend designers to specify requirements on the blockchain features (see  \autoref{sec:features-of-blockchain}) and afterward analyze design options and their attack surfaces at the first three layers of the stacked model of SRA.
We briefly summarize the pros and cons of particular categories within the first three layers of SRA in \autoref{tab:pros-cons}, while security threats and mitigations are covered in \autoref{sec:network}, \autoref{sec:consensus}, and \autoref{sec:smart_contracts}.

On top of that, we recommend the designers of a new \textbf{blockchain application} to analyze particular options and their security implications at the application layer of SRA.
We list the pros and cons of a few categories from the application layer in \autoref{tab:pros-cons-2},\footnote{Note that the table contains only categories with specified sub-categorizations that represent the subject to a comparison.} while security threats and mitigation techniques of this layer are elaborated in \autoref{sec:apps} and \autoref{sec:apps-applications}.
During this process, we recommend the designers to follow security dependencies of the target category on other underlying categories (see \autoref{fig:dependencies}) if their decentralized variants are used (which is a preferable option from the security point-of-view).
For example, if one intends to design a decentralized reputation system, she is advised to study the security threats from the reputation system category and its recursive dependencies on e-voting, identity management, crypto-tokens \& wallets, and (optionally) filesystems.

%

\medskip
\paragraph{Divide and Conquer.}
If a designer of the blockchain application is  also designing a blockchain platform, 
we recommend her to split the functionality of the solution with the divide-and-conquer approach respecting particular layers of our stacked model. 
In detail, if some functionality is specific to the application layer, then it should be implemented at that layer. 
Such an approach minimizes the attack surface of a solution and enables isolating the threats to specific layers, where they are easier to protect from and reviewed by the community.
%
A contra-example is to incorporate a part of application layer functionality/validation into the consensus layer. 
The consensus layer should deal only with the ordering of transactions, and it should be agnostic to the application. 
%

Nevertheless, it is worth noting that the divide-and-conquer approach might not be suitable for some very specific cases.
For example, some decentralized filesystems might combine data storage as an application-layer service with the proof-of-storage consensus algorithm, presented at the consensus layer. 
Therefore, the consensus layer also embeds a part of functionality from the application layer.
However, when filesystems are in security dependencies of the target application other than filesystems, one should realize that they are usually running on a different blockchain or infrastructure than the target application, and this exception is not a concern.

\begin{figure}[t]
	\begin{center}		
		\includegraphics[width=0.850\textwidth]{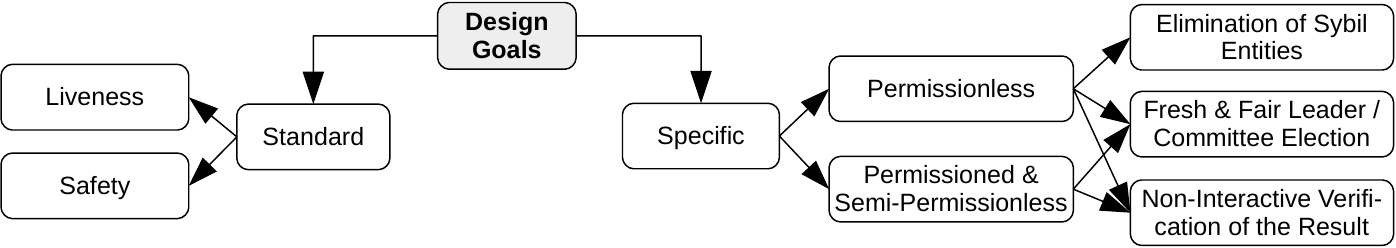} 		
		\caption{Standard and specific design goals of consensus protocols.}
		\label{fig:design-goals}
	\end{center}	
\end{figure}

\subsection{Blockchain Types \& Design Goals}\label{sec:specific-design-goals}
We learned that the type of a blockchain (see \autoref{sec:types-of-blockchains}) implies the specific design goals of its consensus protocol (see \autoref{fig:design-goals}), which must be considered on top of the standard design goals (i.e., liveness and safety) and the inherent features (see \autoref{sec:background-features}) during the design of a particular blockchain platform and its consensus protocol.
In the following, we elaborate on such specific design goals. 

\paragraph{Permissionless Type.} 
The first design goal is to \textit{eliminate Sybil entities} -- such elimination can be done by requiring that some amount of scarce resources is spent for extension of the blockchain, and hence no Sybil entity can participate.
This implies that no pure PoS protocol can be permissionless since it never spends resources on running a consensus protocol.
The next design goal is \textit{a fresh and fair leader/committee election}, which ensures that each consensus node influences the result of a consensus commensurately to the number of scarce resources spent.
Moreover, freshness avoids the prediction of the selected nodes, and therefore elected nodes cannot become the subject of targeted DoS attacks. 
The last design goal is the \textit{non-interactive verification} of the consensus result by any node -- i.e., any node can verify the result of the consensus based on the data presented in the blockchain.

\paragraph{Permissioned and Semi-Permissionless Types.}
These types of blockchains require \textit{fresh and fair leader/committee election} as well as \textit{non-interactive verification} of the result of the consensus.
However, in contrast to the permissionless blockchains, they do not require a means for the elimination of Sybil entities, as permission to enter the system is given by a centralized entity (i.e., permissioned type) or any existing consensus node (i.e., semi-permissionless type).

\paragraph{Blockchain Types and Incentives.}
We observed that no application running on a \textit{public} (permissioned) blockchain has been able to work without introducing crypto-tokens (i.e., an incentive scheme), even if the use case is not financial in nature, e.g., e-voting, notaries, secure timestamping, or reputation systems. 
In these blockchains, incentive schemes serve as a means for the elimination of Sybil entities, besides other purposes.
The situation is different in the context of \textit{private} (permissioned) blockchains, which are usually provisioned by a single organization or a consortium and do not necessarily need crypto-tokens to operate. 
Misaligned incentives can cause consensus-level vulnerabilities, e.g., when it becomes profitable to drop blocks of other nodes to earn higher mining rewards~\cite{eyal2018majority} or transaction fees~\cite{instability_noReward}. 
The design of incentive mechanisms is a research field by itself and we refer the reader to the work of Leonardos et al.~\cite{leonardos2019presto}. 

\subsection{Security-Specific Features of Blockchains}\label{sec:lessons-finality}
We realized that consensus protocols are the target of most financially-oriented attacks on the decentralized infrastructure of blockchains, even if such attacks might originate from the network layer (e.g., routing and eclipse attacks).
The goal of these attacks is to overturn and re-order already ordered blocks while doing double-spending.
Hence, the finality is the most security-critical feature of the consensus layer.
The finality differs per various categories of the consensus layer. 
The best finality is achieved in the pure BFT protocols, and the worse finality is achieved in the single-leader-based PoR and PoS protocols.
On the other hand, combinations of the BFT with PoS protocols (i.e., introducing committees) slightly deteriorate the finality of BFT in a probabilistic ratio that is commensurate to the committee size.
In the case of PoR protocols with partial solutions, finality is improved as opposed to pure PoR protocols; however, it is also probabilistic, depending on the number of partial solutions.

\subsection{Limitations in the Literature and Practice}\label{sec:limitations}

\paragraph{Applications of Blockchains.}
Although the literature contains surveys and overviews \cite{casino2018systematic,zheng2018blockchain,wust2018you} of blockchain-based applications, these works introduce only domain-oriented categorizations (i.e., categories such financial, governance, security, education, supply chain, etc.) and they do not investigate the security aspects and functionalities that these applications leverage on and whether some of the applications do not belong to the same category from the security and functionality point-of-view.
To address this limitation, we provide a security-driven functionality-oriented categorization of blockchain-based applications (see \autoref{sec:apps}), which is agnostic to an application domain and thus can generalize different application scenarios.
Furthermore, our proposed categorization enables us to model security and functionality-based dependencies among particular categories, which is not possible with state-of-the-art categorizations.

\paragraph{Centralization.}
Even though blockchains are meant to be fully decentralized, we have seen that this does not hold at some layers of the SRA -- the network and application layers, in particular.
In the network layer, some attacks are possible due to centralized DNS bootstrapping, while in the application layer a few categories utilize centralized components to ensure some functionality that cannot run on-chain or its provisioning would be too expensive and slow, which, however, forms the trade-off with the security.  
Some applications might depend on components from other application categories (e.g., identity management) but implementing these components in a centralized fashion, even though there exist some decentralized variants that are gaining popularity (e.g., DIDs~\cite{did-w3c} for identity management).

%

\section{Contributing Papers}\label{sec:sra-contributing-papers}
The papers that contributed to this research direction are enumerated in the following, while highlighted papers are attached to this thesis in their original form.

\begingroup
\let\clearpage\relax

\renewcommand\bibname{}

\endgroup

%% file: sec/consensus-protocols.tex

In this chapter, we present our contributions to the area of consensus protocols in block\-chains and their security, which belong to the consensus layer of our security reference architecture (see \autoref{chapter:sra}).
In particular, this chapter is focused on Proof-of-Work (PoW) consensus protocols and is based on the papers \cite{strongchain, perevsini2023incentive, peresini2023incentive, perevsini2023dag,budinsky2023fee} (see also \autoref{sec:cons-papers}). 

First, we address the selfish mining attack by revising the Nakamoto Consensus protocol used in Bitcoin into a new design called StrongChain \cite{strongchain}, and we show that StrongChain mitigates this attack in our simulation experiments. 
Next, we demonstrate the existence of incentive attacks \cite{perevsini2023incentive,perevsini2023dag,perevsini2021dag} on several DAG-based PoW consensus protocols with random transaction selection by game theoretical analysis and simulation, while we also elaborate on a few mitigation techniques.\footnote{Note that we first analyzed this attack on PHANTOM and its optimization GHOSTDAG \cite{perevsini2021dag}, and later we generalized this attack and applied it to more concerned protocols \cite{perevsini2023incentive}.}
Finally, we address the problems of undercutting attacks and the mining gap in a sole transaction-fee-based regime of PoW blockchains by proposing fee-redistribution smart contracts \cite{budinsky2023fee} as a modification of the Nakamoto Consensus.
In the following, we briefly introduce these research directions and summarize our contributions, while in later sections we go into more detail.

\subsubsection{Selfish Mining Attack \& StrongChain}
Selfish mining \cite{eyal2014majority} is a strategy where the attacker holding a fraction $\alpha$ of total mining power can earn more than $\alpha$ of total rewards by occupying more than $\alpha$ of the total mined blocks. 
Selfish mining became a profitable strategy after reaching a certain mining power threshold by the attacker -- e.g., 33\% in Bitcoin.\footnote{Nevertheless, it can be arbitrarily low in the case sole transaction-fee-based regime.} 
The key principle of selfish mining is that the attacker keeps building her secret chain that poses a fork w.r.t., the honest chain and releases it when the honest chain starts to ``catch up'' with the attacker's chain, overriding the honest chain since miners accept the longer (attacker's) chain as the canonical one and continue to mine on it.
This causes honest miners to waste their work/resources on the chain that is abandoned.

Nakamoto consensus \cite{nakamoto2008bitcoin} is PoW consensus protocol that stands behind the Bitcoin -- the most successful cryptocurrency so far.  
However, despite its unprecedented success, Bitcoin suffers from many inefficiencies. 
In particular, Bitcoin's consensus mechanism has been proven to be incentive-incompatible (e.g.,  mostly by selfish mining), its high reward variance
causes centralization (by creating mining pools), and its hardcoded deflation raises questions about its long-term sustainability with regard to the mining gap and sole transaction-fee-based regime.  

Therefore, we revise the Bitcoin consensus mechanism by proposing StrongChain, a scheme that introduces transparency and incentivizes participants to collaborate
rather than to compete.  
The core design of StrongChain is to utilize the computing power aggregated on the blockchain which is invisible and ``wasted'' in Bitcoin by default. 
Introducing relatively easy, although important changes to Bitcoin's design enable us to improve many crucial aspects of
Bitcoin-like cryptocurrencies making them more secure, efficient, and profitable
for participants. 
Most importantly, StrongChain improves the threshold where the selfish mining strategy starts to be profitable by 10\% in contrast to Bitcoin.  
See further details on StrongChain in \autoref{sec:cons-strongchain}.

\subsubsection{Incentive Attacks on DAG-based Blockchains}


Blockchains inherently suffer from the processing throughput bottleneck, as consensus must be reached for each block within the chain.
One approach to solve this problem is to increase the block creation rate.
However, such an approach has drawbacks.
If blocks are not propagated through the network before a new block is created, a \textit{soft fork} might occur, in which two concurrent blocks reference the same parent block.
A soft fork is resolved in a short time by a fork-choice rule, and thus only one block is eventually accepted.
All transactions in an \textit{orphaned} (a.k.a., stale) block are discarded.
As a result, consensus nodes that created orphaned blocks wasted their resources and did not get rewarded. 

As a response to the above issue, several proposals (e.g., Inclusive \cite{lewenberg2015inclusive}, PHANTOM \cite{sompolinsky2020phantom}, GHOSTDAG \cite{sompolinsky2020phantom}, SPECTRE \cite{sompolinsky2016spectre}) have substituted a single chaining data structure for unstructured direct acyclic graph, 
while another proposal in this direction employed structured DAG (i.e., Prism \cite{bagaria2019prism}).
Such a structure can maintain multiple interconnected chains and thus theoretically increase processing throughput. The assumption of concerned \gls{dag}-oriented solutions is to abandon transaction selection purely based on the highest fees since this approach intuitively increases the probability that the same transaction is included in more than one block (hereafter \textit{transaction collision}).
Instead, these approaches use the random transaction selection (i.e., \gls{RTS})~\footnote{Note that \gls{RTS} involves certain randomness in transaction selection but does not necessarily equals to uniform random transaction selection (to be in line with the works utilizing Inclusive \cite{lewenberg2015inclusive}, such as PHANTOM, GHOSTDAG \cite{sompolinsky2020phantom}, SPECTRE \cite{sompolinsky2016spectre}, as well as the implementation of GHOSTDAG called Kaspa \cite{Kaspa}).} strategy to avoid transaction collisions.
Although the consequences of deviating from such a strategy might seem intuitive, no one has yet thoroughly analyzed the performance and robustness of concerned \gls{dag}-based approaches under incentive attacks aimed at transaction selection.

Therefore, we focus on the impact of incentive attacks caused by \textbf{greedy}\footnote{Greedy actors deviate from the protocol to increase their profits.} actors in above-mentioned \gls{dag}-oriented designs of consensus protocols. 
In particular, we study the situation where a greedy attacker deviates from the protocol by not following the \gls{RTS} strategy that is assumed by the mentioned DAG-based approaches.
%
We make a hypothesis stating that the attacker deviating from \gls{RTS} strategy might earn greater rewards as compared to honest participants, and such an attacker harms transaction throughput since transaction collision is increased.
We verify and prove our hypothesis in a game theoretical analysis and show that \gls{RTS} does not constitute Nash equilibrium.
%
Next,  we substantiate conclusions from game theoretical analysis by a few simulation experiments on the abstracted \textsc{DAG-protocol}, which confirm that a greedy actor who selects transactions based on the highest fee profits significantly more than honest miners following the \gls{RTS}.
In another experiment, we demonstrate that multiple greedy actors can significantly reduce the effective transaction throughput by increasing the transaction collision rate across parallel chains of DAGs.
Finally, we show that greedy actors have an incentive to form a mining pool to increase their relative profits, which degrades the decentralization of the concerned DAG-oriented designs. 
See further details on these incentive attacks together with potential mitigation techniques in \autoref{sec:cons-dags}.


\subsubsection{Undercutting Attacks \& Transaction-Fee-based Regime}
In Bitcoin and its numerous clones, the block reward is divided by two approx. every four years (i.e., after every 210k blocks), which will eventually result in a zero block reward around the year 2140 and thus \emph{a pure transaction-fee-based regime}.
There was only very little research made to investigate the properties of transaction-fee-based regimes, which motivated our work.

Before 2016, there was a belief that the dominant source of the miners' income does not impact the security of the blockchain.
However, Carlsten et al.~\cite{carlsten2016instability} pointed out the effects of the high variance of the miners' revenue per block caused by exponentially distributed block arrival time in the transaction-fee-based protocols.
The authors showed that \emph{undercutting} (i.e., forking) a wealthy block is a profitable strategy for a malicious miner. 
Nevertheless, literature~\cite{daian2019flash,mclaughlin2023large} showed that this attack is viable even in blockchains containing traditional block rewards due to front-running competition of arbitrage bots who are willing to extremely increase transaction fees to earn Maximum Extractable Value profits.

Therefore, we focus on mitigation of the undercutting attack in the transaction-fee-based regime of the single chain PoW blockchains -- i.e., blockchains that prefer availability over consistency within the CAP theorem and thus are designed to resolve forks often.
We also discuss related problems present (not only) in a transaction-fee-based regime.
In particular, we focus on minimizing the mining gap~\cite{carlsten2016instability,tsabary2018gap}, (i.e., the situation, where the immediate reward from transaction fees does not cover miners' expenditures) as well as balancing significant fluctuations in miners' revenue.
%
%
To mitigate these issues, we propose a solution that splits transaction fees from a mined block into two parts -- (1) an instant reward for the miner and (2) a deposit sent into one or more fee-redistribution smart contracts ($\mathcal{FRSC}$s).
At the same time, these $\mathcal{FRSC}$s reward the miner of a block with a certain fraction of the accumulated funds over a fixed number of blocks, thereby emulating the moving average on a portion of the transaction fees.
We evaluate our approach using various fractions of the transaction fees (split across the miner and $\mathcal{FRSC}$) and experiment with the various numbers and lengths of $\mathcal{FRSC}$s -- we demonstrate that usage of multiple $\mathcal{FRSC}$s of various lengths has the best advantages mitigating the problems we are addressing.
Finally, we perform a simulation demonstrating that the threshold of \textsc{Default-Compliant} miners who strictly do not execute undercutting attack is lowered from 66\% (as reported in \cite{carlsten2016instability}) to 30\% with our approach. 
See further details on undercutting attacks and our solution in \autoref{sec:cons-undercut}.

\medskip
\noindent
In the following section, we detail the individual directions of our research.

\renewcommand{\name}{StrongChain\xspace}%

\section{StrongChain}\label{sec:cons-strongchain}

\subsection{Bitcoin's Drawbacks}
There are multiple drawbacks of Bitcoin that undermine its security
promises and raise questions about its future.  Bitcoin has been proved to be
incentive-incompatible \cite{eyal2015miner,sapirshtein2016optimal,eyal2014majority,gapgame}.
Namely, in some circumstances, the miners' best strategy is to not announce their
found solutions immediately, but instead withhold them for some time period (e.g., selfish mining \cite{eyal2018majority}).
Another issue is that the increasing popularity of the system tends towards its
centralization.  Strong competition between miners resulted in a high reward variance,
thus to stabilize their revenue miners started grouping their computing power
by forming \textit{mining pools}.  Over time, mining pools have come to dominate the
computing power of the system, and although they are beneficial for miners,
large mining pools are risky for the system as they have multiple ways of
abusing the
protocol \cite{karame2012double,eyal2015miner,eyal2014majority,sapirshtein2016optimal}.
Also, a few researchers rigorously analyzed one of the impacts of Bitcoin's
deflation \cite{longrunBitcoin,carlsten2016instability,gapgame}. Their results
indicate that Bitcoin may be unsustainable in the long term, mainly due to
decreasing miners' rewards that will eventually stop at all. 
Besides that,
unusually for a transaction system, Bitcoin is designed to favor availability
over consistency. This choice was motivated by its open and permissionless
spirit, but in the case of inconsistencies (i.e., \textit{forks} in the
blockchain) the system can be slow to converge.

\subsection{Overview of Proposed Approach}
Motivated by these drawbacks, we propose \name, a simple yet powerful revision
of the Bitcoin consensus mechanism.  Our main intuition is to design
a system such that the mining process is more transparent and collaborative,
i.e., miners get better knowledge about the mining power of the system and they are incentivized to solve puzzles together rather than compete.  In order to achieve
it, in the heart of the \name's design we employ \textit{weak solutions}, i.e.,
puzzle solutions with a PoW that is significant yet insufficient for a standard
solution.  We design our system, such that a) weak solutions are part of the
consensus protocol, b) their finders are rewarded independently, and c) miners
have incentives to announce own solutions and append solutions of others
immediately.  
We show that with these
changes, the mining process is becoming more transparent, collaborative, secure,
efficient, and decentralized.  
Surprisingly, we also show how our approach can improve the freshness properties offered by Bitcoin. 

\subsection{Details}
In our
scheme, miners solve a puzzle as today but in addition to publishing solutions,
they exchange weak solutions too (i.e., almost-solved puzzles). The lucky miner
publishes her solution that embeds gathered weak solutions (pointing to the same
previous block) of other miners.  Such a published block better reflects the aggregated
PoW of a block, which in the case of a fork can indicate that more mining power
is focused on a given branch (i.e., actually it proves that more computing power
``believes'' that the given branch is correct).  Another crucial change is to
redesign the Bitcoin reward system, such that the finders of weak solutions are also
rewarded. Following lessons learned from mining pool attacks, instead of sharing
rewards among miners, our scheme rewards weak solutions proportionally to their
PoW contributed to a given block and all rewards are independent of other
solutions of the block.\footnote{Note, that this change requires a Bitcoin \textit{hard fork}.}

There are a few intuitions behind these design choices.  First, a selfish miner
finding a new block takes a high risk by keeping this block secret.  This is 
because blocks 
have a better granularity due to honest miners exchanging partial solutions and strengthening their
prospective block, which in the case of a fork would be stronger than the older
block kept secret (i.e., the block of the selfish miner).  Secondly, miners are
actually incentivized to collaborate by a) exchanging their weak solutions, and
b) by appending weak solutions submitted by other miners.  For the former case,
miners are rewarded whenever their solutions are appended, hence keeping them
secret can be unprofitable for them.  For the latter case, a miner appending
weak solutions of others only increases the strength of her potential block, and
moreover, appending these solutions does not negatively influence  the miner's
potential reward.  Finally, our approach comes with another benefit.
Proportional rewarding of weak solutions decreases the reward variance, thus
miners do not have to join large mining pools in order to stabilize their
revenue. This could lead to a higher decentralization of mining power on the
network.

\begin{algorithm}[!th]
	\caption{Pseudocode of \name.}
	\label{alg:all}
	
	\scriptsize
	
		\SetKwProg{func}{function}{}{}
		
		\func{mineBlock()}{
			$\mathit{weakHdrsTmp\leftarrow\emptyset}$\;
			\For{$nonce \in \{0,1,2,...\}$}{
				$hdr\leftarrow\mathit{createHeader(nonce)}$\;
				/* check if the header meets the strong target */\\
				$h_{tmp}\leftarrow H(hdr)$\;
				\If{$h_{tmp} < T_s$}{
					$\mathit{B\leftarrow createBlock(hdr,weakHdrsTmp,Txs)}$\;
					$\mathit{broadcast(B)}$\;
					\Return; /* signal to mine with the new block */\\
				}
				/* check if the header meets the weak target */\\
				\If{$h_{tmp} < T_w$}{
					$\mathit{weakHdrsTmp.add(hdr)}$\;
					$\mathit{broadcast(hdr)}$\;
				}
			}
		}

		\func{onRecvWeakHdr($\mathit{hdr})$}{
			$h_w\leftarrow H(hdr)$\;
			\textbf{assert}($T_s \leq h_w<T_w$ \textbf{and} \textit{validHeader(hdr)})\;
			\textbf{assert}($\mathit{hdr.PrevHash ==  H(lastBlock.hdr)}$) \;
			$\mathit{weakHdrsTmp.add(hdr)}$\;
		}

		\func{rewardBlock($B$)}{
			/* reward block finder with $R$ */\\
			reward($B.hdr.Coinbase, R + B.TxFees$)\; 
			$w\leftarrow \gamma*T_s/T_w$; /* reward weak headers proportionally */\\
			\For{$\mathit{hdr \in B.weakHdrSet}$}{
				reward($hdr.Coinbase, w * c * R$)\;
			}
		}

		\func{validateBlock($B$)}{
			\textbf{assert}($H(B.hdr)<T_s$ \textbf{and} \textit{validHeader(B.hdr)})\;
			\textbf{assert}($\mathit{B.hdr.PrevHash ==  H(lastBlock.hdr)}$) \;
			\textbf{assert}($\mathit{validTransactions(B)}$)\;
			\For{$\mathit{hdr \in B.weakHdrSet}$}{
				\textbf{assert}($T_s \leq H(hdr) <T_w$ \textbf{and} \textit{validHeader(hdr)})\;
				\textbf{assert}($\mathit{hdr.PrevHash ==  H(lastBlock.hdr)}$)\;
			}
		}

		\func{chainPoW($chain$)}{
			$sum\leftarrow 0$\;
			\For{$B \in chain$}{
				/* for each block compute its aggregated PoW */\\
				$\mathit{T_s \leftarrow B.hdr.Target}$\;
				$sum\leftarrow sum + T_{max}/T_{s}$\;
				\For{$\mathit{hdr \in B.weakHdrSet}$}{
					$sum\leftarrow sum + T_{max}/T_w$\;
				}
			}
			\Return $sum$\;
		}

		\func{getTimestamp($B$)}{
			$\mathit{sumT\leftarrow B.hdr.Timestamp}$\;
			$\mathit{sumW\leftarrow 1.0}$\;
			/* average timestamp by the aggregated PoW */\\
			$w\leftarrow T_{s}/T_w$\;
			\For{$\mathit{hdr \in B.weakHdrSet}$}{
				$sumT\leftarrow sumT + w*hdr.Timestamp$\;
				$sumW\leftarrow sumW + w$\;
			}
			\Return $sumT/sumW$\;
		}
\end{algorithm}

\subsubsection{Mining}
As in Bitcoin, in \name miners authenticate transactions by collecting them into
blocks whose headers are protected by a certain amount of PoW.  A simplified
description of a block mining procedure in \name is presented as the
\textit{mineBlock()} function in \autoref{alg:all}.  Namely, every miner tries
to solve a PoW puzzle by computing the hash function over a newly created
header.  The header is constantly being changed by modifying its nonce
field,\footnote{In fact, other fields can be modified too if needed.} until a
valid hash value is found.  Whenever a miner finds a header $hdr$ whose hash
value $h=H(hdr)$ is smaller than the \textit{strong target} $T_s$, i.e., a $h$
that satisfies the following:
\begin{equation*}
	h < T_s,
\end{equation*}
then the corresponding block is announced to the network and becomes, with all
its transactions and metadata, part of the blockchain.  We refer to headers of
included blocks as \textit{strong headers}.

One of the main differences with Bitcoin is that our mining protocol handles also
headers whose hash values do not meet the strong target $T_s$, but still are low
enough to prove a significant PoW. We call such a header a \textit{weak header}
and its hash value $h$ has to satisfy the following:
\begin{equation}
	\label{eq:weak}
	T_s \leq h < T_w,
\end{equation}
where $T_w > T_s$ and $T_w$ is called the \textit{weak target}.

Whenever a miner finds such a block header, she adds it to her local list of
weak headers (i.e., \textit{weakHdrsTmp}) and she propagates the header among
all miners.  Then every miner that receives this information first validates it
(see \textit{onRecvWeakHdr()}) by checking whether
\begin{compactitem}
	\item the header points to the last strong header,
	\item its other fields are correct (see \autoref{sec:details:block:layout_valid}),
	\item and \autoref{eq:weak} is satisfied.
\end{compactitem}
Afterward, miners append the header to their lists of weak headers.
We do not limit the number of weak headers appended, although this number is
correlated with the $T_w/T_s$ ratio.
Finally, miners continue the mining process in order to find a strong header.
In this process, a miner keeps creating candidate headers by computing hash values
and checking whether the strong target is met.  Every candidate header ``protects''
all collected weak headers (note that all of these weak headers point to the same previous
strong header).

In order to keep the number of found weak headers close to a constant value,
\name adjusts the difficulty $T_w$ of weak headers every 2016 blocks immediately
following the adjustment of the difficulty $T_s$ of the strong headers according
to \autoref{eqn:adjust_Ts}, such that the ratio $T_w / T_s$ is kept at a
constant.

\subsubsection{Block Layout and Validation}
\label{sec:details:block:layout_valid}
A block in our scheme consists of transactions, a list of weak headers, and a
strong header that authenticates these transactions and weak headers.  Strong
and weak headers in our system inherit the fields from Bitcoin headers and
additionally enrich it by a new field.  A block header consists of the following
fields:
\begin{compactitem}
	\item $\mathit{PrevHash}$\textnormal{:} is a hash of the previous block header, 
	\item $\mathit{Target}$\textnormal{:} is the value encoding the current target defining
	the difficulty of finding new blocks,
	\item $\mathit{Nonce}$\textnormal{:} is a nonce, used to generate PoW,
	\item $\mathit{Timestamp}$\textnormal{:} is a Unix timestamp,
	\item $\mathit{TxRoot}$\textnormal{:} is the root of the Merkle tree \cite{Merkle1988}
	aggregating all transactions of the block, and
	\item $\mathit{Coinbase}$\textnormal{:} represents an address of the miner that will
	receive a reward. 
\end{compactitem}
As our protocol rewards finders of weak headers (see details in
\autoref{sec:details:payments}), every weak header has to be accompanied with the
information necessary to identify its finder.  Otherwise, a finder of a strong
block could maliciously claim that some (or all) weak headers were found by her
and get rewards for them.  For this purpose and for efficiency, we introduced a
new 20B-long header field named \textit{Coinbase}.  With the
introduction of this field, \name headers are 100B long.  But on the other
hand, there is no longer any need for Bitcoin coinbase transactions. 

Weak headers are exchanged among nodes as part of a block, hence
it is necessary to protect the integrity of all weak headers associated with the
block. 
To realize it, we introduce a
special transaction, called a \textit{binding transaction}, which contains a hash
value computed over the weak headers.  This transaction is the first transaction
of each block and it protects the collected weak headers. 
Whenever a strong header is found, it is announced
together with all its transactions and collected weak headers, therefore, this
field  protects all associated weak headers. To encode this field we utilize
the \texttt{OP\_RETURN} operation as follows:
\begin{equation}
	\label{eq:weak_hdr_hash}
	\texttt{OP\_RETURN}\quad \mathit{H(hdr_0 \| hdr_1 \| ... \| hdr_n)},
\end{equation}
where $hdr_i$ is a weak header pointing to the previous strong header.  
Since  weak headers have redundant fields (the \mbox{\textit{PrevHash}}, \textit{Target}, and
\textit{Version} fields have the same values as the strong header), 
we propose to save bandwidth and storage by not including these fields into the data of a block. 
This modification reduces the size of a weak header from 100B to 60B only, which is especially important for SPV clients who keep downloading new block headers.

With our approach, a newly mined and announced block can encompass
multiple weak headers.  Weak headers, in contrast to strong headers, are not used to
authenticate transactions, and they are even stored and exchanged \textit{without} their
corresponding transactions. Instead, the main purpose of including weak headers
is to contribute and reflect the aggregated mining power concentrated on a given
branch of the blockchain.  We present a fragment of a blockchain of \name in
\autoref{fig:block}.  As depicted in the figure, each block contains a single strong
header, transactions, and a set of weak headers aggregated via a binding
transaction.

\begin{figure}[t!]
	\centering
	\includegraphics[width=0.6\columnwidth]{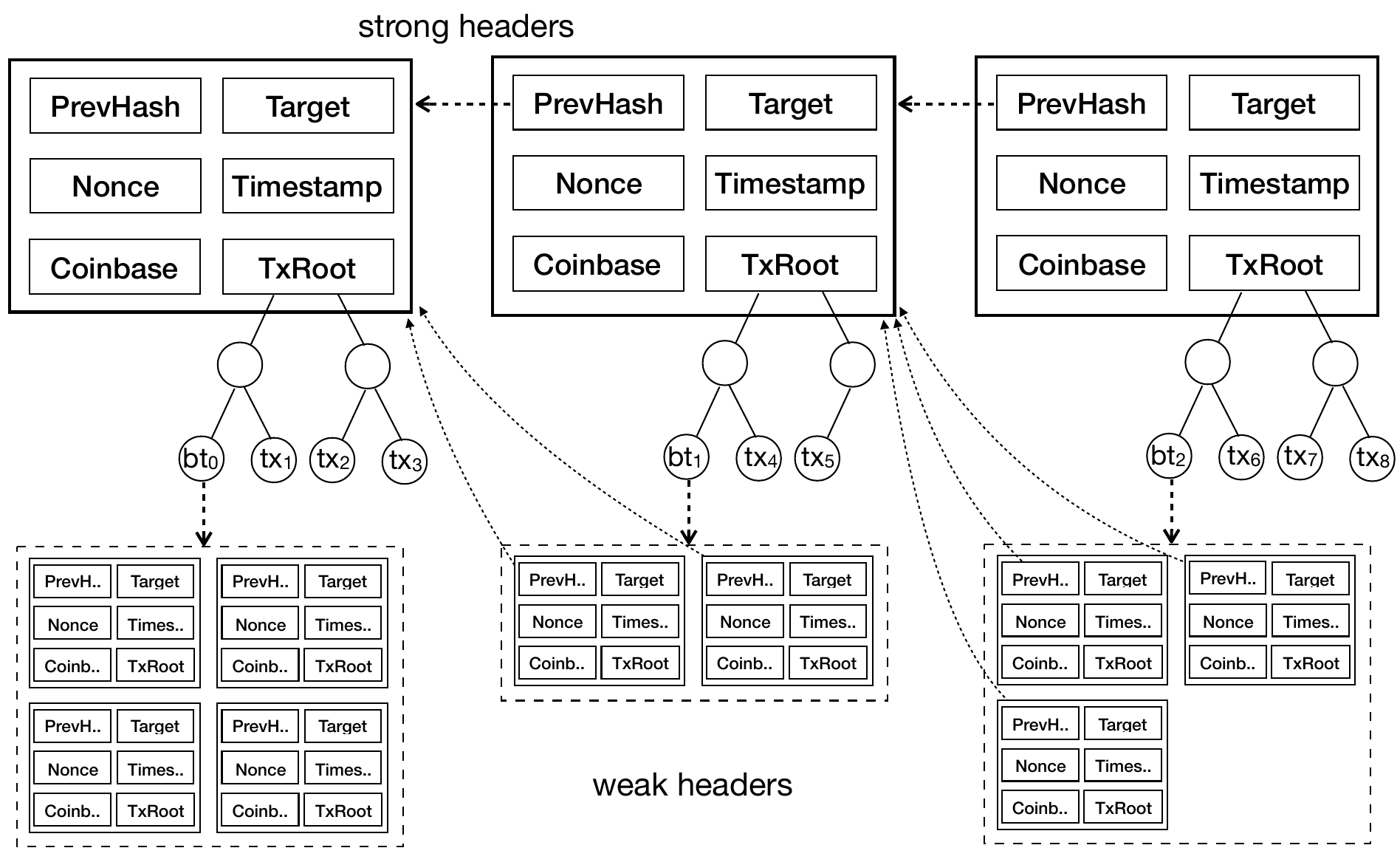}
	\caption{An example of a blockchain fragment with strong headers, weak
		headers, and binding and regular transactions.}
	\label{fig:block}
\end{figure}

On receiving a new block, miners validate the block by checking the following
(see \textit{validateBlock()} in \autoref{alg:all}):
\begin{compactenum}
	\item The strong header is protected by the PoW and points to the previous
	strong header.
	\item Header fields have correct values (i.e., the version, target, and
	timestamp are set correctly).
	\item All included transactions are correct and protected by the strong
	header. This check also includes checking that all weak headers
	collected are protected by a binding transaction included in the block.
	
	\item All included weak headers are correct:
		a) they meet the targets as specified in \autoref{eq:weak},
		b) their \textit{PrevHash} fields point to the
		previous strong header,
		and c) their version, targets, and timestamps have correct values.
\end{compactenum}
If the validation is successful, the block is accepted as part of the
blockchain.

\subsubsection{Forks}
\label{sec:details:forks}
One of the main advantages of our approach is that blocks reflect their aggregated mining
power  more precisely.  Each block beside its strong header contains
multiple weak headers that contribute to the block's PoW.  In the case of a
fork, our scheme relies on the strongest chain rule, however, the PoW is
computed differently than in Bitcoin.  For every chain its PoW is calculated as
presented by the \textit{chainPoW()} procedure in \autoref{alg:all}.  Every
chain is parsed and for each of its blocks the PoW is calculated by adding:
\begin{compactenum}
	\item the PoW of the strong header, computed as $T_{max}/T_s$, where $T_{max}$ is the maximum target value, and
	\item the accumulated PoW of all associated weak headers,
	counting each weak header equally as $T_{max}/T_w$.
\end{compactenum}
Then the chain's PoW is expressed as just the sum of all its blocks' PoW. Such
an aggregated chain's PoW is compared with the competing chain(s). The chain with
the largest aggregated PoW is determined as the current one.  As difficulty in
our protocol changes over time, the strong target $T_s$ and PoW of weak headers
are relative to the maximum target value $T_{max}$.  We assume that nodes of the
network check whether every difficulty window is computed correctly (we skipped
this check in our algorithms for easy description).

Including and empowering weak headers in our protocol moves away from
Bitcoin's ``binary'' granularity and gives blocks better expression of the PoW
they convey. An example is presented in \autoref{fig:chain}. For instance, nodes
having the blocks $B_i$ and $B'_i$ can immediately decide to follow the
block $B_i$ as it has more weak headers associated, thus it has accumulated more
PoW than the block $B'_i$.

An exception to this rule is when miners solve conflicts.  Namely, on
receiving a new block, miners run the algorithm as presented, however, they also
take into consideration PoW contributions of known weak headers that
point to the last blocks. For instance, for a one-block-long fork within the
same difficulty window, if a block $B$ includes $l$ weak headers and a miner
knows of $k$ weak headers pointing to $B$, then that miner will select $B$ over any competing block $B'$ that includes $l'$ weak and has $k'$ known weak headers pointing to it if $l+k>l'+k'$.  Note that this rule incentivizes miners to propagate
their solutions as quickly as possible as competing blocks become ``stronger'' over
time.

\begin{figure}[t!]
	\centering
	\includegraphics[width=.6\linewidth]{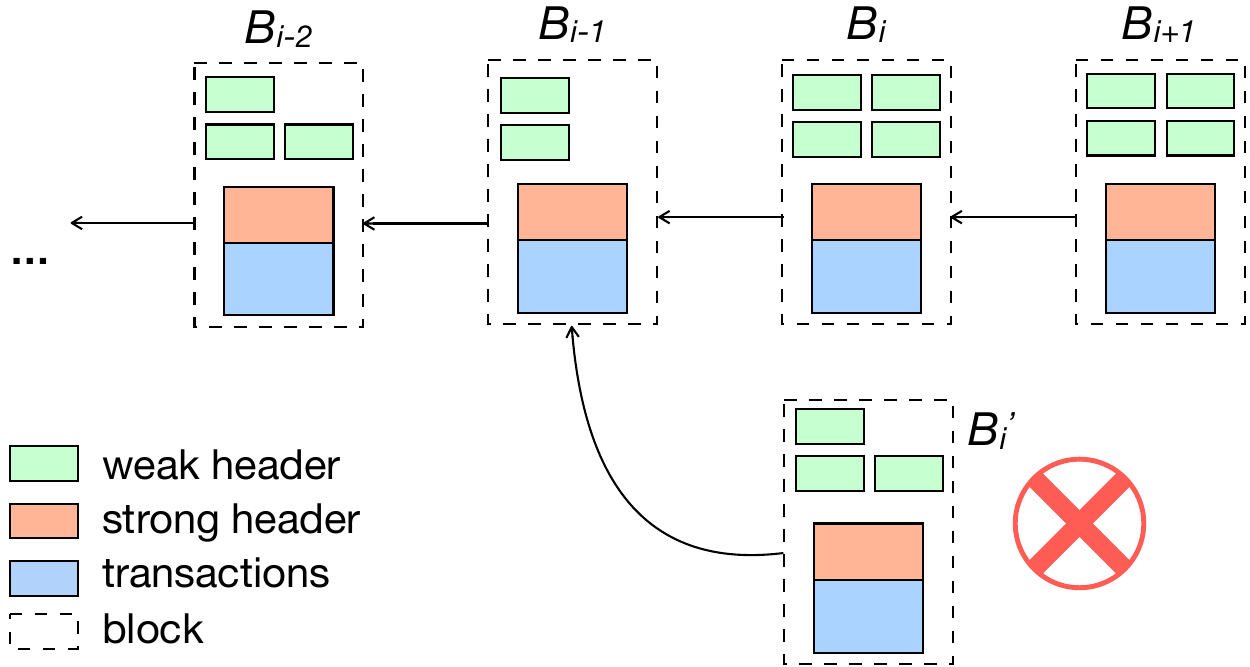}
	\caption{An example of a forked blockchain in \name.}
	\label{fig:chain}
\end{figure}

\subsubsection{Rewarding Scheme}
\label{sec:details:payments}
The rewards distribution is another crucial aspect of \name and it is presented
by the \textit{rewardBlock()} procedure from \autoref{alg:all}.  The miner that
found the strong header receives the full reward $R$. Moreover,  in contrast to
Bitcoin, where only the ``lucky'' miner is paid the full reward, in our scheme
all miners that have contributed to the block's PoW (i.e., whose weak headers
are included) are paid by commensurate rewards to the provided PoW.    A weak
header finder receive a fraction of $R$, i.e., $\gamma*c*R*T_s/T_w$, as a reward
for its corresponding solution contributing to the total PoW of a particular
branch, where the $\gamma$ parameter influences the relative impact of weak header rewards and $c$ is just a scaling constant.  
Moreover, we
do not limit weak header rewards and miners can get multiple rewards for their
weak headers within a single block.  Similar reward mechanisms are present in
today's mining pools, but unlike them, weak header
rewards in \name are independent of each other.  
Therefore, the reward scheme is
not a zero-sum game and miners cannot increase their own rewards by dropping
weak headers of others  -- they
can only lose since their potential strong blocks would have less aggregated PoW without others' weak headers.  Furthermore, weak header rewards decrease significantly the mining variance as miners can get steady revenue, making the system more decentralized and collaborative.
In \autoref{tab:pools}, we estimate a size reduction of the largest Bitcoin mining pools in 2019 with \name while maintaining the same reward variance.
In sum, \name can offer 75x-105x size reduction.
\begin{table}[t]
	\centering
	\footnotesize
	\begin{tabular}{lrrc}
		\toprule 
		\multirow{2}*{Mining Pool}            &  \multicolumn{2}{c}{Pool Size} & Size \\
		&  Bitcoin &   \name & Reduction\\
		\midrule
		BTC.com & 18.1\% & 0.245\%    & 74$\times$ \\
		F2Pool & 14.1\% & 0.172\%     & 82$\times$ \\
		AntPool & 11.7\% & 0.135\%    & 87$\times$ \\
		SlushPool & 9.1\% & 0.099\%   & 92$\times$ \\
		ViaBTC & 7.5\% & 0.079\%      & 95$\times$ \\
		BTC.TOP & 7.1\% & 0.074\%     & 96$\times$ \\
		BitClub & 3.1\% & 0.030\%     & 103$\times$ \\
		DPOOL & 2.6\% & 0.025\%       & 104$\times$ \\
		Bitcoin.com & 1.9\% & 0.018\% & 106$\times$ \\
		BitFury & 1.7\% & 0.016\%     & 106$\times$ \\
		\bottomrule
	\end{tabular}
	\caption{Largest Bitcoin mining pools and the corresponding pool sizes in
		StrongChain offering the same relative reward variance (\mbox{$T_w/T_s=1024$} and
		$\gamma=10$).}
	\label{tab:pools}
\end{table}

As mentioned before, the number of weak headers of a block is unlimited, they
are rewarded independently (i.e., do not share any reward), and all block
rewards in our system are proportional to the PoW contributed.  In such a
setting, a mechanism incentivizing miners to terminate a block creation is
needed (without such a mechanism, miners could keep creating huge blocks with weak
headers only).  In order to achieve this,  \name always attributes block
transaction fees ($B.TxFees$)  to the finder of the strong header (who also
receives the full reward $R$).

Note that in our rewarding scheme, the amount of newly minted coins is always at
least $R$, and consequently, unlike Bitcoin or Ethereum \cite{wood2014ethereum},
the total supply of the currency in our protocol is not upper-bounded. This
design decision is made in accordance with recent results on the long-term
instability of deflationary
cryptocurrencies \cite{longrunBitcoin,carlsten2016instability,gapgame}.

\subsubsection{Timestamps}
\label{sec:details:timestamps}
In \name, we follow the Bitcoin rules on constraining timestamps (see
\autoref{sec:pre:bitcoin}), however, we redefine how block timestamps are
interpreted.  Instead of solely relying on a timestamp put by the miner who
mined the block, block timestamps in our system are derived from the strong
header and all weak headers included in the corresponding block. The algorithm
to derive a block's timestamp is presented as \textit{getTimestamp()} in
\autoref{alg:all}.  A block's timestamp is determined as a weighted average
timestamp over the strong header's timestamp and all timestamps of the weak
headers included in the block.  The strong header's timestamp has a weight of
$1$, while weights of weak header timestamps are determined as their PoW
contributed (namely, a weak header's timestamp has a weight of the ratio between
the strong target and the weak target).  Therefore, the timestamp value is
adjusted proportionally to the mining power associated with a given block.  That
change reflects an average time of the block creation and mitigates miners that
intentionally or misconfigured put incorrect timestamps into the blockchain. 

\subsubsection{SPV Clients}
Our protocol supports light SPV clients. With every new block, an SPV client is
updated with the following information:
\begin{equation}
	\mathit{hdr, hdr_0, hdr_1, ..., hdr_n, BTproof},
\end{equation}
where \textit{hdr} is a strong header, $\mathit{hdr_i}$ are associated weak headers, and \textit{BTproof} is an inclusion proof of a binding transaction
that contains a hash over the weak headers (see \autoref{eq:weak_hdr_hash}).
Note that headers contain redundant fields, thus as described in
\autoref{sec:details:block:layout_valid}, they can be provided to SPV clients
efficiently.

With this data, the client verifies fields of all headers, computes the PoW of
the block (analogous, as in \textit{chainPoW()} from \autoref{alg:all}), and
validates the \textit{BTproof} proof to check whether all weak headers are
correct, and whether the transaction is part of the blockchain (the proof is
validated against \textit{TxRoot} of \textit{hdr}).  Afterward, the client saves the strong header
$\textit{hdr}$ and its computed PoW, while other messages (the weak headers and
the proof) can be dropped.

\subsection{Experiments}
We executed a few experiments in the original paper.
Nevertheless, for the purpose of this thesis we show only the results of experiments that demonstrate improved threshold of selfish mining profitability.
We will consider the selfish mining strategy of \cite{eyal2014majority}, described as follows:
\setdefaultleftmargin{1em}{2em}{}{}{}{}
\begin{itemize}
	\item The attacker does not propagate a newly found block until she finds at least a second block on top of it, and then only if the difference in difficulty between her chain and the strongest known alternative chain is between zero and $R$.
	\item The attacker adopts the strongest known alternative chain if its difficulty is at least greater than her own by $R$.
\end{itemize}
\setdefaultleftmargin{2em}{2em}{}{}{}{}

\begin{figure}[t!]
	\centering

	
	\begin{subfigure}[t]{0.432\textwidth}
		\centering
		
		\includegraphics[width=\textwidth, trim={2cm 0 1.7cm 0}, clip]{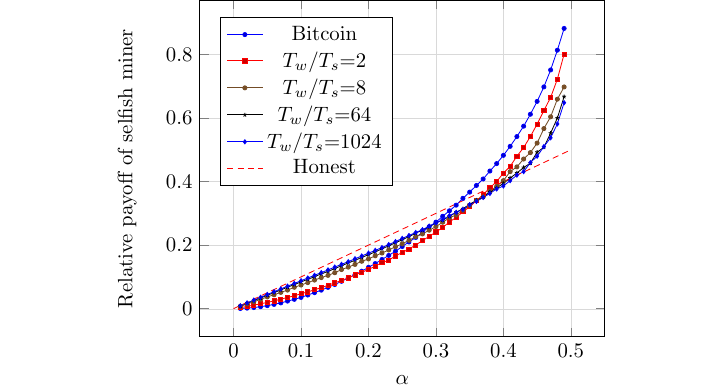}
		\caption{\emph{Relative} payoff of a \emph{selfish} miner following the strategy of \cite{eyal2014majority}, compared to an $(1-\alpha)$-strong honest miner.}\label{fig:selfish_payoffs1} 
	\end{subfigure}
	\hfill
	\begin{subfigure}[t]{0.432\textwidth}
		\centering
		
		\includegraphics[width=\textwidth, trim={2cm 0 1.7cm 0}, clip]{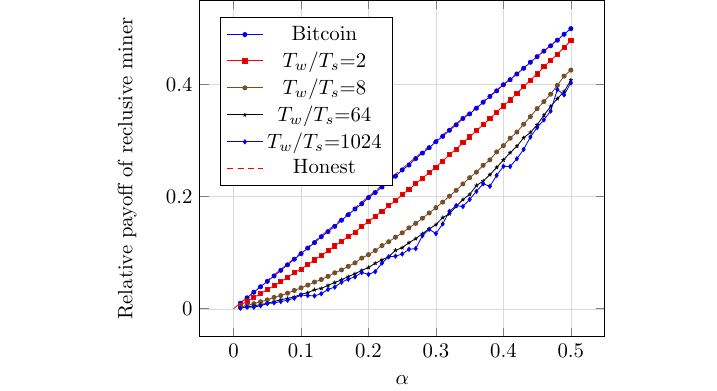}
		\caption{Relative payoff of a \emph{reclusive} miner who does not broadcast her weak blocks.}\label{fig:selfish_payoffs2} 
	\end{subfigure}
	\\
	\vspace{0.5cm}
	\begin{subfigure}[t]{0.432\textwidth}
		\centering
		
		\includegraphics[width=\textwidth, trim={2cm 0 1.7cm 0}, clip]{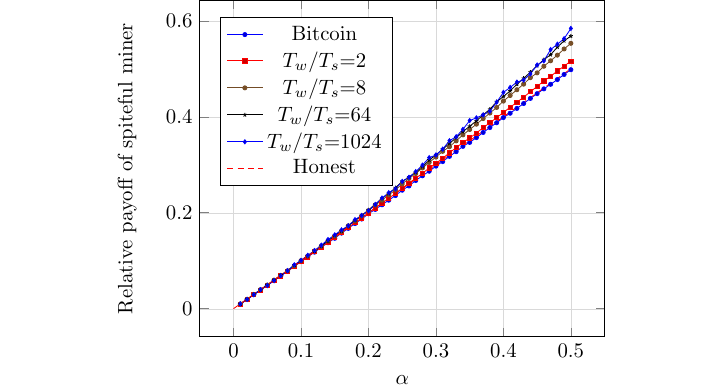}
		\caption{\emph{Relative} payoff (with respect to the rewards of all miners combined) of a \emph{spiteful} miner, who does not include other miners' weak blocks unless necessary. }\label{fig:selfish_payoffs3} 
	\end{subfigure}	
	\hfill	
	\begin{subfigure}[t]{0.432\textwidth}
		\centering
		
		\includegraphics[width=\textwidth, trim={2cm 0 1.7cm 0}, clip]{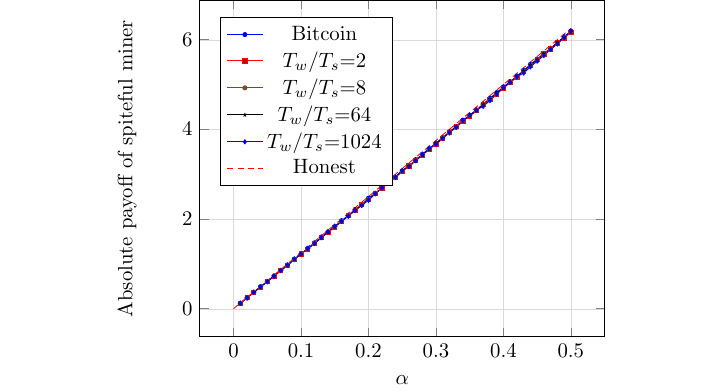}
		\caption{\emph{Absolute} payoff of a \emph{spiteful} miner, with 12.5 BTC on average awarded per block.}\label{fig:selfish_payoffs4} 
	\end{subfigure}	

	\vspace{0.5cm}
	\caption{Payoffs of an $\alpha$-strong adversarial miner for different strategies. We consider Bitcoin and \name with different choices of $T_w/T_s$, with \mbox{$\gamma = \log_2(T_w/T_s)$}.
	}\label{sec:cons-strong-exp1}
	
\end{figure}

In \autoref{fig:selfish_payoffs1}, we have depicted the profitability of this selfish mining strategy for different choices of $T_w/T_s$. As we can see, for $T_w/T_s = 1024$ the probability of being `ahead' after two strong blocks is so low that the strategy only begins to pay off when the attackers' mining power share is close to $43\%$ --- this is an improvement over Bitcoin, where the threshold is closer to $33\%$.

\name does introduce new adversarial strategies based on the mining of new weak headers.  
Some examples include not broadcasting any newly found weak blocks (``reclusive'' mining), refusing to include the weak headers of other miners (``spiteful'' mining), and postponing the publication of a new strong block and wasting the weak headers found by other miners in the meantime. In the former case, the attacker risks losing their weak blocks, whereas in both of the latter two cases, the attacker risks their strong block going stale as other blocks and weak headers are found. Hence, these are not cost-free strategies. Furthermore, because the number of weak headers does not affect the difficulty rescale, the attacker's motive for increasing the stale rate of other miners' weak headers is less obvious (although in the long run, an adversarial miner could push other miners out of the market entirely, thus affecting the difficulty rescale). 

In \autoref{fig:selfish_payoffs2}, we have displayed the relative payout (with respect to the total rewards) of a reclusive $\alpha$-strong miner --- this strategy does not pay for any $\alpha<0.5$. In~\autoref{fig:selfish_payoffs3}, we have depicted the relative payoff of a spiteful mine who does not include other miners' weak blocks unless necessary (i.e., unless others' weak blocks together contribute more than $R$ to the difficulty, which would mean that any single block found by the spiteful miner would always go stale). For low latencies (the graphs were generated with an average latency of 0.53 seconds), the strategy is almost risk-free, and the attacker does manage to hurt other miners more than herself, leading to an increased relative payout. However, as displayed in \autoref{fig:selfish_payoffs4}, there are no absolute gains, even mild losses. As mentioned earlier, the weak headers do not affect the difficulty rescale so there is no short-term incentive to engage in this behavior --- additionally there is little gain in computational overhead as the attacker still needs to process her own weak headers. 

%
%

\section{Incentive Attacks on DAG-Based Blockchains}\label{sec:cons-dags}

\subsection{Problem Definition \& Scope}\label{sec:dags-problem}
Let there be a PoW blockchain network that uses the Nakamoto consensus (NC) and consists of honest and greedy miners, with the greedy miners holding a fraction \gls{adversarial-mining-power} of the total mining power (i.e., adversarial mining power).
Then, we denote the network propagation delay in seconds as \gls{network-propagation-delay} and the block creation time in seconds as \gls{block-creation-rate}.
We assume that the minimum value of \gls{block-creation-rate} is constrained by \gls{network-propagation-delay} of the blockchain network.
%
%
It is well-known that Nakamoto-style blockchains generate stale blocks (a.k.a., orphan blocks).
As a result, a fraction of the mining power is wasted. The rate at which stale blocks are generated increases when \gls{block-creation-rate} is decreased, which is one of the reasons why Bitcoin maintains a high \gls{block-creation-rate} of 600s.

\subsubsection{\textbf{DAG-Oriented Designs}}
Many DAG-oriented designs were proposed to allow a decrease of \gls{block-creation-rate} while utilizing stale blocks in parallel chains, which should increase the transaction throughput.
Although there are some \gls{dag}-oriented designs that do not address the problem of increasing transaction throughput (e.g., IoTA \cite{silvano2020iota}, Nano \cite{Lemahieu2018NanoA}, Byteball \cite{Churyumov2016}),
we focus on the specific group of solutions addressing this problem, such as Inclusive \cite{lewenberg2015inclusive}, GHOSTDAG, PHANTOM \cite{sompolinsky2020phantom}, SPECTRE \cite{sompolinsky2016spectre}, and Prism \cite{bagaria2019prism}.
We are targeting the  \gls{RTS} strategy, which is a common property of this group of protocols.
In the \gls{RTS}, the miners do not take into account transaction fees of all included transactions; instead, they select transactions to blocks randomly -- although not necessarily uniformly at random (e.g., \cite{Kaspa}).
In this way, these designs aim to eliminate transaction collision within parallel blocks of the \gls{dag} structure.
Nevertheless, the interpretation of randomness in \gls{RTS} is not enforced/verified by these designs, and miners are trusted to ignore fees of all (or the majority (e.g., \cite{Kaspa}) of) transactions for the common ``well-being'' of the protocol.
%
Contrary, miners of blockchains such as Bitcoin use a well-known transaction selection mechanism that maximizes profit by selecting transactions of the block based on the highest fees -- we refer to this strategy as the \textit{greedy strategy} in this work. 

%



\subsubsection{\textbf{Assumptions}}\label{sec:attack-assumptions}
We assume a generic DAG-oriented consensus protocol using the \gls{RTS} strategy (denoted as \textsc{DAG-Protocol}).
Then, we assume that the incentive scheme of \textsc{DAG-Protocol} relies on transaction fees (but additionally might also rely on block rewards),\footnote{Note that block rewards would not change the applicability of our incentive attacks, and the constraints defined in the game theoretic model (see \autoref{sec:game-theory:model-analysis}) would remain met even with them.} and transactions are of the same size.\footnote{Note that this assumption serves only for simplification of the follow-up sections. Transactions of different sizes would require normalizing fees by the sizes of transactions to obtain an equivalent setup (i.e., a fee per Byte).}
Let us assume that the greedy miners may only choose a different transaction selection strategy to make more profit than honest miners.
Then, we assume that \textsc{DAG-Protocol} uses rewarding where the miner of the block \gls{phantom-block} gets rewarded for all unique not-yet-mined transactions in \gls{phantom-block} (while she is not rewarded for transaction duplicates mined before).

\subsubsection{\textbf{Identified Problems -- Incentive Attacks}}
Although the assumptions stated above might seem intuitive, there is no related work studying the impact of greedy miners deviating from the  \gls{RTS} strategy on any of the considered \textsc{DAG-protocol}s (\cite{sompolinsky2020phantom},\cite{sompolinsky2016spectre},\cite{lewenberg2015inclusive},\cite{bagaria2019prism})
and the effect it might have on the throughput of these protocols as well as a fair distribution of earned rewards. 
Note that we assume GHOSTDAG, PHANTOM, and SPECTRE are utilizing the \gls{RTS} strategy that was proposed in the Inclusive protocol \cite{sompolinsky2016spectre}, as recommended by the (partially overlapping) authors of these works -- this is further substantiated by the practical implementation of GHOSTDAG/PHANTOM called Kaspa \cite{Kaspa}, which utilizes a variant of \gls{RTS} strategy that selects a majority portion of transactions in a block uniformly random, while a small portion of the block capacity is seized by the transaction selected based on the highest fees.
Nevertheless, besides potentially increased transaction collision rate, even such an approach enables more greedy behavior.
We make a hypothesis for our incentive attacks:
\begin{hypothesis}\label{hypo:problem-definition}
	A greedy transaction selection strategy will decrease the relative profit of honest miners as well as transaction throughput in the \textsc{DAG-Protocol}.\footnote{Note that the greedy transaction selection strategy deviates from the \textsc{DAG-protocol} and thus is considered adversarial.}
	
\end{hypothesis}

\subsection{Game Theoretical Analysis}
\label{sec:gametheory}

In this section, we model a \textsc{DAG-protocol}\footnote{Note that we consider DAG-based designs under this generic term of \textsc{DAG-protocols} to simplify the description but not to claim that all \textsc{DAG-protocols} (with \gls{RTS}) can be modeled as we do.}
as a two-player game, in which the honest player/phenotype ($\phon$) uses the \gls{RTS} strategy and the greedy player/phenotype ($\pmal$) uses the greedy transaction selection strategy.
We assume that the fees of transactions vary -- the particular variance of fees is agnostic to this analysis.
We present the game theoretical approach widely used to analyze interactions of players (i.e., consensus nodes) in the blockchain.
Several works attempted to study the outcomes of different scenarios in blockchain networks (e.g., \cite{2019-Ziyao-survey-game,2020-Wang-game-mining,singh2020game}) but none of them addressed the case of \textsc{DAG-protocols} and their transaction selection. 
%
%
\noindent In game theoretic terms, we examine the following hypothesis:
\begin{hypothesis}\label{hypo:game1}
	So-called (honest) \gls{honest}-behavior with \gls{RTS} is a Subgame Perfect Nash Equilibrium (\gls{SPNE}) in an infinitely repeated \textsc{DAG-Protocol} game. 
	This was presented in Inclusive \cite{lewenberg2015inclusive} and we will contradict it. 
\end{hypothesis}

\subsubsection{Model of the \textsc{DAG-Protocol}} 
%
%
Players in \textsc{DAG-Protocol} receive transaction fees after a delay.
To simplify analysis, we can divide the flow of transactions into rounds of the game.
This allows us to study player behavior within defined time.
In each round, players make decisions and receive payoffs.
Since no round is explicitly marked as the last one, this game is repeated infinitely.

We model \textsc{DAG-Protocol} in the form of \textit{an infinitely repeated two players game with a base game}
\begin{equation}
	\Gamma = (\{\phon,\pmal\};\{\gls{honest},\gls{greedy}\};U_{hon},U_{mal}),
	\label{def-basegame}
\end{equation}
where $\phon$ is the player's determination to play \gls{honest} strategy and $\pmal$ the player's determination to the \gls{greedy}-behavior.
Pure strategy \gls{honest} is interpreted as the \gls{RTS}, while \gls{greedy} strategy represents picking the transactions with the highest fees.
Payoff functions are depicted in \autoref{tab:game-base},
where the profits in the strategic profiles $(\gls{honest},\gls{honest})$ and $(\gls{greedy},\gls{greedy})$ are uniformly distributed between players.
In the following, we analyze the model in five possible scenarios with generic levels $a, b, c, d$ of the payoffs.
%

\subsubsection{Analysis of the Model}\label{sec:game-theory:model-analysis}

For purposes of our analysis, lets start with the assumption that \gls{greedy}-behavior is more attractive and profitable than \gls{honest}-behavior.
Otherwise, there would be no reason to investigate Hypothesis~\autoref{hypo:game1}.
Thus, let us consider $c>a$ as the basic constraint.
We also assume $c>b$, meaning that \gls{honest}-behavior loses against \gls{greedy}-behavior in the cases of $(\gls{honest},\gls{greedy})$ and  $(\gls{greedy},\gls{honest})$ profiles.
%
These basic constraints yield the following scenarios:
\begin{compactitem}
	\item \textbf{Scenario 1}: $d>c>a>b$,
	\item \textbf{Scenario 2}: $c>d>a>b$,
	\item \textbf{Scenario 3}: $c>a>d>b$,
	\item \textbf{Scenario 4}: $c>a>b>d$,
	\item \textbf{Scenario 5}: where $a=d$ and $c>a$, $c>b$.
\end{compactitem}
Note that we do not assume the case $a=b$ since the presence of $\pmal$ will drain all high-fee transactions that $\phon$ would originally obtain. 
%

\noindent The following provides a high-level summary of the scenarios. For a more comprehensive analysis, we refer the refer to the full version of our paper \cite{peresini2023incentive}.
\begin{itemize}
	\item \underline{Scenarios 1} and \underline{2} are covered just for a sake of completeness.
	If the transaction fees were to cause such game outcomes, there would be no need to trust in \gls{honest}-behavior, and the system would settle in the unique $(\gls{greedy},\gls{greedy})$ Pure Nash Equilibrium (\gls{PNE}).
	\smallskip

	\item \underline{Scenario 3A} \textbf{Purely Non-Cooperative Interpretation}.
	In Scenario 3, both players ($\phon$ and $\pmal$) are incentivized to choose the greedy \gls{greedy} strategy, even though this leads to a worse overall outcome for both of them.
	This is because each player can do better by betraying the other player than by cooperating.
	This situation is known as a \textit{Prisoner's dilemma} \cite{Osborne1994}.
	\begin{proof}
		(Informal) Strategy $\gls{greedy}$ strictly dominates $\gls{honest}$ and thus $(\gls{greedy},$  $\gls{gree\-dy})$ is the unique \gls{PNE}.
	\end{proof}
	\begin{corollary}
		If $\phon$ is willing to follow the social norm of using the DAG protocol, then $\pmal$'s best response is also to use the $\gls{greedy}$ strategy.
		This is because $\phon$'s cooperation is not credible, and $\pmal$ can always benefit from betraying $\phon$.
	\end{corollary}
	\smallskip

	\item \underline{Scenario 3B} \textbf{When Some Coordination is Allowed}.
	It is possible for players to coordinate their behavior and achieve a better outcome for both of them, both playing $\gls{honest}$ strategy.
	It must be common knowledge to the players that $\phon$ uses \textit{grim trigger strategy} \cite{Osborne1994, Mailath}.
	This means that $\phon$ will cooperate as long as $\pmal$ cooperates.
	However, if $\pmal$ defects even once (playing $\gls{greedy}$), then $\phon$ will switch to the $\gls{greedy}$ strategy forever.
	$\pmal$ must also have a high discount factor.
	This means that she must value future payoffs more than immediate payoffs.
	If $\pmal$'s discount factor is too low, then she will be tempted to defect even if she knows it will lead to punishment in the long run.
	\smallskip
	
	%
	
	\begin{table}[t]
		\centering{\bazka{a}{b}{c}{d}}
		\caption{The utility functions $U_{hon},U_{mal}$ in the \textit{base game.}} 
	\label{tab:game-base}
	\vspace{-0.3cm}
\end{table}

\item \underline{Scenario 4A} \textbf{Purely Non-Cooperative Interpretation}. We choose utility functions: $a=2$, $b=1$, $c=3$ and $d=0$.
This scenario is an anti-coordination game \cite{Osborne1994} instance, so the game has two \gls{PNE}s $(\gls{honest},\gls{greedy})$ \& $(\gls{greedy},\gls{honest})$, and one Mixed Nash Equilibrium (\gls{MNE}) in mixed strategic profile $\bracks{\halfhalf}{\halfhalf}$.
\begin{claim}
	The most reasonable behavior in Scenario 4 is to play $\halfhalf$ for both players.
\end{claim}
\begin{proof}
	(Informal) Both players have two equally good choices: either be \gls{honest} or be \gls{greedy}.
	From $\phon$'s perspective, mixed behavior $\halfhalf$  guarantees the best stable outcome.
	If $\pmal$ expects $\halfhalf$ behavior from $\phon$, then $\pmal$'s best response is to play the same mixed behavior that establishes \gls{MNE}. 
	The players gain $(\frac{3}{2}, \frac{3}{2})$ in that \gls{MNE}, which
	is the highest expectation they can obtain.
\end{proof}
\smallskip
Therefore, the most reasonable behavior for both players is to play a mixed strategy where they are half-\gls{honest} and half-\gls{greedy}.
\smallskip

\item \underline{Scenario 4B} \textbf{When Some Coordination is Allowed}.
Similarly to Scenario 3, it is possible for players to coordinate their behavior and agree to always be $\gls{honest}$.
This would be a good outcome for both players, as they would both get a payoff of 2.
The same principle and consequences apply as in scenario 3(B) (Grimm trigger strategy).
This will make the $\pmal$ player regret defecting, and it will make her more likely to cooperate in the future.
Therefore, the conclusion from Scenario 3 applies here as well.
\smallskip

\item \underline{Scenario 5A} \textbf{Purely Non-Cooperative Interpretation}.
In this scenario, the game is a zero-sum game, which means that no player can gain more than 100\% profit, regardless of their chosen strategy.
This is because the sum of all incoming transaction fees is fixed in any set of rounds.
As a result, the total profit for all players is always constant if they all play the honest or greedy strategy.
Therefore, the only rational outcome of this scenario is for both players to play the $\gls{greedy}$ strategy.
\smallskip

If we consider a social norm, it may be tempting to appeal to players' sense of responsibility and ask them to refrain from playing the $\gls{greedy}$ strategy.
However, this is unlikely to be effective, as the $\gls{honest}$ strategy does not benefit either player.
Scenario 5 is highly similar to the classic game-theoretical model called \textit{The Tragedy of Commons} \cite{miller2003game}.
In this model, individuals are incentivized to use a shared resource to the maximum extent possible, even if this depletes the resource and harms the group as a whole.
In anonymous environments, where individuals cannot be held accountable for their actions, it is even more likely that they will prioritize their own interests over the interests of the group.
This is because they know that they will not be punished for acting in their self-interest, meaning there is no harm to play $\gls{greedy}$ strategy.
\end{itemize}

\subsubsection{Summary}
We conclude that Hypothesis \autoref{hypo:game1} is not valid. The $(\gls{honest},\gls{honest})$ profile is not a \gls{PNE} in any of our scenarios. Incentives enforcing \gls{honest}-behavior are hardly feasible in the anonymous (permissionless) environment of blockchains.
A community of honest miners can follow the \textsc{DAG-Protocol} until the attacker appears. The attacker playing the \gls{greedy} strategy can parasite on the system and there is no defense against such a behavior (since greedy miners can leave the system anytime and mine elsewhere, which is not assumed in \cite{lewenberg2015inclusive}). 
Therefore, \gls{honest} is not an \emph{evolutionary stable strategy} \cite{smith_1982}, and thus \gls{honest} does not constitute a stable equilibrium.
For more details about game theoretical analysis, we refer the reader to the full extended version of our paper \cite{perevsini2023incentive}, which is not yet published.

\subsection{Simulation Model}\label{sec:cons-model}

We created a simulation model to conduct various experiments investigating the behavior of \textsc{DAG-Protocol} under incentive attacks related to the problems identified in \autoref{sec:dags-problem} and thus Hypothesis~\autoref{hypo:problem-definition}.
Some experiments were designed to provide empirical evidence for 
the conclusions from \autoref{sec:gametheory}.
%

\subsubsection{\textbf{Abstraction of \textsc{DAG-Protocol}}}
For evaluation purposes, we simulated the \textsc{DAG-protocol} (with \gls{RTS}) by modeling the following aspects:
\begin{compactitem}
	\item  All blocks in DAG are deterministically ordered. 
	\item  The mining rewards consist of transaction fees only.
	\item  A fee of a particular transaction is awarded only to a miner of the block that includes the transaction as the first one in the sequence of totally ordered blocks. 
\end{compactitem}
Also, in terms of PHANTOM/GHOSTDAG terminology, we generalize and do not reduce transaction fees concerning the delay from ``appearing'' of the block until it is strongly connected to the DAG.
Hence, we utilize \gls{discount-function} = 1. 
In other words, for each block \gls{phantom-block}, the discount function does not penalize a block according to its gap parameter \(\gls{gap-parameter}(\gls{phantom-block})\), i.e. \(\gls{discount-function}(\gls{gap-parameter}(\gls{phantom-block})) = 1\).
Such a setting is optimistic for honest miners and maximizes their profits from transaction fees when following the \gls{RTS} strategy.
This abstraction enables us to model the concerned problems of considered \textsc{DAG-Protocols}.

\subsubsection{(Simple) Network Topology}\label{sec:network-topology-simple}

We created a simple network topology that is convenient for proof-of-concept simulations and encompasses some important aspects of the real-world blockchain network.
In particular, we were interested in emulating the network propagation delay \gls{network-propagation-delay} to be similar to Bitcoin (i.e., $\sim5s$ at most of the time in 2022), but using a small ring topology.
To create such a topology, we assumed that the Bitcoin network contains \(7592\) nodes, according to the snapshot of reachable Bitcoin nodes found on May 24, 2022.\footnote{\url{https://bitnodes.io/nodes/}}
In Bitcoin core, the default value of the consensus node's peers is set to  8 (i.e., the node degree).\footnote{Nevertheless, the node degree is often higher than~8 in reality \cite{mivsic2019modeling}.} 
%
Therefore, the maximum number of hops that a gossiped message requires to reach all consensus nodes in the network is $\sim4.29$ (i.e., $log_8(7592)$).
Moreover, if we were to assume $2-3x$ more independent blockchain clients (that are not consensus nodes), then this number would be increased to $4.83$--$4.96$.
%
To model this environment, we used the ring network topology with 10 consensus nodes,
which sets the maximum value of hops required to propagate a message to~$5$.
Next, we set the inter-node propagation delay $\partial \tau$ to $1s$, which fits assumed \gls{network-propagation-delay} (i.e., 5s / 5 hops = 1s).
%



\subsubsection{Simulator}
There are simulators \cite{paulavivcius2021systematic} that model blockchain protocols, mainly focusing on network delays, different consensus protocols, and behaviors of specific attacks (e.g., SimBlock \cite{Aoki2019SimBlockAB}, Blocksim \cite{BlockSim:Alharby}, Bitcoin-Simulator \cite{sim_bitcoin-simulator}).
However, none of these simulators was sufficient for our purposes due to missing support  for multiple chains  and incentive schemes assumed in \textsc{DAG-protocols}.
To verify Hypothesis~\autoref{hypo:problem-definition}, we built a simulator that focuses on the mentioned problems of \textsc{DAG-protocols}.
In detail, we started with the Bitcoin mining simulator \cite{gavinsimulator}, which is a discrete event simulator for the PoW mining on a single chain, enabling a simulation of network propagation delay within a specified network topology.
%
%
We extended this simulator to support \textsc{DAG-Protocol}s, enabling us to monitor transaction duplicity, throughput, and relative profits of miners with regard to their mining power.
The simulator is written in \verb!C++! (see details and its evaluation in \cite{perevsini2023sword}.
%
%
%
In addition, we added more simulation complexity to simulate each block, including the particular transactions (as opposed to simulating only the number of transactions in a block \cite{gavinsimulator}). 
Most importantly, we implemented two different transaction selection strategies -- greedy and random.
For demonstration purposes, we implemented the exponential distribution of transaction fees in mempool, based on several graph cuts of fee distributions in mempool of Bitcoin from \cite{bitcoin-mempool-stats}.\footnote{Distribution of transaction fees in mempool might change over time; however, it mostly preserves the low number of high-fee transactions.} 
Our simulator is available at \url{https://github.com/Tem12/DAG-simulator}.

\subsection{Evaluation}
\label{sec:evaluation-dags}


We designed a few experiments with our simulator, which were aimed at investigating the relative profit of greedy miners and transaction collision rate (thus throughput) to investigate Hypothesis~\autoref{hypo:problem-definition}.
In all experiments, honest miners followed the \gls{RTS},
while greedy miners followed the greedy strategy.
Unless stated otherwise, the block creation time was set to \(\gls{block-creation-rate} = 20s\). 
However, we abstracted from \gls{network-propagation-delay} of transactions and ensured that the mempools of nodes were regularly filled (i.e., every 60s) by the same set of new transactions, while the number of transactions in the mempool was always sufficient to fully satisfy the block capacity that was set to 100 transactions.
We set the size of mempool equal to 10000 transactions, and thus the ratio between these two values is similar to Bitcoin \cite{bitcoin-mempool-stats} in common situations.
In all experiments, we executed multiple runs and consolidated their results; however, in all experiments with the simple topology, the spread was negligible, and therefore we do not depict it in graphs. 





\subsubsection{Experiment I}\label{sec:experiment-1}
\paragraph{Goal.}
The goal of this experiment was to compare the relative profits earned by two miners/phenotypes in a network, corresponding to our game theoretical settings (see \autoref{sec:gametheory}). 
Thus, one miner was greedy and followed the greedy strategy, while the other one was honest and followed the \gls{RTS}.

\paragraph{\textbf{Methodology and Results.}}
The ratio of total mining power between the two miners was varied with a granularity of \(10\%\), and the network consisted of 10 miners, where only the two miners had assigned the mining power.
Other miners acted as relays, emulating the maximal network delay of 5 hops between the two miners in a duel.
The relative profits of the miners were monitored as their profit factor \(\mathbb{P}\) w.r.t. their mining power.
We conducted 10 simulation runs and averaged their results (see \autoref{fig:duel-profit}).
Results show that the greedy miner earned a profit disproportionately higher than her mining power, while the honest miner's relative profit was negatively affected by the presence of the greedy miner.
We can observe that \(\mathbb{P}\)  of greedy miner was indirectly proportional to her \gls{adversarial-mining-power}, which was caused by the exponential distribution of transaction fees that contributed more significantly to the higher \(\mathbb{P}\) of a smaller miner.
In sum, the profit advantage of the greedy miner aligns with the conclusions from the game theoretical model (Scenario 5, see \autoref{sec:gametheory}) in particular, which represents the case of \gls{adversarial-mining-power}=50\%.
Nevertheless, our results indicate that the greedy strategy is more profitable than the \gls{RTS} for any non-zero~\gls{adversarial-mining-power}.

\begin{figure}[t]
	\centering
	\includegraphics[width=0.5\linewidth]{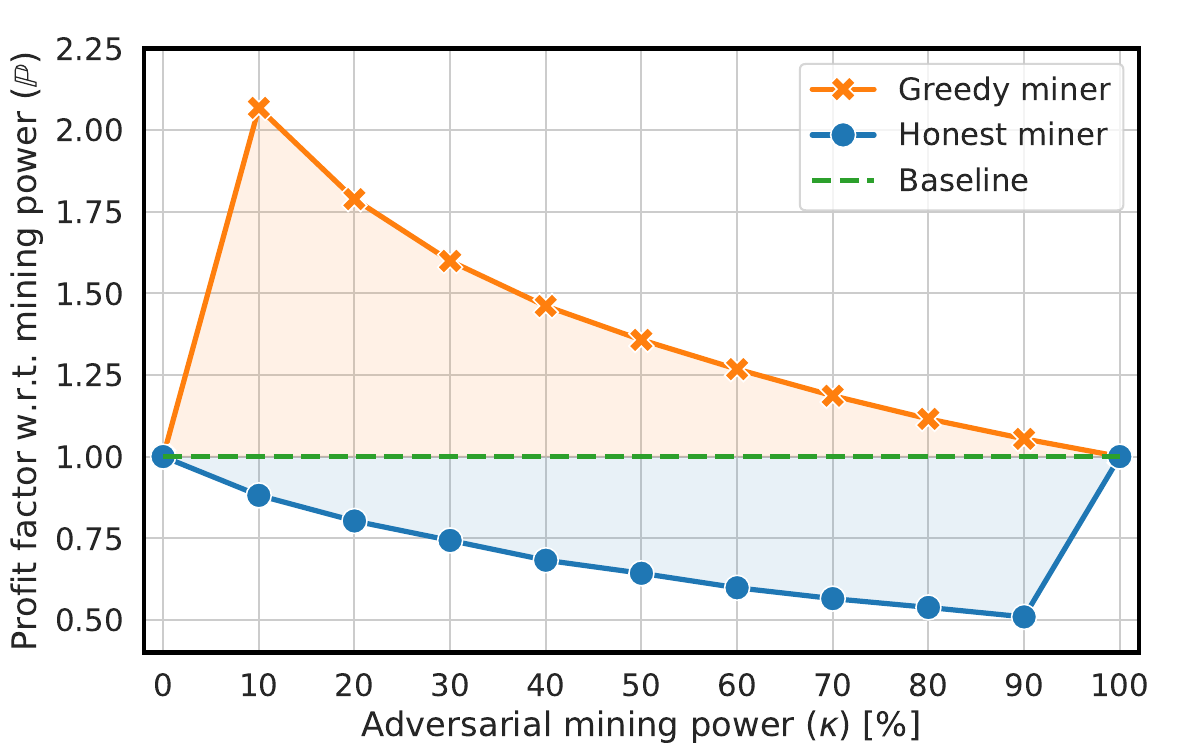}
	\caption{The profit factor $\mathbb{P}$ of an honest vs. a greedy miner with their mining powers of 100\% - \gls{adversarial-mining-power} and \gls{adversarial-mining-power}, respectively.
		The baseline shows the expected $\mathbb{P}$ of the honest miner; \(\gls{block-creation-rate}=20s\).}
	\label{fig:duel-profit}
\end{figure}




\subsubsection{Experiment II}\label{sec:experiment-2}
\paragraph{\textbf{Goal.}}
The goal of this experiment was investigation of the relative profits of a few greedy miners following the greedy strategy in contrast to honest miners following the \gls{RTS}.

\paragraph{\textbf{Methodology and Results.}}
We experimented with 10 miners, where
the number of greedy miners $\gls{cnt-malicious-miners}$ vs. the number of honest miners (i.e., 10 - \gls{cnt-malicious-miners}) was varied, and each held \(10\%\) of the total mining power.
We monitored their profit factor $\overline{\mathbb{P}}$ averaged per miner. 
%
We conducted 10 simulation runs and averaged their results (see \autoref{fig:malicious-miners-earn-more-profit}).
%
%
Alike in \autoref{sec:experiment-1}, we can see that greedy miners earned profit disproportionately higher than their mining power.
Similarly, this experiment showed that the profit advantage of greedy miners decreases as their number increases.
This is similar to increasing \gls{adversarial-mining-power} in a duel of two miners from \autoref{sec:experiment-1}; however, in contrast to it, $\overline{\mathbb{P}}$ of greedy miners is slightly lower with the same total \gls{adversarial-mining-power} of all greedy miners, while  $\overline{\mathbb{P}}$ of honest miners had not suffered with such a decrease.
%
Intuitively, this happened because multiple greedy miners increase transaction collision.
In detail, since miners are only rewarded for transactions that were first to be included in a new block, the profit for the second and later miners is lost if a duplicate transaction is included.
This observation might be seen as beneficial for the protocol as it disincentivizes multiple miners to use the greedy transaction selection strategy, which would support the sequential equilibrium from \cite{lewenberg2015inclusive}.
However, the authors of \cite{lewenberg2015inclusive} do not assume cooperating players, which is unrealistic since miners can cooperate and create the pool to avoid collisions and thus maximize their profits (resulting in a similar outcome, as in \autoref{sec:experiment-1}).

\begin{figure*}[t]
	\centering
	\begin{subfigure}[t]{0.49\textwidth}
		\centering
		\includegraphics[width=\textwidth]{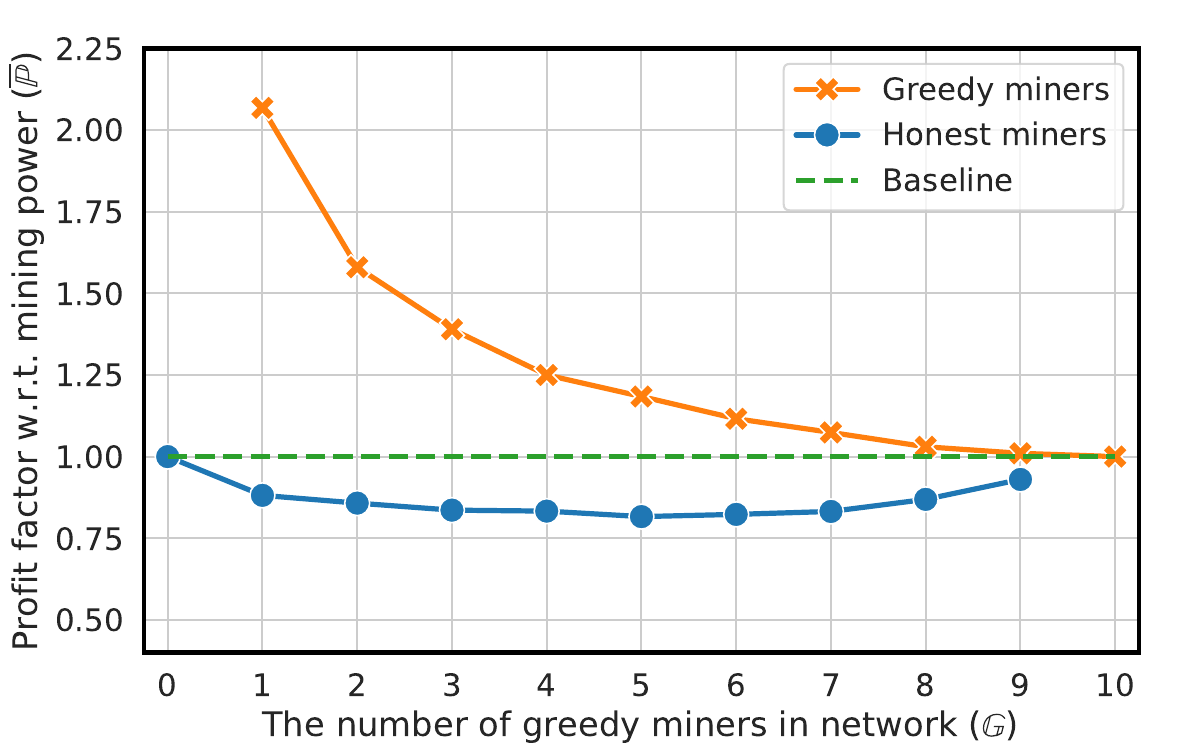}
		\caption{The averaged profit factor $\overline{\mathbb{P}}$ per honest miner and greedy miner, each with \(10\%\) of mining power.
			The number of honest miners is $10$ - \gls{cnt-malicious-miners}.
			The baseline shows the expected $\overline{ \mathbb{P}}$ of an honest miner with \(10\%\) of mining power.}
	\label{fig:malicious-miners-earn-more-profit} 
\end{subfigure}
\hfill
\begin{subfigure}[t]{0.49\textwidth}
	\centering
	\includegraphics[width=\textwidth]{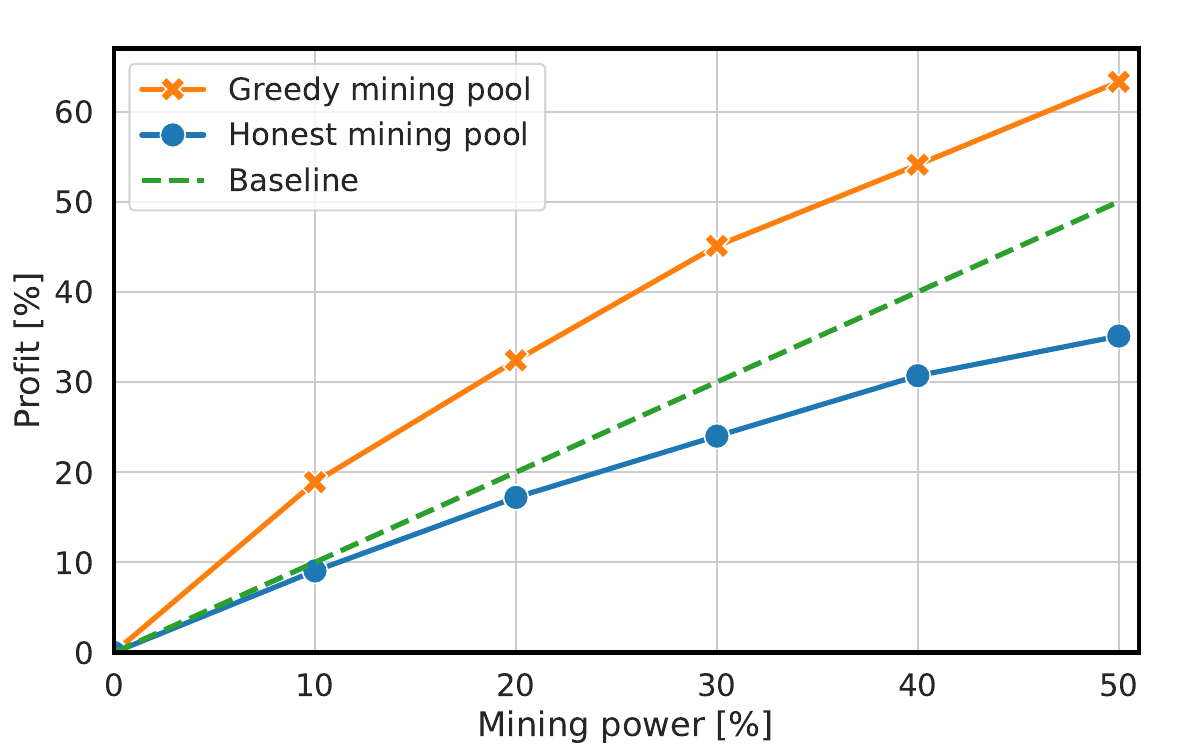}
	\caption{The relative profit of the honest pool and the greedy pool, both with equal mining power (i.e.,~\gls{adversarial-mining-power}), w.r.t. the total mining power of the network.
		The baseline shows the expected profit of the honest mining pool, and \(\gls{block-creation-rate}=20s\).}
	\label{fig:duel-miners-earn-profit}
\end{subfigure}
\\
\vspace{0.3cm}
\begin{subfigure}[b]{0.49\textwidth}
	\centering
	\includegraphics[width=\textwidth]{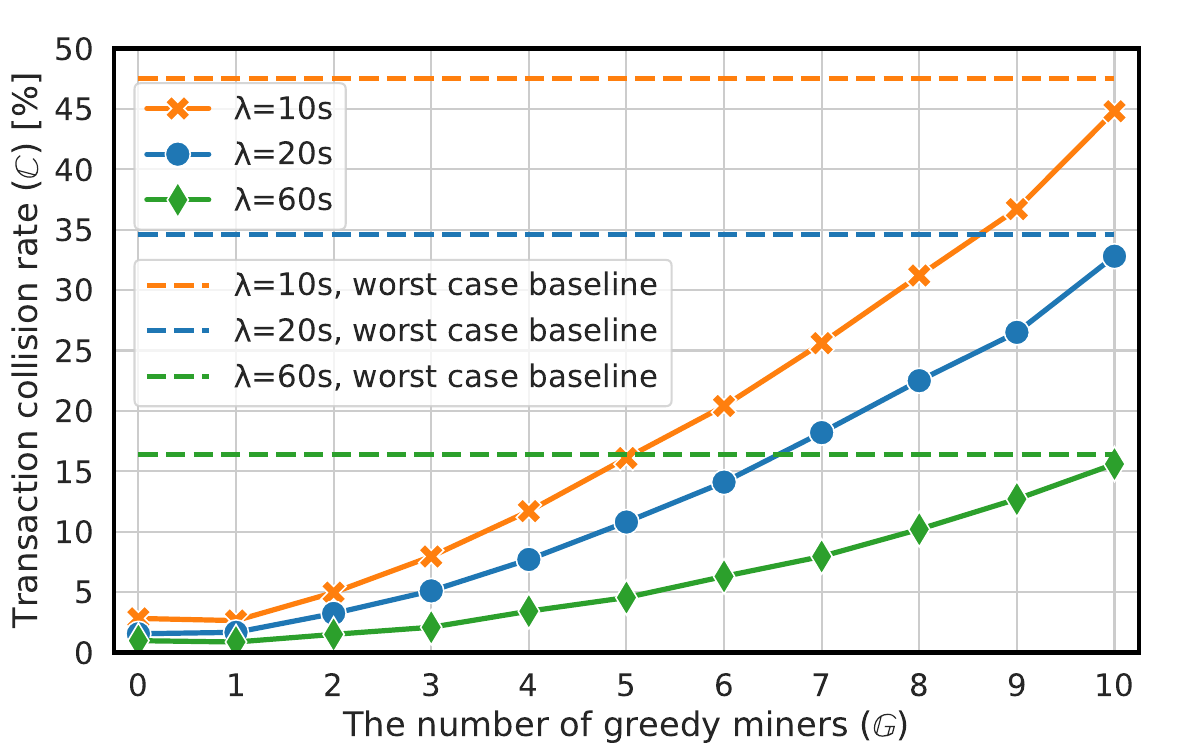}
	\caption{The transaction collision rate $\mathbb{C}$ w.r.t. \# of greedy miners \gls{cnt-malicious-miners} (each with \gls{adversarial-mining-power} = \(10\%\)), where \# of	 honest miners was $10 - \gls{cnt-malicious-miners}$ and \gls{block-creation-rate} $\in \{10s, 20s, 60s\}$.
		The worst case baseline shows $\mathbb{C}$ when all transactions are duplicates.}
	\label{fig:transaction-collision-rate-proportional-to-amount-of-malicious-miners}
\end{subfigure}
\vspace{0.4cm}

\caption{Experiment II, Experiment III (i.e., duel of mining pools) and Experiment IV (i.e., transaction collision rate \& throughput).}
\label{big-figure}
\end{figure*}


\subsubsection{Experiment III}\label{sec:exp-2.5}
\paragraph{\textbf{Goal.}}
The goal of this experiment was to investigate the relative profit of the greedy mining pool depending on its \gls{adversarial-mining-power} versus the honest mining pool with the same mining power.
It is equivalent to Scenario 5 of game theoretical analysis (see \autoref{sec:gametheory}) although there is the honest rest of the network.

\paragraph{\textbf{Methodology and Results.}}
We experimented with 10 miners, and out of them, we choose one greedy miner and one honest miner, both having equal mining power, while the remaining miners in the network were honest and possessed the rest of the network's mining power.
Thus, we emulated a duel of the greedy pool versus the honest pool.
We conducted 10 simulation runs and averaged their results (see \autoref{fig:duel-miners-earn-profit}).
The results demonstrate that the greedy pool's relative earned profit grows proportionally to \gls{adversarial-mining-power} as compared to the honest pool with equal mining power, supporting our conclusions from \autoref{sec:gametheory}.


\subsubsection{Experiment IV}\label{sec:exp-3}
\paragraph{\textbf{Goal.}}
The goal of this experiment was to investigate the transaction collision rate under the occurrence of greedy miners who selected transactions using the greedy strategy.

\paragraph{\textbf{Methodology and Results.}}
In contrast to the previous experiments, 
we considered three different values of block creation time (\gls{block-creation-rate} $ \in \{10s, 20s, 60s\}$).
We experimented with 10 miners, where
the number of greedy miners $\gls{cnt-malicious-miners}$ vs. the number of honest miners (i.e., 10 - \gls{cnt-malicious-miners}) was varied, and each held \(10\%\) of the total mining power.
For all configurations, we computed the transaction collision rate (see \autoref{fig:transaction-collision-rate-proportional-to-amount-of-malicious-miners}).
%
We can see that the increase of $\gls{cnt-malicious-miners}$ causes the increase in the transaction collision rate.
Note that lower \gls{block-creation-rate} has a higher impact on the collision rate, and DAG protocols are designed with the intention to have small \gls{block-creation-rate} (i.e., even smaller than \gls{network-propagation-delay}).
Consequently, the increased collision rate affected the overall throughput of the network (which is complementary to \autoref{fig:transaction-collision-rate-proportional-to-amount-of-malicious-miners}).

\subsection{Experiments with Complex Topology}\label{sec:complex-topology-settings}
Additionally, we conducted more than 500 experiments in complex topology with \(7592\) nodes 
in various configurations (such as different connectivity and positions of greedy miners in the topology).
We emulated weakly and strongly connected miners by setting a different node degree -- we utilized a node degree distribution from \cite{mivsic2019modeling} and projected it into our network by setting the weakly connected edge and a highly connected core.
The results of these experiments confirm the conclusions from the game theoretic analysis (see \autoref{sec:game-theory:model-analysis}) as well as they match the experiments with the simple topology (see \autoref{sec:evaluation-dags}). 
The details of these experiments are presented in the extended version of our paper \cite{perevsini2023incentive}, which is not yet published.

\begin{figure}[t]
	\centering
	\begin{subfigure}[t]{0.48\linewidth}
		\centering
		\includegraphics[width=\linewidth]{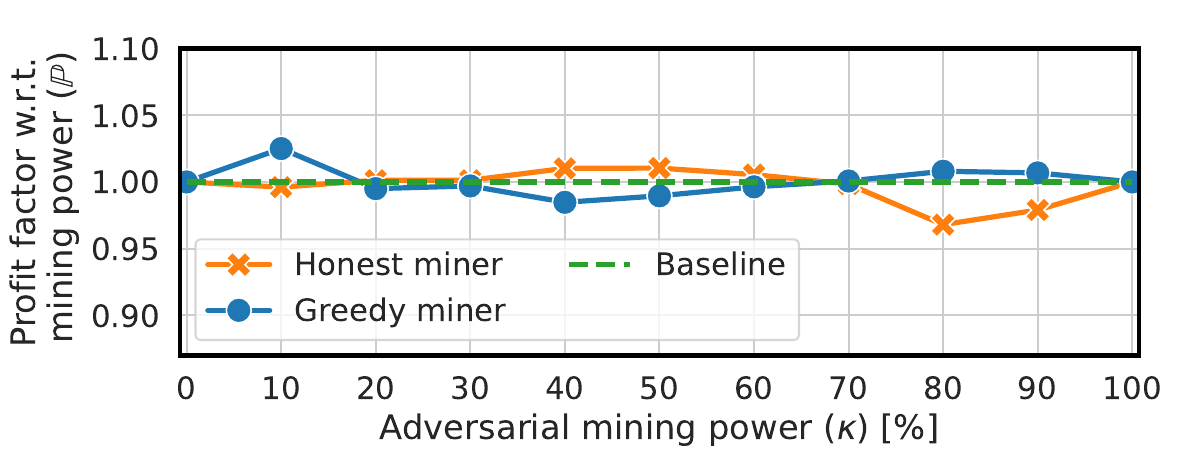}
		\caption{The profit factor $\mathbb{P}$ of a honest vs. a greedy miner with the mining power of 100\% - \gls{adversarial-mining-power} and \gls{adversarial-mining-power}, respectively.}
		\label{fig:flat-fees}
	\end{subfigure}
	\hfill
	\begin{subfigure}[t]{0.48\linewidth}
		\centering
		\includegraphics[width=\linewidth]{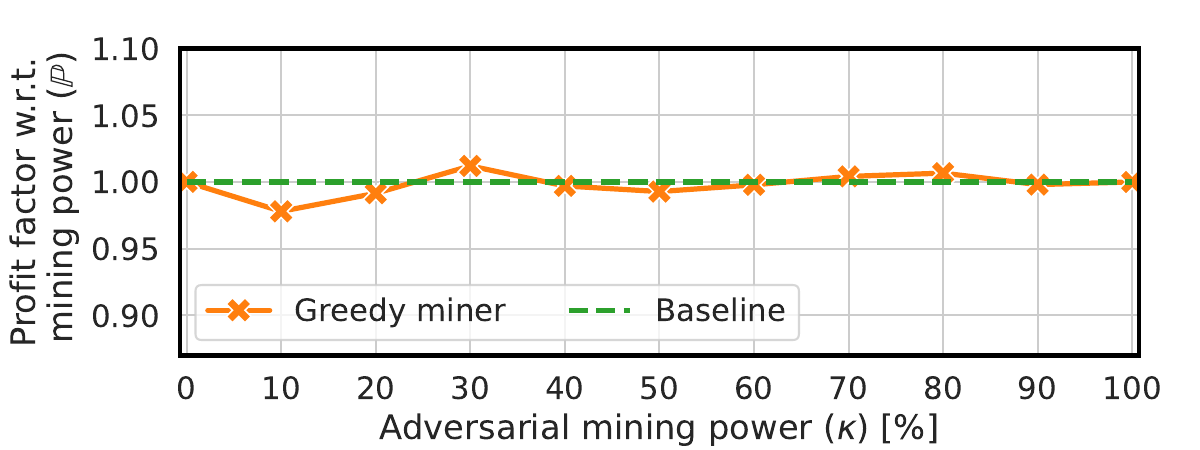}
		\caption{The averaged profit factor $\overline{\mathbb{P}}$ of a greedy miner with \gls{adversarial-mining-power}.
			The rest of the network consisted of $9$ honest miners, each equipped with $\frac{100\% - \gls{adversarial-mining-power}}{9} \%$ of mining power.}
		\label{fig:flat-fees2}
	\end{subfigure}
	\caption{Profit factors of honest and greedy miners. The baseline shows the expected $\mathbb{P}$ of the honest miner; \(\gls{block-creation-rate}=20s\).}
\end{figure}

\subsection{Countermeasures}\label{sec:countermeasures}
Experiments supported Hypothesis \autoref{hypo:problem-definition}.
The main problem is \textbf{not sufficiently enforcing the \gls{RTS}}, i.e.,  verifying that transaction selection was indeed random at the protocol level.
Therefore, using the \gls{RTS} in the \textsc{DAG-Protocol} that does not enforce the interpretation of randomness will never avoid the occurrence of attackers from greedy transaction selection that increases their individual (or pooled) profits.

%

\paragraph{\textbf{Enforcing Interpretation of the Randomness.}}
One countermeasure how to avoid arbitrary interpretation of the randomness in the \gls{RTS} is to enforce it by the consensus protocol.
An example of a DAG-based design using this approach is Sycomore \cite{Anceaume2018Sycomore}, which utilizes the prefix of cryptographically-secure hashes of transactions as the criteria for extending a particular chain in \gls{dag}.
The PoW mining in Sycomore is further equipped with the unpredictability of a chain that the miner of a new block extends, avoiding the concentration of the mining power on ``rich'' chains.
Note that transactions are evenly spread across all chains of the DAG, which happens because prefixes of transaction hashes respect the uniform distribution -- transactions are created by clients (different from miners) who have no incentives for biasing their transactions.




\paragraph{\textbf{Fixed Transaction Fees.}}
Another option how to make the \gls{RTS} viable is to employ fixed fees for all transactions as a blockchain network-adjusted parameter.
In the case of the full block capacity utilization within some period, the fixed fee parameter would be increased and vice versa in the case of not sufficiently utilized block capacity.
In contrast to the previous countermeasure, this mechanism does not enforce the interpretation of randomness while at the same time does not make incentives for greedy miners to  follow other than the \gls{RTS} strategy.
Therefore, miners using other than the \gls{RTS} would not earn extra profits -- we demonstrate it in \autoref{fig:flat-fees} and \autoref{fig:flat-fees2}, considering one honest vs. one greedy miner and one greedy vs. 9 honest miners, respectively.
Note that small deviations from the baseline are caused by the inherent simulation error that is present in the original simulator that we extended.
On the other hand, greedy miners may still cause increased transaction collision rate, and thus decreased throughput.
Therefore, we consider the fixed transaction fee option weaker than the previous one.

\section{Undercutting Attacks}\label{sec:cons-undercut}
In transaction fee-based regime schemes, a few problems have emerged, which we can observe even nowadays in Bitcoin protocol~\cite{carlsten2016instability}.
We have selected three main problems and aim to lower their impact for protocols relying on transaction fees only. 
In detail, we focus on the following problems:

\begin{enumerate}
	\item \textbf{Undercutting attack.}
	In this attack (see \autoref{fig:undercutting}), a malicious miner attempts to obtain transaction fees by re-mining a top block of the longest chain, and thus motivates other miners to mine on top of her block~\cite{carlsten2016instability}.
	In detail, consider a situation, where an honest miner mines a block containing transactions with substantially higher transaction fees than is usual. 
	The malicious miner can fork this block while he leaves some portion of the ``generous'' transactions un-mined. 
	These transactions motivate other miners to mine on top of the attacker's chain, and thus undercut the original block. 
	Such a malicious behavior might result in higher orphan rate, unreliability of the system, and even double spending.
	%
	
	\item \textbf{The mining gap.}
	As discussed in~\cite{carlsten2016instability}, the problem of mining gap arises once the mempool does not contain enough transaction fees to motivate miners in mining.
	Suppose a miner succeeds at mining a new block shortly after the previous block was created, which can happen due to well known exponential distribution of block creation time in PoW blockchains.
	Therefore, the miner might not receive enough rewards to cover his expenses because most of the transactions from the mempool were included in the previous block, while new transactions might not have yet arrived or have small fees.
	Consequently, the miners are motivated to postpone mining until the mempool is reasonably filled with enough transactions (and their fees).
	The mining gap was also analyzed by the simulation in the work of Tsabary and Eyal~\cite{tsabary2018gap}, who further demonstrated that mining gap incentivizes larger mining coalitions (i.e., mining pools), negatively impacting decentralization.
	
	\item \textbf{Varying transaction fees over time.}
	In the transaction-fee-based regime, any fluctuation in transaction fees directly affects the miners' revenue.
	High fluctuation of transaction fees during certain time frames, e.g., in a span of a day or a week~\cite{b5}, can lead to an undesirable lack of predictability in miners' rewards and indirectly affect the security of the underlying protocol.
	
	
\end{enumerate}

\begin{figure}[t]
	\centering
	\includegraphics[width=0.65\textwidth]{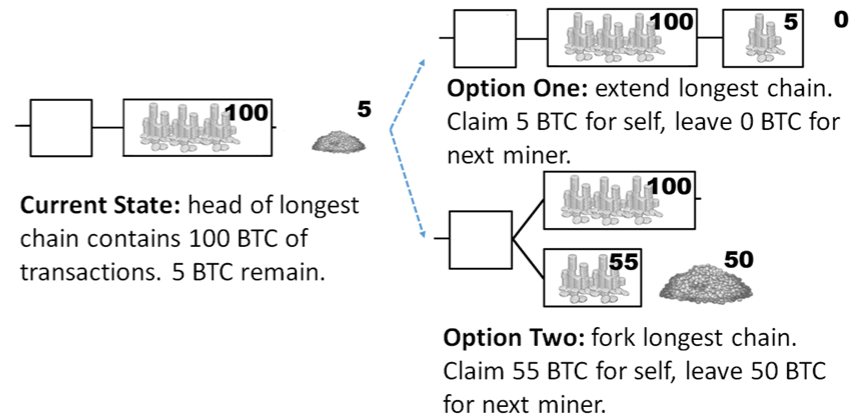}
	\caption{The undercutting attack, according to Carlsten et al.~\cite{carlsten2016instability}. }
	\label{fig:undercutting}
\end{figure}

\subsection{Overview of Proposed Approach}


%
We propose a solution that collects a percentage of transaction fees in a native cryptocurrency from the mined blocks into one or multiple fee-redistribution smart contracts (i.e., $\mathcal{FRSC}$s).
Miners of the blocks who must contribute to these contracts are at the same time rewarded from them, while the received reward approximates a moving average of the incoming transaction fees across the fixed sliding window of the blocks.
The fraction of transaction fees (i.e., $\mathbb{C}$) from the mined block is sent to the $\mathcal{FRSC}$ and the remaining fraction of transaction fees (i.e., $\mathbb{M}$) is directly assigned to the miner, such that $ \mathbb{C} + \mathbb{M} = 1.$
The role of $\mathbb{M}$ is to incentivize the miners in prioritization of the transactions with the higher fees while the role of $\mathbb{C}$ is to mitigate the problems of undercutting attacks and the mining gap.
Our solution can be deployed with hard-fork and imposes only negligible performance overhead.

\medskip
We depict the overview of our approach in \autoref{fig:overview}, and it consists of the following steps:
\begin{enumerate}
	
	\item Using $\mathcal{FRSC}$, the miner calculates the reward for the next block $B$ (i.e., $nextClaim\-(\mathcal{FRSC})$ -- see \autoref{eq:nextClaim}) that will be payed by $\mathcal{FRSC}$ to the miner of that block.
	
	\item The miner mines the block $B$ using the selected set of the highest fee transactions from her mempool.
	
	\item The mined block $B$ directly awards a certain fraction of the transaction fees (i.e., $B.fees ~*~ \mathbb{M}$) to the miner and the remaining part (i.e., $B.fees ~*~ \mathbb{C}$) to $\mathcal{FRSC}$.
	
	\item The miner obtains $nextClaim$ from $\mathcal{FRSC}$.
	
\end{enumerate}
Our approach is embedded into the consensus protocol, and therefore consensus nodes are obliged to respect it in order to ensure that their blocks are valid.
It can be implemented with standard smart contracts of the blockchain platform or within the native code of the consensus protocol.

\begin{figure}[t]
	\centering
	\includegraphics[width=0.65\textwidth]{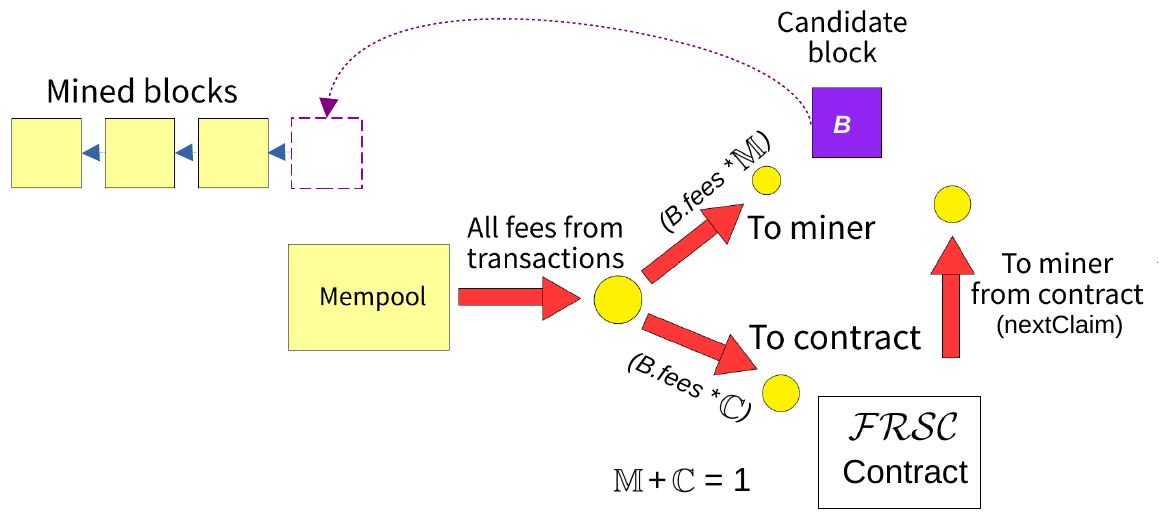}
	\caption{Overview of our solution. }
	\label{fig:overview-undercut}
\end{figure}

\subsection{Details of Fee-Redistribution Smart Contracts}\label{sub:redistrib}
We define the fee-redistribution smart contract as a tuple 
\begin{eqnarray}
	\mathcal{FRSC} = (\nu, \lambda, \rho), 
\end{eqnarray}
where
$\nu$ is the accumulated amount of tokens in the contract, $\lambda$ denotes the size of $\mathcal{FRSC}$'s sliding window in terms of the number of preceding blocks that contributed to $\nu$,
and $\rho$ is the parameter defining the ratio for redistribution of incoming transaction fees among multiple contracts (if there are multiple $\mathcal{FRSC}$s), while the sum of $\rho$ across all $\mathcal{FRSC}$s must be equal to 1:
\begin{eqnarray}\label{eq:frsc-redistrib-ratios}
	\sum_{x ~\in~ \mathcal{FRSC}s} x.\rho &=& 1.
\end{eqnarray}
In contrast to a single $\mathcal{FRSC}$, we envision multiple $\mathcal{FRSC}$s to enable better adjustment of compensation to miners during periods of higher transaction fee fluctuations or in an unpredictable environment (we show this in \autoref{sec:exp3}).


\medskip
We denote the state of $\mathcal{FRSC}$s at the blockchain height $H$ as $\mathcal{FRSC}_{[H]}$.
%
Then, we determine the reward from $\mathcal{FRSC}_{[H]} \in \mathcal{FRSC}s_{[H]}$ for the miner of the next block with height $H+1$ as follows:
\begin{equation}
	\partial Claim_{[H+1]}^{\mathcal{FRSC}_{[H]}} = \frac{\mathcal{FRSC}_{[H]}.\nu}{\mathcal{FRSC}_{[H]}.\lambda},
\end{equation}
while the reward obtained from all $\mathcal{FRSC}$s is
\begin{equation} \label{eq:nextClaim}	
	nextClaim_{[H+1]} = \sum_{\mathcal{X}_{[H]} ~\in~ \mathcal{FRSC}s_{[H]}^{}} \partial Claim_{[H+1]}^{\mathcal{X}_{[H]}}.
\end{equation}


\noindent
Then, the total reward of the miner who mined the block $B_{[H+1]}$ with all transaction fees $B_{[H+1]}.fees$ is
\begin{equation} \label{eq:rewardtotal}
	rewardT_{[H+1]} = nextClaim_{[H+1]} + \mathbb{M} * B_{[H+1]}.fees.
\end{equation}
The new state of contracts at the height $H + 1$ is 
\begin{eqnarray}
	\mathcal{FRSC}s_{[H+1]} = \{\mathcal{X}_{[H+1]}(\nu, \lambda, \rho)\ ~|~
\end{eqnarray}
\begin{eqnarray}
	\lambda &=& \mathcal{X}_{[H]}.\lambda,\\
	\rho &=& \mathcal{X}_{[H]}.\rho,\\
	\nu &=& \mathcal{X}_{[H]}.\nu - \partial Claim_{[H+1]} +  deposit * \rho,\\
	deposit &=& B_{[H+1]}.fees * \mathbb{C}\},
\end{eqnarray}
where $deposit$ represents the fraction $\mathbb{C}$ of all transaction fees from the block $B_{[H+1]}$ that are deposited across all $\mathcal{FRSC}$s in ratios respecting \autoref{eq:frsc-redistrib-ratios}.

\subsubsection{Example}
We consider Bitcoin~\cite{nakamoto2008bitcoin} with the current height of the blockchain $H$. 
We utilize only a single $\mathcal{FRSC}$:
\[
\mathcal{FRSC}_{[H]} = (2016, 2016, 1).
\]
We set $\mathbb{M} = 0.4$ and $\mathbb{C} = 0.6$, which means a miner directly obtains 40\% of the $B_{[H+1]}.fees$ and $\mathcal{FRSC}$ obtains 60\%.
\noindent
Next, we compute the reward from $\mathcal{FRSC}$ obtained by the miner of the block with height $H+1$ as 
\begin{equation*}
	\partial Claim_{[H+1]} = \frac{\mathcal{FRSC}_{[H]}.\nu}{\mathcal{FRSC}_{[H]}.\lambda} = \frac{2016}{2016} = 1~\text{BTC},
\end{equation*}
resulting into 
\begin{equation*}
	nextClaim_{[H+1]} = \partial Claim_{[H+1]}~=~1~\text{BTC}.
\end{equation*}

\noindent
Further, we assume that the total reward collected from transactions in the block with height $H+1$ is $B_{[H+1]}.fees = 2$ BTC.
Hence, the total reward obtained by the miner of the block $B_{[H+1]}$ is
\begin{eqnarray*}
	rewardT_{[H+1]} &=& nextClaim_{[H+1]}  + \mathbb{M} * B_{[H+1]}.fees \\
	&=& 1 + 0.4 * 2 ~=~ 1.8 ~\text{BTC},
\end{eqnarray*}
and the contribution of transaction fees from $B_{[H+1]}$ to the $\mathcal{FRSC}$ is 
\[
deposit = B_{[H+1]}.fees * \mathbb{C} = 1.2 ~\text{BTC}.
\]
Therefore, the value of $\nu$ in $\mathcal{FRSC}$ is updated at height H~+~1 as follows:
\begin{eqnarray*}
	v_{[H+1]} &=& \mathcal{FRSC}_{[H]}.\nu - nextClaim_{[H+1]} + deposit \\
	&=& 2016 - 1 + 1.2  ~\text{BTC} ~=~ 2016.2 ~\text{BTC}. 
\end{eqnarray*}


%

\vspace{-0.2cm}
\subsubsection*{\textbf{Traditional Way in Tx-Fee Regime}}
In traditional blockchains, $rewardT_{[H+1]}$ would be equal to the sum of all transaction fees $B_{[H+1]}.fees$ (i.e., $2$ BTC); hence, using $\mathbb{M} = 1$. 
In our approach, $rewardT_{[H+1]}$ is equal to the sum of all transaction fees in the block $B_{[H+1]}$, if:
\begin{eqnarray}
	B_{[H+1]}.fees = \frac{nextClaim_{[H+1]}}{\mathbb{C}}.
\end{eqnarray}
In our example, a miner can mine the block $B_{[H+1]}$ while obtaining the same total reward as the sum of all transaction fees in the block if the transactions carry 1.66 BTC in fees:
\begin{equation*}
	B_{[H+1]}.fees = \frac{1}{0.6} = 1.66 ~\text{BTC}.
\end{equation*}

\subsubsection{Initial Setup of $\mathcal{FRSC}$s Contracts}
To enable an even start, we propose to initiate $\mathcal{FRSC}$s of our approach by a genesis value.
The following formula calculates the genesis values per $\mathcal{FRSC}$ and initializes starting state of $\mathcal{FRSC}s_{[0]}$:
\begin{equation} \label{eq:setup}
	\{\mathcal{FRSC}_{[0]}^{x}(\nu, \lambda, \rho)\ |\ \nu = \overline{fees} * \mathbb{C} * \rho * \lambda\},
\end{equation}
where $\overline{fees}$ is the expected average of incoming fees.

\subsection{Evaluation}
\label{sec:evaluation}

We base on Bitcoin Mining Simulator~\cite{bitcoin_mining_simulator:kalodner}, introduced in~\cite{carlsten2016instability}, which we modified for our purposes.
We have created a configuration file to simulate custom scenarios of incoming transactions instead of the accumulated fees in the original design~\cite{carlsten2016instability}.
We added an option to switch simulation into a mode with a full mempool, and thus bound the total fees (and consequently the total number of transactions) that can be earned within a block -- 
this mostly relates to blocks whose mining takes longer time than the average time to mine a block.\footnote{Note that the original simulator~\cite{carlsten2016instability} assumes that the number of transactions (and thus the total fees) in the block is constrained only by the duration of a time required to mine the block, which was also criticized in~\cite{gong2022towards}.}
%
%
Next, we moved several parameters to arguments of the simulator to eliminate the need for frequent recompilation of the program, and therefore simplified the process of running various experiments with the simulator.
Finally, we integrated our $\mathcal{FRSC}$-based solution into the simulator.
$\mathcal{FRSC}$s are initiated from a corresponding configuration file.
The source code of our modified simulator is available at \url{https://github.com/The-Huginn/mining_simulator}.

\paragraph{Experiments.}
We evaluated our proof-of-concept implementation of $\mathcal{FRSC}$s on a custom long-term scenario designed to demonstrate significant changes in the total transaction fees in the mempool evolving across the time.
This scenario is depicted in the resulting graphs of most of our experiments, represented by the ``\textit{Fees in mempool}'' series -- see \autoref{sec:exp1} and \autoref{sec:exp2}.

\begin{figure*}[!ht]
	\centering
	\begin{subfigure}{0.45\textwidth}
		\includegraphics[width=\textwidth]{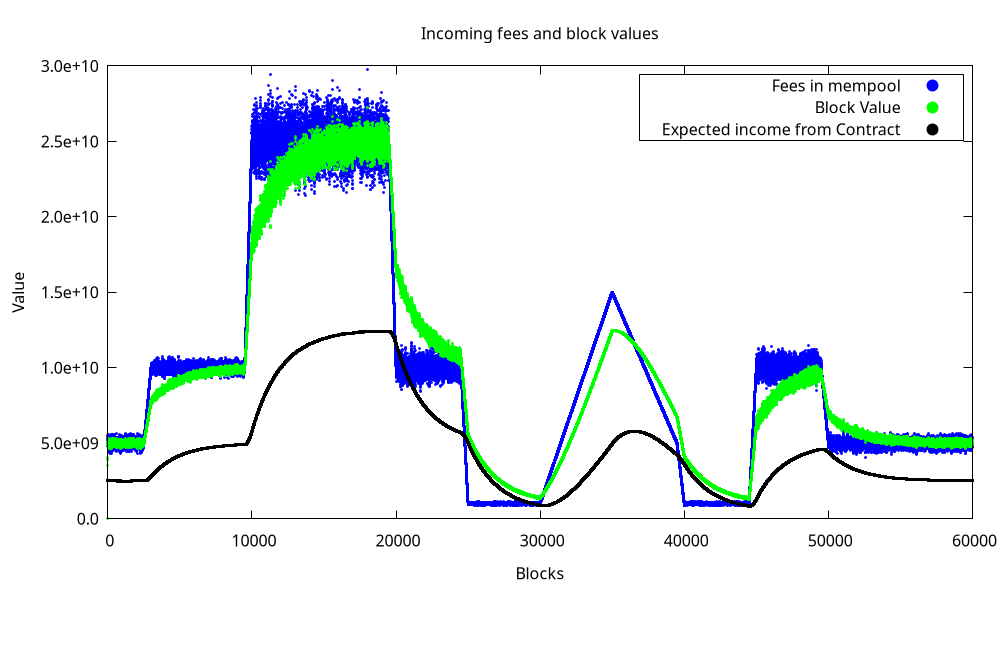}
		\vspace{-0.7cm}
		\caption{$\mathcal{FRSC}^{1}$ and $\mathbb{C} = 0.5$.}

	\end{subfigure}
	\begin{subfigure}{0.45\textwidth}
		\includegraphics[width=\textwidth]{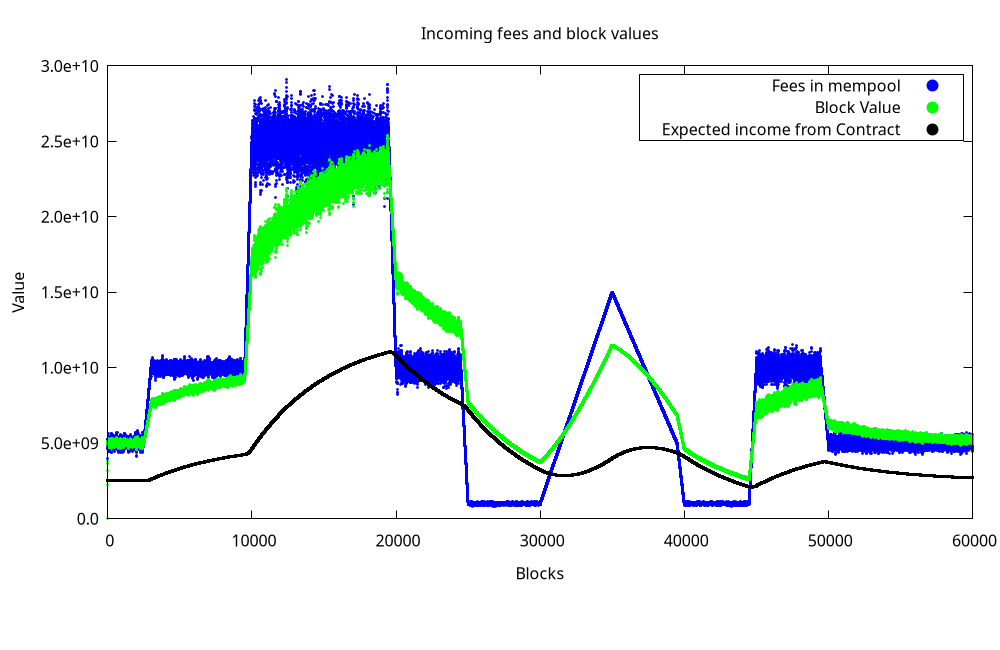}
		\vspace{-0.7cm}
		\caption{$\mathcal{FRSC}^{2}$ and $\mathbb{C} = 0.5$.}

	\end{subfigure}
	\vspace{0.4cm}
	
	\begin{subfigure}{0.45\textwidth}
		\includegraphics[width=\textwidth]{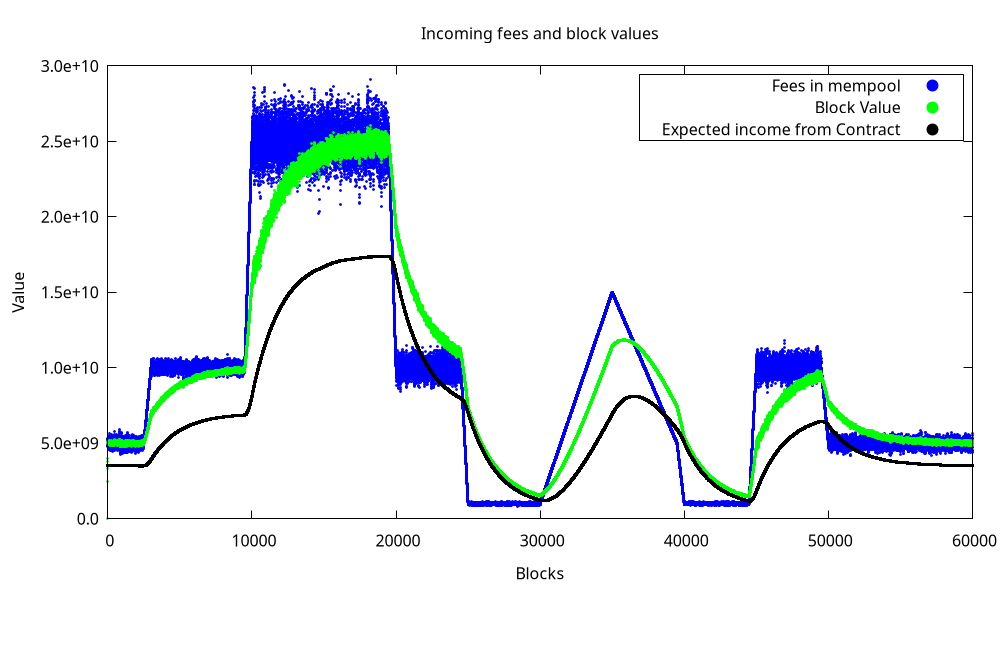}
		\vspace{-0.7cm}
		\caption{$\mathcal{FRSC}^{1}$ and $\mathbb{C} = 0.7$.}

	\end{subfigure}
	\begin{subfigure}{0.45\textwidth}
		\includegraphics[width=\textwidth]{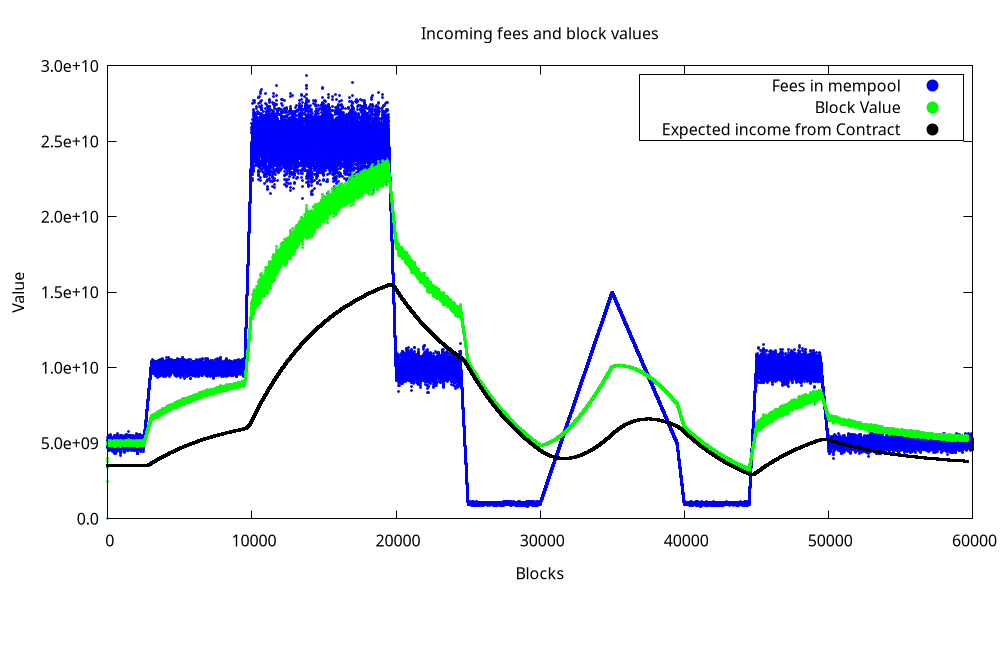}
		\vspace{-0.7cm}
		\caption{$\mathcal{FRSC}^{2}$ and $\mathbb{C} = 0.7$.}

	\end{subfigure}
	\vspace{0.4cm}
	
	\begin{subfigure}{0.45\textwidth}
		\includegraphics[width=\textwidth]{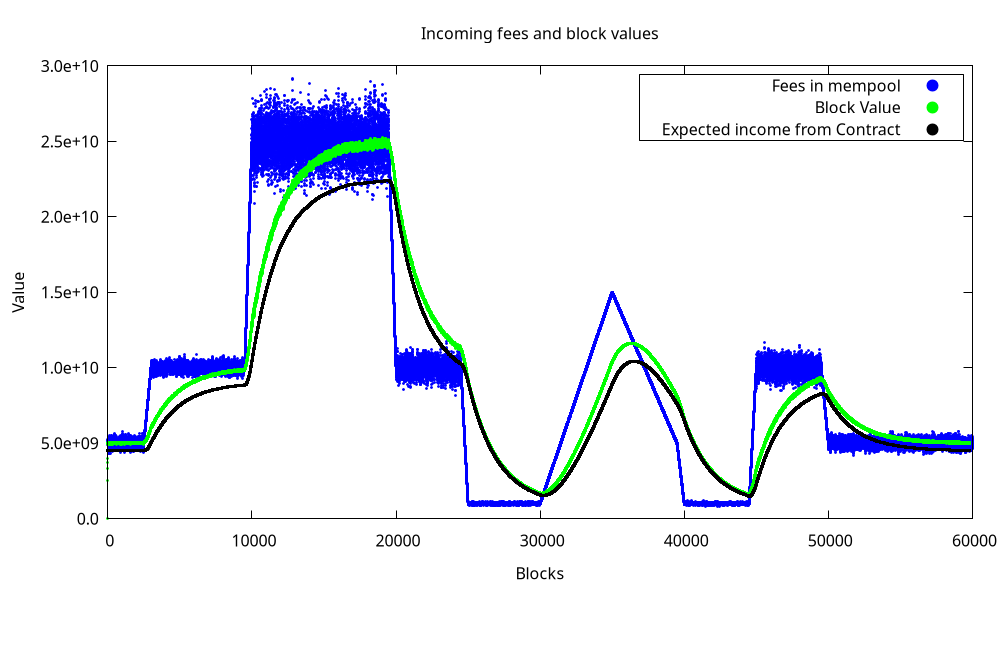}
		\vspace{-0.7cm}
		\caption{$\mathcal{FRSC}^{1}$ and $\mathbb{C} = 0.9$.}

	\end{subfigure}
	\begin{subfigure}{0.45\textwidth}
		\includegraphics[width=\textwidth]{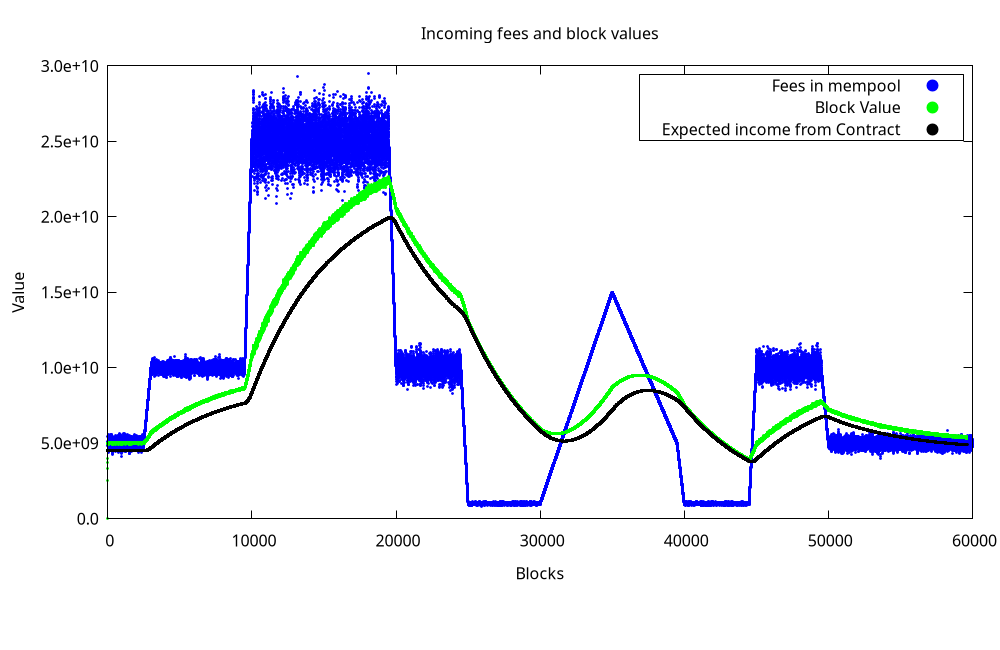}
		\vspace{-0.7cm}
		\caption{$\mathcal{FRSC}^{2}$ and $\mathbb{C} = 0.9$.}

	\end{subfigure}
	\vspace{0.4cm}
	\caption{Experiment I investigating various $\mathbb{C}s$ and $\lambda$s of a single $\mathcal{FRSC}$, where
		$\mathcal{FRSC}^{1}.\lambda = 2016$ and 		$\mathcal{FRSC}^{2}.\lambda = 5600$.
		\textit{Fees in mempool} show the total value of fees in the mined block (i.e., representing the baseline).
		\textit{Block Value} is the reward a miner received in block $B$ as a sum of the fees he obtained directly (i.e. $\mathbb{M} * B.fees$) and the reward he got from $\mathcal{FRSC}$ (i.e., $nextClaim_{[H]}$).
		\textit{Expected income from Contract} represents the reward of a miner obtained from $\mathcal{FRSC}$ (i.e., $nextClaim_{[H]}$).}\label{fig:50tocontract}
\end{figure*}

We experimented with different parameters and investigated how they influenced the total rewards of miners coming from $\mathcal{FRSC}$s versus the baseline without our solution.
Mainly, these included a setting of $\mathbb{C}$ as well as different lengths $\lambda$ of $\mathcal{FRSC}$s.
For demonstration purposes, we used the value of transaction fees per block equal to 50 BTC, the same as Carlsten et al.~\cite{carlsten2016instability} used.
Across all our experiments but the last one (i.e., \autoref{sec:exp5}), we enabled the full mempool option to ensure more realistic conditions.

\begin{figure}[t]

	\centering
	\begin{subfigure}{0.32\textwidth}
		\includegraphics[width=\textwidth]{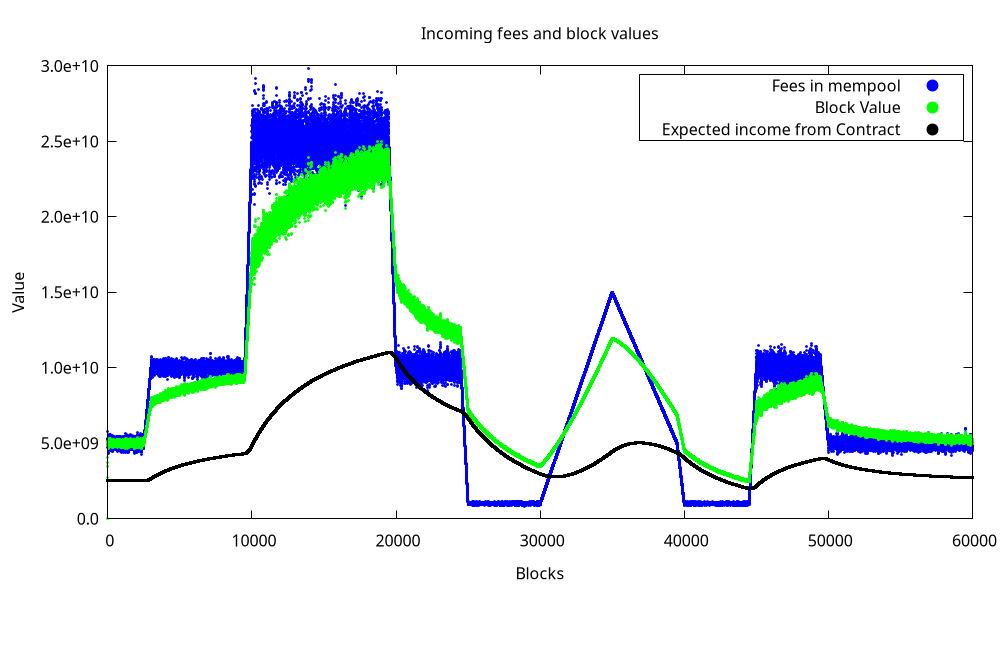}

		\caption{Scenario with 4 $\mathcal{FRSC}$s,\\ $\mathbb{C} = 0.5$.}
	\end{subfigure}
	\begin{subfigure}{0.32\textwidth}
		\includegraphics[width=\textwidth]{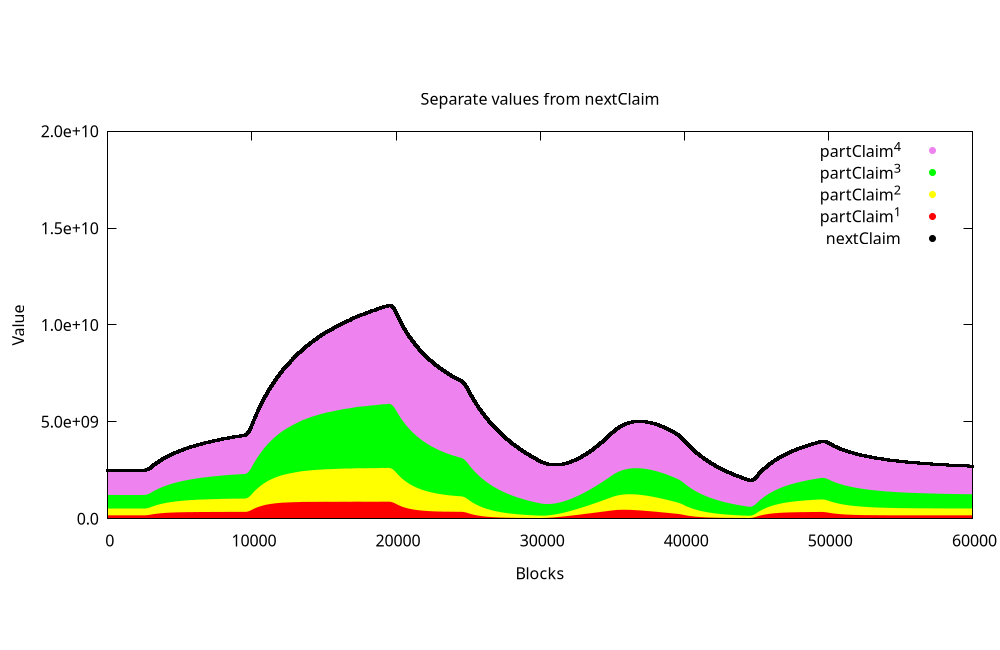}

		\caption{$\partial Claim$s and $nextClaim$, \\ $\mathbb{C} = 0.5$.}
	\end{subfigure}
	\begin{subfigure}{0.32\textwidth}
		\includegraphics[width=\textwidth]{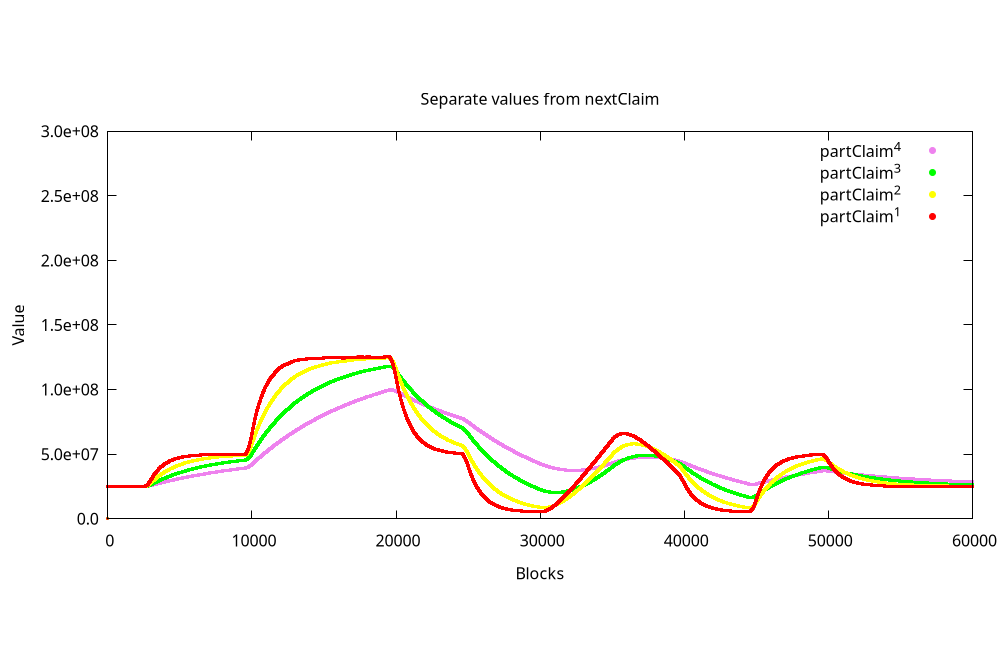}

		\caption{$\partial Claim$s normalized by $\rho$, \\ $\mathbb{C} = 0.5$.}
	\end{subfigure}
	\\
	\vspace{0.4cm}
	
	\begin{subfigure}{0.32\textwidth}
		\includegraphics[width=\textwidth]{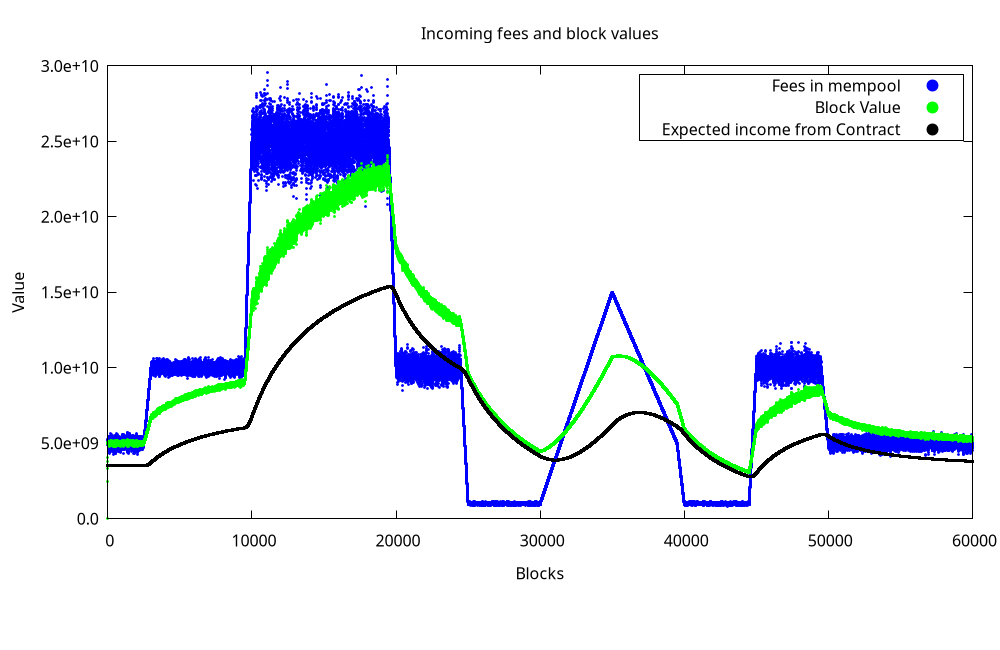}

		\caption{Scenario with 4 $\mathcal{FRSC}$s,\\ $\mathbb{C} = 0.7$.}
	\end{subfigure}
	\begin{subfigure}{0.32\textwidth}
		\includegraphics[width=\textwidth]{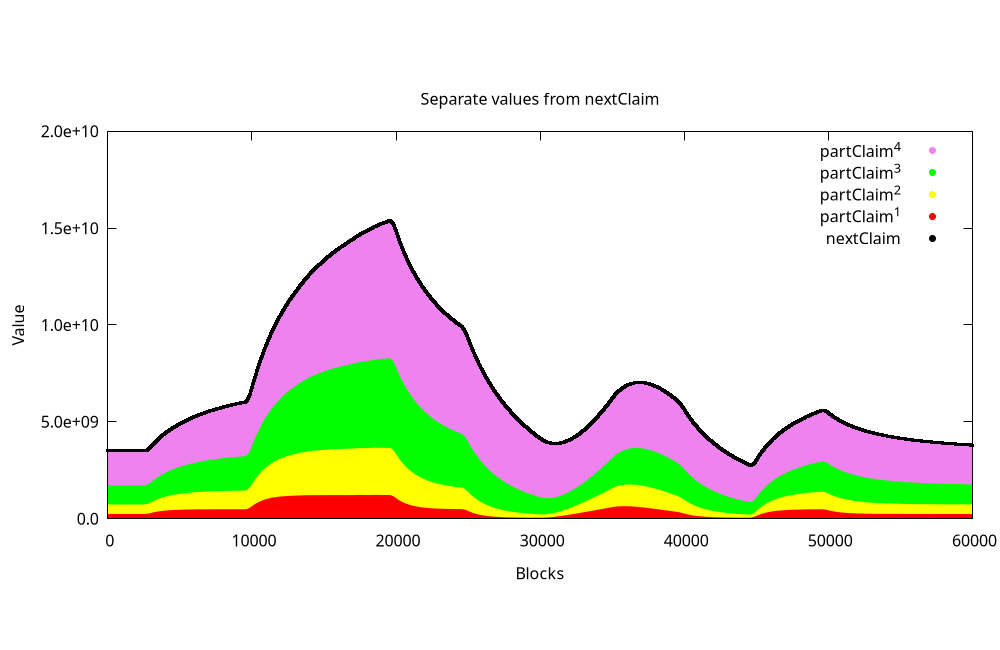}

		\caption{$\partial Claim$s and $nextClaim$, \\ $\mathbb{C} = 0.7$.}
	\end{subfigure}
	\begin{subfigure}{0.32\textwidth}
		\includegraphics[width=\textwidth]{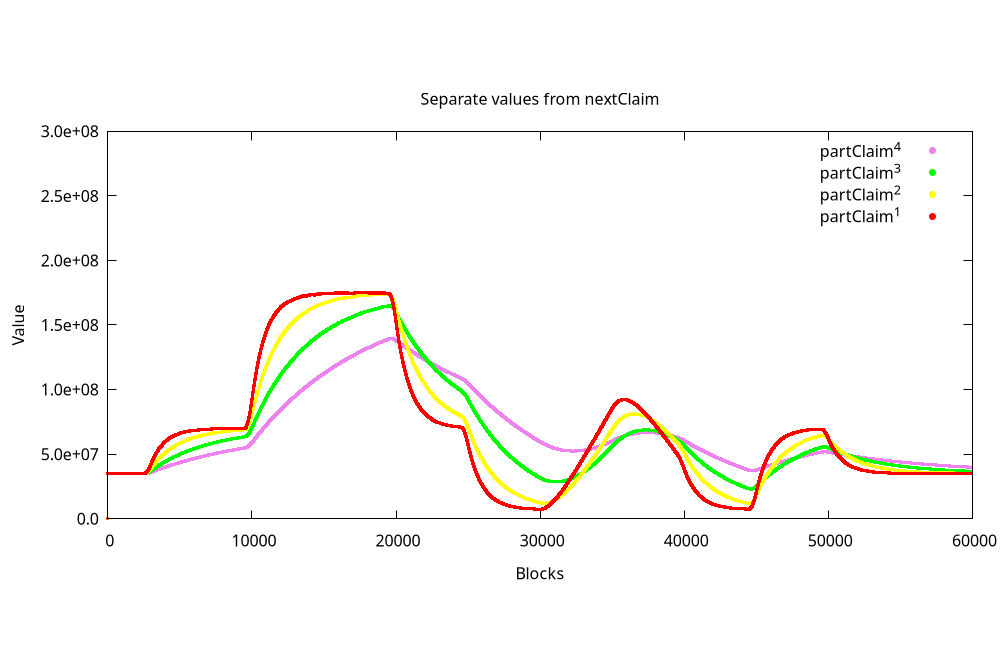}

		\caption{$\partial Claim$s normalized by $\rho$, \\ $\mathbb{C} = 0.7$.}
	\end{subfigure}
		\vspace{0.4cm}
	
	\begin{subfigure}{0.32\textwidth}
		\includegraphics[width=\textwidth]{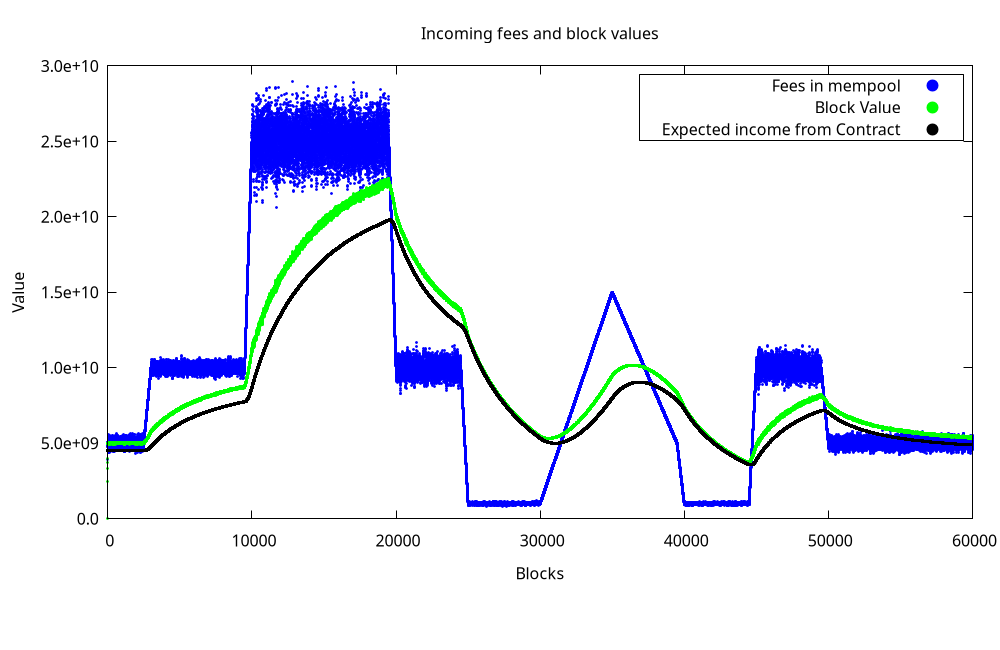}

		\caption{Scenario with 4 $\mathcal{FRSC}$s, \\ $\mathbb{C} = 0.9$.}
	\end{subfigure}
	\begin{subfigure}{0.32\textwidth}
		\includegraphics[width=\textwidth]{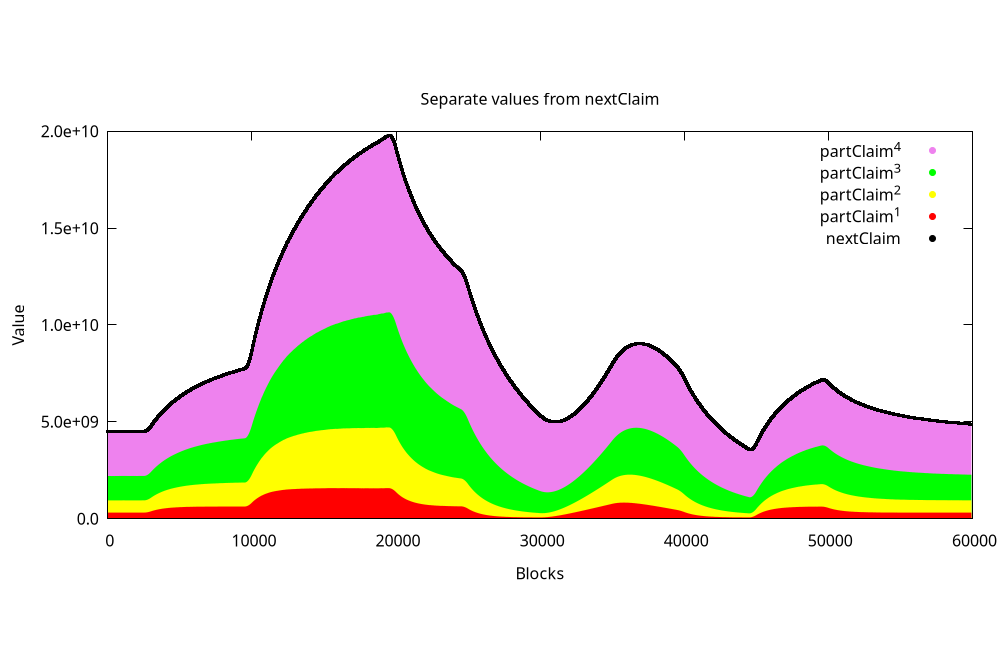}

		\caption{$\partial Claim$s and $nextClaim$, \\ $\mathbb{C} = 0.9$.}
	\end{subfigure}
	\begin{subfigure}{0.32\textwidth}
		\includegraphics[width=\textwidth]{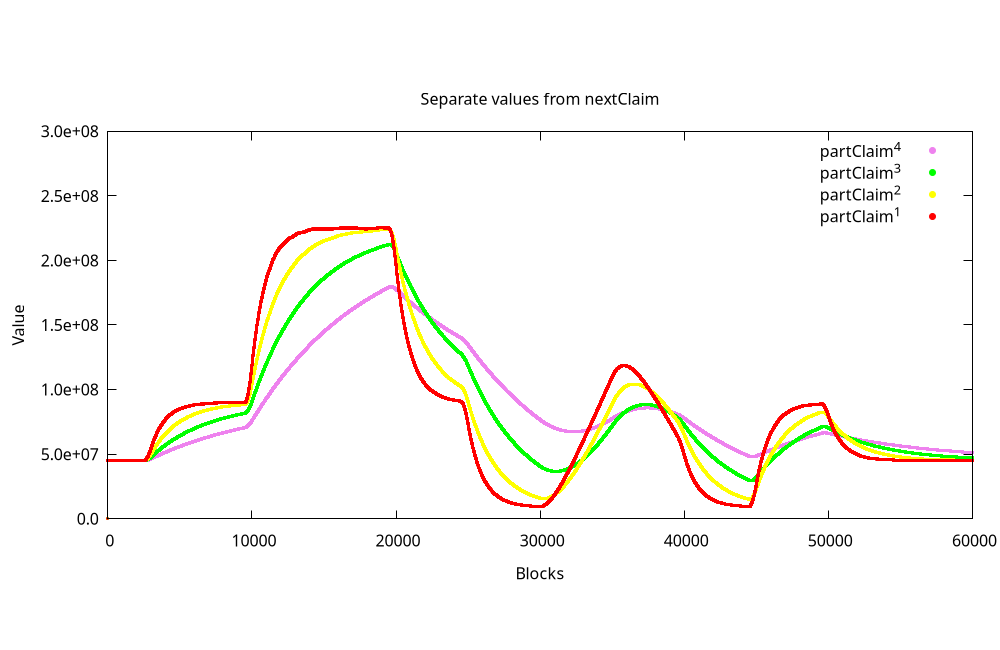}

		\caption{$\partial Claim$s normalized by $\rho$, \\ $\mathbb{C} = 0.9$.}
	\end{subfigure}
		\vspace{0.4cm}

	\caption{Experiment II investigating various $\mathbb{C}$s in the setting with multiple $\mathcal{FRSC}$s with their corresponding $\lambda$ = $\{1008, 2016, 4032, 8064\}$ and $\rho$ = $\{0.07, 0.14, 0.28, 0.51\}$. $\partial Claim$s represents contributions of individual $\mathcal{FRSC}$s to the total reward of the miner (i.e., its $nextClaim$ component).}\label{fig:exp-II}
\end{figure}

\subsubsection{Experiment I}
\label{sec:exp1}


\paragraph{\textbf{Methodology.}}
The purpose of this experiment was to investigate the amount of the reward a miner earns with our approach versus the baseline (i.e., the full reward is based on all transaction fees).
We investigated how $\mathbb{C}$ influences the total reward of the miner and how $\lambda$ of the sliding window averaged the rewards.
In detail, we created two independent $\mathcal{FRSC}$s with different $\lambda$ -- one was set to 2016 (i.e., $\mathcal{FRSC}^1$), and the second one was set to 5600 (i.e., $\mathcal{FRSC}^2$). 
We simulated these $\mathcal{FRSC}$s with three values of $\mathbb{C} \in \{0.5,~0.7,~0.9\}$.

%

\paragraph{\textbf{Results.}}
The results of this experiment are depicted in \autoref{fig:50tocontract}.
Across all runs of our experiment, we can observe that $\mathcal{FRSC}^{2}$ adapts slower as compared to $\mathcal{FRSC}^{1}$, which leads to a more significant averaging of the total reward paid to the miner.


\subsubsection{Experiment II}
\label{sec:exp2}


\paragraph{\textbf{Methodology.}}
In this experiment, we investigated how multiple $\mathcal{FRSC}$s dealt with the same scenario as before -- i.e., varying $\mathbb{C}$.
In detail, we investigated how individual $\mathcal{FRSC}$s  contributed to the $nextClaim_{[H+1]}$ by their individual $\partial Claim^{\mathcal{FRSC}_{[H]}}_{[H+1]}$.
This time, we varied only the parameter $\mathbb{C} \in \{0.5, ~0.7, ~0.9\}$, and we considered four $\mathcal{FRSC}$s:
\begin{center}
	\vspace{-0.3cm}
	$\mathcal{FRSC}s=\{$\\
	$\mathcal{FRSC}^{1}(\_, 1008, 0.07), \mathcal{FRSC}^{2}(\_, 2016, 0.14),$\\
	$\mathcal{FRSC}^{3}(\_, 4032, 0.28), \mathcal{FRSC}^{4}(\_, 8064, 0.51)\},$
	\vspace{-0.3cm}
\end{center}
where their lengths $\lambda$ were set to consecutive multiples of 2 (to see differences in more intensive averaging across longer intervals), and their redistribution ratios $\rho$ were set to maximize the potential of averaging by longer $\mathcal{FRSC}$s. 


\paragraph{\textbf{Results.}}
The results of this experiment are depicted in \autoref{fig:exp-II}.
We can observe that the shorter $\mathcal{FRSC}$s quickly adapted to new changes and the longer $\mathcal{FRSC}$s kept more steady income for the miner.
In this sense, we can see that $\partial Claim^{4}$ held steadily over the scenario while for example $\partial Claim^{1}$ fluctuated more significantly.
Since the scenarios of fees evolution in the mempool was the same across all our experiments (but \autoref{sec:exp3}), we can compare the $\mathcal{FRSC}$ with $\lambda$ = 5600 from \autoref{sec:exp1} and the current setup involving four $\mathcal{FRSC}$s -- both had some similarities.
This gave us intuition for replacing multiple $\mathcal{FRSC}$s with a single one (see \autoref{sec:exp3}).


\begin{figure*}[t]

	\centering
	\begin{subfigure}{0.32\textwidth}
		\includegraphics[width=\textwidth]{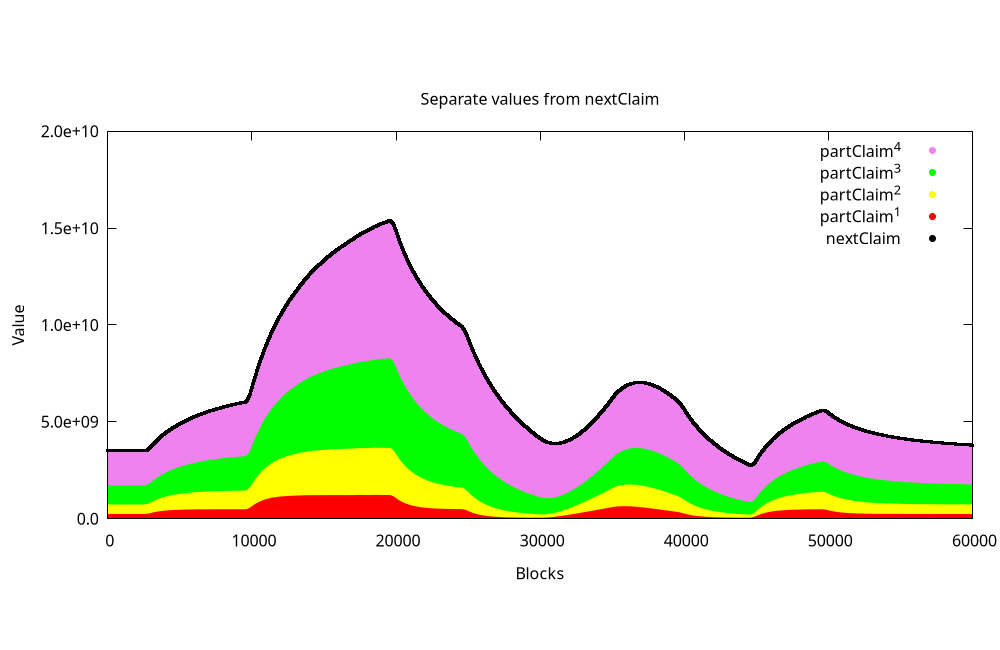}

		\caption{$\rho$ correlates with $\lambda$.}
	\end{subfigure}
	\begin{subfigure}{0.32\textwidth}
		\includegraphics[width=\textwidth]{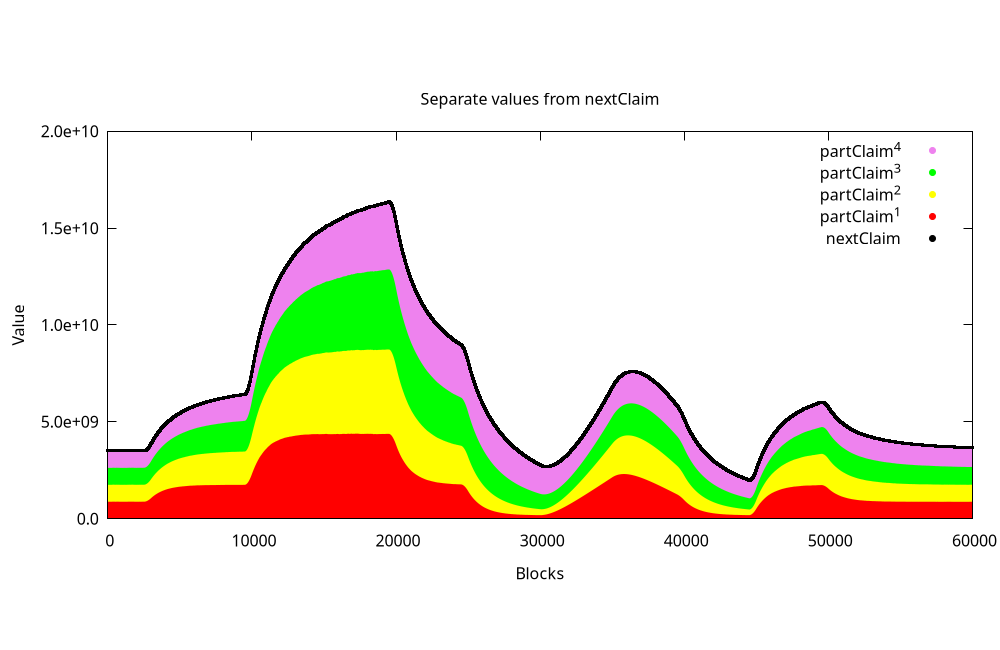}

		\caption{$\rho$ equal for every $\mathcal{FRSC}$.}
	\end{subfigure}
	\begin{subfigure}{0.32\textwidth}
		\includegraphics[width=\textwidth]{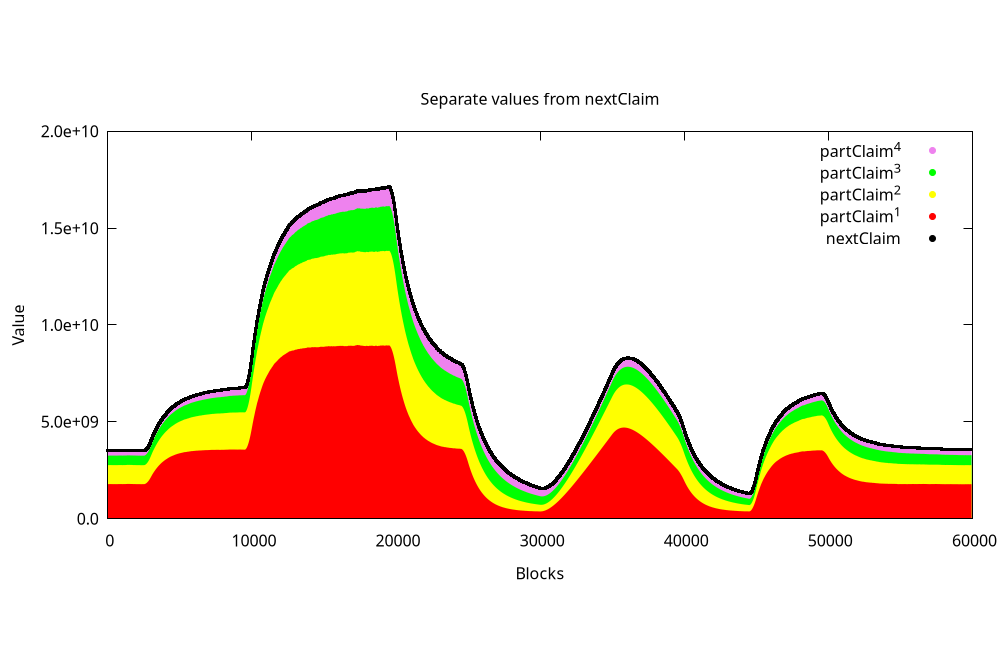}

		\caption{$\rho$ negatively correlates with $\lambda$.}
	\end{subfigure}
	
	\vspace{0.4cm}
	\caption{Experiment II -- multiple $\mathcal{FRSC}$s using various distributions of $\rho$ and their impact on $\partial Claim$, where $\mathbb{C}$ = 0.7.}	\label{fig:percentages}
		\vspace{0.5cm}
\end{figure*}

\begin{figure*}[t]
	
	\centering
	\begin{subfigure}[t]{0.45\textwidth}
		\includegraphics[width=\textwidth]{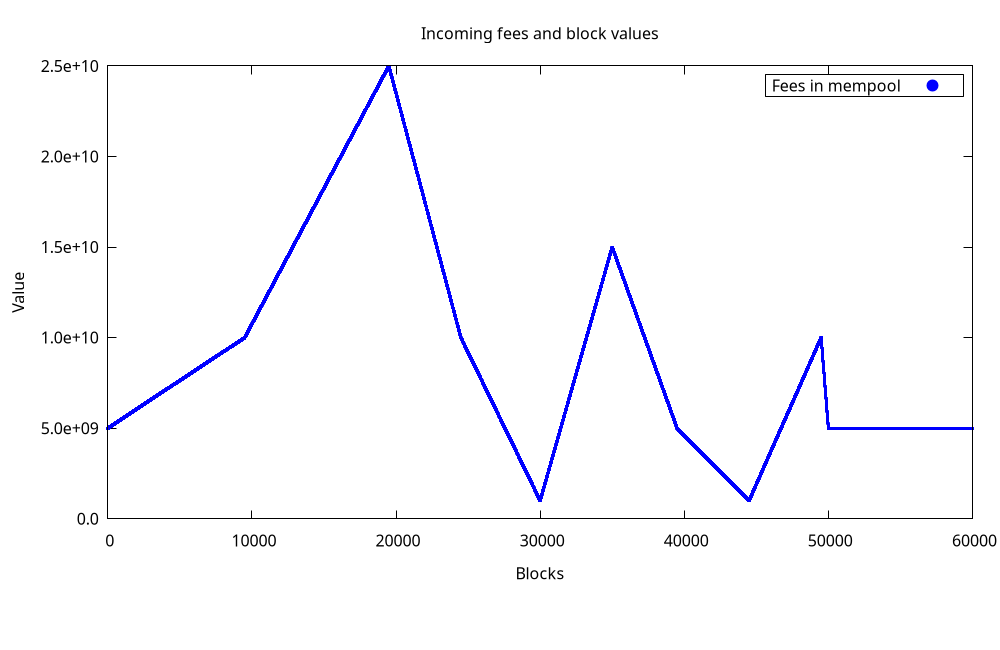}
		\caption{A custom fee scenario for Experiment III.}\label{fig:exp3-fees}
	\end{subfigure}
	\hspace{0.2cm}
	\begin{subfigure}[t]{0.45\textwidth}
	
		\includegraphics[width=\textwidth]{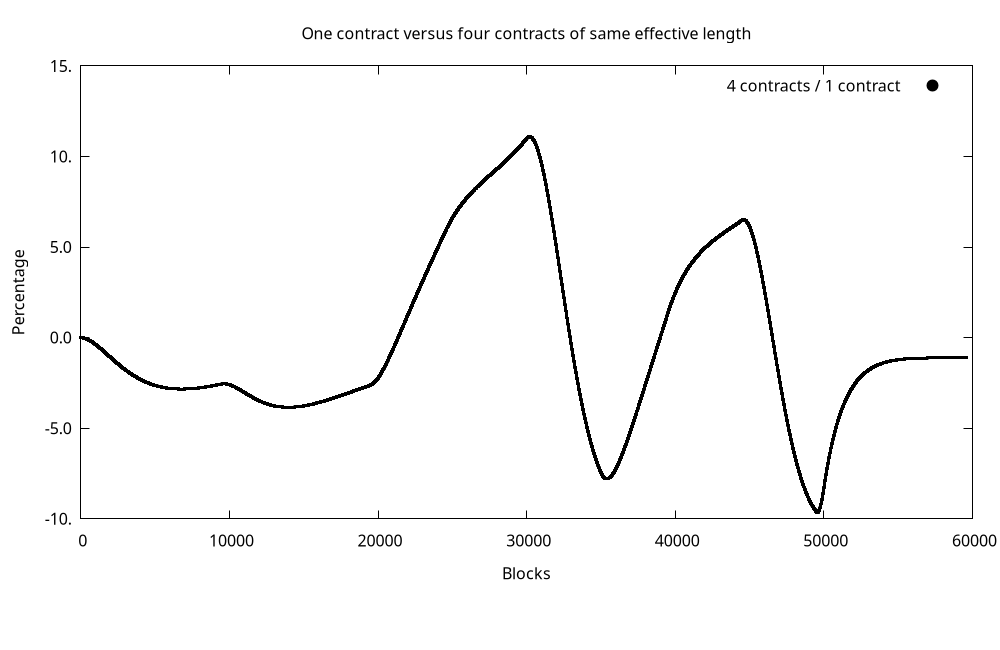}	
		\caption{A relative difference in $nextClaim$ between 4 $\mathcal{FRSC}$s and a single $\mathcal{FRSC}$.}\label{fig:exp3-relative-diff}
	\end{subfigure}
	
	\vspace{0.4cm}
	\caption{Experiment III comparing 4 $\mathcal{FRSC}$s and 1 $\mathcal{FRSC}$, both configurations having the same $\text{effective}\_\lambda$.  }\label{fig:effective_length}

\end{figure*}

\vspace{-0.1cm}
\subsubsection{\textbf{Different Fee Redistribution Ratios Across $\mathcal{FRSC}$s}}\label{sec:different-rhoss}
In \autoref{fig:percentages} we investigated different values of $\rho$ in the same set of four contracts and their impact on $\partial Claim$s.
The results show that the values of $\rho$ should correlate with $\lambda$ of multiple $\mathcal{FRSC}$s to maximize the potential of averaging by longer $\mathcal{FRSC}$s.

\subsubsection{Experiment III}\label{sec:exp3}
\paragraph{\textbf{Methodology.}}
In this experiment, we investigated whether it is possible to use a single $\mathcal{FRSC}$ setup to replace a multiple $\mathcal{FRSC}$s while preserving the same effect on the $nextClaim$.
To quantify a difference between such cases, we introduced a new metric of $\mathcal{FRSC}s$, called $\text{effective}\_\lambda$, which can be calculated as follows:
\begin{eqnarray} \label{eq:effective_length}	
	\text{effective}\_\lambda (\mathcal{FRSC}s) = \sum_{x ~\in~ \mathcal{FRSC}s} x.\rho* x.\lambda.
\end{eqnarray}
We were interested in comparing a single $\mathcal{FRSC}$ with 4 $\mathcal{FRSC}$s, both configurations having the equal $\text{effective}\_\lambda$. 
The configurations of these two cases are as follows: 
\begin{center}
	(1) $\mathcal{FRSC}(\_, 5292, 1)$ and\\
	(2)
	$\mathcal{FRSC}s=\{$\\
	$\mathcal{FRSC}^{1}(\_, 1008, 0.07), \mathcal{FRSC}^{2}(\_, 2016, 0.19),$\\
	$ \mathcal{FRSC}^{3}(\_, 4032, 0.28), \mathcal{FRSC}^{4}(\_, 8064, 0.46)\}.$\\
\end{center}
We can easily verify that the $\text{effective}\_\lambda$ of 4 $\mathcal{FRSC}$s is the same as in a single $\mathcal{FRSC}$ using \autoref{eq:effective_length}:
$0.07 * 1008 + 0.19 * 2016 + 0.28 * 4032 + 0.46 * 8064 = 5292$.

We conducted this experiment using a custom fee evolution scenario involving mainly linearly increasing/decreasing fees in the mempool (see \autoref{fig:exp3-fees}), and we set $\mathbb{C}$ to 0.7 for both configurations.
The custom scenario of the fee evolution in mempool in this experiment was chosen to contain extreme changes in fees, emphasizing possible differences in two investigated setups.

\paragraph{\textbf{Results.}}
In \autoref{fig:exp3-relative-diff}, we show the relative difference in percentages of $nextClaim$ rewards between the settings of 4 $\mathcal{FRSC}$s versus 1 $\mathcal{FRSC}$.
It is clear that the setting of 4 $\mathcal{FRSC}$s  in contrast to a single $\mathcal{FRSC}$ provided better reward compensation in times of very low fees value in the mempool, while it provided smaller reward in the times of higher values of fees in the mempool.
Therefore, we concluded that it is not possible to replace a setup of multiple $\mathcal{FRSC}$s with a single one while retaining the same fee redistribution behavior. 

\subsubsection{Experiment IV}\label{sec:exp5}
We focused on reproducing the experiment from Section 5.5 of~\cite{carlsten2016instability}, while utilizing our approach. 
The experiment is aimed on searching for the minimal ratio of \textsc{Default-Compliant} miners, at which the undercutting attack is no longer profitable strategy.
\textsc{Default-Compliant} miners are honest miners who follow the rules of the consensus protocol such as building on top of the longest chain.
%
We executed several simulations, each consisting of multiple games (i.e., 300k as in~\cite{carlsten2016instability}) with various fractions of \textsc{Default-Compliant} miners.
From the remaining miners we evenly created \textit{learning miners}, who learn on the previous runs of games and switch with a certain probability the best strategy out of the following:
\begin{compactitem}
	\item \textsc{PettyCompliant}: This miner behaves as \textsc{Default-Compliant} except one difference.
	In the case of seeing two chains, he does not mine on the oldest block but rather the most profitable block.
	Thus, this miner is not the (directly) attacking miner.
	
	\item \textsc{LazyFork}:
	This miner checks which out of two options is more profitable: (1) mining on the longest-chain block or (2) undercutting that block.
	In either way, he leaves half of the mempool fees for the next miners, which prevents another \textsc{LazyFork} miner to undercut him.
	
	\item \textsc{Function-Fork()}
	The behavior of the miner can be parametrized with a function f(.) expressing the level of his undercutting.
	The higher the output number the less reward he receives and more he leaves to incentivize other miners to mine on top of his block.
	This miner undercuts every time he forks the chain. 
\end{compactitem}

\paragraph{\textbf{Methodology.}}
With the missing feature for difficulty re-adjustment  (in the simulator from~\cite{carlsten2016instability} that we extended) the higher orphan rate occurs, which might directly impact our $\mathcal{FRSC}$-based approach. 
If the orphan rate is around 40\%, roughly corresponding to~\cite{carlsten2016instability}, our blocks would take on average 40\% longer time to be created, increasing the block creation time (i.e., time to mine a block).
This does not affect the original simulator, as there are no $\mathcal{FRSC}s$ that would change the total reward for the miner who found the block.


\begin{figure*}

	\centering
	\begin{subfigure}[t]{0.47\textwidth}
		\includegraphics[width=\textwidth]{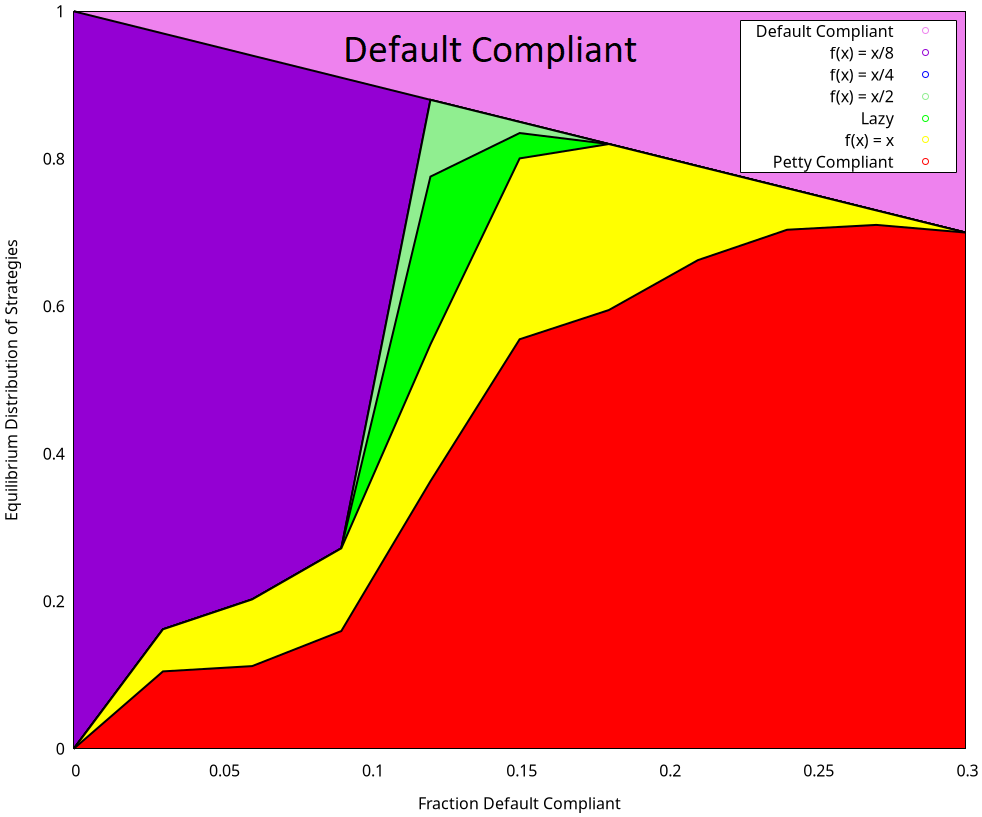}
		\caption{Simulations of our approach.}
	\end{subfigure}
	\begin{subfigure}[t]{0.46\textwidth}
		\includegraphics[width=\textwidth]{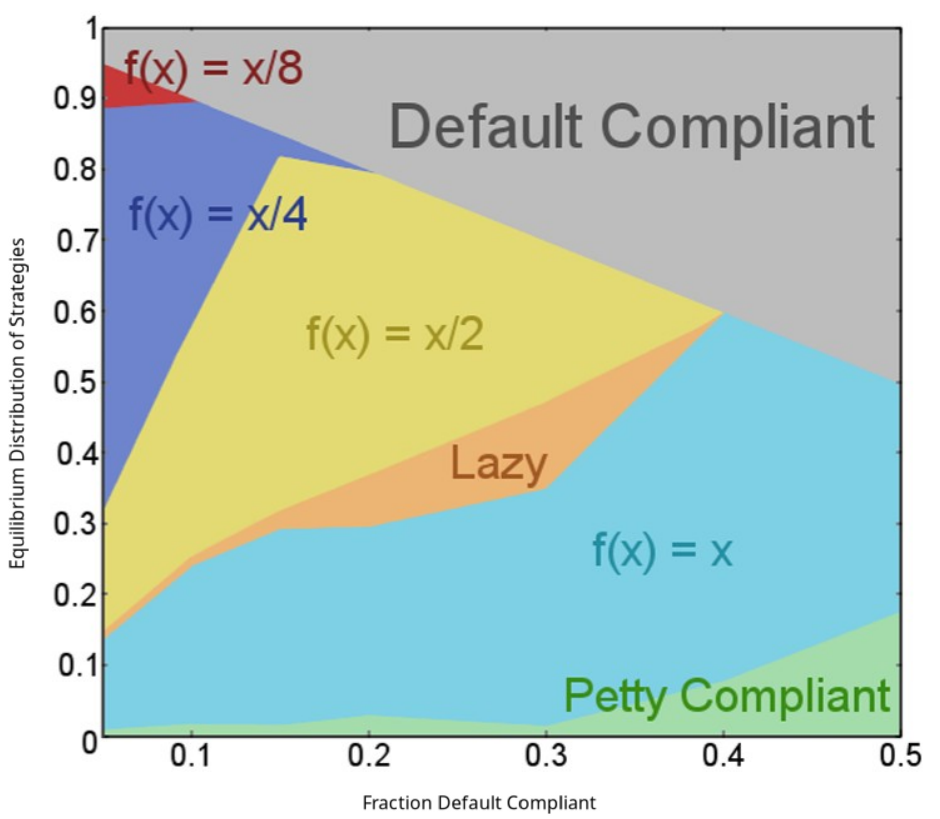}
		\caption{Simulations of the original work~\cite{carlsten2016instability}.}
	\end{subfigure}

	\vspace{0.4cm}
	\caption{Experiment IV -- The ratio of \textsc{Default-Compliant}  miners in our approach is $\sim$30\% (in contrast to $\sim66\%$ of~\cite{carlsten2016instability}).}\label{fig:improvement}

\end{figure*}

Nevertheless, this is not true for $\mathcal{FRSC}$-based simulations as the initial setup of $\mathcal{FRSC}s$ is calculated with $\overline{fees} = 50$ BTC (as per the original simulations).
However, with longer block creation time and transaction fees being calculated from it, the amount of $\overline{fees}$ also changes.
With no adjustments, this results in $\mathcal{FRSC}s$ initially paying smaller reward back to the miner before $\mathcal{FRSC}s$ are saturated. 
To mitigate this problem, we increased the initial values of individual $\mathcal{FRSC}$s by the orphan rate from the previous game before each run.
This results in very similar conditions, which can be verified  by comparing the final value in the longest chain of our simulation versus the original simulations.
We decided to use this approach to be as close as possible to the original experiment.
This is particularly important when the full mempool parameter is equal to $false$, which means that the incoming transaction fees to mempool are calculated based on the block creation time. 
In our simulations, we used the following parameters: 100 miners,
10 000 blocks per game, 
300 000 games (in each simulation run),
exp3 learning model, and
$\mathbb{C} = 0.7$.
Modeling of fees utilized the same parameters as in the original paper~\cite{carlsten2016instability}: the full mempool parameter disabled, a constant inflow of 5 000 000 000 Satoshi (i.e., 50 BTC) every 600s.
For more details about the learning strategies and other parameters, we refer the reader to~\cite{carlsten2016instability}.

\paragraph{\textbf{Setup of $\mathcal{FRSC}s$}.}
Since we have a steady inflow of fees to the mempool, we do not need to average the income for the miner.
Therefore, we used only a single $\mathcal{FRSC}$ defined as $\mathcal{FRSC}$(7 056 000 000 000, 2016, 1), where the initial value of $\mathcal{FRSC}.\nu$ was adjusted according to \autoref{eq:setup}, assuming $\overline{fees} = 50$ BTC.
In the subsequent runs of each game, $\mathcal{FRSC}.\nu$ was increased by the orphan rate from the previous runs. 


\paragraph{\textbf{Results.}}
The results of this experiment, depicted in \autoref{fig:improvement}, demonstrate, that with our approach using $\mathcal{FRSC}$s, we decreased the number of \textsc{Default-Compliant} miners from the original $66\%$ to $30\%$. 
This means that the profitability of undercutting miners is avoided with at least $30\%$ of \textsc{Default-Compliant} miners, indicating more robust results. 

\medskip
\noindent
For other details, we refer the reader to our paper~\cite{budinsky2023fee}.

\section{Contributing Papers}\label{sec:cons-papers}
The papers that contributed to this research direction are enumerated in the following, while highlighted papers are attached to this thesis in their original form.

\begingroup
\let\clearpage\relax

\renewcommand\bibname{}

\endgroup

%% file: sec/wallets.tex

In this chapter, we present our contributions to the area of authentication for blockchain and decentralized applications, which belong to the application layer of our security reference architecture (see \autoref{chapter:sra}).
In particular, this chapter is focused on cryptocurrency wallets and their subcategory of smart contract wallets, and it is based on the papers \cite{homoliak2020smartotps,homoliak2020air-extend} (see also \autoref{sec:wallets-papers}).

First, we review existing cryptocurrency wallet solutions (with their security issues) and propose a classification scheme based on authentication factors validated against the blockchain or a centralized party \cite{homoliak2020air-extend}. 
We apply the proposed classification to the existing wallet solutions and also cross-compare other security features of them.
Next, we propose SmartOTPs \cite{homoliak2020smartotps}, a 2FA authentication scheme against the blockchain, which, on top of using a hardware wallet, introduces the authenticator App (or a device) generating OTPs that are transferred in an air-gapped fashion to the client.

\subsubsection{Notation}\label{sec:wallets-notation}
We denote the user by $\mathbb{U}$, the client (e.g., the user agent/browser) by $\mathbb{C}$, a wallet holding a private key by $\mathbb{W}$, the authenticator device or App as $\mathbb{A}$, and an adversary by $\mathcal{A}$.

\section{Security Issues in Authentication Schemes of Wallets}\label{sec:wallets-review}
According to works \cite{eskandari2018first,2015-Bitcoin-SOK}, there are a few categories of key management approaches. 
In pass\-word-protected wallets, private keys are encrypted with selected passwords.
Unfortunately, users often choose weak passwords that can be brute-forced if stolen by malware \cite{2015-CCSM-SecureWorks}; optionally, such malware may use a keylogger for capturing a passphrase \cite{2015-Bitcoin-SOK,2017-keylogger-bc-malware}.
Another similar option is to use password-derived wallets that generate keys based on the provided password.
However, they also suffer from the possibility of weak passwords \cite{courtois2016speed}. 
Hardware wallets are a category that promises the provision of better security by introducing devices that enable only the signing of transactions, without revealing the private keys stored on the device.
However, these wallets do not provide protection from an attacker with full access to the device \cite{kraken-trezor-hack,kraken-keepkey-hack,donjon-ellipal-hack}, and more importantly, wallets that do not have a secure channel for informing the user about the details of a transaction being signed (e.g., \cite{ledger-nano}) may be exploited by malware targeting IPC mechanisms \cite{bui2018man}.

A popular option for storing private keys is to deposit them into  server-side hosted (i.e., custodial) wallets and currency-exchange services \cite{CoinbaseWallet,binance-exchange,poloniex-exchange,kraken-exchange,luno-wallet,paxful-wallet}.
In contrast to the previous categories, server-side wallets imply trust in a provider, which is a potential risk of this category.
Due to many cases of compromising server-side wallets \cite{2018-coindesk-bithumb,2014-Mt-Gox,2016-Bitfinex-hack,moore2013beware,binance-hack-2019} or fraudulent currency-exchange operators \cite{vasek2015there}, client-side hosted wallets have started to proliferate.
In such wallets, the main functionality, including the storage of private keys, has moved to the user side \cite{mycelium-wallet,CarbonWallet,CitoWiseWallet,coinomi-wallet,InfinitoWallet};
hence, trust in the provider is reduced but the users still depend on the provider's infrastructure.

To increase security of former wallet categories, multi-factor authentication (MFA) is often used, which enables spending crypto-tokens only  when a number of secrets are used together.
Wallets from a split control category \cite{eskandari2018first} provide MFA against the blockchain.
This can be achieved by threshold cryptography wallets \cite{goldfeder2015securing,mycelium-entropy}, multi-signature wallets \cite{Armory-SW-Wallet,Electrum-SW-Wallet,TrustedCoin-cosign,copay-wallet}, and state-aware smart-contract wallets \cite{TrezorMultisig2of3,parity-wallet,ConsenSys-gnosis}.
Nevertheless, these schemes might impose additional usability implications, performance overhead, or cost of wallet devices.

\section{Classification of Authentication Schemes}\label{sec:wallets-classification}
We introduce the notion of $k$-factor authentication against the blockchain and $k$-factor authentication against the authentication factors.
Using these notions, we propose a classification of authentication schemes, and we apply it to examples of existing key management solutions (see \autoref{sec:soa-wallet-types} and \autoref{appendix:classification}).

In the context of the blockchain, we distinguish between k-factor authentication  \textit{against the blockchain} and k-factor authentication \textit{against the authentication factors} themselves.
For example, an authentication method may require the user to perform 2-of-2 multi-signature in order to execute a transfer, while $\mathbb{U}$ may keep each private key stored in a dedicated device -- each requiring a different password.
In this case, 2FA is performed against the blockchain, since both signatures are verified by all miners of the blockchain.
Additionally, a one-factor authentication is performed once in each device of $\mathbb{U}$ by entering a password in each of them.
For clarity, we classify authentication schemes by the following notation:
\begin{equation*}\label{eqn:auth-simple}
	\Bigg ( 
	Z 
	+  X_1
	\big/ \ldots \big/ 
	X_Z
	\Bigg),
\end{equation*}
where $Z \in \{0, 1, \ldots \}$ represents the number of authentication factors against the block\-chain and
$X_i \in \{0, 1, \ldots \} ~|~i \in [1,\ldots, Z]$ represents the number of authentication factors against the i-th factor of $Z$.
With this in mind, we remark that the previous example provides $\left( 2 + 1/1 \right)$-factor authentication: twice against the blockchain (i.e., two signatures), once for accessing the first device (i.e., the first password), and once for accessing the second device (i.e., the second password).

\medskip
Since the previous notation is insufficient for authentication schemes that use secret sharing \cite{shamir1979share}, we
extend it as follows:
\begin{equation*}
	\Bigg ( 
	Z^{(W_1, \dots, W_Z)} 
	+ \left( X_1^{1}, \ldots , X_1^{W_1} \right) 
	\big/ \ldots \big/ 
	\left(X_Z^{1}, \ldots, X_Z^{W_Z} \right) 
	\Bigg),
\end{equation*}
where $Z$ has the same meaning as in the previous case, 
$W_i \in \{0, 1, \ldots \}$ $|~i \in [1,\ldots, Z]$ denotes the minimum number of secret shares required to use the complete i-th secret $X_i$.
With this in mind, we remark that the aforementioned example provides $\left( 2^{(1, 1)} + (1)/(1) \right)$-factor authentication: twice against the blockchain (i.e., two signatures), once for accessing the first device (i.e., the first password), and once for accessing the second device (i.e., the second password).
We consider an implicit value of $W_i = 1$; hence, the classification $(2 + 1/1)$ represents the same as the previous one (the first notation suffices).
If one of the private keys were additionally split into two shares, each encrypted by a password, then such an approach would provide $\left( 2^{(2, 1)} + (1, 1)/(1) \right)$-factor authentication.

\subsection{Review of Wallet Types Using the Classification}\label{sec:soa-wallet-types}
We extend the previous work of Eskandari et al. \cite{eskandari2018first} and Bonneau et al. \cite{2015-Bitcoin-SOK}, by categorizing and reviewing a few examples of key management solutions, while assuming our classification.

\paragraph{\textbf{Keys in Local Storage.}}
In this category of wallets, the private keys are stored in plaintext form on the local storage of a machine, thus providing $(1+0)$-factor authentication.
Examples that enable the use of unencrypted private key files are Bitcoin Core \cite{BitcoinCore} or MyEtherWallet \cite{MyEtherWallet} wallets.

\paragraph{\textbf{Password-Protected Wallets.}}
These wallets require the user-spe\-ci\-fied password to encrypt a private key stored on the local storage, thus providing $(1+1)$-factor authentication.
Examples that support this functionality are Armory Secure Wallet \cite{Armory-SW-Wallet}, Electrum Wallet \cite{Electrum-SW-Wallet}, MyEtherWallet \cite{MyEtherWallet}, Bitcoin Core \cite{BitcoinCore}, and Bitcoin Wallet \cite{BitcoinWallet}.
This category addresses physical theft, yet enables the brute force  of passwords and digital theft (e.g., keylogger).

\paragraph{\textbf{Password-Derived Wallets.}}
Password-derived wallets \cite{maxwell2011deterministic} (a.k.a., brain wallets or hierarchical deterministic wallets) can compute a sequence of private keys from only a single mnemonic string and/or password.
This approach takes advantage of the key creation in the ECDSA signature scheme that is used by many blockchain platforms.
Examples of password-derived wallets are Electrum \cite{Electrum-SW-Wallet}, Armory Secure Wallet \cite{Armory-SW-Wallet}, Metamask \cite{MetamaskWallet}, and Daedalus Wallet \cite{daedalus-wallet}.
The wallets in this category provide $(1+X_1)$-factor authentication (usually $X_1 = 1$) and also suffer from weak passwords \cite{courtois2016speed}.

\paragraph{\textbf{Hardware Storage Wallets.}}
In general, wallets of this category include devices that can only sign transactions by private keys stored inside sealed storage, while the keys never leave the device. 
To sign a transaction, $\mathbb{U}$ connects the device to a machine and enters his passphrase.
When signing a transaction, the device displays the transaction's data to $\mathbb{U}$, who may verify the details. 
Thus, wallets of this category usually provide $(1+1)$-factor authentication.
Popular USB (or Bluetooth) hardware wallets containing displays are offered by Trezor \cite{trezor-hw-wallet},  Le\-dger \cite{ledger-nano-s}, KeepKey \cite{keep-key}, and BitLox \cite{BitLox}.
An example of a USB wallet that is not resistant against tampering with $\mathbb{C}$ (e.g., keyloggers) is Ledger Nano \cite{ledger-nano}
-- it does not have a display, hence $\mathbb{U}$ cannot verify the details of transactions being signed.
An air-gapped transfer of transactions using QR codes is provided by ELLIPAL wallet \cite{ellipal-hw-wallet}.
In ELLIPAL, both $\mathbb{C}$ (e.g., smartphone App) and the hardware wallet must be equipped with cameras and display.
$(1+0)$-factor authentication is provided by a credit-card-shaped hardware wallet from CoolBitX \cite{CoolWalletS}. 
A hybrid approach that relies on a server providing a relay for 2FA is offered by BitBox \cite{BitBox}.
Although a BitBox device does not have a display, after connecting to a machine, it communicates with $\mathbb{C}$ running on the machine and at the same time, it communicates with a smartphone App through BitBox's server; 
each requested transaction is displayed and confirmed by $\mathbb{U}$ on the smartphone. 
One limitation of this solution is the lack of self-sovereignty.

\paragraph{\textbf{Split Control -- Threshold Cryptography.}}
In threshold cryptography \cite{shamir1979share,threshold-mackenzie2001two,threshold-gennaro2007secure,blakley1979safeguarding}, a key is split into several parties which enables the spending of crypto-tokens only when n-of-m parties collaborate.
Threshold cryptography wallet provide $\left( 1^{(W_1, \ldots, W_n)}\- + (X_1, \dots, X_n) \right)$-factor authentication, as only a single signature verification is made on a blockchain, but $n$ verifications are made by parties that compute a signature.
Therefore, all the computations for co-signing a transaction are performed off-chain, which provides anonymity of access control policies (i.e., a transaction has a single signature) in contrast to the multi-signature scheme that is publicly visible on the blockchain.
An example of this category is presented by Goldfeder et al. \cite{goldfeder2015securing}.
One limitation of this solution is a computational overhead that is directly proportional to the number of involved parties $m$ (e.g., for $m = 2$ it takes $13.26$s).
Another example of this category is a USB dongle called Mycelium Entropy \cite{mycelium-entropy}, which, when connected to a printer, generates triplets of paper wallets using 2-of-3 Shamir's secret sharing; providing $(1^{(2)} + (0, ~0))$-factor authentication.

\paragraph{\textbf{Split Control -- Multi-Signature Wallets.}}
In the case of multi-signature wallets, n-of-m owners of the wallet must co-sign the transaction made from the multi-owned address.
Thus, the wallets of this category provide $(n + X_1/\ldots/X_n)$-factor authentication.
One example of a multi-owned address approach is Bitcoin's Pay to Script Hash (P2SH).\footnote{We refer to the term  \textit{multi-owned address of P2SH} for clarity, although it can be viewed as Turing-incomplete smart contract.}
Examples supporting multi-owned addresses are Lockboxes of Armory Secure Wallet \cite{Armory-SW-Wallet} and Electrum Wallet \cite{Electrum-SW-Wallet}.
A property of multi-owned address is that each transaction with such an address requires off-chain communication. 
A hybrid instance of this category and client-side hosted wallets category is Trusted Coin's cosigning service \cite{TrustedCoin-cosign}, which provides a 2-of-3 multi-signature scheme -- $\mathbb{U}$ owns a primary and a backup key, while TrustedCoin owns the third key.
Each transaction is signed first by user's primary key and then, based on the correctness of the OTP from Google Authenticator, by TrustedCoin's key.
Another hybrid instance of this category and client-side hosted wallets is Copay Wallet \cite{copay-wallet}.
With Copay, the user can create a multi-owned Copay wallet, where $\mathbb{U}$ has all keys in his machines and each transaction is co-signed by n-of-m keys.
Transactions are resent across user's machines during multi-signing through Copay.

\paragraph{\textbf{Split-Control -- State-Aware Smart Contracts.}}\label{sec:state-aware=contracts}
State-aware smart contracts provide ``rules'' for how crypto-tokens of a contract can be spent by owners, while they keep the current setting of the rules on the blockchain.
The most common example of state-aware smart contracts is the 2-of-3 multi-signature scheme that provides $(2+X_1/X_2)$-factor authentication.
An example of the 2-of-3 multi-signature approach that only supports Trezor hardware wallets is \textit{TrezorMultisig2of3} from Unchained Capital \cite{TrezorMultisig2of3}.
One disadvantage of this solution is that $\mathbb{U}$ has to own three Trezor devices, which may be an expensive solution that, moreover, relies only on a single vendor.
Another example of this category, but using the n-of-m multi-signature scheme, is Parity Wallet \cite{parity-wallet}. 
However, two critical bugs \cite{parity-bug-July-17,parity-bug-November-17} have caused the multi-signature scheme to be currently disabled.
The n-of-m multi-signature scheme is also used in \textit{Gnosis Wallet} from ConsenSys \cite{ConsenSys-gnosis}.

\paragraph{\textbf{Hosted  Wallets.}}
Common features of hosted wallets are that they provide an online interface for interaction with the blockchain, managing crypto-tokens, and viewing transaction history, while they also store private keys at the server side.
If a hosted wallet has full control over private keys, it is referred to as a \textit{server-side wallet}. 
A server-side wallet acts like a bank -- the trust is centralized.
Due to several cases of compromising such server-side wallets \cite{2018-coindesk-bithumb}, \cite{2014-Mt-Gox}, \cite{2016-Bitfinex-hack}, \cite{moore2013beware}, the hosted wallets that provide only an interface for interaction with the blockchain (or store only user-encrypted private keys) have started to proliferate.
In such wallets, the functionality, including the storage of private keys, has moved to $\mathbb{U}$'s browser (i.e., client).
We refer to these kinds of wallets as \textit{client-side wallets} (a.k.a., hybrid wallets \cite{eskandari2018first}).

\paragraph{\textbf{Server-Side Wallets.}}
Coinbase \cite{CoinbaseWallet} is an example of a server-side hosted wallet, which also provides exchange services.
Whenever the user logs in or performs an operation, he authenticates himself against Coinbase's server using a password and obtains a code from Google Authenticator/Authy app/SMS. 
Other examples of server-side wallets having equivalent security level to Coinbase are Circle Pay Wallet \cite{CircleWallet} and Luno Wallet \cite{luno-wallet}.
The wallets in this category provide $(0+2)$-factor authentication when 2FA is enabled.

\paragraph{\textbf{Client-Side Wallets.}}
An example of a client-side hosted wallet is Blockchain Wallet \cite{BlockchainInfoWallet}. 
Blockchain Wallet is a password-derived wallet that provides 1-factor authentication against the server based on the knowledge of a password and additionally enables 2FA against the server through one of the options consisting of Google Authenticator, YubiKey, SMS, and email.
When creating a transaction, $\mathbb{U}$ can be authenticated by entering his secondary password. 
Equivalent functionality and security level as in Blockchain Wallet are offered by BTC Wallet \cite{BTC-com-wallet}. 
In contrast to Blockchain Wallet, BTC wallet uses 2FA also during the confirmation of a transaction. 
Other examples of this category are password-derived wallets, like Mycelium Wallet \cite{mycelium-wallet}, CarbonWallet \cite{CarbonWallet}, Citowise Wallet \cite{CitoWiseWallet}, Coinomi Wallet \cite{coinomi-wallet}, and Infinito Wallet \cite{InfinitoWallet}, which, in contrast to the previous examples, do not store backups of encrypted keys at the server.
A 2FA feature is provided additionally to password-based authentication, in the case of CarbonWallet.
In detail, the 2-of-2 multi-signature scheme uses the machine's browser and the smartphone's browser (or the app) to co-sign transactions.

\setlength{\tabcolsep}{2.0pt}
\begin{table*}[!h]
	\scriptsize{
		\vspace{-0.3cm}
		\scalebox{0.72}{
\begin{tabular}{lccccccccccccc}
				\toprule
				&
				\multicolumn{2}{c}{\specialcell{\textbf{Authentication Scheme}\\\\}} & \rotatebox[origin=l]{90}{\textbf{\specialcell{Air-Gapped Property}}} &
				\rotatebox[origin=l]{90}{\textbf{\specialcell{Resilience to Tampering\\ with the Client}}} &
				\rotatebox[origin=l]{90}{\textbf{\specialcell{Post-Quantum Resilience}}} &
				\rotatebox[origin=l]{90}{\textbf{\specialcell{No Need for Off-Chain\\Communication}}} &
				\rotatebox[origin=l]{90}{\textbf{\specialcell{Malware Resistance}}} &
				\rotatebox[origin=l]{90}{\textbf{\specialcell{Secret(s) Kept Offline}}} &
				\rotatebox[origin=l]{90}{\textbf{\specialcell{Independence of\\Trusted Third Party}}} &
				\rotatebox[origin=l]{90}{\textbf{\specialcell{Resilience to Physical Theft}}} &
				\rotatebox[origin=l]{90}{\textbf{\specialcell{Resilience to Password Loss}}} &
				\rotatebox[origin=l]{0}{\textbf{\specialcell{Comments}}} \\
				\expandafter\cline\expandafter{\expandafter2\string-3\smallskip}
				&
				\textbf{Classification} &
				\textbf{Details} &
				& & & & & & & & &  \\
				\toprule
				
				\smallskip
				\textbf{Keys in Local Storage}                       & $1 + (0)$          & Private key                                                   &                                               &                                                         &                                                   &                                                   &                                             &                                                  &                                                        &                                                       &                                                      &                                                                                                                  \\
\bottomrule
			\end{tabular}
		}
	}
	\caption{Comparison of state-of-the-art cryptocurrency wallets using our classification (see \autoref{sec:wallets-classification}) and other security features.}	
	\label{tab:wallets-state-of-the-art}
\end{table*}
\setlength{\tabcolsep}{1.4pt}

\subsection{A Comparison of Security Features of Wallets}\label{appendix:classification}
\hfill

\noindent
We present a comparison of wallets and approaches from \autoref{sec:soa-wallet-types} in \autoref{tab:wallets-state-of-the-art}.
We apply our proposed classification on authentication schemes, while we also survey a few selected security and usability properties of the wallets from the work of Eskandri et al. \cite{eskandari2018first}. 
In the following, we briefly describe each property and explain the criteria stating how we attributed the properties to particular wallets. 

\paragraph{\textbf{Air-Gapped Property.}}
We attribute this property (Y) to approaches that involve at least one hardware device storing secret information, which do not need a connection to a machine in order to operate.

\paragraph{\textbf{Resilience to Tampering with the Client.}}
We attribute this property (Y) to all hardware wallets that sign transactions within a device, while they require $\mathbb{U}$ to confirm transaction's details at the device (based on displayed information).
Then, we attribute this property to wallets containing multiple clients that collaborate in several steps to co-signs transactions (a chance that all of them are tampered with is low).

\paragraph{\textbf{Post-Quantum Resilience.}}
We attribute this property (Y) to approaches that utilize hash-based cryptography that is known to be resilient against quantum computing attacks \cite{amy2016estimating}.

\paragraph{\textbf{No Need for Off-Chain Communication.}}
We attribute this property (Y) to approaches that do not require an off-chain communication/transfer of transaction among parties/devices
to build a final (co-)signed transaction, before submitting it to a blockchain (applicable only for $Z \geq 2$ or $W_i \ge 2$).

\paragraph{\textbf{Malware Resistance (e.g., Key-Loggers).}}
We attribute this property (Y) to approaches that either enable signing  transactions inside of a sealed device or split signing control over secrets across multiple devices. 

\paragraph{\textbf{Secret(s) Kept Offline.}}
We attribute this property (Y) to approaches that keep secrets inside their sealed storage, while they expose only signing functionality.
Next, we attribute this property to paper wallets and fully air-gapped devices.

\paragraph{\textbf{Independence of Trusted Third Party}}
We attribute this property (Y) to approaches that do not require trusted party for operation, while we do not attribute this property to all client-side and server-side hosted wallets.
We partially (P) attribute this property to approaches requiring an external relay server for their operation.

\paragraph{\textbf{Resilience to Physical Theft.}}
We attribute this property (Y) to approaches that are protected by an encryption password or PIN.
We partially (P) attribute this property to approaches that do not provide password and PIN protection but have a specific feature to enforce uniqueness of an environment in which they are used (e.g., bluetooth pairing). 

\paragraph{\textbf{Resilience to Password Loss.}}
We attribute this property (Y) to approaches that provide means for recovery of secrets (e.g., a seed of hierarchical deterministic wallets).

\section{SmartOTPs}\label{sec:wallets-SmartOTPs}
In this section, we propose SmartOTPs, a smart-contract wallet framework that gives a flexible, usable, and secure way of managing crypto-tokens in a self-sovereign fashion.
The proposed framework consists of four components (i.e., an authenticator, a client, a hardware wallet, and a smart contract), and it provides 2-factor authentication (2FA) performed in two stages of interaction with the blockchain.
To the best of our knowledge, our framework is the first one that utilizes one-time passwords (OTPs) in the setting of the public blockchain.
In SmartOTPs, the OTPs are aggregated by a Merkle tree and hash chains whereby for each authentication only a short OTP (e.g., 16B-long) is transferred from the authenticator to the client.
Such a novel setting enables us to make a fully air-gapped authenticator by utilizing small QR codes or a few mnemonic words, while additionally offering resilience against quantum cryptanalysis.
We have made a proof-of-concept based on the Ethereum platform. 
Our cost analysis shows that the average cost of a transfer operation is comparable to existing 2FA solutions using smart contracts with multi-signatures.

\subsubsection{Notation}\label{sec:notation}
By the term \textit{operation} we refer to an action with a smart-contract wallet using SmartOTPs, which may involve, for instance, a transfer of crypto-tokens or a change of daily spending limits.
Then, we use the term \textit{transfer} for the indication of transferring crypto-tokens. 
By $\{msg\}_\mathbb{U}$ we denote the message $msg$ digitally signed by $\mathbb{U}$, and by $msg.\sigma$ we refer to the signature;
$\mathcal{RO}$ is the random oracle;
$h(.)$: stands for a cryptographic hash function;
$h^i(.)$ substitutes $i$-times chained function $h(.)$, e.g., $h^2(.) \equiv h(h(.))$;
$\|$ is the string concatenation;
$h_{\mathcal{D}}^i(.)$ substitutes $i$-times chained function $h(.)$ with embedded domain separation, e.g., $h_{\mathcal{D}}^2(.) = h(2 ~||~ h(1~||~.))$;
$F_k(.) \equiv h(k ~\|~ .)$ denotes a pseudo-random function that is parametrized by a secret seed $k$;
$\%$ represents modulo operation over integers; 
$\Sigma. \{KeyGen, Verify, Sign\}$ represents a signature scheme of the blockchain platform;
$SK_\mathbb{U}$, $PK_\mathbb{U}$ is the private/public key-pair of $\mathbb{U}$, under $\Sigma$,
and $a~|~b$ represents bitwise OR of arguments $a$ and $b$.

\subsection{Problem Definition}
The main goal of this research is to propose a cryptocurrency wallet framework that provides a secure and usable way of managing crypto-tokens. 
In particular, we aim to achieve:
\begin{compactdesc}
	\item[Self-Sovereignty:] 
	ensures that the user does not depend on the 3rd party's infrastructure, and the user does not share his secrets with anybody.
	Self-sovereign (i.e., non-custodial) wallets do not pose a single point of failure in contrast to server-side (i.e., custodial) wallets, which when compromised, resulted in huge financial loses \cite{2018-coindesk-bithumb,2014-Mt-Gox,2016-Bitfinex-hack,moore2013beware,binance-hack-2019}.
	\item[Security:] the insufficient security level of some self-sovereign wallets has caused significant financial losses for individuals and companies \cite{2016-brainwallets,courtois2016speed,CHHMPRSS18,parity-bug-July-17}. 
	We argue that wallets should be designed with security in mind and in particular, we point out 2FA solutions, which have successfully contributed to the security of other environments \cite{aloul2009two,schneier2005two}.
	Our motivation is to provide a cheap security extension of the hardware wallets (i.e., the first factor) by using OTPs as the second factor in a fashion similar to Google Authenticator.   
\end{compactdesc}

\subsection{Threat Model}\label{sec:threat-model}
For a generic cryptocurrency, we assume an adversary $\mathcal{A}$ whose goal is to conduct unauthorized operations on the user's behalf or render the user's wallet unusable. 
$\mathcal{A}$ is able to eavesdrop on the network traffic as well as to participate in the underlying consensus protocol. 
However, $\mathcal{A}$ is unable to take over the cryptocurrency platform nor to break the used cryptographic primitives.
We further assume that $\mathcal{A}$ is able to intercept and ``override'' the user's transactions, e.g., by launching a man-in-the-middle (MITM) attack or by creating a conflicting malicious transaction with a higher fee, which will incentivize miners to include $\mathcal{A}$'s transaction and discard the user's one; this attack is also referred to as \textit{transaction front-running}.
We assume three types of exclusively occurring attackers, each targeting one of the three components of our framework: (1) $\mathcal{A}$ with access to the user's private key hardware wallet $\mathbb{W}$, (2) $\mathcal{A}$ that tampers with the client $\mathbb{C}$, and for completeness we also assume (3) $\mathcal{A}$ with access to the authenticator $\mathbb{A}$.
Next, we assume that the legitimate user correctly executes the proposed protocols and $h(.)$ is an instantiation of random oracle $\mathcal{RO}$.

\subsection{Design Space}\label{sec:design-space}
There are many types of wallets with different properties (see \autoref{sec:soa-wallet-types}).
In our context, to achieve self-sovereignty we identify smart-contract wallets as a promising category.
These wallets manage crypto-tokens by the functionality of smart contracts, enabling users to have customized control over their wallets. 
The advantages of these solutions are that spending rules can be explicitly specified and then enforced by the cryptocurrency platform itself.
Therefore, using this approach, it is possible to build a flexible wallet with features such as daily spending limits or transfer limits.

\subsubsection{General OTPs}
With spending rules encoded in a smart contract, it is feasible to design custom security features, such as OTP-based authentication serving as the second factor. 
In such a setting, the authenticator produces OTPs to authenticate transactions in the smart contract.
However, in contrast to digital signatures,
OTPs do not provide non-repudiation of data present in a transaction with an OTP; moreover, they can be intercepted and misused by the front-running or the MITM attacks.
To overcome this limitation, we argue that a two-stage protocol $\Pi_O^{<G>}$ must be employed, enabling secure utilization of general OTPs in the context of blockchains.
In the first stage of $\Pi_O^{<G>}$, an operation $O$, signed by the user $\mathbb{U}$, is submitted to the blockchain, where it obtains an identifier $i$. 
Then, in the second stage, $O_{i}$ is executed on the blockchain upon the submission of $OTP_{i}$ that is unambiguously associated with the operation initiated in the first stage. 		

\subsubsection{Requirements of General and Air-Gapped OTPs}
Based on the above, we define the necessary security requirements of general OTPs used in the blockchain as follows:
\begin{compactenum}
	\item \textbf{Authenticity:} each OTP must be associated only with a unique authenticator instance. 
	
	\item \textbf{Linkage:} each $OTP_{i}$ must be linked with exactly a single operation $O_i$, ensuring that $OTP_i$ cannot be misused for the authentication of $O_j, i \neq j$.
	
	\item \textbf{Independence:} $OTP_i$ linked with the operation $O_i$
	cannot be derived from $OTP_j$  of an operation $O_j$, where $i \neq j$, or an arbitrary set of other OTPs. 	
\end{compactenum}
Nevertheless, in the air-gapped setting (important for a high usability and security), one more requirement comes into play: \textbf{the short length of OTPs}. 
Short OTPs allow the users to use a relatively small number of mnemonic words or a small QR code to transfer an OTP in an air-gapped fashion.
This requirement is of high importance especially in the case when the authenticator is implemented as a resource-constrained embedded device with a small display (e.g., credit-card-shaped wallet, such as CoolBitX \cite{CoolWalletS}).

\subsubsection{Analysis of Existing Solutions}
We argue that not all solutions meet the requirements of air-gapped OTPs.
Asymmetric cryptography primitives such as digital signatures or zero-knowledge proofs are inadequate in this setting, despite meeting all general OTP requirements.
State-of-the-art signature schemes with reasonable performance overhead \cite{bernstein2012high,johnson2001elliptic} and short signature size produce a 48B-64B long output.
The BLS signatures \cite{boneh2001short} go even beyond the previous constructs and might produce signatures of size 32B. 
Nevertheless, BLS signatures are unattractive in the setting of the smart contract platforms that put high execution costs for BLS signature verification, which is $\sim$33 times more expensive than in the case of ECDSA with the equivalent security level \cite{RFC-BLS-signatures}.
Hence, we assume 48B as the minimal feasible OTP size for assymetric cryptography.

However, transferring even 48B in a fully air-gapped environment by transcription of mnemonic words \cite{bipMnemonic} would lack usability for regular users -- considering study from Dhakal et al. \cite{Dhakal-typing-study}, transcription of 36 English words takes 42s on average, which is much longer than users are willing to ``sacrifice.''
We note that the situation is better with QR code, but on the other hand it has two limitations: 
(1) when the authenticator is implemented as a simple embedded device, its display might be unable to fit a requested QR code with sufficient scanning properties (to preserve the maximal scanning distance of QR code, the ``denser'' QR code must be displayed in a larger image \cite{QR2011size}) and 
(2) occasionally, the users might not have a camera in their devices, thus, they can proceed only with a fallback method that uses mnemonics.
Finally, most of the currently deployed asymmetric constructions are vulnerable to quantum computing \cite{bernstein2009introduction}.

The problem of long signatures also exists in hash-based signature constructs \cite{lamport1979constructing-lamportSigs,dods2005hash-winternitz,merkle1989certified}. 
Lamport-Diffie one-time signatures (LD-OTS) \cite{lamport1979constructing-lamportSigs} produce an output of length $2|h(.)|^2$, which, for example in the case of $|h(.)| = 16B$ yields $4kB$-long signatures. 
The signature size of LD-OTS can be reduced by using one string of one-time key for simultaneous signing of several bits in the message digest (i.e., Winternitz one-time signatures (W-OTS) \cite{dods2005hash-winternitz}), but at the expense of exponentially increased number of hash computations (in the number of encoded bits) during a signature generation and verification.
The extreme case minimizing the size of W-OTS to $|h(.)|$ (for simplicity omitting checksum) would require $2^{|h(.)|}$ hash computations for signature generation, which is unfeasible. 

Approaches based on symmetric cryptography primitives produce much shorter outputs, but it is challenging to implement them with smart-contract wallets.
Widely used one-time passwords like HOTP \cite{m2005hotp} or TOTP \cite{m2011totp} require the user to share a secret key $k$ with the authentication server. 
Then, with each authentication request the user proves that he possesses $k$ by returning the output of an $F_k(.)$ computed with a nonce (i.e., HOTP) or the current timestamp (i.e., TOTP).
This approach is insecure in the setting of the blockchain, as the user would have to share the secret $k$ with a smart-contract wallet, making $k$ publicly visible. 

A solution that does not publicly disclose secret information and, at the same time, provides short enough OTPs (e.g., $16B \simeq 12$ mnemonic words $\simeq$ QR code v1), can be implemented by Lamport's hash chains \cite{lamport1981password} or other single hash-chain-based constructs, such as T/Key \cite{kogan2017t}.
A hash chain enables the production of many OTPs by the consecutive execution of a hash function, starting from $k$ that represents a secret key of the authenticator.
Upon the initialization, a smart contract is preloaded with the last generated value $h^n(k)$.
When the user wants to authenticate the $i$th operation, he sends the $h^{n-i}(k)$ to the smart contract in the second stage of $\Pi_O^{<G>}$. 
The smart contract then computes $h(.)$ consecutively $i$ times and checks to ascertain whether the obtained value equals the stored value. 
However, the main drawback of this solution is that \textit{each OTP can be trivially derived from any previous one}, and thereby this scheme does not meet the requirement of OTPs on independence. 
To detail an attack misusing this flaw, assume the MITM attacker possessing $SK_{\mathbb{U}}$ (i.e., the first factor) is able to initiate operations in the first stage of $\Pi_O^{<G>}$.
The attacker $\mathcal{A}$ initiates operation $O_i$ and waits for $\mathbb{U}$ to initiate and confirm an arbitrary follow-up operation $O_j, j > i$. 
When $\mathbb{U}$ sends $OTP_j$ in the second stage of $\Pi_O^{<G>}$, $\mathcal{A}$ intercepts and ``front-runs'' the user's transaction by a malicious transaction with $OTP_i$ computed as $h^{j-i}(OTP_j)$.
Although one may argue that this scheme can be hardened by a modification denying to confirm older operations than the last initiated one, it would bring a race condition issue in which $\mathcal{A}$ might keep initiating operations in the first stage of $\Pi_O^{<G>}$ each time he intercepts a confirmation transaction from $\mathbb{U}$, causing the DoS attack on the wallet.

\subsection{Proposed Approach}\label{sec:wallets-SmartOTPs-overview}
For a generic cryptocurrency with Turing-complete smart contract platform, we propose SmartOTPs, a 2FA against the blockchain, which consists of: (1) a client $\mathbb{C}$, (2) a private key hardware wallet $\mathbb{W}$ equipped with a display, (3) a smart-contract $\mathbb{S}$, and (4) an air-gapped authenticator $\mathbb{A}$ that might be implemented as an embedded device with limited resources.
First, we explain the key idea of our approach, which enables us to construct $\mathbb{A}$ as a fully air-gapped device.
Then, we present the base version of SmartOTPs, and finally, we describe modifications.

\subsubsection{Design of an Air-Gapped Authenticator}
In our approach, OTPs are generated by a pseudo-random function $F_k(.)$ and then aggregated by a Merkle tree, providing a single value, the root hash ($\mathcal{R}$).
$\mathcal{R}$ is stored at $\mathbb{S}$ and serves as a PK for OTPs.
Assuming the two stage protocol $\Pi_O^{<G>}$ (further denoted as $\Pi_O$), the user $\mathbb{U}$ might confirm the initiated operation $O_{opID}$ by a corresponding $OTP_{opID}$ (provided by $\mathbb{A}$) in the second stage of $\Pi_O$, whereby $\mathbb{S}$ verifies the correctness of $OTP_{opID}$ with use of $\mathcal{R}$.
A challenge of such an approach is the size of an OTP.

\paragraph{\textbf{From Straw-Man to the Base Version.}}
Using the straw-man version, a 2FA requires $\mathbb{A}$ to provide an OTP and its proof.
However, in such a straw-man version, the user $\mathbb{U}$ has to transfer $\frac{(S + S \times H)}{8}$ bytes from $\mathbb{A}$ each time he confirms an operation, where $S$ represents the bit-length of an OTP as well as the output of $h(.)$, and $H$ represents the height of a Merkle tree with $N$ leaves; hence $H = log_2(N$).
For example, if $S = 256$ and $H = 10$, then $\mathbb{U}$ would have to transfer 352B each time he confirms an operation, which has very low usability in an air-gapped setting utilizing transcription of mnemonic words \cite{bipMnemonic} (i.e., 264 words) or scanning of several QR codes (e.g., 21 QR codes v1) displayed on an embedded device with a small display.
Even further reduction of $S$ to 128 bits would not help to resolve this issue, as the amount of user transferred data would be equal to 176B $\simeq$ 132 mnemonic words $\simeq$ 11 QR codes v1.

We make the observation that it is possible to decouple providing OTPs from providing their proofs.
The only data that need to be kept secret are OTPs, while any node of a Merkle tree may potentially be disclosed -- no OTP can be derived from these nodes.
Therefore, we propose providing OTPs by $\mathbb{A}$, while their proofs can be constructed at $\mathbb{C}$ from stored hashes of OTPs. 
This modification enables us to fetch the nodes of the proof from the storage of $\mathbb{C}$, while $\mathbb{U}$ has to transfer only the OTP itself from $\mathbb{A}$ when confirming an operation (i.e, $S = 128 \simeq 12$ mnemonic words by default).

\subsubsection{Base Version}\label{sec:base-version}

\begin{algorithm}[t]	
	\caption{Smart contract $\mathbb{S}$ with 2FA}\label{alg:wallet-overview}
	\scriptsize 
	
	\SetKwProg{func}{function}{}{}
	
	$\triangleright$ \textsc{Variables and functions of environment:} \\
	\hspace{1em} \textit{tx}: a current transaction processed by $\mathbb{S}$,\\
	\hspace{1em} \textit{balance}: the current balance of a contract, \\
	\hspace{1em} \textit{transfer(r, v)}: transfer \textit{v} crypto-tokens from a smart contract to \textit{r},\\
	
	\smallskip
	$\triangleright$ \textsc{Declaration of types:}\\
	\hspace{1em} \textbf{Operation} \{ addr, param, pending, type $\in \{ \text{TRANSFER}, \ldots\}$ \} \\
	
	\smallskip
	$\triangleright$ \textsc{Declaration of functions:}
	
	\func{$constructor$(root, pk) \textbf{public} }{
		operations $\leftarrow$ []; \Comment{An append-only list} \\ 
		$PK_\mathbb{U}$ $\leftarrow$ \textit{pk},
		$~$$\mathcal{R}$ $\leftarrow$ root, 
		$~$nextOpID $\leftarrow$ 0; \\
		\textbf{return} $\mathbb{S}^{ID};$ \Comment{Computed by a blockchain platform.} \\
	}

	\func{$initOp$(a, p, type) \textbf{public} }{
		\textbf{assert} $\Sigma.verify(tx.\sigma, PK_\mathbb{U})$; \Comment{1st factor of 2FA} \\						
		opID $\leftarrow$ nextOpID$\texttt{++}$; \\
		operations[opID] $\leftarrow$ \textbf{new} Operation(a, p, \textbf{true}, type); \\
	}
	
	\func{$confirmOp$(otp, $\pi$, opID) \textbf{public} } { 
		\textbf{assert}  operations[opID].pending; \\
		verifyOTP(otp, $\pi$, opID); \Comment{2nd factor of 2FA} \\									
		execOp(operations[opID]); \\
		operations[opID].pending $\leftarrow$ \textbf{false}; \\
	}
	
	\func{$verifyOTP$(otp, $\pi_{opID}$, opID) \textbf{private} } {
		\textbf{assert} deriveRootHash(otp, $\pi_{opID}$, opID) = $\mathcal{R}$; \\
	}
	
	\func{$execOp$(oper) \textbf{private} }{
		\If{TRANSFER = oper.type}{
			\textbf{assert} oper.param  $\leq$ \textit{balance}; \\
			\textit{transfer}(oper.addr, oper.param); \\
			
		} 
	}						
\end{algorithm}

\paragraph{\textbf{Secure Bootstrapping.}}\label{sec:bootstraping-sec}
As common in other schemes and protocols, by default, we assume a secure environment for bootstrapping protocol $\Pi_{B}^\mathcal{S}$ (see \autoref{fig:smart-contract-deploy}). 
First, $\mathbb{A}$ generates a secret seed $k$, which  is stored as a recovery phrase by $\mathbb{U}$.
$\mathbb{W}$ generates a key-pair $SK_\mathbb{U}, PK_\mathbb{U} \leftarrow \Sigma.KeyGen()$.
Next, $\mathbb{U}$ transfers $k$ from $\mathbb{A}$ to $\mathbb{C}$ in an air-gapped manner (i.e., transcribing a few mnemonic words or scanning a QR code). 
Then, $\mathbb{C}$ generates OTPs by computing $F_k(i)~|~i \in \{0, 1, \ldots, N-1\}$, where $N$ is the number of leaves (equal to the number of OTPs in the base version).
Next, $\mathbb{C}$ computes and stores the leaves of the tree -- i.e., the hashes of the OTPs (i.e., $hOTPs$), which do not contain any confidential data.\footnote{To improve performance during provisioning of proofs, $\mathbb{C}$ might additionally store non-leaf nodes, increasing the requirement on $\mathbb{C}$'s storage 2x.}
After this step, $k$ and the OTPs are deleted from $\mathbb{C}$, and $\mathbb{C}$ computes $\mathcal{R}$ from the stored hashes of the OTPs.
Then, $\mathbb{C}$ creates a transaction containing constructor of $\mathbb{S}$ (see \autoref{alg:wallet-overview}) with $\mathcal{R}$ as the argument and passes it to $\mathbb{W}$ for appending $PK_\mathbb{U}$.
Finally, $\mathbb{C}$ sends the transaction with the constructor to the blockchain where the deployment of $\mathbb{S}$ is made.\footnote{$\mathbb{C}$ has the template of $\mathbb{S}$ and the deployment process is unnoticeable for the users.}  
In the constructor, $\mathcal{R}$ with $PK_\mathbb{U}$ are stored and ID of $\mathbb{S}$ (i.e., $\mathbb{S}^{ID}$) is assigned by a blockchain platform and returned in a response.\footnote{Note that $\mathbb{S}^{ID}$ represents a public identification of $\mathbb{S}$, which serves as a destination for sending crypto-tokens to $\mathbb{S}$ by any party.}
Storing $\mathcal{R}$ and $PK_\mathbb{U}$ binds an instance of $\mathbb{S}$ with the user's authenticator $\mathbb{A}$ and the user's private key wallet $\mathbb{W}$, respectively. 
In detail, $PK_\mathbb{U}$ enables $\mathbb{S}$ to verify whether an arbitrary transaction was signed by the user who created $\mathbb{S}$, while $\mathcal{R}$ enables the verification whether the given OTP was produced by the user's $\mathbb{A}$.

\paragraph{\textbf{Operation Execution.}}
When the wallet framework is initialized, it is ready for executing operations by a two-stage protocol $\Pi_{O}$ (see \autoref{fig:smart-contract-execution}): 
\begin{figure}[t]
	\begin{center}
		\vspace{-0.2cm}
		\includegraphics[width=0.5\textwidth]{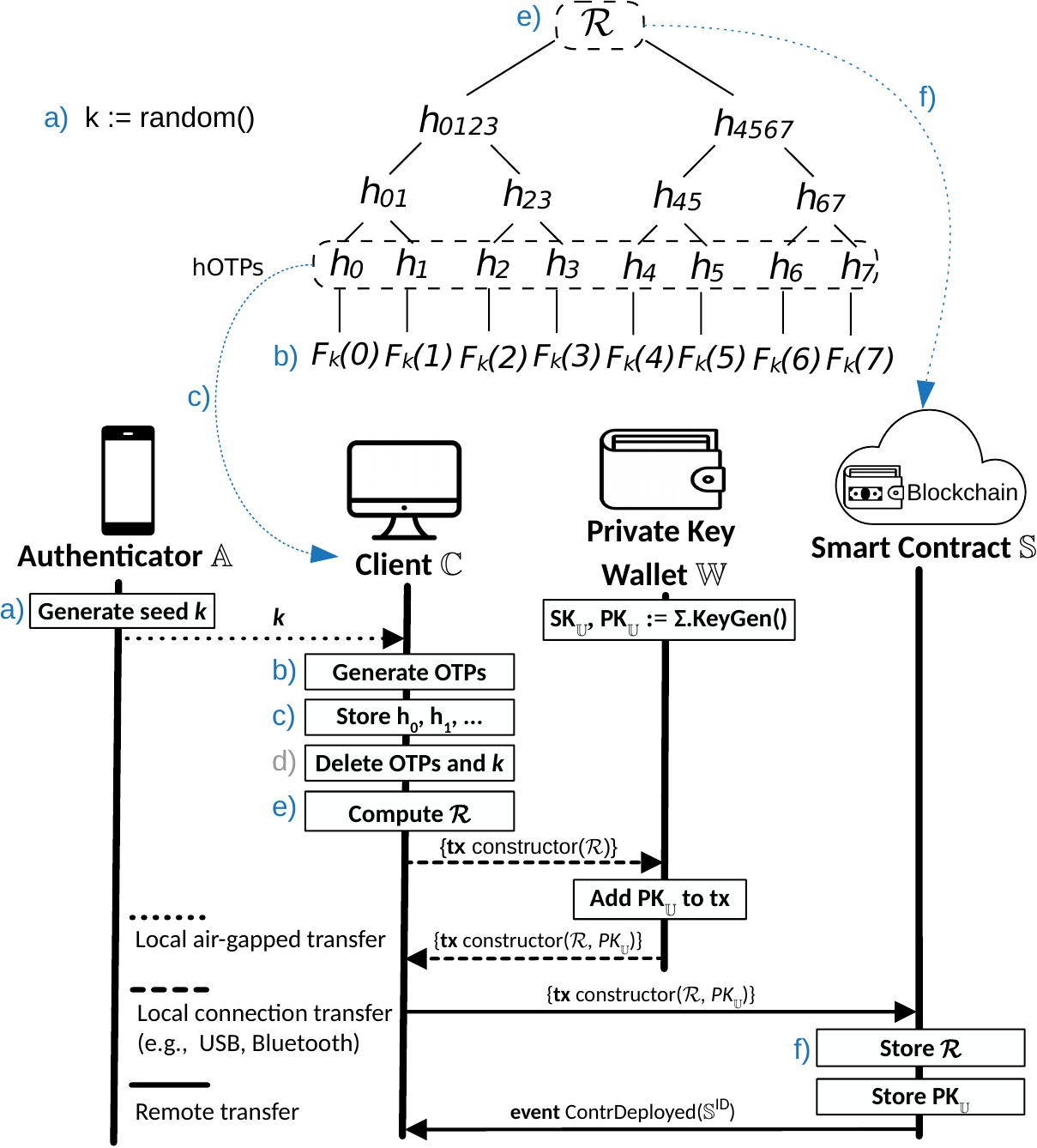} 
		\caption{Bootstrapping of SmartOTPs in a secure environment ($\Pi_{B}^\mathcal{S}$).}
		
		\label{fig:smart-contract-deploy}
		\vspace{-0.1cm}
	\end{center}	
\end{figure}
\begin{compactenum}
	\item \textbf{Initialization Stage}.
	When $\mathbb{U}$ decides to execute an operation with SmartOTPs, he enters the details of the operation into $\mathbb{C}$ that creates a transaction calling \textit{initOp()}, which is provided with operation-specific parameters -- the type of operation (e.g., transfer), a numerical parameter (e.g., amount or daily limit), and an address parameter (e.g., recipient). 
	Then, $\mathbb{C}$ sends this transaction to $\mathbb{W}$, which displays the details of the transaction and prompts $\mathbb{U}$ to confirm signing by a hardware button.
	Upon confirmation, $\mathbb{W}$ signs the transaction by $SK_{\mathbb{U}}$ and sends it back to $\mathbb{C}$.
	$\mathbb{C}$ forwards the transaction to $\mathbb{S}$. 
	In the function \textit{initOp()}, $\mathbb{S}$ verifies whether the signature was created by $\mathbb{U}$ (the first factor), stores the parameters of the operation, and then assigns a sequential ID (i.e., $opID$) to the initiated operation.
	In the response from $\mathbb{S}$, $\mathbb{C}$ is provided with an $opID$.

	\item \textbf{Confirmation Stage}.
	After the transaction (that initiated the operation) is persisted on the blockchain, $\mathbb{U}$ proceeds to the second stage of $\Pi_{O}$.
	$\mathbb{U}$ enters $opID$ to $\mathbb{A}$, which, in turn, computes and displays $OTP_{opID}$ as $F_k(opID)$.
	Storing $hOTPs$ computed from OTPs at $\mathbb{C}$ enables $\mathbb{U}$ to transfer only the displayed OTP from $\mathbb{A}$ to $\mathbb{C}$, which can be accomplished in an air-gapped manner. 
	Considering the mnemonic implementation \cite{bipMnemonic}, this means an air-gapped transfer of 12 words in the case of $O=\text{16B}$.
	Then, $\mathbb{C}$ computes and appends the corresponding proof $\pi_{opID}$ to the OTP.
	The proof of the OTP is computed from stored $hOTPs$ in the $\mathbb{C}$'s storage (or directly fetched from the storage if $\mathbb{C}$ stores all nodes of the Merkle tree).
	Next, $\mathbb{C}$ sends a transaction with $OTP_{opID}$ and its proof $\pi_{opID}$ to the blockchain, calling the function \textit{confirmOp()} of $\mathbb{S}$, which handles the second factor. 
	This function verifies the authenticity of the OTP (i.e., the first requirement of OTPs) and its association with the requested operation (i.e., the second requirement of OTPs), which together implies the correctness of the provided OTP.\footnote{Note that SmartOTPs meet the third requirement of OTPs by the design.}
	In detail, upon calling the \textit{confirmOp()} function with $opID$, $OTP_{opID}$, and $\pi_{opID}$ as the arguments, $\mathbb{S}$ reconstructs the root hash from the provided arguments by the function \textit{deriveRootHash()} that is presented in Appendix of \cite{homoliak2020smartotps}.
	If the reconstructed value matches the stored value $\mathcal{R}$, the operation is executed (e.g., crypto-tokens are transferred).
	
\end{compactenum}

\begin{figure}[t]
	\begin{center}
		\vspace{-0.2cm}
		\includegraphics[width=0.52\textwidth]{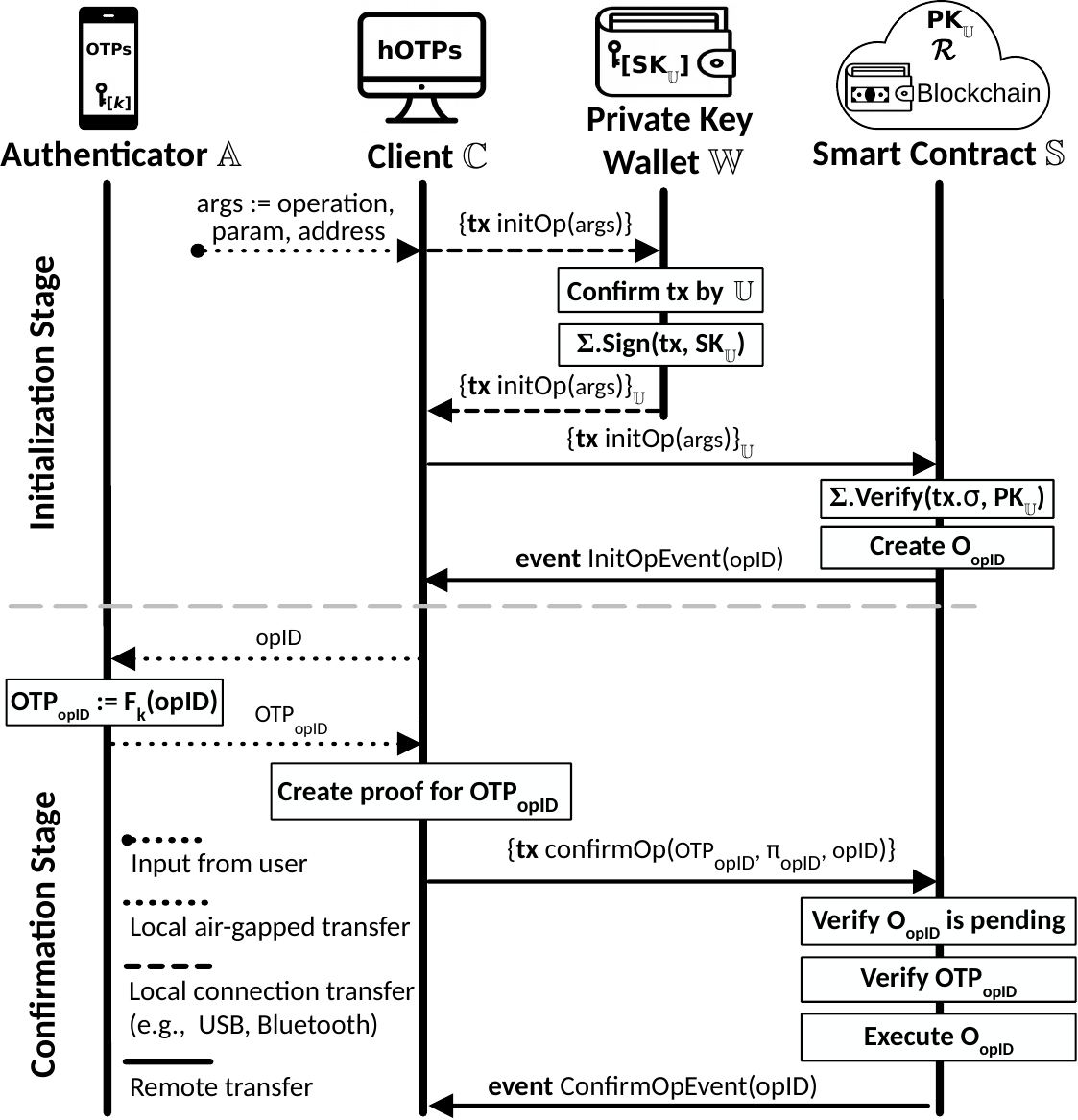} 
		\caption{Execution of an operation ($\Pi_{O}$).}
		\label{fig:smart-contract-execution}		
	\end{center}	
	\vspace{-0.1cm}
\end{figure}

\noindent
In the following, we present extensions of SmartOTPs, improving its efficiency and usability, and introducing new features.

\subsubsection{\textbf{Bootstrapping in an Insecure Environment}}\label{sec:bootstraping-insec}
The main advantage of $\Pi_{B}^\mathcal{S}$ described above is its high usability, requiring only an air-gapped transfer of $k$ and connected $\mathbb{W}$.
However, $\Pi_{B}^\mathcal{S}$ is not resistant against $\mathcal{A}$ tampering with $\mathbb{C}$; $\mathcal{A}$ might intercept $k$ or forge $\mathcal{R}$ for $\mathcal{R}'$.
Similarly, $\mathcal{A}$ might forge $PK_\mathbb{U}$ for $PK_\mathcal{A}$, while staying unnoticeable for $\mathbb{U}$ who expects that $\mathbb{S}^{ID}$ obtained is correct.
Therefore, we propose an alternative bootstrapping protocol $\Pi_{B}^\mathcal{I}$ (see Appendix of \cite{homoliak2020smartotps}), assuming that $\mathcal{A}$ can tamper with $\mathbb{C}$ during bootstrapping.
In this protocol, first we protect SmartOTPs from the interception of $k$ and then from forging $\mathcal{R}$ and $PK_\mathbb{U}$.

To avoid the interception of $k$, instead of transferring $k$, $\mathbb{U}$ performs a transfer of all leaves of the Merkle tree (i.e., $hOTPs$) from $\mathbb{A}$ to $\mathbb{C}$, which can be achieved with a microSD card.
Note that the leaves are hashes of OTPs, hence they do not contain any confidential data.
Next, to protect SmartOTPs from forging of $PK_\mathbb{U}$ and $\mathcal{R}$,
we require a deterministic computation of $\mathbb{S}^{ID}$ by a blockchain platform using $PK_\mathbb{U}$ and $\mathcal{R}$, hence $\mathbb{S}^{ID}$ can be computed and displayed together with $\mathcal{R}$ in $\mathbb{W}$ before the deployment of $\mathbb{S}$.
In detail, $\mathbb{S}^{ID}$ is computed as $h(PK_\mathbb{U} ~\|~ \mathcal{R})$, thus each pair consisting of a public key and a root hash  maps to the only $\mathbb{S}^{ID}$.
However, even with this modification, $\mathcal{R}$ can still be forged by $\mathbb{C}$. 
Therefore, when transaction with the constructor is sent to $\mathbb{W}$, $\mathbb{U}$ has to compare $\mathcal{R}$ displayed at $\mathbb{W}$ with the one computed and displayed by $\mathbb{A}$. 
In the case of equality, $\mathbb{U}$ records $\mathbb{S}^{ID}$ displayed in $\mathbb{W}$.

\subsubsection{Increasing the Number of OTPs}\label{sec:increasingNoOfOTPs}
A small number of OTPs can have negative usability and security implications.
First, users executing many transactions\footnote{E.g., several smart contracts in Ethereum have over $2^{20}$ transactions made.}
would need to create new
OTPs often, and thus change their addresses.
Second, an attacker possessing $SK_{\mathbb{U}}$ can flood $\mathbb{S}$ with initialized operations, rendering all the OTPs unusable.
Therefore, we need to increase the number of OTPs to make the attack unfeasible.
However,  increasing the number of OTPs linearly increases the amount of data
that $\mathbb{C}$ needs to preserve in its storage. 
For example, if the number of OTPs is $2^{20}$, then $\mathbb{C}$ has to
store $33.6MB$ of data (considering $S=16B$ and $\mathbb{C}$ storing all leaves), which is feasible even on storage-limited devices.
However, e.g., for $2^{32}$ OTPs, $\mathbb{C}$ needs to store $137.4GB$ of data, which might be infeasible even on PCs, especially when $\mathbb{C}$ handles multiple instances of SmartOTPs.

To resolve this issue, we modify the base approach by applying a
time-space trade-off \cite{hellman1980cryptanalytic} for OTPs.
Namely, we introduce  hash chains of which last items are aggregated by the Merkle tree. 
With such a construction, OTPs can be encoded as elements of chains and revealed layer by layer in the reverse order of creating the chains. 
This allows multiplication of the number of OTPs by the chain length without
increasing the $\mathbb{C}$'s storage but imposing a larger number of hash computations on $\mathbb{S}$ and $\mathbb{A}$.
Nonetheless, smart contract platforms set only a low execution cost for $h(.)$.

An illustration of this construction is presented in the bottom left part of
\autoref{fig:smartotps-overview}.
A hash chain of length $P$ is built from each OTP assumed so far. 
Then, the last items of all hash chains are used as the first iteration layer, which provides $\frac{N}{P}$ OTPs.\footnote{For simplicity, we assume  that $GCD(N,P) = P$.} 
Similarly, the penultimate items of all the hash chains are used as the second iteration layer, etc., until the last iteration layer consisting of the first items of hash chains (i.e., outputs of $F_k(.)$) has been reached (see the middle part of \autoref{fig:smartotps-overview}).
We emphasize that introducing hash chains may cause a violation of the requirement on the independence of OTPs if implemented incorrectly; 
i.e., OTPs from upper iteration layers can be derived from lower layers.
Therefore, to enforce this requirement, we invalidate all the OTPs of all the previous iteration layers by a sliding window at $\mathbb{S}$.

Furthermore, if a hash chain were to use the same hash function throughout the entire chain, it would be vulnerable to birthday attacks \cite{hu2005efficient}.
To harden a hash chain against a birthday attack, a \textit{domain separation} proposed by Leighton and Micali \cite{leighton1995large} can be used: a different hash function is applied in each step of a hash chain.
Note that without domain separation, inverting the $i$th iterate of $h(.)$ is $i$ times easier than inverting a single hash function (see the proof in \cite{haastad2001practical}).
Therefore, we use a different hash function for all but the last iteration layer $1 \leq i < P$ as follows:
\begin{eqnarray}\label{eqn:hash-with-domain-sep}
h_{\mathcal{D}[i]}(x) &=& h(P - i + 1~||~x),
\end{eqnarray}
where $x$ represents the OTP from the next iteration layer.

Although domain separation hardens a single hash chain against the birthday attack, this attack is still possible within the current iteration layer, which is an inevitable consequence of using multiple hash chains.
Therefore, the number of leaves $\mathcal{L}$ (i.e., N/P) is the parameter that must be considered when quantifying the security level of our scheme (see \autoref{sec:analysis}).

With this improvement, $\mathbb{A}$ is updated to provide OTPs by
\begin{eqnarray}\label{eqn:hash-chains}
getOTP(i) &=& h^{\alpha(i)}_{\mathcal{D}} \Bigg( F_k \Big(\beta(i) \Big) \Bigg),
\end{eqnarray}
where $i$ is the operation ID, $\alpha(i)$ determines the index in a hash chain, 
and $\beta(i)$ determines the index in the last iteration layer of OTPs.
We provide concrete expressions for $\alpha(i)$ and $\beta(i)$ in \autoref{eqn:alpha-beta-final}, which involves all proposed improvements and optimizations.
A derivation of $\mathcal{R}$ from the OTP at $\mathbb{S}$ needs to be updated as well (see  Appendix of \cite{homoliak2020smartotps}).
In detail, $\mathbb{S}$ executes $P ~-~ \alpha(i) ~-~ 1 = \left\lfloor \frac{i P}{N} \right\rfloor$ hash computations, which is a complementary number to the number of hash computations at $\mathbb{A}$ with regard to $P$.
Also, $\mathbb{C}$ has to be modified, requiring computation of a proof to use the leaf index relative to the current iteration layer of OTPs (i.e., $i ~\%~ \frac{N}{P})$.

With this improvement, given the number of leaves equal to $2^{20}$ and $P=2^{12}$, $\mathbb{C}$ stores only $33.6MB$ of data and it has $2^{32}$ OTPs available.
On the other hand, this modification implies, on average, the execution of additional $P/2$ hash computations  at $\mathbb{S}$, imposing additional costs.
However, our experiments show the benefits of this approach (see \autoref{sec:costs-analysis}).

\begin{figure*}[t]
	\begin{center}
		\vspace{-0.3cm}
		\includegraphics[width=1.02\textwidth]{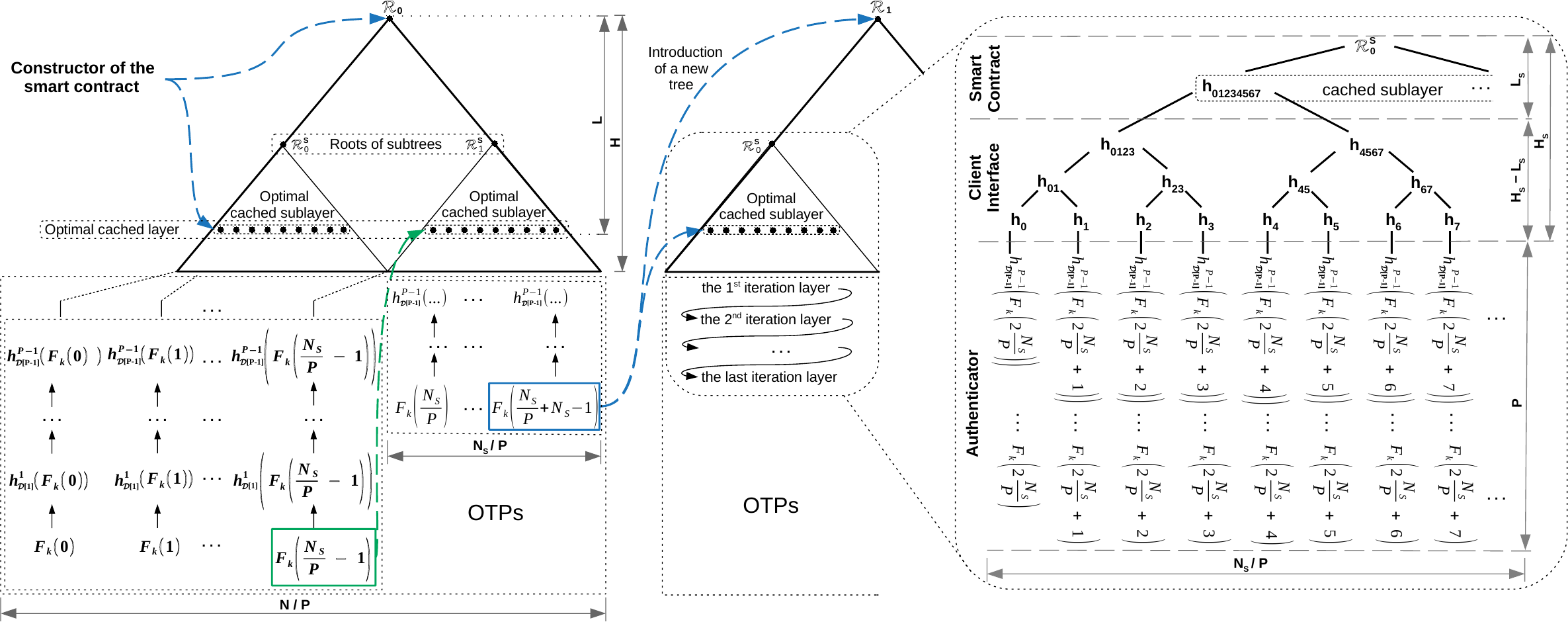} 
		\caption{An overview of our approach and its improvements.
		}
		\label{fig:smartotps-overview}
		\vspace{-0.3cm}
	\end{center}	
\end{figure*}

\begin{algorithm}[!b]
	\scriptsize
	
	\SetKwProg{func}{function}{}{}
	
	$L_1$ $\leftarrow$ []; \Comment{Items have form $<h(\mathcal{R}^{new} ~\|~ OTP)>$}\\
	$L_2$ $\leftarrow$ []; \Comment{Items have form $<$ $\mathcal{R}^{new}$$>$}\\
	
	\smallskip
	\func{$1\_newRootHash$(hRootAndOTP) \textbf{public}}{
		\textbf{assert} $\Sigma.verify(tx.\sigma, PK_\mathbb{U})$;\\
		\textbf{assert} $nextOpID ~\%~ N = N - 1$; \Comment{The last oper. of tree} \\
		$L_1$.append(hRootAndOTP); \\
	}					
	
	\func{$2\_newRootHash$($\mathcal{R}^{new}$) \textbf{public}}{
		\textbf{assert} $\Sigma.verify(tx.\sigma, PK_\mathbb{U})$;\\
		\textbf{assert} nextOpID$~\%~N = N - 1$; \Comment{The last oper. of tree}\\
		$L_2$.append($\mathcal{R}^{new}$);\\
	}
	
	\func{$3\_newRootHash$(otp, $\pi$) \textbf{public}}{
		\textbf{assert} nextOpID$~\%~N = N - 1$; \Comment{The last oper. of tree} \\
		verifyOTP(otp, $\pi$, nextOpID); \\
		\If{$L_1\text{.len} > LEN_{MAX} ~|~ L_2\text{.len} > LEN_{MAX}$}{			
			$L_1$, $L_2$ $\leftarrow$ [], []; \\
			\textbf{return}; \Comment{To avoid $\mathcal{A}$ DoS-ing $\mathbb{S}$ by gas depletion.} \\
		}
		
		\For{$\{j \leftarrow 0; ~j < L_1\text{.len} ; ~j\texttt{++}\}$}{
			\For{$\{i \leftarrow 0; ~i < L_2\text{.len} ; ~i\texttt{++}\}$}{
				\If{h($L_2$[i] $\|$ otp) = $L_1$[j]}{
					$\mathcal{R}$ $\leftarrow$ $L_2[i]$; \\
					$L_1$, $L_2$ $\leftarrow$ [], []; \\
					nextOpID++; \\
				}
			}
		}	        
		
	}					
	\caption{Introduction of a new $\mathcal{R}$ in $\mathbb{S}$}\label{alg:new-parent-tree}
	\vspace{-0.1cm}
\end{algorithm}

\subsubsection{Depletion of OTPs}\label{sec:depletion-of-otps}
Even with the previous modification, the number of OTPs remains bounded, therefore they may be depleted.
We propose handling of depleted OTPs by a special operation that replaces the
current tree with a new one.
To introduce a new tree securely, we propose updating $\mathcal{R}$ value while using the last OTP of the current tree for confirmation.  
Nevertheless, for this purpose we cannot use $\Pi_O$ consisting of two stages, as $\mathcal{A}$ possessing $SK_{\mathbb{U}}$ could be ``faster'' than the user and might initialize the last operation and thus block all the user's funds.
If we were to allow repeated initialization of this operation, then we would create a race condition issue.

To avoid this race condition issue, we propose a protocol $\Pi_{NR}$ that replaces $\mathcal{R}$ during three stages of interaction with the blockchain, which requires two append-only lists $L_1$ and $L_2$ (see \autoref{alg:new-parent-tree}):
\begin{compactenum}
	\item $\mathbb{U}$ enters $OTP_{N-1}$ to $\mathbb{C}$. 
	$\mathbb{C}$ sends  $h(OTP_{N-1} ~\|~ \mathcal{R}^{new})$ to $\mathbb{S}$, which appends it to $L_1$.
	
	\item $\mathbb{C}$ sends $\mathcal{R}^{new}$ to $\mathbb{S}$, which appends it to $L_2$.
	
	\item $\mathbb{C}$ passes $OTP_{N-1}$ with $\pi_{N-1}$ to $\mathbb{S}$, where the first matching entries of $L_1$ and $L_2$ are located to perform the introduction of $\mathcal{R}^{new}$.
	Finally, the lists are cleared for future updates.
\end{compactenum}

\noindent
Locating the first entries in the lists relies on the append-only feature of lists, hence no $\mathcal{A}$ can make the first valid pair of entries in the lists.
Similarly as in $\Pi_{B}$, we propose two variants of $\Pi_{NR}$ intended for secure (i.e., $\Pi_{NR}^\mathcal{S}$) and insecure environment (i.e., $\Pi_{NR}^\mathcal{I}$).
In $\Pi_{NR}^\mathcal{I}$ (see Appendix of \cite{homoliak2020smartotps} for detailed description), $\mathbb{A}$ must compute and display $h(OTP_{N-1} ~\|~ \mathcal{R}^{new})$ and $\mathcal{R}^{new}$ to enable protection against $\mathcal{A}$ that tampers with $\mathbb{C}$. 
Hence, $\mathbb{U}$ can verify the equality of items displayed at $\mathbb{W}$ with the ones displayed at $\mathbb{A}$ during the first and the second stage of $\Pi_{NR}^{\mathcal{I}}$, preventing $\mathcal{A}$ from forging the tree.
To adapt this improvement at $\mathbb{C}$, $\mathbb{C}$ needs to store all nodes of the new tree.
Therefore, $\mathbb{U}$ provides $\mathbb{C}$ with all nodes of the new tree, transferred from $\mathbb{A}$ on a microSD card.
In the case of $\Pi_{NR}^\mathcal{S}$, the nodes of the new tree are transferred by a transcription of $k$ from $\mathbb{A}$ to $\mathbb{C}$ and no values are displayed at $\mathbb{W}$ and $\mathbb{A}$ for $\mathbb{U}$'s verification.

\subsubsection{Cost \& Security Optimizations }\label{subsec:caching-smart-contract}

\paragraph{\textbf{Caching in the Smart Contract.}}\label{sec:caching-smart-contract}
With a high Merkle tree, the reconstruction of $\mathcal{R}$ from a leaf node may be costly.
Although the number of hash computations stemming from the Merkle tree is logarithmic in the number of leaves, the cost imposed on the blockchain platform may be significant for higher trees.
We propose to reduce this cost by caching an arbitrary tree layer of depth $L$ at $\mathbb{S}$ and do proof verifications against a cached layer.
Hence, every call of \textit{deriveRootHash()} will execute $L$ fewer hash computations in contrast to the version that reconstructs $\mathcal{R}$, while $\mathbb{C}$ will transfer by $L$ fewer elements in the proof. 

The minimal operational cost can be achieved by directly caching leaves of the tree, which accounts only for hash computations coming from hash chains, not a Merkle tree.
However, storing such a high amount of cached data on the blockchain is
too expensive. 
Therefore, this cost optimization must be viewed as a trade-off between the depth $L$ of the cached layer and the price required for the storage of such a cached layer on the blockchain (see
\autoref{sec:costs-analysis}). 

We depict this modification in the left part of \autoref{fig:smartotps-overview}, and we show that an optimal caching layer can be further partitioned into caching sublayers of subtrees (introduced later).
To enable this optimization, the cached layer of the Merkle tree must be stored in the constructor of $\mathbb{S}$. 
From that moment, the cached layer replaces the functionality of $\mathcal{R}$,  reducing the size of proofs.
During the confirmation stage of $\Pi_{O}$, an OTP and its proof are used for the reconstruction of a particular node in the cached layer, instead of $\mathcal{R}$.
Then the reconstructed value is compared with an expected node of the cached layer.
The index of an expected node is computed as 
\begin{eqnarray}
idxInCache(i) &=&	\left\lfloor \left(i ~\%~ \frac{N}{P} \right)  ~/~ 2^{H - L} \right\rfloor, 
\end{eqnarray}
where $i$ is the ID of an operation.

\vspace{-0.1cm}
\paragraph{\textbf{Partitioning to Subtrees.}}
The caching of the optimal layer minimizes the operational costs of SmartOTPs, but on the other hand, it requires prepayment for storing the cache on the blockchain.
If the cached layer were to contain a high number of nodes, then the initial deployment cost could be prohibitively high, and moreover, the user might not deplete all the prepaid OTPs.
On top of that, after revealing the first iteration layer of OTPs, the security of our scheme described so far is decreased by $log_2(N/P)$ bits due to the birthday attack (see \autoref{sec:analysis}) on OTPs. 
Hence, bigger trees suffer from higher security loss than smaller trees. 

To overcome the prepayment issue and to mitigate the birthday attack, we propose partitioning an optimal cached layer to smaller groups having the same size, forming sublayers that belong to subtrees (see the left part of \autoref{fig:smartotps-overview}).
The obtained security loss is $log_2(N_S/P)$, $N_S \ll N$.

Starting with the deployment of $\mathbb{S}$, the cached sublayer of the first subtree and the ``parent'' root hash (i.e., $\mathcal{R}$) are passed to the constructor; the cached sublayer is stored on the blockchain and its consistency against $\mathcal{R}$ is verified.
Then during the operational stage of $\Pi_{O}$, when confirmation of operation is performed, the passed OTP is verified against an expected node in the cached sublayer of the current subtree, saving costs for not doing verification against $\mathcal{R}$ (see Appendix of \cite{homoliak2020smartotps}).

If the last OTP of the current subtree is reached, then no operation other than the introduction of the next subtree can be initialized (see the green dashed arrow in \autoref{fig:smartotps-overview}).
We propose a protocol $\Pi_{ST}$ for the introduction of the next subtree (see Appendix of \cite{homoliak2020smartotps} for the detailed description).
Namely, $\mathbb{C}$ introduces the next subtree in a single step by calling a function \textit{nextSubtree()} of $\mathbb{S}$ with the arguments containing:
(1) the last OTP of the current subtree $OTP_{(N_S - 1) + \delta N_S}, ~\delta \in \{1, ~\ldots,~ N/N_S - 1\}$, 
(2) its proof $\pi_{otp}$, 
(3) the cached sublayer of the next subtree, and 
(4) the proof $\pi_{sr}$  of the next subtree's root; all items but OTP are computed by $\mathbb{C}$. 
\begin{algorithm}[t]
	\caption{Introduction of the next subtree at $\mathbb{S}$}\label{alg:alg-next-child-tree}	
	
	\scriptsize
	
	\SetKwProg{func}{function}{}{}
	
	currentSubLayer[]; \Comment{Adjusted in the constructor} \\
	\smallskip		
	
	\func{$nextSubtree$(nextSubLayer, otp, $\pi_{otp}$, $\pi_{sr}$) \textbf{public}}{ 
		\textbf{assert} nextOpID $\%~N$ $\neq N - 1$; \Comment{Not the last op. of parent}\\
		\textbf{assert} nextOpID $\%~N_S$ $ = N_S - 1 $; \Comment{The last op. of subtree}\\
		\textbf{assert} currentSubLayer.len = nextSubLayer.len;\\
		
		\textbf{assert} deriveRootHash(otp,  $\pi_{otp}$, nextOpID) = $\mathcal{R}$;\\				
		
		currentSubLayer $\leftarrow$  nextSubLayer;\\
		
		$\mathcal{R}^s$ $\leftarrow$ reduceMT(currentSubLayer, currentSubLayer.len);\\
		\textbf{assert} subtreeConsistency($\mathcal{R}^s$, $\pi_{sr}$, $\mathcal{R}$);\\
		nextOpID++;  \Comment{Accounts for this introduction of a subtree}\\			
	}			
\end{algorithm}
The pseudo-code of the next subtree  introduction at $\mathbb{S}$ is shown in \autoref{alg:alg-next-child-tree}. 
The current subtree's cached sublayer is replaced by the new one, which is verified by the function $subtreeConsistency()$
against $\mathcal{R}$ with the use of the passed proof 
$\pi_{sr}$ of the new subtree's root hash $\mathcal{R}^s$.
Note that introducing a new subtree invalidates all initialized yet to be confirmed operations of the previous subtree.

At $\mathbb{A}$, this improvement requires accommodating the iteration over layers of hash chains in shorter periods.
Hence, $\mathbb{A}$ provides OTPs by \autoref{eqn:hash-chains} with the following expressions:
\begin{eqnarray}\label{eqn:alpha-beta-final}
\begin{split}
\alpha(i) &=& P - \left\lfloor   \frac{(i ~\%~ N_S) P}{N_S} \right\rfloor - 1,\\
\beta(i) &=&  \left\lfloor \frac{i}{N_S} \right\rfloor \frac{N_S}{P} +  \left( i ~\%~ \frac{N_S}{P} \right),
\end{split}
\end{eqnarray}
where $i$ is an operation ID and $N_S$ is the number of OTPs provided by a single subtree.
We remark, that due to this optimization, the update of a new parent root $\mathcal{R}$ as well as the constructor of $\mathbb{S}$ requires, additionally to \autoref{alg:new-parent-tree} and \autoref{alg:wallet-overview}, the introduction of a cached sublayer of the first subtree (omitted here for simplicity).

\subsection{Security Analysis}
\label{sec:analysis}
We analyze the security of SmartOTPs and its resilience to attacker models under the assumption of random oracle model $\mathcal{RO}$.

\subsubsection{Security of OTPs}\label{sec:security-of-otps}
OTPs in our scheme are related to two cryptographic constructs: a list of hash chains and the Merkle tree aggregating their last values. 
In this subsection, we assume an adversary $\mathcal{A}$ who is trying to invert OTPs, and we give a concrete expressions for security of our scheme. 
Since we employ the hash domain separation technique \cite{leighton1995large} for hash chains, each hash execution can be seen as an execution of an independent hash function.
For such a construction, Kogan et al. give the following upper bound (see Theorem 4.6 in \cite{kogan2017t}) on the advantage of $\mathcal{A}$ breaking a chain:
\begin{eqnarray}\label{eqn:adv-breaks-single-chain}
Pr[\mathit{\mathcal{A}~breaks~a~chain}]\leq\frac{2Q+2P+1}{2^S},
\end{eqnarray}
where $Q$ is the number of queries that $\mathcal{A}$ can make to $h(.)$, $P$ is the chain length, and $S$ is the bit-length of OTPs (and the output of $h(.)$).  
Kogan et al. \cite{kogan2017t} proved that inverting a hash chain hardened by the domain separation imposes a loss of security equal to the factor of 2.
Therefore, to make a hardened hash chain as secure as $\lambda\varDash{-} bit$ $\mathcal{RO}$, it is enough to set $S = \lambda + 2$.
E.g., to achieve 128-bit security, $S$ should be equal to 130.

\paragraph{SmartOTPs without Subtrees.}
This scheme (see \autoref{subsec:caching-smart-contract}) uses a Merkle tree that aggregates $\mathcal{L}=\frac{N}{P}$ hash chains, where
the chains are created independently of each other; they have the same length and the same number of OTPs.
$\mathcal{A}$ can win by inverting any of the chains; hence, the probability that this scheme is secure is
\begin{eqnarray}\label{eqn:scheme-is-secure}
Pr[\mathit{Scheme~is~secure}] = \bigg(1-\frac{2Q+2P+1}{2^S}\bigg)^\mathcal{L}.
\end{eqnarray}
We can apply the alternative form of Bernoulli's inequality
$(1 - x)^\mathcal{L} \geq 1 - x\mathcal{L},$
where $\mathcal{L} \geq 1$ and $0 \leq x \leq 1$ must hold.
In our case, the input conditions hold since the number of hash chains is always greater than one and the probability that $\mathcal{A}$ breaks a single chain from \autoref{eqn:adv-breaks-single-chain} fits the range of $x$ (i.e.,  $0\leq\frac{2Q+2P+1}{2^S}\leq 1$).
Hence, we lower-bound the probability from \autoref{eqn:scheme-is-secure} as follows:
\begin{eqnarray}\label{eqn:scheme-is-secure-nosubtrees}
Pr[\mathit{Scheme~is~secure}] \geq 1 -\frac{\mathcal{L}(2Q+2P+1)}{2^S}.
\end{eqnarray}
\begin{corollary}
	To make SmartOTPs without partitioning into subtrees as secure as $\lambda\varDash{-} bit$ $\mathcal{RO}$, it is enough to set $S = \lambda + 2 + log_2(\mathcal{L})$.
\end{corollary}
\noindent For example, to achieve 128-bit security with $\mathcal{L}=64$ and $P \geq 1$, $S$ should be equal to 136, and thus an OTP can be transferred by one QR code v1 or 13 mnemonic words.

\paragraph{Full SmartOTPs.}
The full SmartOTPs scheme contains partitioning into subtrees, in which all leaves of the next subtree ``are visible'' only after depleting OTPs of the current subtree (and using OTPs from the 1st iteration layer of the next subtree). 
This improves the security of our scheme under the assumption that $\mathbb{C}$'s storage is not compromised by $\mathcal{A}$, which is true for $\mathcal{A}$ that possesses $PK_{\mathbb{U}}$ or~$\mathbb{A}$.
Therefore, we replace $\mathcal{L}$ in \autoref{eqn:scheme-is-secure-nosubtrees} for $\mathcal{L_S} = \frac{N_S}{P},$ $N_S \ll N$.
\begin{corollary}
	To make the full scheme of SmartOTPs as secure as $\lambda\varDash{-} bit$ $\mathcal{RO}$, it is enough to set $S = \lambda + 2 + log_2(\mathcal{L_S})$.
\end{corollary}
\noindent Therefore, to achieve 128-bit security with $\mathcal{L} = \frac{N}{N_S} \mathcal{L_S}$, $\mathcal{L_S}= 64$, and $P \geq 1$, $S$ should be equal to 136, and thus an OTP can be transferred by a QR code v1 or 13 mnemonic words.
To achieve the same security with $\mathcal{L_S}=1024$, we need to set $S= 140$, and thus an OTP can be transferred in a QR code v2 or 13 mnemonic words.

\subsubsection{The Attacker Possessing $SK_{\mathbb{U}}$}

\begin{theorem}
	$\mathcal{A}$ with access to $SK_{\mathbb{U}}$ is able to initiate operations by $\Pi_{O}$ but is unable to confirm them.
\end{theorem}

\begin{proof1}
	The security of $\Pi_{O}$ is achieved by meeting all requirements on general OTPs (see \autoref{sec:design-space}).
	In detail, the requirement on the \textit{independence} of two different OTPs is satisfied by the definition of $F_k(.) \equiv h(k ~\|~ .)$, where $h(.)$ is instantiated by  $\mathcal{RO}$. 
	This is applicable when $P=1$.
	However, if $P>1$, then items in previous iteration layers of OTPs can be computed from the next ones.
	Therefore, to enforce this requirement, we employ an explicit invalidation of OTPs belonging to all previous iteration layers by a sliding window at $\mathbb{S}$ (see \autoref{sec:increasingNoOfOTPs}).
	The requirement on the \textit{linkage} of each $OTP_i$ with operation $O_i$ 
	is satisfied due to (1) $\mathcal{RO}$ used for instantiation of $h(.)$ and (2) by the definition of the Merkle tree, preserving the order of its aggregated leaves. 
	By meeting these requirements, $\mathcal{A}$ is able to initiate an operation $O_j$ in the first stage of $\Pi_{O}$ but is unable to use an $OTP_i$ intercepted in the second stage of $\Pi_{O}$ to confirm $O_j$, where $j \neq i$.
	Finally, the requirement on the \textit{authenticity} of OTPs is ensured by a random generation of $k$ and by anchoring $\mathcal{R}$ associated with $k$ at the constructor of $\mathbb{S}$.
\end{proof1}

\begin{theorem}
	Assuming $\delta \in \{0,\ldots, \frac{N}{N_S} - 2\}$, $\mathcal{A}$ with access to $SK_{\mathbb{U}}$ 
	is unable to deplete all OTPs or misuse a stolen OTP that introduces the $(\delta + 1)$th subtree by $\Pi_{ST}$.
\end{theorem}

\vspace{-0.2cm}
\begin{proof1}
	When all but one OTPs of the $\delta$th subtree are depleted, the last remaining operation $O_{(N_S - 1) + \delta N_S}, ~\delta \in \{0, ~\ldots,~ \frac{N}{N_S} - 2\}$ is enforced by $\mathbb{S}$ to be the introduction of the next subtree. 
	This operation is executed in a single transaction calling the function $nextSubtree()$ of $\mathbb{S}$ (see \autoref{alg:alg-next-child-tree}) requiring the corresponding $OTP_{(N_S - 1) + \delta N_S}$ that is under control of $\mathbb{U}$; hence $\mathcal{A}$ cannot execute the function to proceed with a further depletion of OTPs in the $(\delta + 1)th$ subtree. 
	If $\mathcal{A}$ were to intercept $OTP_{(N_S - 1) + \delta N_S}$ during the execution of $\Pi_{ST}$ by $\mathbb{U}$, he could use the intercepted OTP only for the introduction of the next valid subtree since the function $nextSubtree()$ also checks a valid cached sublayer of the $(\delta+1)$th subtree against the parent root hash $\mathcal{R}$.	
\end{proof1}	

\begin{theorem}
	Assuming $\delta = \frac{N}{N_S} - 1$,
	$\mathcal{A}$ with access to $SK_{\mathbb{U}}$ is neither able to deplete all OTPs nor introduce a new parent tree nor render SmartOTPs unusable. 
\end{theorem}	

\vspace{-0.2cm}
\begin{proof1}
	In contrast to the adjustment of the next subtree, the situation here is more difficult to handle, since the new parent tree cannot be verified at $\mathbb{S}$ against any paramount field. 
	If we were to use $\Pi_{O}$ while constraining to the last initialized operation $O_{(N-1) + \eta N}, ~\eta \in \{0,1,\ldots\}$ of the parent tree, then $\mathcal{A}$ could render SmartOTPs unusable by submitting an arbitrary $\mathcal{R}$ in $initOp()$, thus blocking all the funds of the user.
	If we were to allow repeated initialization of this operation, then we would create a race condition issue.
	Therefore, this operation needs to be handled outside of the protocol $\Pi_{O}$, using two unlimited append-only lists $L_1$ and $L_2$ that are manipulated in three stages of interaction with the blockchain (see \autoref{sec:depletion-of-otps}).
	In the first stage, $h(\mathcal{R}^{new} ~\|~ OTP_{(N-1) + \eta N})$ is appended to $L_1$, hence $\mathcal{A}$ cannot extract the value of OTP.
	In the second stage, $\mathcal{R}^{new}$ is appended to $L_2$,
	and finally, in the third stage, the user reveals the OTP for confirmation of the first matching entries in both lists. 
	Although $\mathcal{A}$ might use an intercepted OTP from the third stage for appending malicious arguments into $L_1$ and $L_2$, when he proceeds to the third stage and submits the intercepted OTP to $\mathbb{S}	$, the user's entries will match as the first ones. 
\end{proof1}

\subsubsection{The Attacker Tampering with the Client}

\begin{theorem}
	If $\mathbb{C}$ is tampered with after $\Pi_{B}$, $\mathbb{U}$ can detect such a situation and prevent any malicious operation from being initialized.
\end{theorem}	

\begin{proof1}
	If we were to assume that $\mathbb{W}$ is implemented as a software wallet (or hardware wallet without a display), then $\mathcal{A}$ tampering with $\mathbb{C}$ might also tamper with the $\mathbb{W}$'s software running on the same machine. 
	This would in turn enable a malicious operation to be initialized and further confirmed by $\mathbb{U}$, since $\mathbb{U}$ would be presented with a legitimate data in $\mathbb{C}$ and $\mathbb{W}$, while the transactions would contain malicious data.
	Therefore, we require that $\mathbb{W}$ is implemented as a hardware wallet with a display, which exposes only signing capabilities, while $SK_{\mathbb{U}}$ never leaves the device (e.g., \cite{trezor-hw-wallet,keep-key,BitLox,ellipal-hw-wallet}).
	Due to it, $\mathbb{U}$ can verify the details of a transaction being signed in $\mathbb{W}$ and confirm signing only if the details match the information shown in $\mathbb{C}$ (for $\Pi_{O}$) or $\mathbb{A}$ (for $\Pi_{NR}^{\mathcal{I}}$).	
	We refer the reader to the work of Arapinis et al. \cite{ArapinisGKK19} for the security analysis of hardware wallets with displays. 
\end{proof1}

\begin{theorem}
	If $\mathbb{C}$ is tampered with during an execution of  $\Pi_{B}^{\mathcal{I}}$, $\mathcal{A}$ can neither intercept $k$ nor forge $\mathcal{R}$ nor forge $PK_\mathbb{U}$.
\end{theorem}	

\vspace{-0.2cm}
\begin{proof1}
	When the protocol $\Pi_{B}^{\mathcal{I}}$ is used, instead of an air-gapped transfer of $k$ from $\mathbb{A}$ to $\mathbb{C}$, $\mathbb{U}$ transfers leaves of the Merkle tree by microSD card.
	The leaves represent hashes of OTPs in the base version or the hashes of the last items of hash chains in the full version of SmartOTPs.  
	In both versions, the transferred data do not contain any secrets, hence $\mathcal{A}$ cannot take advantage of intercepting them.
	The next option that $\mathcal{A}$ may seek for is to forge $\mathcal{R}$ for $\mathcal{R}'$ and $PK_\mathbb{U}$ for $PK_\mathcal{A}$, which results in different $\mathbb{S}^{ID}$ than in the case of $\mathcal{R}$ and $PK_\mathbb{U}$, since $\mathbb{S}^{ID}$ is computed as $h(PK_\mathbb{U} ~\|~ \mathcal{R})$.
	While $PK_\mathbb{U}$ is stored at $\mathbb{W}$, the authenticity of $\mathcal{R}$ needs to be verified by $\mathbb{U}$ who compares displays of $\mathbb{A}$ and $\mathbb{W}$.
	Only in the case of equality, $\mathbb{U}$ knows that $\mathbb{S}^{ID}$ displayed in $\mathbb{W}$ maps to legitimate $PK_\mathbb{U}$ and~$\mathcal{R}$.
\end{proof1}

\subsubsection{The Attacker Possessing the Authenticator}
It is trivial to see that $\mathcal{A}$ with access to $\mathbb{A}$ is unable to initialize any operation with SmartOTPs since he does not hold $PK_\mathbb{U}$.

\subsubsection{Further Properties and Implications}

\paragraph{\textbf{Requirement on Block Confirmations.}}
Most cryptocurrencies suffer from long time to finality, potentially enabling the accidental forks, which create parallel inconsistent blockchain views.
On the other hand, this issue is not present at blockchain platforms with fast finality, such as Algorand \cite{gilad2017algorand}, HoneyBadgerBFT \cite{miller2016honey}, or StrongChain \cite{strongchain}. 
In blockchains with long time to finality, overly fast confirmation of an operation may be dangerous, as, if an operation were initiated in an ``incorrect'' view, an attacker holding $SK_{\mathbb{U}}$ would hijack the OTP and reuse it for a malicious operation settled in the ``correct'' view. 
To prevent this threat, the recommendation is to wait for several block confirmations to ensure that an accidental fork has not happened. 
For example, in Ethereum, the recommended number of block confirmations to wait is 12 (i.e., $\sim$3 minutes).
Note that such waiting can be done as a background task of $\mathbb{C}$, hence $\mathbb{U}$ does not have to wait: 
(1) considering that $\mathcal{A}$ possesses $SK_{\mathbb{U}}$, $\mathbb{C}$ can detect such a fork during the wait and resubmit the $initOp()$ transaction, 
(2) in the case of $\mathcal{A}$ tampering with $\mathbb{C}$, no operation can be initialized since $\mathbb{U}$ never signs $\mathcal{A}$'s transaction (due to the hardware wallet), and 
(3) $\mathcal{A}$ possessing $\mathbb{A}$ cannot initialize any operation as well. 

\paragraph{Attacks with a Post Quantum Computer.}
Although a resilience to quantum computing (QC) is not the focus of our work, it is of worthy to note that our scheme inherits a resilience to $QC$ from the hash-based cryptography.
The resilience of our scheme to QC is dependent on the output size of $h(.)$. 
A generic QC attack against $h(.)$ is Grover's algorithm \cite{grover1996fast}, providing a quadratic speedup in searching for the input of the black box function.
As indicated by Amy et al. \cite{amy2016estimating}, using this algorithm under realistic assumptions, the security of SHA-3 is reduced from 256 to 166 bits.
Applying these results to OTPs with 128-bit security from examples in \autoref{sec:security-of-otps}, we obtain 98-bits post-QC security. 
Further, when assuming the example with $\mathcal{L}=64$ from \autoref{sec:security-of-otps} and \cite{amy2016estimating}, to achieve 128-bits of post-QC security, we estimate the length of OTPs to 205-bits (i.e., 19 mnemonic words).

\subsection{Implementation}
\label{sec:implementation}

We have selected the Ethereum platform and the Solidity language for the implementation of $\mathbb{S}$,
HTML/JS for DAPP of $\mathbb{C}$, Java for smartphone App of $\mathbb{A}$,
and Trezor T\&One \cite{trezor-hw-wallet} for $\mathbb{W}$.
We selected $S=128$ bits, which has practical advantages for an air-gapped $\mathbb{A}$, producing OTPs that are 12 mnemonic words long or a QR code v1 (with a capacity of 17B).
Next, we used SHA-3 with truncated output to 128 bits as $h(.)$. 
We selected the size of $k$ equal to 128 bits, fitting 12 mnemonic words $\simeq$ 1 QR code v1.

So far, we have considered only the crypto-token transfer operation.
However, our proposed protocol enables us to extend the set of operations.
For demonstration purposes, we extended the operation set by supporting daily limits and last resort information (see Appendix of \cite{homoliak2020smartotps}). 
In addition, we made a hardware implementation of $\mathbb{A}$ using Node\-MCU \cite{node-mcu} equipped with ESP8266 with the overall cost below \$5 (see Appendix of \cite{homoliak2020smartotps}).
The source code of our implementation and videos are available at~\url{https://github.com/ivan-homoliak-sutd/SmartOTPs}.

\begin{figure}[t]
	\begin{center}		
		\includegraphics[width=0.45\textwidth]{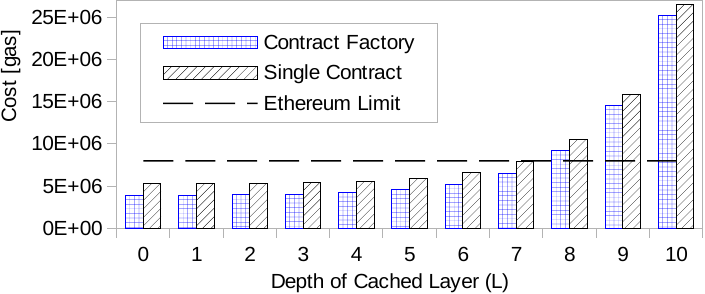} 
		\caption{Deployment costs ($H = H_S$).}
		\label{fig:MT-deploy-costs}
	\end{center}	
\end{figure}

\subsubsection{Analysis of the Costs}
\label{sec:costs-analysis}
In this section, we analyze the costs of our approach using the same bit-length $S$ for $h(.)$ as well as for OTPs.
$S$ significantly influences the gas consumption for storing the cached layer on the blockchain. 
We remark that measured costs can also be influenced by EVM internals (e.g., 32B-long words/alignment).
We assumed 8M as the maximum gas limit at the Ethereum main network, which, however, has already changed since this research was made.

\paragraph{\textbf{Costs Related to the Merkle Tree.}}
First, we abstracted from the concept of subtrees and hashchains to analyze a single tree (i.e., $H = H_S$).
The deployment costs of our scheme with respect to the depth $L$ ($\equiv L_S$) of the cached layer are presented in \autoref{fig:MT-deploy-costs}, which contains also the variant with the contract factory (saving $\sim1.3M$, regardless of $L_S$).
Although the cost of each operation supported by $\Pi_{O}$ is similar, here we selected the transfer of crypto-tokens $O^t$, 
and we measured the total cost of $O^t$ as follows: 
\vspace{-0.1cm}
\medmuskip=0mu\thinmuskip=0mu\thickmuskip=-0mu
\begin{eqnarray*}
	O^t\_cost(L, ~N, ~P)&=&  
	\overline{cost} \left(O^t(L, ~N, ~P) \right)
	+  \frac{cost \left(O^d(L) \right)}{N}, \\
	\overline{cost} \left( O^t(L, ~N, ~P) \right) &=& \frac{1}{N} \sum_{i=1}^{N}{ cost \left(O^{t}_{i}(L, ~N, ~P) \right)},  \\
	cost \left(O^{t}_{i}(L, ~N, ~P) \right) &=& cost \left( O^{t.init}_i \right) + cost \left( O^{t.confirm}_i(L, ~N, ~P) \right),
\end{eqnarray*}
\medmuskip=3mu\thinmuskip=1mu\thickmuskip=5mu
where \textit{cost()} measures the cost of an operation in gas units, and $O_d$ represents the deployment operation.
As the purpose of the cached layer is to reduce the number of hash computations in \textit{confirmOp()}, the size of an optimal cached layer is subject to a trade-off between the cost of storing the cached layer on the blockchain and the savings benefit of the caching.
In \autoref{fig:MT-expenses}, we can see that the total average cost per transfer decreases with the increasing number of OTPs, as the deployment cost is spread across more OTPs.
The optimal point depicted in the figure minimizes $O^t$ by balancing $cost(O^d(L))$) and  $\overline{cost}(O^t(L, ~N, ~P))$.
We see that $L=H-3$ for such an optimal point.

\begin{figure}[t]
	\centering

	\includegraphics[width=0.45\textwidth]{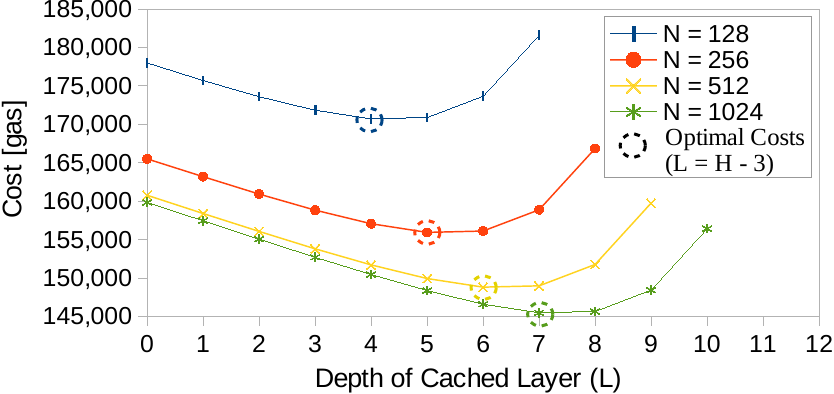}

	\caption{Average total cost per transfer ($H = H_S$).}\label{fig:MT-expenses}
\end{figure}

\begin{figure}[t]
	\centering	
	\vspace{0.3cm}
	
	\subfloat[\label{fig:fig:MT-current-benefit-of-caching-256} $H = 8$]{
		\hspace{-0.25cm}
		\includegraphics[width=0.33\textwidth]{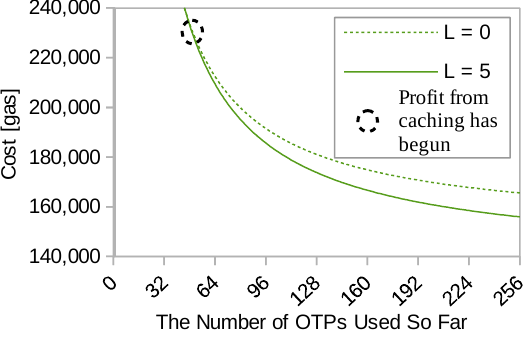} 
	}
	\subfloat[\label{fig:fig:MT-current-benefit-of-caching-512} $H = 9$]{
		\hspace{-0.25cm}
		\includegraphics[width=0.33\textwidth]{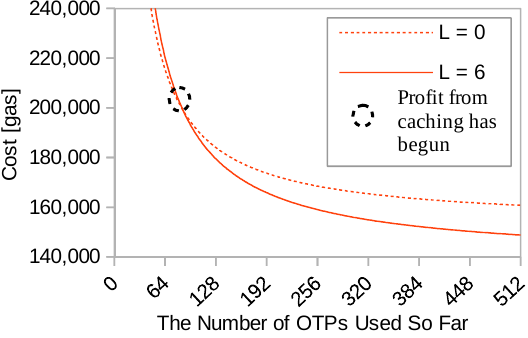} 
	}
	\subfloat[\label{fig:fig:MT-current-benefit-of-caching-1024} $H = 10$]{
		\hspace{-0.25cm}
		\includegraphics[width=0.33\textwidth]{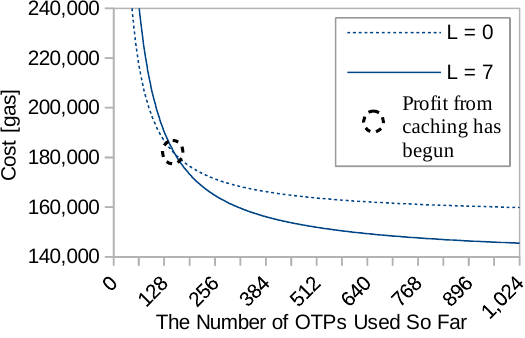} 
	}
	\vspace{0.3cm}
	\caption{Rolling average cost per transfer ($H = H_S$).}
	\label{fig:MT-current-benefit-of-caching}
\end{figure}
Next, we explored the number of transfer operations to be executed until a profit of the caching has begun (see \autoref{fig:MT-current-benefit-of-caching}).
We computed a rolling average cost per $O^t$, while distinguishing between the optimal caching layer and disabled caching.
We measured the cost of introducing the next subtree within a parent tree depending on $L_S$, while we set $H = 20$ and $H_S = 10$ (see \autoref{fig:cost-of-new-subtree}).
We found out that when subtrees (and their cached sublayers) are introduced within a dedicated operation, it is significantly cheaper compared to the introduction of a subtree during the deployment.

\begin{figure}[t]
	\centering

	\includegraphics[width=0.45\textwidth]{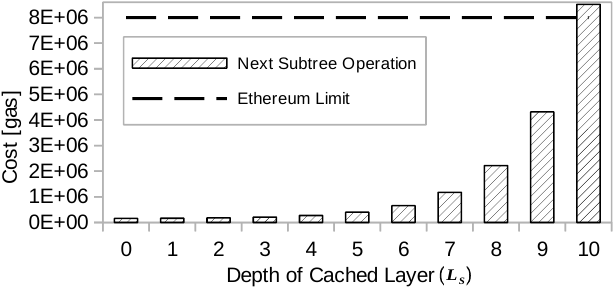} 	
	
	\caption{Cost of introducing the next subtree ($H = 20,~ H_S = 10$).}\label{fig:cost-of-new-subtree}
\end{figure}

\paragraph{\textbf{Costs Related to Hash Chains.}}
Since each iteration layer of hash chains contributes to an average cost of $confirmOp()$ with around the same value, we measured this value on a few trees with $P$ up to 512. 
Next, using this value and the deployment cost, we calculated the average total cost per transfer by adding layers of hash chains to a tree with $H = H_S$, thus increasing $N$ by a factor of $P$ until the minimum cost was found.
As a result, the optimal caching layer shifted to the leaves of the tree (see \autoref{fig:expenses-hash-chain-optim}), which would however, exceed the gas limit of Ethereum.
To respect the gas limit, we adjusted $L = 7$, as depicted in \autoref{fig:expenses-hash-chain-maxCache}.
In contrast to the configurations with $L = 0~\text{and}~ P = 1$ (from \autoref{fig:MT-expenses}), we achieved savings of $27.80\%$, $19.61\%$, $14.95\%$, and $12.51\%$ for trees with $H$ equal to $7$, $8$, $9$, and $10$, respectively.
For completeness, we calculated costs for $L=0$ as well (see \autoref{fig:expenses-hash-chain-zero-cache}).
Note that for $L = 0$ and $L = 7$, smaller trees are ``less expensive,'' as they require less operations related to the proof verification in contrast to bigger trees; these operations consume substantially more gas than operations related to hash chains.
Although we minimized the total cost per transfer by finding an optimal $P$, we highlight that increasing $P$ contributes to the cost only minimally but on the other hand, it increases the variance of the cost.
Hence, one may set this parameter even at higher values, depending on the use case.

\begin{figure*}[t]
	\centering
	
	\hspace{-0.5cm}	
	
	\subfloat[\label{fig:expenses-hash-chain-optim}$L = H$]{
		\hspace{-0.25cm}
		\includegraphics[width=0.35\textwidth]{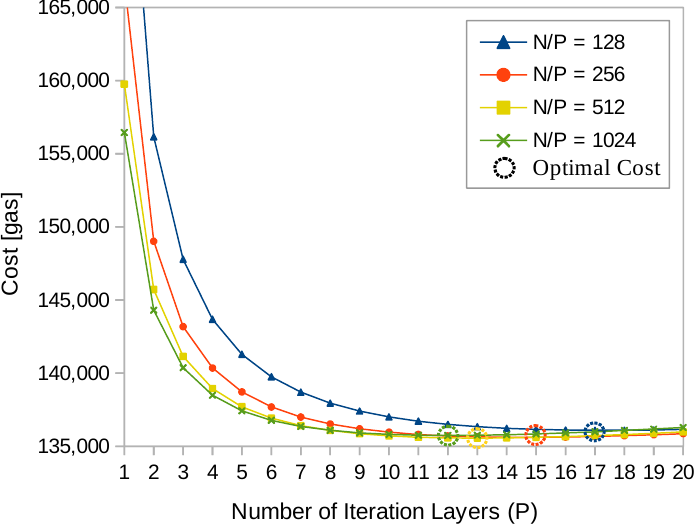} 
	}
	\subfloat[\label{fig:expenses-hash-chain-maxCache}$L = 7$]{
		\hspace{-0.25cm}
		\includegraphics[width=0.33\textwidth]{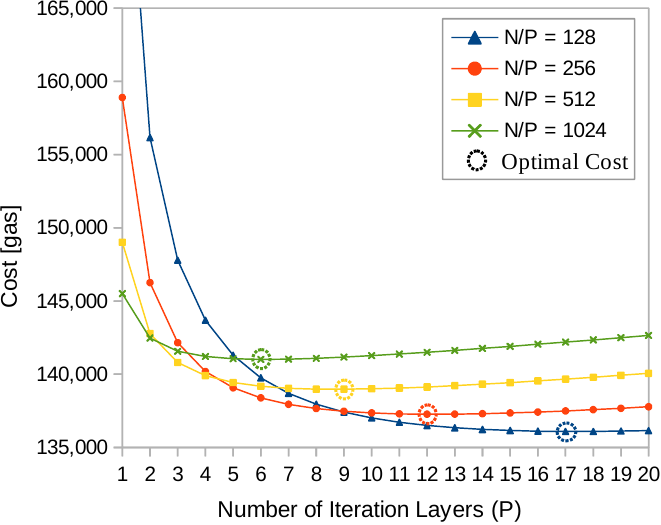} 
	}
	\subfloat[\label{fig:expenses-hash-chain-zero-cache}$L = 0$]{
		\hspace{-0.25cm}
		\includegraphics[width=0.33\textwidth]{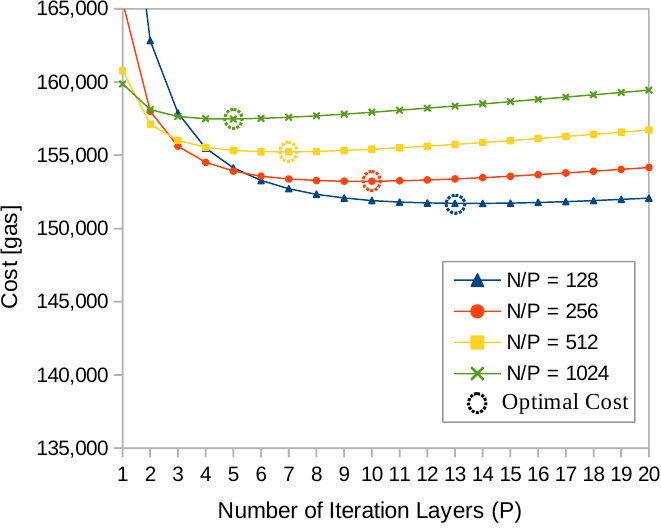} 
	}

	\hspace{1.5cm}
	\vspace{0.3cm}
	\caption{Average total cost per transfer with regards to the length $P$ of hash chains.}
	\label{fig:hash-chain-expenses}
\end{figure*}

\section{Contributing Papers}\label{sec:wallets-papers}
The papers that contributed to this research direction are enumerated in the following, while highlighted papers are attached to this thesis in their original form.

\begingroup
\let\clearpage\relax

\renewcommand\bibname{}

\endgroup

%% file: sec/evoting.tex

In this chapter, we present our contributions to the area of decentralized applications for remote electronic voting using blockchain, which belongs to the application layer of our security reference architecture (see \autoref{chapter:sra}).
In particular, this chapter is focused on privacy, scalability, security, and other practical aspects of blockchain-based e-voting, and it is based on the papers \cite{homoliak2023bbb,stanvcikova2023sbvote,venugopalan2023always} (see also \autoref{sec:evoting-papers}).

First, we review the desired properties of e-voting and then we introduce BBB-Voting, a blockchain-based boardroom voting that provides 1-out-of-$k$ choices and the maximum voter privacy in the setting that outputs the full tally of votes, while additionally offering a mechanism for resolution of faulty participants (see \autoref{sec:evoting-bbbvoting}).

Next, we aim to resolve the scalability limitation of the self-tallying approaches for boardroom voting (such as OVN \cite{McCorrySH17} and BBB-Voting \cite{homoliak2023bbb}) while maintaining security and maximum voter privacy.
Therefore, we introduce SBvote \cite{stanvcikova2023sbvote}, a decentralized blockchain-based e-voting protocol providing scalability in the number of participants by grouping them into voting booths instantiated as dedicated smart contracts that are controlled and verified by the aggregation smart contract.
We base our work on BBB-Voting since it enables more than two vote choices and recovery of faulty participants in contrast to OVN. 
See details of SBvote in \autoref{sec:voting-sbvote}.

Finally, we identify two shortcomings in present governance systems for voting: 
(a) the inability of participants to change their vote between two consecutive elections (e.g., that might be a few years apart), and 
(b) a manipulation of participants via peak-end effect \cite{Dash2018,Healy2014,Wlezien2015}.
As a response, we propose Always-on-Voting (AoV) framework \cite{venugopalan2023always} that has four key features: (1) it works in repetitive epochs, (2) voters are allowed to change their vote anytime before the end of each epoch, (3) ends of epochs are randomized and unpredictable, and (4) only the supermajority of votes can change the previous winning vote choice. 
AoV uses public randomness and commitments to a future event to determine when the current epoch should end. 
See details of AoV in \autoref{sec:evoting-aov}, where we analyze two different adversaries. 


\subsection{Properties of E-Voting}\label{sec:evoting-preliminaries}
A voting protocol is expected to meet several properties.
A list of such properties appears in the works of Kiayias and Yung~\cite{Kiayias2002}, Groth~\cite{Groth2004}, and Cramer et al.~\cite{cgs97}. 

\begin{compactdesc}
\item \textbf{(1) Privacy of Vote:} 
ensures the secrecy of the ballot contents~\cite{Kiayias2002}.
Hence, a participant's vote must not be revealed other than by the participant herself upon her discretion (or through the collusion of all remaining participants).
Usually, privacy is ensured by trusting authorities in traditional elections or by homomorphic encryption in some decentralized e-voting solutions (e.g.,~\cite{Kiayias2002,HaoRZ10,McCorrySH17,DBLP:conf/fc/SeifelnasrGY20,icissp:DMMM18}). 
%
%
\item \textbf{(2) Perfect Ballot Secrecy:} is an extension of the privacy of the vote, stating that a partial tally (i.e., prior to the end of voting) is available only if all remaining participants are involved in its computation.	
\item \textbf{(3) Fairness:} ensures that a tally may be calculated only after all participants have submitted their votes.
Therefore, no partial tally can be revealed to anyone before the end of the voting protocol~\cite{Kiayias2002}.
%
\item \textbf{(4a) Universal Verifiability:} any involved party can verify that all cast votes are correct and they are correctly included in the final tally~\cite{Kiayias2002}.
%
\item \textbf{(4b) End-to-End (E2E) Verifiability:}
The verifiability of voting systems is also assessed by E2E verifiability~\cite{benaloh2015end} as follows: 
\begin{compactitem}
	\item \emph{\textbf{cast-as-intended}:} a voter can verify the encrypted vote contains her choice of candidate,
	\item \emph{\textbf{recorded-as-cast}:} a voter can verify the system recorded her vote correctly,
	\item \emph{\textbf{tallied-as-recorded}:} any interested party is able to verify whether the final tally corresponds to the recorded votes.
\end{compactitem}
A voting system that satisfies all these properties is end-to-end verifiable.

\item \textbf{(5) Dispute-Freeness:} extends the notion of verifiability.
A dispute-free~\cite{Kiayias2002} voting protocol contains built-in mechanisms eliminating disputes between participants; therefore, anyone can verify whether a participant followed the protocol. 
Such a scheme has a publicly-verifiable audit trail that contributes to the reliability and trustworthiness of the scheme.
\item \textbf{(6) Self-Tallying:} once all the votes are cast, any involved party can compute the tally without unblinding the individual votes.
Self-tallying systems need to deal with the fairness issues (see (3) above) because the last participant is able to compute the tally even before casting her vote.
This can be rectified with an additional verifiable dummy vote~\cite{Kiayias2002}.
%
\item \textbf{(7) Robustness (Fault Tolerance):} the voting protocol is able to recover from faulty (stalling) participants, where faults are publicly visible and verifiable due to dispute-freeness~\cite{Kiayias2002}. 		
Fault recovery is possible when all the remaining honest participants are involved in the recovery.		
%
%
\item \textbf{(8) Resistance to Serious Failures:}
Serious failures are defined as situations in which voting results were changed either by a simple error or an adversarial attack. 
Such a change may or may not be detected. 
If detected in non-resistant systems, it is irreparable without restarting the entire voting~\cite{park2021going}.

\item \textbf{(9) Receipt-Freeness:} A participant is unable to supply a receipt of her vote after casting the vote.
The goal of receipt-freeness is to prevent vote-selling.
This property also prevents a post election coercion~\cite{Benaloh1994,Hirt2000}, where a participant may be coerced to reveal her vote.

\item{\textbf{(10) Dispute-Freeness.}}
The protocol's design prevents any disputes among involved parties by allowing anyone to verify whether a participant followed the protocol.

\end{compactdesc}

\renewcommand{\name}{BBB-Voting\xspace}
\section{\name}\label{sec:evoting-bbbvoting}

\subsection{System Model}

Our system model of \name \cite{homoliak2023bbb} has the following actor/components:
(1) \textit{a participant (P)} who votes,
(2) \textit{a voting authority ($VA$)} who is responsible for validating the eligibility of $P$s to vote, their registration, and 
(3) \textit{a smart contract (SC)}, which collects the votes, acts as a verifier of zero-knowledge proofs, enforces the rules of voting, and verifies the tally.

\subsubsection{Adversary Model}\label{sec:adversary-model}
The adversary $\mathcal{A}$ has bounded computing power, is unable to break used cryptographic primitives, and can control at most $t$ of $n$ participants during the protocol, where $t\leq n-2$ $\wedge$ $n\geq3$ (see the proof in the appendix of \cite{venugopalan2020bbbvoting}).
%
Any $P$ under the control of $\mathcal{A}$ can misbehave during the protocol execution. 
$\mathcal{A}$ is also a passive listener of all communication entering the blockchain network but cannot block it or replace it with a malicious message since all transactions sent to the blockchain are authenticated by signatures of $P$s or $VA$.
Finally, $VA$ \textbf{is only trusted in terms of identity management}, i.e., it performs identity verification of $P$s honestly, and neither censor any $P$ nor register any spoofed $P$.
Nevertheless, no other trust in $VA$ is required. 

\subsection{Overview of Proposed Approach}
\label{sec:bbbvoting-protocol}
\name scheme provides all properties mentioned in \autoref{sec:evoting-preliminaries} but receipt-free\-ness.
Similar to OVN~\cite{McCorrySH17}, \name publishes the full tally at the output and uses homomorphic encryption to achieve privacy of votes and perfect ballot secrecy.
%
In detail, we extend the protocol of Hao et al.~\cite{HaoRZ10} to support $k$ choices utilizing the 1-out-of-$k$ proof verification proposed by Kiayias and Yung~\cite{Kiayias2002}, and we accommodate this approach to run on the blockchain. 
Additionally, we extend our protocol to support the robustness, based on Khader et al.~\cite{KhaderSRH12},
which enables the protocol to recover (without a restart) from faulty participants who did not submit their votes.
As a consequence, robustness increases the resistance of our protocol to serious failures. 

\vspace{-0.2cm}
\subsection{Base Variant}
\label{ssec:basicprotocol}
\begin{figure}[ht]
	\centering
	
	\includegraphics[width=0.55\linewidth]{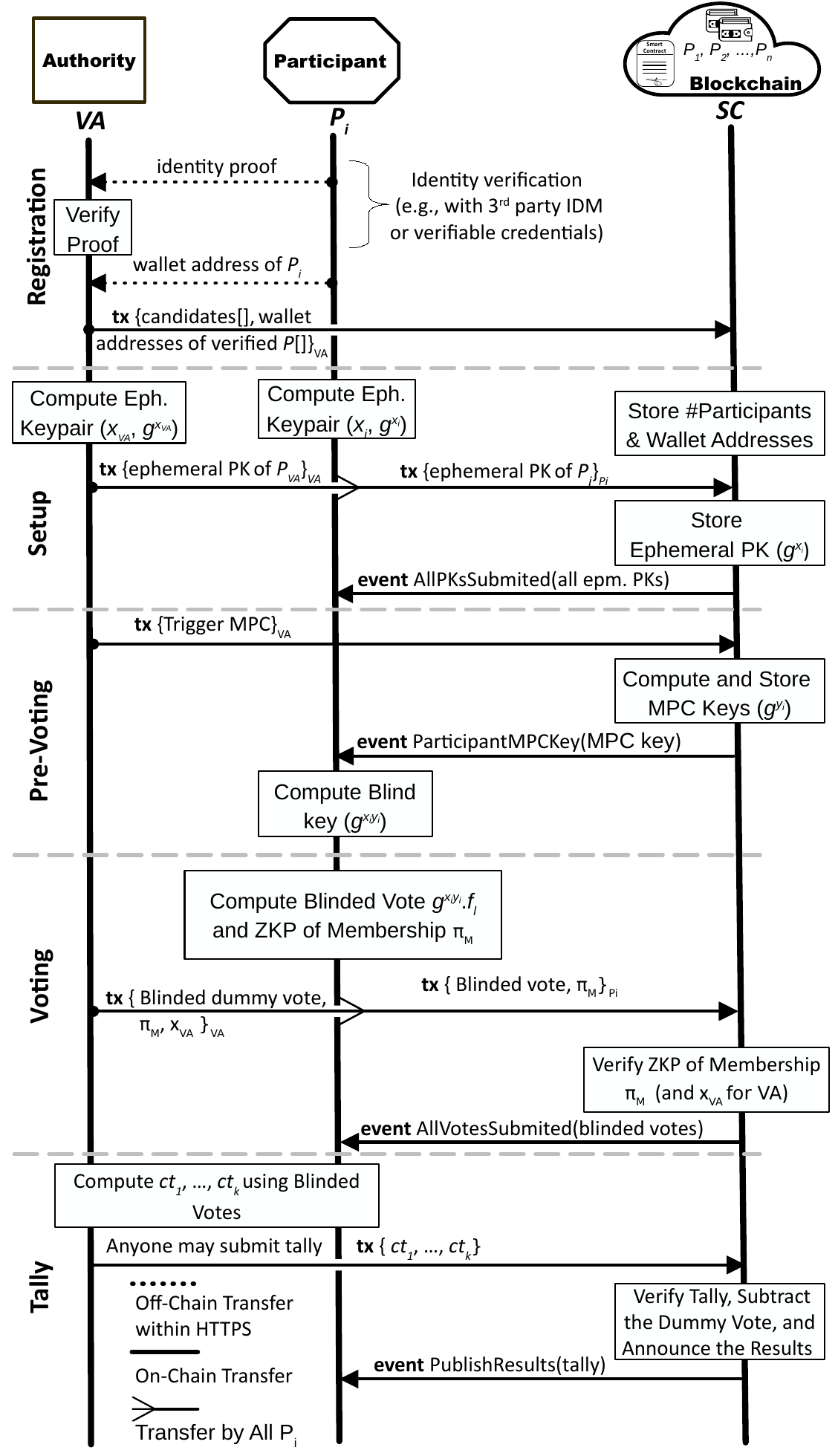}
	\caption{Basic protocol of \name.
	}
	\label{fig:bbb:operation-scheme-bw}
\end{figure}
%
The base variant of \name (see \autoref{fig:bbb:operation-scheme-bw}) does not involve a fault recovery and is divided into five stages: 
\textbf{registration} (identity verification,
key ownership verification, enrollment at $SC$), a \textbf{setup} (an agreement on system parameters, submission of ephemeral public keys), \textbf{pre-voting} (computation of MPC keys), \textbf{voting} (vote packing, blinding, and verification), and \textbf{tally} phases.
%
All faulty behaviors of $P$s and $VA$ are subject to deposit-based penalties. 
In detail, $P$ who submitted her ephemeral key (in the setup phase) and then has not voted within the timeout will lose the deposit. 
To achieve fairness, $VA$ acts as the last $P$ who submits a ``dummy'' vote with her ephemeral private key\footnote{Privacy for a dummy vote is not guaranteed since it is subtracted.} after all other $P$s cast their vote (or upon the voting timeout expiration). 

\subsubsection{\textbf{Phase 1 (Registration)}}
%
$VA$ first verifies the identity proof of each $P$.
For decentralized identity management (IDM), the identity proof is represented by the verifiable credentials (VC)~\cite{verifiable-credentails} signed by the issuer, while in a centralized IDM the identity proof is interactively provided by a third-party identity provider (e.g., Google).
First, $VA$ verifies the issuer's signature on the identity proof. 
Next, $VA$ challenges $P$ to prove (using her VC) that she is indeed the owner of the identity.
Further, each $P$ creates her blockchain wallet address (i.e., the blockchain public key (PK)) and provides it to $VA$.
The $VA$ locally stores a bijective mapping between a $P$'s identity and her wallet address.\footnote{Note that the address of $P$ must not be part of identity proof -- avoiding $VA$ to possess a proof of identity to blockchain address mapping (see \autoref{sec:blockchain-specific-issues}).} 
Next, $VA$ enrolls all verified $P$s by sending their wallet addresses to $SC$.


\subsubsection{\textbf{Phase 2 (Setup)}}
$P$s agree on system parameters that are universal to voting -- the parameters for voting are publicly visible on $SC$ (deployed by $VA$ in a transaction).
Therefore, any $P$ may verify these parameters before joining the protocol.
Note the deployment transaction also contains the specification of timeouts for all further phases of the protocol as well as deposit-based penalties for misbehavior of $VA$ and $P$s.
The parameters for voting are set as follows:

\setlength{\parskip}{5pt} \setlength{\itemsep}{5pt}
\begin{compactenum}[$\;$]
	\item[1)] $VA$ selects a common generator $g\in \mathbb{F}_p^*$. 
	The value of $p$ is chosen to be a safe prime, i.e., $p=2\cdot  q +1$, where $q$ is a prime. 
	A safe prime is chosen to ensure the multiplicative group of order  $p - 1 = 2\cdot q$, which has no small subgroups that are trivial to detect.\footnote{We use modular exponentiation by repeated squaring to compute $g^x$ $mod$ $p$, which has a time complexity of $\mathcal{O}$(($log$ $x$)$\cdot (log^2$ $p$))~\cite{Koblitz}. } 
	Let $n < p - 1$.	
	\item[2)] Any participant $P_i$ is later permitted to submit a vote $\{v_i ~|~ i\in \{1,2,...,k\}\}$ for one of $k$ choices. This is achieved by selecting $k$ independent generators $\{f_1,...,f_k\}$  in $\mathbb{F}_p^*$ (one for each choice).
	These generators for choices should meet a property described by Hao et al.~\cite{HaoRZ10} to preclude having two different valid tallies that fit \autoref{eqn:eq3}:
	\begin{eqnarray}
		f_i = 
		\begin{cases}
			g^{2^0}       & \quad \text{for choice 1},\\
			g^{2^m}       & \quad \text{for choice 2},\\
			\quad \quad \quad \cdots \\
			g^{2^{(k-1)m}} & \quad \text{for choice k},
		\end{cases}
	\end{eqnarray}
	where $m$ is the smallest integer such that $2^m > n$ (the number of participants).		
\end{compactenum}


\noindent\textbf{Ephemeral Key Generation \& Committing to Vote.}
Each $P_i$ creates her ephemeral private key as a random number $x_i\in_{R} \mathbb{F}_p^*$ and ephemeral public key as $g^{x_i}$.
Each $P_i$ sends her ephemeral public key to $SC$ in a transaction signed by her wallet, thereby, committing to submit a vote later.\footnote{In contrast to OVN~\cite{McCorrySH17} (based on the idea from~\cite{HaoRZ10}), we do not require $P_i$ to submit ZKP of knowledge of $x_i$ to $SC$ since $P_i$ may only lose by submitting $g^{x_i}$ to which she does not know $x_i$ (i.e., a chance to vote + deposit). 
}
Furthermore, $P_i$ sends a deposit in this transaction, which can be retrieved back after the end of voting.
However, if $P_i$ does not vote within a timeout (or does not participate in a fault recovery (see \autoref{sec:bbbvoting-faulttolerance})), the deposit is lost, and it is split to the remaining involved parties.
$P$s who do not submit their ephemeral keys in this stage are indicating that they do not intend to vote; the protocol continues without them and they are not subject to penalties.		
Finally, each $P$ obtains (from $SC$) the ephemeral public keys of all other verified $P$s who have committed to voting.
Ephemeral keys are one-time keys, and thus can be used only within one run of the protocol to ensure privacy of votes (other runs require fresh ephemeral keys). 

\subsubsection{\textbf{Phase 3 (Pre-Voting)}}
This phase represents multiparty computation (MPC), which is run to synchronize the keys among all $P$s and achieve the \textit{self-tallying} property.
However, no direct interaction among $P$s is required since all ephemeral public keys are  published at $SC$. The MPC keys are computed by $SC$, when $VA$ triggers the compute operation via a transaction.
The $SC$ computes and stores the MPC key for each $P_i$ as follows:
\begin{equation}
	\label{eqn:eq1}
	h=g^{y_i}=\prod\limits_{j=1}^{i-1} g^{x_j}/\prod\limits_{j=i+1}^{n} g^{x_j},
\end{equation}
where $y_i=\sum_{j<i}x_j- \sum_{j>i}x_j$
and
$\sum_{i}x_i y_i=0$ (see Hao et al.~\cite{HaoRZ10} for the proof).
%
While anyone can compute $g^{y_i}$, to reveal $y_i$, all $P$s $\setminus $ $P_i$ must either collude or solve the DLP for \autoref{eqn:eq1}.
%
As the corollary of \autoref{eqn:eq1}, the protocol preserves vote privacy if there are at least 3 $P$s with at least 2 honest (see the proof in the extended version of our paper \cite{venugopalan2020bbbvoting}).


\begin{figure}[t]
	\begin{center}
		\scriptsize
		\fbox{
			\begin{protocolm}{2}						
				\participants{\underline{Participant $P_i$}}{\underline{Smart Contract}}
				\participants{($~h\leftarrow g^{y_i},~ v_i$)}{($~h\leftarrow g^{y_i}$)}
				\hline	
				& &\\		
				Select~v_i\in\{1,...,k\}, & & \\	 				
				Use~choice~generators~\\
				f_l\in\{f_1,...,f_k\}\subseteq\mathbb{F}_p^*,\\			
				Publish~ x\leftarrow g^{x_i} && 	\\
				Publish~ B_i\leftarrow h^{x_i}{f_l} & & \\
				
				
				w\in_{R} \mathbb{F}_p^*\\

				\forall l\in \{1,..,k\}\setminus {v_i}:\\
				
				\quad 1.~r_l,d_l\in_{R} \mathbb{F}_p^* \\
				\quad 2.~a_l\leftarrow x^{-d_l}g^{r_l}\\ 				
				\quad 3.~b_l\leftarrow h^{r_l}(\frac{B_i}{f_l})^{-d_l} \\

				for\enspace {v_i}:\\
				
				\quad 1.~a_{v_i}\leftarrow g^{w}\\
				\quad 2.~b_{v_i}\leftarrow h^{w} \\
				\hdashline
				c\gets\Hzk(\{ \{a_l\},\{b_l\} \}_{l})\\
				\hdashline
				
				for\enspace {v_i}:\\
				
				\quad 1.~d_{v_i}\leftarrow \sum_{l\neq{v_i}}d_l\\
				\quad 2.~d_{v_i}\leftarrow c - d_{v_i}\\
				
				\quad 3.~r_{v_i}\leftarrow w+x_id_{v_i} \\			
				\quad 4.~q\leftarrow p-1\\
				\quad 5.~{r_{v_i}\leftarrow r_{v_i}\mod{q}}\\

				& \forall l: \{a_l\}, \\
				& \sends{\{b_l\}, \{r_l\},\{d_l\}}
				&\\
				
				& &\Psi \gets \{\forall l: \{a_l\},\{b_l\}\}  \\
				&& c \gets\Hzk(\Psi)\\
				
				& &\sum_{l} d_{l} \stackrel{?}{=} c\\
				
				&& \forall l\in \{1,..,k\}\\

				&&\quad 1.~g^{r_l}\stackrel{?}{=}{a_l}x^{d_l}\\
				&&\quad 2.~h^{r_l}\stackrel{?}{=}{b_l}(\frac{B_i}{{f_l}})^{d_l}\\
			\end{protocolm}
		}
	\end{center}
	\caption{ZKP of set membership for 1-out-of-$k$ choices. }
	\label{fig:ZKPMultiCandidate}
\end{figure}	

\subsubsection{\textbf{Phase 4 (Voting)}}
In this phase, each $P_i$  blinds and submits her
vote to $SC$.
These steps must ensure the recoverability of the tally, vote privacy, and well-formedness of the vote.
Vote privacy is achieved by multiplying the $P_i$'s blinding key with her vote choice.
The blinded vote of the participant $P_i$ is
\begin{equation}
	\label{eqn:eq2}
	B_i=\begin{cases}
		g^{x_iy_i}f_1 & \text{if $P_i$ votes for choice 1},\\
		g^{x_iy_i}f_2 & \text{if $P_i$ votes for choice 2},\\
		~~~~~~~~~~~~~...\\
		g^{x_iy_i}f_k & \text{if $P_i$ votes for choice \textit{k}.}
	\end{cases}
\end{equation}
The $P$ sends her choice
within a blinded vote along with a  1-out-of-$k$ non-interactive zero-knowledge (NIZK) proof of set membership to $SC$ (i.e., proving that the vote choice $\in \{1, ..., k\}$).
We modified the approach proposed by Kiayias et al.~\cite{Kiayias2002} to the form used by Hao et al.~\cite{HaoRZ10}, which is convenient for practical deployment on existing smart contract platforms.
The verification of set membership using this protocol is depicted in \autoref{fig:ZKPMultiCandidate}, where $P_i$ is a prover and $SC$ is the verifier.
Hence, $SC$ verifies the correctness of the proof and then stores the blinded vote.
In this stage, it is important to ensure that no re-voting is possible, which is to avoid any inference about the final vote of $P$ in the case she would change her vote choice during the voting stage.
Such a re-voting logic can be enforced by $SC$, while user interface of the $P$ should also not allow it.
Moreover, to ensure fairness, $VA$ acts as the last $P$ who submits a dummy vote and her ephemeral private key.

\subsubsection{\textbf{Phase~5~(Tally)}}
When the voting finishes (i.e., voting timeout expires or all $P$s and $VA$ cast their votes), the tally of votes received for each of $k$ choices is computed off-chain by any party and then submitted to $SC$.
When $SC$ receives the tally, it verifies whether \autoref{eqn:eq3} holds, subtracts a dummy vote of $VA$, and notifies all $P$s about the result. 
The tally is represented by vote counts $ct_i,\forall i\in\{1,...,k\}$ of each choice, which are computed using an exhaustive search fitting 
\begin{equation}
	\label{eqn:eq3}
	\prod\limits_{i=1}^{n}B_i=
	\prod\limits_{i=1}^{n}g^{x_iy_i}f=
	g^{\sum_{i}x_i y_i}f=
	{f_1}^{ct_1}{f_2}^{ct_2}...{f_k}^{ct_k}.
\end{equation}
The maximum number of attempts is bounded by combinations with repetitions to ${n+k-1}\choose{k}$.
%
Although the exhaustive search of 1-out-of-$k$ voting is more computationally demanding in contrast to 1-out-of-2 voting~\cite{McCorrySH17},~\cite{HaoRZ10}, this process can be heavily parallelized. 
See time measurements in \autoref{sec:tally-exps}.

\subsection{Variant with Robustness}\label{sec:bbbvoting-faulttolerance}

We extend the base variant of \name by a fault recovery mechanism.
If one or more $P$s stall (i.e., are faulty) and do not submit their blinded vote in the voting stage despite committing in doing so, the tally cannot be computed directly.
To recover from faulty $P$s, we adapt the solution proposed by Khader et al.~\cite{KhaderSRH12}, and we place the fault recovery phase immediately after the voting phase. 
All remaining honest $P$s are expected to repair their vote by a transaction to $SC$, which contains key materials shared with all faulty $P$s and their NIZK proof of correctness.
$SC$ verifies all key materials with proofs (see  \autoref{fig:ZKPdiffie}), and then they are used to invert out the counter-party keys from a blinded vote of an honest $P$ who sent the vote-repairing transaction to $SC$.
Even if some of the honest (i.e., non-faulty) $P$s would be faulty during the recovery phase (i.e., do not submit vote-repairing transaction), it is still possible to recover from such a state by repeating the next round of the fault recovery, but this time only with key materials related to new faulty $P$s.
To disincentivize stalling participants, they lose their deposits, which is split across remaining $P$s as a compensation for the cost of fault recovery. 

\begin{figure}[t]
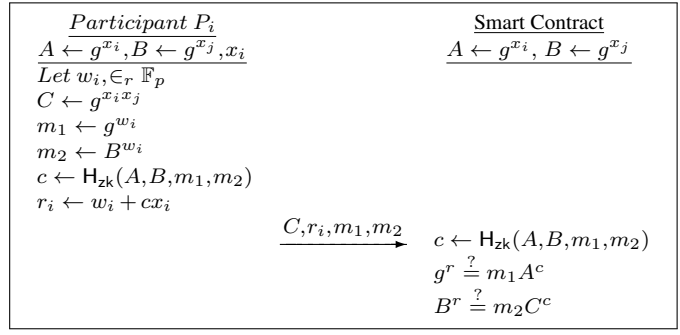

	\begin{center}
		\fbox{\scriptsize
			\begin{protocolm}{2}	
				
				\participants{\underline{$Participant~P_i$}}{\underline{Smart Contract}}
				\participants{\underline{$A\leftarrow g^{x_i},B
						\leftarrow g^{x_j},x_i$}}{\underline{$A\leftarrow g^{x_i},~B\leftarrow g^{x_j}$}}
				Let~w_i,\in_{r}\mathbb{F}_p\\
				C\leftarrow g^{x_ix_j}\\
				m_1\leftarrow {g}^{w_i}\\
				m_2\leftarrow {B}^{w_i}  \\ 

				c\leftarrow  \Hzk(A,B,m_1,m_2)\\ 	
				r_i\leftarrow w_i+ cx_i &&\\
				&\sends{C,r_i,m_1,m_2} & c\leftarrow  \Hzk(A,B,m_1,m_2)\\			

				&  & g^r\stackrel{?}{=}{m_1}{A^c}\\
				&  & B^r\stackrel{?}{=}{m_2}{C^c}\\

			\end{protocolm}
		}
	\end{center}
	\caption{ZKP verifying correspondence of $g^{x_ix_j}$ to public keys $A=g^{x_i}, B=g^{x_j}$. 
	}
	\label{fig:ZKPdiffie}
\end{figure}


\subsection{Implementation}\label{sec:bbbvoting-implementation}

We selected the Ethereum-based environment for evaluation due to its wide\-spread adoption and open standardized architecture (driven by the Enterprise Ethereum Alliance \cite{EEA}), which is incorporated by many blockchain projects.
We implemented $SC$ components in Solidity, while $VA$ and $P$ components were implemented in Javascript. 
In this section, we analyze the costs imposed by our approach, perform a few optimizations, and compare the costs with OVN~\cite{McCorrySH17}.
%
In the context of this work, we assume 10M as the block gas limit.
With block gas limit assumed, our implementation supports up to 135 participants (see \autoref{fig:mpc-comparison}), up to 7 vote choices (see \autoref{fig:submit-vote}), and up to 9 simultaneously stalling faulty participants (see \autoref{fig:repair-vote}).\footnote{The max. corresponds to a single recovery round but the total number of faulty participants can be unlimited since the fault recovery round can repeat.}
%


We made two different implementations, the first one  is based on DLP for integers modulo $p$ (denoted as integer arithmetic (IA)), and the second one is based on the elliptic curve DLP (denoted as ECC).
In the ECC, we used a standardized \textit{Secp256k1} curve from existing libraries~\cite{witnet-ecc-solidity}, \cite{McCorrySH17}.
In the case of IA, we used a dedicated library~\cite{bignumber-solidity} for operations with big numbers since EVM natively supports only 256-bit long words, which does not provide sufficient security level with respect to the DLP for integers modulo $p$.\footnote{Since this DLP was already computed for 795-bit long safe prime in 2019~\cite{integer-dlp-2019}, only values higher than 795-bit are considered secure enough.}
We consider 1024 bits the minimal secure (library-supported) length of numbers in IA.
As we will show below, IA implementation even with minimal 1024 bits is overly expensive, and thus in many cases does not fit the block gas limit by a single transaction. 
Therefore, in our experiments, we put emphasis on ECC implementation. 
The source code of our implementation is available at \url{https://github.com/ivan-homoliak-sutd/BBB-Voting}.

\begin{figure}[t]
	\centering
	\includegraphics[width=0.55\linewidth]{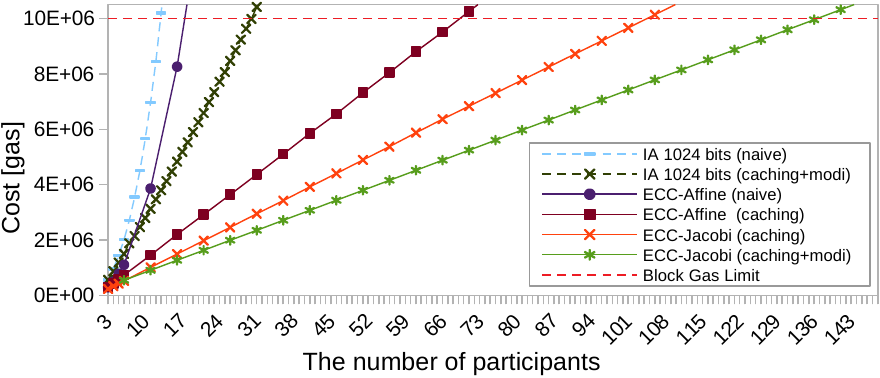}
	\caption{A computation of MPC keys by the authority.} 
\label{fig:mpc-comparison}
\end{figure}

\begin{figure}[t]
	\centering	
	\subfloat[Vote submission by $P_i$.\label{fig:submit-vote}]{	
		\includegraphics[width=0.5\columnwidth]{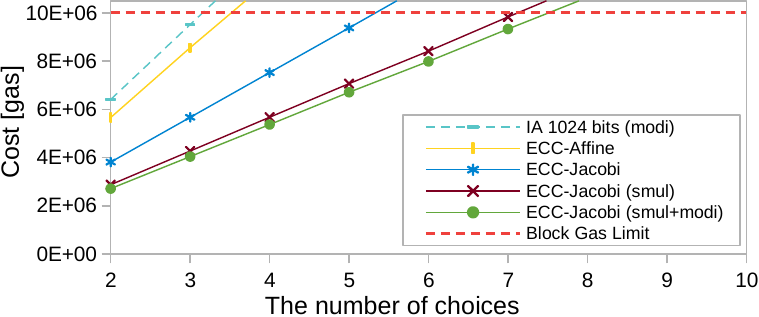}
		\vspace{0.3cm}
	}
	\subfloat[Vote repair by $P_i$.\label{fig:repair-vote}]{	
		\includegraphics[width=0.5\columnwidth]{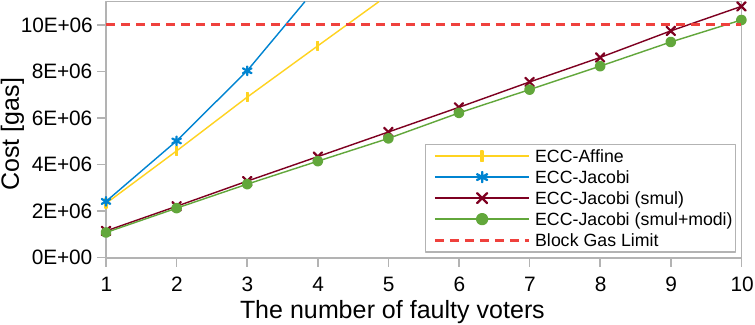}	
		\vspace{0.3cm}
	}		
	\vspace{0.2cm}
	\caption{Vote submission and vote repair (i.e., fault recovery) with various optimizations.}
	\label{fig:submit-and-repair-vote}
	
\end{figure}

\subsection{Evaluation \& Cost Optimizations}
\label{ssec:optimizatations}

Since ECC operations in ZKP verifications and computation  of MPC keys impose a high execution cost when running at the blockchain, we have made several cost optimizations.

\subsubsection*{\textbf{(1) Caching in MPC Key Computation.}}
If implemented nai\-vely, the computation of all MPC keys in $SC$ would contain a high number of overlapped additions, and hence the price would be excessively high (see series \textit{``ECC-Affine (naive)''} in \autoref{fig:mpc-comparison}).\footnote{The same phenomenon occurs in IA (see \autoref{eqn:eq1}) but with overlapped multiplications (see series \textit{``IA 1024 bits (naive)''}).}
Therefore, in the code of $SC$ we accumulate and reuse the value of the left side of MPC key computation during iteration through all participants. 
Similarly, the sum at the right side can be computed when processing the first participant, and then in each follow-up iteration, it can be subtracted by one item. 
However, we found out that subtraction imposes non-negligible costs since it contains one affine transformation (which we later optimize). 
In the result, we found pre-computation of all intermediary right items in the expression during the processing of the first participant as the most optimal. 
The resulting savings are depicted as \textit{``ECC-Affine (caching)''} series in \autoref{fig:mpc-comparison}.
We applied the same optimization for IA; however, even after adding a further optimization (i.e., pre-computation of modular inverses; see \autoref{ssec:optimizatations}.4), the costs were still prohibitively high (see \textit{``IA 1024 bits (caching+modi)''} series in \autoref{fig:mpc-comparison}), with the max. number of participants fitting the gas limit only 29.


\subsubsection*{\textbf{(2) Affine vs. Jacobi Coordinates.}}
In the ECC libraries employed~\cite{witnet-ecc-solidity} \cite{McCorrySH17}, by default, all operations are performed with points in Jacobi coordinates and then explicitly transformed to affine coordinates.
However, such a transformation involves one modular inversion and four modular multiplications over 256-bit long integers, which is costly.
Therefore, we maximized the utilization of internal functions from the Witnet library~\cite{witnet-ecc-solidity}, which do not perform affine transformation after operation execution but keep the result in Jacobi coordinates.
This is possible only until the moment when two points are compared -- a comparison requires affine coordinates. 
Hence, a few calls of the affine transformation are inevitable.
%
This optimization is depicted in \autoref{fig:mpc-comparison} (series \textit{``ECC-Jacobi (caching)''}) and \autoref{fig:submit-vote} (series \textit{``ECC-Jacobi''}). 
In the case of computation  of MPC keys, this optimization brought improvement of costs by $23\%$ in contrast to the version with affine coordinates and caching enabled.
Due to this optimization, up to $111$ participants can be processed in a single transaction not exceeding the block gas limit.
%
In the case of vote submission, this optimization brought improvement of costs by $33\%$ in contrast to the version with affine coordinates.

\subsubsection*{\textbf{(3) Multiplication with Scalar Decomposition.}}
The most expensive operation on an elliptic curve is a scalar multiplication; based on our experiments, it is often 5x-10x more expensive than the point addition since it involves several point additions (and/or point doubling).
The literature proposes several ways of optimizing the scalar multiplication, where one of the most significant ways is w-NAF (Non-Adjacent Form) scalar decomposition followed by two simultaneous multiplications with scalars of halved sizes~\cite{hankerson2005guide}.
This approach was also adopted in the Witnet library~\cite{witnet-ecc-solidity} that we base on.
The library boosts the performance (and decreases costs) by computing the sum of two simultaneous multiplications $kP + lQ$, where $k = (k_1 + k_2 \lambda)$, $l = (l_1 + l_2 \lambda)$, and $\lambda$ is a specific constant to the endomorphism of the elliptic curve.
To use this approach, a scalar decomposition of $k$ and $l$ needs to be computed beforehand.
Nevertheless, such \textbf{a scalar decomposition can be computed off-chain (and verified on-chain)}, while only a simultaneous multiplication is computed on-chain.
However, to leverage the full potential of the doubled simultaneous multiplication, one must have the expression $kP + lQ$, which is often not the case.
In our case, we modified the check at $SC$ to fit this form.
Alike the vote submission, this optimization can be applied in vote repair.
We depict the performance improvement brought by this optimization as series \textit{``ECC-Jacobi (smul)''} in \autoref{fig:submit-and-repair-vote}.
%


\subsubsection*{\textbf{(4) Pre-Computation of Modular Inversions.}}\label{sec:optim-modular-inv}
Each affine transformation in the vote submission contains one operation of modular inversion -- assuming previous optimizations, ZKP verification of one item in 1-out-of-$k$ ZKP requires three affine transformations (e.g., for $k=5$, it is $15$). 
Similarly, the ZKP verification of correctness in the repair vote requires two affine transformations per each faulty participant submitted.
The modular inversion operation runs the extended Euclidean algorithm, which imposes non-negligible costs.
However, all modular inversions \textbf{can be pre-computed off-chain, while only their verification can be made on-chain} (i.e., modular multiplication), which imposes much lower costs.
We depict the impact of this optimization as \textit{``ECC-Jacobi (smul+modi)''} series in \autoref{fig:submit-and-repair-vote} and ``...modi...'' series in \autoref{fig:mpc-comparison}.
In the result, it has brought $5\%$ savings of costs in contrast to the version with the simultaneous multiplication.


\vspace{-0.2cm}
\subsubsection{Tally Computation}\label{sec:tally-exps}
In \autoref{table:bruteforcetime}, we provide time measurements of tally computation through the entire search space on 1 core vs. all cores of the i7-10510U CPU laptop.\footnote{In some cases we estimated the time since we knew the number of attempts.} 
We see that for $n \le 100$ and $k \le 6$, the tally can be computed even on a commodity PC in a reasonable time.
However, for higher $n$ and $k$, we recommend using a more powerful machine or distributed computation across all $P$s. 
One should realize that our measurements correspond to the upper bound, and if some ranges of tally frequencies are more likely than other ones, they can be processed first -- in this way, the computation time can be significantly reduced.
Moreover, we emphasize that an exhaustive search for tally computation is not specific only to our scheme but to homomorphic-encryption-based schemes providing perfect ballot secrecy and privacy of votes (e.g.,~\cite{Kiayias2002},~\cite{HaoRZ10},~\cite{McCorrySH17}).

\renewcommand{\arraystretch}{0.9}
\setlength{\tabcolsep}{0.1cm}
\begin{table}[t]
	\centering	
	\scriptsize
	
	\subfloat[1 core]{	
		\begin{tabular}{r l l l l}		
			\toprule
			\textbf{Voters} & \multicolumn{4}{c}{\textbf{Choices}}  \\ [0.5ex]
			
			(n) & ~$k=2$ & ~$k=4$ & ~$k=6$ & ~$k=8$  \\ [0.5ex]
			\midrule
			\textbf{20} & $0.01s$ & $0.01s$ &  $0.07s$ & $0.07s$\\						
			\textbf{30} & $0.01s$ & $0.01s$ & $0.53s$ & $13.3s$\\								
			\textbf{40} & $0.01s$ & $0.04s$ & $02.6s$ & $112s$\\							
			\textbf{50} &$0.01s$ & $0.08s$ & $10.0s$ & $606s$\\
			
			\textbf{60} &$0.01s$ & $0.16s$ & $28.2s$ & $2424s$\\
			\textbf{70} &$0.01s$ & $0.48s$ & $69.6s$ & $\sim2.1h$\\
			\textbf{80} &$0.01s$ & $0.82s$ & $160s$ & $\sim5.8h$\\
			\textbf{90} &$0.01s$ & $1.08s$ & $320s$ & $\sim14.2h$\\

			\textbf{100}& $0.01s$ & $1.2s$ & $722s$ & $\sim33h$\\
			
			\bottomrule
		\end{tabular}
	}
	\hspace{0.2cm}
	\subfloat[8 cores]{
		
		\begin{tabular}{r l l}		
			\toprule
			\textbf{Voters} & \multicolumn{2}{c}{\textbf{Choices}}  \\ [0.5ex]
			
			(n)  & ~$k=6$ & ~$k=8$  \\ [0.5ex]
			\midrule
			\textbf{20}  &  $0.01s$ & $0.01s$\\						
			\textbf{30}  & $0.08s$ & $2.0s$\\								
			\textbf{40} & $0.39s$ & $16.8s$\\							
			\textbf{50}  & $1.5s$ & $90.9s$\\					
			\textbf{60}  & $4.44s$ & $267s$\\					
			\textbf{70}  & $11.85s$ & $773s$\\								
			\textbf{80}  & $19.46s$ & $2210s$\\								
			\textbf{90}  & $44.02s$ & $\sim2.7h$\\								
			\textbf{100} &  $108.3s$ & $\sim4.9h$\\			
			\bottomrule
		\end{tabular}
	}
	\vspace{0.3cm}
	\caption{Upper time bound for tally computation.}	\label{table:bruteforcetime}
	\vspace{-0.3cm}
	
\end{table}

\renewcommand{\arraystretch}{1.1}
\setlength{\tabcolsep}{0.15cm}
\begin{table}[b]
	\centering
	\scriptsize
	
	\begin{tabular}  {l c l l}
		\toprule
		& \textbf{Gas Paid by} & \textbf{OVN} &\textbf{\name}   \\
		\midrule
		
		\textbf{\specialcell{Deployment of Voting $SC$}}   	& $VA$ &  3.78M  & 4.8M  \\
		
		\textbf{\specialcell{Deployment of\\Cryptographic $SCs$}}  & $VA$ &  \specialcell{2.44M}  & \specialcell{2.15M\\(1.22M+0.93M)}   \\
		
		\textbf{Enroll voters} & $VA$ & \specialcell{2.38M\\(2.15M+0.23M)}  & 1.93M \\
		
		\textbf{\specialcell{Submit Ephemeral PK}} & $P$ &  0.76M & 0.15M  \\
		
		\textbf{Cast Vote} & $P$ &  2.50M & 2.72M \\
		
		\textbf{Tally} & $VA ~(or ~P)$ & 0.75M  & 0.39M  \\				
		
		\midrule
		\textbf{Total Costs for} $\mathbf{P}$ &   & 3.26M  & 2.87M \\
		
		\textbf{Total Costs for} $\mathbf{VA}$ &  &  9.35M  & 9.27M  \\
		
		
		\bottomrule
	\end{tabular}
	\caption{A normalized cost comparison of \name with OVN for $n=40$ and $k=2$.}\label{tab:costcomparisons}
\end{table}

\subsubsection{Cost Comparison}
In \autoref{tab:costcomparisons}, we made a cost comparison of \name (using ECC) with OVN~\cite{McCorrySH17}, where we assumed two choices and 40 participants (the same setting as in~\cite{McCorrySH17}).
We see that the total costs are similar but \name improves $P$'s costs by $13.5\%$ and $VA$'s cost by $0.9\%$ even though using more complex setting that allows 1-out-of-k voting.
We also emphasize that the protocol used for vote casting in \name \textbf{contains more operations} than OVN but regardless of it, the costs are close to those of OVN, which is mostly caused by the proposed optimizations.\footnote{To verify 1-out-of-$k$ ZKP in vote casting, \name computes $5 \cdot k$ multiplications and $3 \cdot k$ additions on the elliptic curve -- i.e., 10 multiplications and 6 additions for  $k=2$. In contrast, OVN computes only 8 multiplications and 5 additions for $k=2$.}
Next, we found that OVN computes tally on-chain, which is an expensive option.
In contrast, \name computes tally off-chain and $SC$ performs only verification of its correctness, which enables us to minimize the cost of this operation.
Another gas saving optimization of \name in contrast to OVN (and Hao et al.~\cite{HaoRZ10}) is that we do not require voters to submit ZKP of knowledge of $x_i$ in $g^{x_i}$ during the registration phase to $SC$ since $P_i$ may only lose by providing incorrect ephemeral public key $g^{x_i}$ -- she might lose the chance to vote and her deposit.
Finally, we note that we consider the deployment costs of our $SC$ equal to 4.8M units of gas; however, our $SC$ implementation contains a few auxiliary view-only functions for a pre-computation of modular inverses, with which, the deployment costs would increase to 7.67M due to code size.
Nevertheless, these operations can be safely off-chained and we utilized them on-chain only for simplicity. 

\subsection{Limitations \& Extensions}
\label{sec:discussion}
In this section, we discuss the extensions addressing the scalability and performance limitations of \name. 

\begin{figure}
	\centering
	\includegraphics[width=0.6\linewidth]{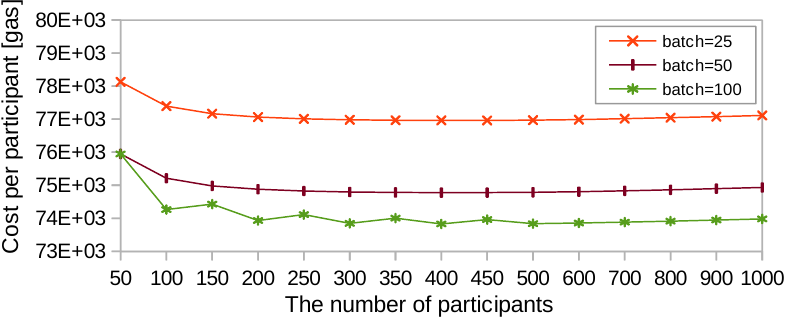}
	\caption{A computation of MPC keys by the authority $VA$ using various batch sizes and the most optimized ECC variant.
	}
	\label{fig:mpc-batches}
\end{figure}

\subsubsection{Scalability Limitation \& Extension}
\label{ssec:scalability-improvements}
The limitation of \name (like in OVN) is a lack of scalability, where the block gas limit might be exceeded with a high number of $P$s.
Therefore, we primarily position our solution as boardroom voting; however, we will  show in this section that it can be extended even to larger voting.
Our voting protocol (see \autoref{ssec:basicprotocol}) has a few platform-specific bottlenecks due to the block gas limit of Ethereum, when either $k$ or $n$ reach particular values.
%
%
Nevertheless, transaction batching can be introduced for the elimination of all bottlenecks.
To realize a batching of pre-voting, voting, and fault-recovery  phases, the additional integrity preservation logic across batches needs to be addressed while the verification of integrity has to be made by $SC$.
For demonstration purposes, we addressed the bottleneck of the pre-voting stage (see \autoref{fig:mpc-batches})  and setup stage, 
which further improves the vote privacy in \name (see \autoref{sec:blockchain-specific-issues}) and causes only a minimal cost increase due to the overhead of integrity preservation (i.e., 1\%).
With this extension, $n$ is limited only by the expenses paid by $VA$ to register $P$s and compute their MPC keys, and the computing power to obtain the tally.\footnote{E.g., for $n = 1000$, $k=2 ~(\text{and} ~k=4)$, it takes $0.15$s (and $\sim4$h) to obtain the tally on a commodity PC with 8 cores, respectively.}
If a certain combination of high $n$ and $k$ would make computation of a tally overly computationally expensive (or the cost of its verification by $SC$), it is further possible to partition $P$s into multiple groups (i.e., voting booths), each managed by an instance of \name, while the total results could be summed across instances.
Scalability extension is a subject of our next work \cite{stanvcikova2022sbvote}, which will be described in \autoref{sec:voting-sbvote}. 


\subsubsection{Cost and Performance Limitations}\label{ssec:cost-perf-limitations}
Although we thoroughly optimized the costs of our implementation (see \autoref{ssec:optimizatations}), the expenses imposed by a public permissionless smart contract platform might be still high, especially during peaks of the gas price and/or in the case of a larger voting than boardroom voting (see \autoref{ssec:scalability-improvements}).
Besides, the transactional throughput of such platforms might be too small for such larger voting instances to occur in a specified time window.
Therefore, to further optimize the costs of \name and its performance, it can run on a 
public permissioned Proof-of-Authority platforms, e.g., using Hyperledger projects (such as safety-favored Besu with Byzantine Fault Tolerant (BFT) protocol).
Another option is to use smart contract platforms utilizing the trusted computing that off-chains expensive computations (e.g., Ekiden~\cite{cheng2019ekiden}, TeeChain~\cite{lind2019teechain}, Aquareum~\cite{homoliak2020aquareum}), 
or other partially-decentralized second layer solutions (e.g., Plasma \cite{plasma-eth}, Polygon Matic~\cite{polygon-matic-investopedia}, Hydra~\cite{chakravarty2020hydra}). 
Even though these solutions might preserve most of the blockchain features harnessed in e-voting, availability and decentralization might be decreased, which is the security/performance trade-off.

\subsection{Security Analysis}
\label{sec:bbbvoting:security}
\vspace{-0.2cm}
We first analyze security of \name with regard to the voting properties specified in \autoref{sec:evoting-preliminaries}.
Next, we analyze blockchain-specific security \& privacy issues and discuss their mitigations. 
Also, we compare voting and other features of \name with a few related works in \autoref{tab:bbbvoting:comparison}.
%

\renewcommand{\arraystretch}{1.1}
\setlength{\tabcolsep}{6.5pt}
\begin{table}[t]
	\vspace{-0.3cm}
	\centering
	\scriptsize
	
	\begin{tabular}[t]  {l c c c c c c c c c c}
		\toprule
		\textbf{\specialcell{Approach\\}} & \trot{\textbf{Privacy of Votes}} & \trot{\textbf{\specialcell{Perfect Ballot Secrecy}}} & \trot{\textbf{Fairness}}  & \trot{\textbf{Self-Tallying}} & \trot{\textbf{Robustness}} &  \trot{\textbf{\specialcell{ Uses Blockchain}}} & \trot{\textbf{Uni. Verifiability}} &  \trot{\textbf{E2E Verifiability}}  & \trot{\textbf{Open Source}} & \trot{\textbf{Choices}}\\
		\toprule
		
		Hao et al.~\cite{HaoRZ10} & \hfil\cmark & \hfil\cmark & \hfil\xmark & \hfil\cmark & \hfil\xmark & \hfil\xmark  & \hfil\cmark & \hfil\cmark  & \hfil\xmark &  \hfil $2$\\
		
		Khader et.~\cite{KhaderSRH12} & \hfil\cmark & \hfil\cmark & \hfil\cmark & \hfil\cmark & \hfil\cmark  & \hfil\xmark  & \hfil\cmark & \hfil\cmark & \hfil\xmark & \hfil $2$\\
		
		Kiayias and Yung~\cite{Kiayias2002} & \hfil\cmark & \hfil\cmark & \hfil\cmark  & \hfil\cmark & \hfil\cmark & \hfil\xmark  & \hfil\cmark & \hfil\cmark  & \hfil\xmark & \hfil $2 / k$\\
		
		McCorry et.~\cite{McCorrySH17} (OVN) & \hfil\cmark & \hfil\cmark & \hfil\cmark  & \hfil\cmark & \hfil\xmark &  \hfil\cmark & \hfil\cmark & \hfil\cmark  & \hfil\cmark & \hfil $2$\\
		
		\specialcell{Seifelnasr et al.~\cite{DBLP:conf/fc/SeifelnasrGY20}\\~~~~(sOVN)} & \hfil\cmark & \hfil\cmark & \hfil\cmark & \hfil\cmark & \hfil\xmark & \hfil\cmark & \hfil\cmark & \hfil\cmark & \hfil\cmark & \hfil $2$  \\

		Li et al.~\cite{Li2020} & \hfil\cmark & \hfil\cmark & \hfil\cmark & \hfil\cmark & \hfil\cmark &  \hfil \cmark  & \hfil\cmark & \hfil\cmark & \hfil\xmark & \hfil $2$\\
		
		
		Baudron et al.~\cite{Baudron2001} & \hfil\cmark & \hfil\xmark & \hfil\xmark & \hfil\xmark & \hfil\cmark &  \hfil\xmark & \hfil\cmark & \xmark & \hfil\xmark & \hfil $k$\\ 
		
		Groth~\cite{Groth2004} & \hfil\cmark & \hfil\cmark & \hfil\cmark & \hfil\cmark & \hfil\xmark & \hfil\xmark  & \hfil\cmark & \cmark & \hfil\xmark & \hfil $k$\\
		
		Adida~\cite{Adida08} (Helios) & \hfil\cmark$^*$  & \hfil\xmark & \hfil\xmark & \hfil\xmark & \hfil\cmark & \hfil\xmark & \hfil\cmark & \cmark & \hfil\cmark & \hfil $k$  \\
		
		Matile et al.~\cite{ICBC:MRSS19} (CaIV) & \hfil\cmark & \hfil\xmark & \hfil\xmark & \hfil\xmark & \hfil\cmark & \hfil\cmark & \hfil\xmark & \hfil\xmark & \hfil\cmark &  \hfil $k$\\
		
		Killer~\cite{killerprovotum} (Provotum) & \hfil \cmark & \xmark & \cmark & \xmark & \cmark & \cmark & \cmark & \cmark & \hfil \cmark & $2$ \\			
		
			
			\specialcell{Dagher et al.~\cite{icissp:DMMM18}\\ ~~~~(BroncoVote)}			 & \hfil \cmark & \cmark$^*$ & \cmark$^*$ & \xmark & \xmark & \cmark & \xmark & \xmark & \xmark & \hfil $k$  \\

			Kostal et al.~\cite{kovst2019blockchain} & \cmark$^*$  &  \cmark$^*$ & \xmark & \cmark & \cmark$^*$ & \cmark & \cmark  & \cmark & \xmark & $k$  \\

			
				
				\specialcell{Zagorski et al.~\cite{zagorski2013remotegrity}\\ ~~~~(Remotegrity)}		 & \cmark$^*$ & \xmark & \xmark & \xmark & \cmark & \xmark & \xmark & \cmark & \cmark & $k$  \\
				
					
					
					
					
					Yu et al.~\cite{yu2018platform}	& \hfil\cmark & \hfil\cmark$^*$ & \hfil\cmark$^*$ & \hfil\cmark & \hfil\cmark$^*$ & \hfil\cmark & \hfil\cmark & \hfil\xmark  &  \hfil\xmark & \hfil $k$  \\

					\textbf{BBB-Voting} & \hfil\cmark & \hfil\cmark & \hfil\cmark & \hfil\cmark & \hfil\cmark &  \hfil\cmark & \hfil\cmark & \hfil\cmark & \hfil\cmark & \hfil $k$\\
					
					\bottomrule
				\end{tabular}
				\caption{A comparison of various remote voting protocols. $^*$Assuming a trusted $VA$.}
				\label{tab:bbbvoting:comparison}
				\vspace{-0.3cm}
			\end{table}

%

\subsubsection{Properties of Voting}
\par\noindent\textbf{(1) Privacy}
in \name requires at least 3 $P$s, out of which at least 2 are honest (see \autoref{ssec:basicprotocol}).
Privacy in BBB-voting is achieved by blinding votes using ElGamal encryption~\cite{Elgamal85}, whose security is based on the decisional Diffie-Hellman assumption. Unlike the conventional ElGamal algorithm, a decryption operation is not required to unblind the votes. 
Instead, we rely on the self-tallying property of the voting protocol. 
The ciphertext representing a blinded vote is a tuple ($c_1,c_2)$, where $c_1=g^{xy}.f$ and $c_2=g^y$, where the purpose of $c_2$ is to assist with the decryption.
Decryption involves computing $(c_2)^{-x}\cdot c_1$ to reveal $f$, which unambiguously identifies a vote choice. 
As a result, the blinding operation for participant $P_i$ in \autoref{eqn:eq2} is equivalent to ElGamal encryption involving the computation of $c_1$ but not the decryption component $c_2$.
Furthermore, the blinding keys are ephemeral and used exactly once for encryption (i.e., blinding) of the vote within a single run of voting protocol\footnote{As a consequence, \name can be utilized in a repetitive voting~\cite{venugopalan2023always} with a limitation of a single vote within an epoch.} 
-- i.e., if the protocol is executed correctly, there are no two votes $f_l$ and $f_m$ encrypted with the same ephemeral blinding key of $P_i$, such that 
\begin{equation}
	\frac{(g^{xy}\cdot f_l)} {(g^{xy}\cdot f_m)} = \frac{f_l}{f_m},  
\end{equation}
from which the individual votes could be deduced.  
For the blockchain-specific privacy analysis, see also \autoref{sec:blockchain-specific-issues}.

\par\noindent\textbf{(2) Ballot Secrecy.} It is achieved by blinding the vote using ElGamal homomorphic encryption~\cite{cgs97}, and it is not required to possess a private key to decrypt the tally because of the self-tallying property ($g^{\sum_{i}x_i y_i}=1$).
Therefore, given a homomorphic encryption function, it is possible to record a sequence of encrypted votes without being able to read the votes choices. 
However, if all $Ps$ are involved in the recovery of a partial tally consisting of a recorded set of votes, these votes can be unblinded (as allowed by ballot secrecy).
Even a subset of $n-2$ $P$s\footnote{Note that at least 2 $P$s are required to be honest (see \autoref{sec:adversary-model}).}  who have already cast their votes cannot recover a partial tally that reveals their vote choices because of the self-tallying property ($g^{\sum_{i}x_i y_i}=1$) has not been met. 




\par\noindent\textbf{(3) Fairness}. 
If implemented naively, the last voting $P$ can privately reveal the full tally by solving \autoref{eqn:eq3} before she casts her vote since all remaining blinded votes are already recorded on the blockchain (a.k.a., the last participant conundrum).
This can be resolved by $VA$ who is required to submit the final dummy vote including the proof of her vote choice, which is later subtracted from the final tally by $SC$.\footnote{Note that If $VA$ were not to execute this step, the fault recovery would exclude $VA$'s share from MPC keys, and the protocol would continue.}



\par\noindent\textbf{(4a) Universal Verifiability}. 
Any involved party can check whether all recorded votes in the blockchain are correct and are correctly included in the final tally~\cite{Kiayias2002}.
Besides, the blinded votes are verified at $SC$, which provides correctness of its execution and public verifiability, relying on the honest majority of the consensus power in the blockchain.


\par\noindent\textbf{(4b) E2E Verifiability}. 
To satisfy E2E verifiability~\cite{EPRINT:PanRoy18}: 
(I) each $P$ can verify whether her vote was cast-as-intended and recorded-as-cast, 
(II) anyone can verify whether all votes are tallied-as-recorded.
\name meets (I) since each $P$ can locally compute her vote choice (anytime) and compare it against the one recorded in the blockchain.
\name meets (II) since $SC$ executes the code verifying that the submitted tally fits \autoref{eqn:eq3} that embeds all recorded votes.

\par\noindent\textbf{(5) Dispute-Freeness}. 
Since the blockchain acts as a tamper-resistant bulletin board (see \autoref{sec:blockchain-specific-issues}), and moreover it provides correctness of code execution (i.e., on-chain execution of verification checks for votes, tally, and fault recovery shares) and verifiability, the election remains dispute-free under the standard blockchain assumptions about the honest majority and waiting the time to finality.

\par\noindent\textbf{(6) Self-Tallying}. 
\name meet this property since in the tally phase of our protocol (and anytime after), all cast votes are recorded in $SC$; therefore any party can use them to fit \autoref{eqn:eq3}, obtaining the final tally.

%

\par\noindent\textbf{(7) Robustness (Fault Tolerance)}. 
\name is robust since it enables to remove (even reoccurring) stalling $P$s by its fault recovery mechanism (see \autoref{sec:bbbvoting-faulttolerance}).
Removing of stalling $P$s involves $SC$ verifiability of ZKP submitted by $P$s along with their counter-party shares corresponding to stalling $P$s.
%

\par\noindent\textbf{(8) Resistance to Serious Failures}. 
The resistance of \name to serious failures relies on the integrity and append-only features of the blockchain, which (under its assumptions \autoref{sec:blockchain-specific-issues}) does not allow the change of already cast votes. 

\par\noindent\textbf{(9) Receipt-Freeness}. 
\name does not meet this property since any $P$ can, using her ephemeral private key, recreate the blinded vote with the original vote choice, which can be compared to the recorded blinded vote. 
Moreover, $P$ can reveal her ephemeral private key to the coercer, who can then verify how $P$ voted.

\par\noindent\textbf{(10) Dispute-Freeness}. 
\name provides dispute-freeness, since the all protocol stages of \name are executed on the blockchain, and thus following of the protocol is self-enforcing.


\subsubsection{Blockchain-Specific Aspects and Issues}\label{sec:blockchain-specific-issues}
In the following, we focus on the most important blockchain-specific aspects and issues related to \name.


\par\noindent\textbf{(1) Bulletin Board vs. Blockchains}. 
The definition of a bulletin board~\cite{Kiayias2002} assumes its immutability and append-only feature, which can be provided by blockchains that moreover provide correct execution of code.
%
CAP theorem~\cite{brewer2000towards} enables a distributed system (such as the blockchain) to select either \textbf{\underline{c}}onsistency or  \textbf{\underline{a}}vailability during the time of network \textbf{\underline{p}}artitions.
If the system selects consistency (e.g., Algorand~\cite{gilad2017algorand}, BFT-based blockchains such as~\cite{hyperledger1}), it stalls during network partitions and does not provide liveness (i.e., the blocks are not produced) but provides safety (i.e., all nodes agree on the same blocks when some are produced). 
On the other hand, if the system selects availability (e.g., Bitcoin~\cite{Satoshi2009}, Ethereum~\cite{wood2014}), it does not provide safety but provides liveness, which translates into possibility of creating \textit{accidental forks} and eventually accepting one as valid.
Many public blockchains favor availability over consistency, and thus do not guarantee immediate immutability. 
%
%
Furthermore, blockchains might suffer from \textit{malicious forks} that are longer than accidental forks and are expensive for the attacker. 
Usually, their goal is to execute double-spending or selfish mining~\cite{eyal2018majority}, violating the assumptions of the consensus protocol employed -- more than 51\% / 66\% of honest nodes presented in PoW / BFT-based protocols.
%
To prevent accidental forks and mitigate malicious forks in liveness-favoring blockchains, it is recommended to wait for a certain number of blocks (a.k.a., block confirmations). 
%
Another option to cope with forks is to utilize  safety-favoring blockchains (e.g., \cite{gilad2017algorand,hyperledger1}). 
%

Considering \name, we argue that these forks are not critical for the proposed protocol since any transaction can be resubmitted if it is not included in the blockchain after a fork.
Waiting for the time to finality (with a potential resubmission) can be done as a background task of the client software at $P$s' devices, so $P$s do not have to wait.
Finally, we emphasize that the time to finality is negligible in contrast to timeouts of the protocol phases; therefore, there is enough time to make an automatic resubmission if needed.

\vspace{-0.1cm}
\par\noindent\textbf{(2) Privacy of Votes}. 
In \name, the privacy of vote choices can be ``violated'' only in the case of unanimous voting by all $P$s, assuming $\mathcal{A}$ who can link the identities of $P$s (approximated by their IP addresses) to their blockchain addresses by  passive monitoring of network traffic.
However, this is the acceptable property in the class of voting protocols that provide the full tally of votes at the output, such as \name and other protocols (e.g.,~\cite{McCorrySH17,Kiayias2002,HaoRZ10,killerprovotum,yu2018platform}).
%
Moreover, $\mathcal{A}$ can do deductions about the probability of selecting a particular vote choice by $P$s.
For example, in the case that the majority $m$ of all participants $n$ voted for a winning vote choice, then $\mathcal{A}$ passively monitoring the network traffic can link the blockchain addresses of $Ps$ to their identities (i.e., IP addresses), 
and thus $\mathcal{A}$ can infer that each $P$ from the group of all $P$s cast her vote to the winning choice with the probability equal to $\frac{m}{n} > 0.5$.
%
However, it does not violate the privacy of votes and such an inferring is not possible solely from the data publicly stored at the blockchain since it stores only blinded votes and blockchain addresses of $P$s, not the identities of $P$s.
To mitigate these issues, $P$s can use anonymization networks or VPN services for sending transactions to the blockchain. 
Moreover, neither $\mathcal{A}$ nor $VA$ can provide the public with the indisputable proof that links $P$'s identity to her blockchain address. 

\vspace{-0.1cm}
\par\noindent\textbf{(3) Privacy of Votes in Larger  Voting}. 
The privacy issue of unanimous and majority voting (assuming $\mathcal{A}$ with network monitoring capability) are less likely to occur in the larger voting than boardroom voting since the voting group of $P$s is larger and potentially more divergent.
We showed that \name can be extended to such a large voting by integrity-preserving batching in \autoref{ssec:scalability-improvements}. 
We experimented with batching up to 1000 $P$s, which is a magnitude greater voting than the boardroom voting.
We depict the gas expenses paid by $VA$ (per $P$) in \autoref{fig:mpc-batches}, where we distinguish various batch sizes.
In sum, the bigger the batch size, the lower the price per~$P$.


\renewcommand{\name}{SBvote\xspace}

\section{\name}\label{sec:voting-sbvote}
We introduce SBvote \cite{stanvcikova2023sbvote}, a blockchain-based self-tallying e-voting protocol that enables scalability in the number of voters and is based on BBB-Voting protocol. 
SBvote introduces multiple voting smart contracts booths that are managed and aggregated by the main smart contract. 
Our extended solution maintains the most properties of decentralized e-voting, including public verifiability, perfect ballot secrecy, and fault tolerance (but excluding receipt-freeness). 
Moreover, it improves the privacy of voters within booths.

\subsection{System Model}
\label{ssec:sys-model}
We focus on a decentralized e-voting that provides desired properties of e-voting schemes mentioned in \autoref{sec:evoting-preliminaries} as well as scalability in the number of the participants.
We assume a centralized authority that is responsible for the enrollment of the participants and shifting the stages of the protocol. 
However, the authority can neither change nor censor the votes of the participants, and it cannot compromise the privacy of the votes.
We assume that a public bulletin board required for e-voting is instantiated by a blockchain platform that moreover supports the execution of smart contracts. 
We assume that all participants of voting have their thin clients that can verify the inclusion of their transactions in the blockchain as well as the correct execution of the smart contract code.

\subsubsection{\textbf{Adversary Model}}
We consider an adversary that passively listens to a communication on the blockchain network.
The adversary cannot modify or replace any honest transactions since she does not hold the private keys of the participants.
Next, we assume that the adversary cannot block an honest transaction due to the censorship-resistance property of the blockchain.
The adversary can link a voter's IP address to her blockchain address. 
However, she does not possess the computational resources to break the cryptographic primitives used in the blockchain platform and the voting protocol.
The adversary cannot access or compromise the voter's device or the user interface of the voting application.
We assume that in each voting group of $n$ participants, at most $t$ of them can be controlled by the adversary and disobey the voting protocol, where $t \leq n - 2$ and $n \geq 3$.
This eliminates the possibility of \emph{full collusion} against a single voter~\cite{HaoRZ10}. 

\subsection{Proposed Approach}
\paragraph{\textbf{Involved Parties.}}
Our proposed approach has the following actors and components: (1) \textit{a participant $\mathbb{P}$} (\textit{a voter}) who chooses a candidate (i.e., a voting choice) and casts a vote,
(2) \textit{a voting authority $\mathbb{VA}$} responsible for the registration of participants and initiating actions performed by smart contracts,
(3) \textit{a booth contract} $\mathbb{BC}$, which is replicated into multiple instances, where each instance serves a limited number of participants. 
New instances might be added on-demand to provide scalability.
(4) \textit{The main contract $\mathbb{MC}$}, which assigns participants to voting booths, deploys booth contracts, and aggregates the final tally from booth contracts.
%

\subsubsection{\textbf{Protocol}}
We depict our protocol in \autoref{fig:operation-scheme-bw}. \name follows similar phases as BBB-Voting but with several alterations that enable better scalability.
The registration phase requires  $\mathbb{VA}$ to authenticate users and generate a list of eligible voters. 
In BBB-Voting, the setup phase of the protocol allows users to submit their ephemeral public keys. 
However, in contrast to BBB-Voting, \name requires additional steps to set up the booth contracts. 
First, eligible voters are assigned to voting groups and then $\mathbb{BC}$ is deployed for each voting group.
Once the setup is finished, voters proceed to submit their ephemeral public keys during a sign-up phase.
These keys are further used to compute multi-party computation (MPC) keys within each voting group during a pre-voting phase.
In the voting phase, voters cast their blinded votes along with corresponding NIZK proofs.
The NIZK proof allows $\mathbb{BC}$ to verify that a blinded vote correctly encrypts one of the valid candidates.
If some of the voters who submitted their ephemeral public keys have failed to cast their vote, the remaining active voters repair their votes in the subsequent fault recovery phase.
This is achieved by removing the key material of stalling voters  from the encryption of the correctly cast votes.
The key material has to be provided by each active voter along with NIZK proof of correctness.
After the repair of votes, the tallies for individual voting groups are computed during the tally phase of a booth.
Then, partial tally results are aggregated to obtain the final tally by $\mathbb{MC}$.

In the following, we describe the phases of our protocol in more detail.
Phases 2--6 are executed independently (thus in parallel) within each of the voting groups/booth contracts. \\

\paragraph{\textbf{Registration.}}
\label{ssec:registration}
In this phase, the participants interact with $\mathbb{VA}$ to register as eligible voters. 
A suitable identity management (IDM) system is required, allowing $\mathbb{VA}$ to verify participants' identities and eligibility to vote.\footnote{The details of IDM are out-of-scope for this work.}
Each participant creates her blockchain wallet address and registers it with  $\mathbb{VA}$ that stores a mapping between a participant's identity and her wallet address. 

\begin{figure}[!t]
	\centering
	\includegraphics[width=0.65\columnwidth]{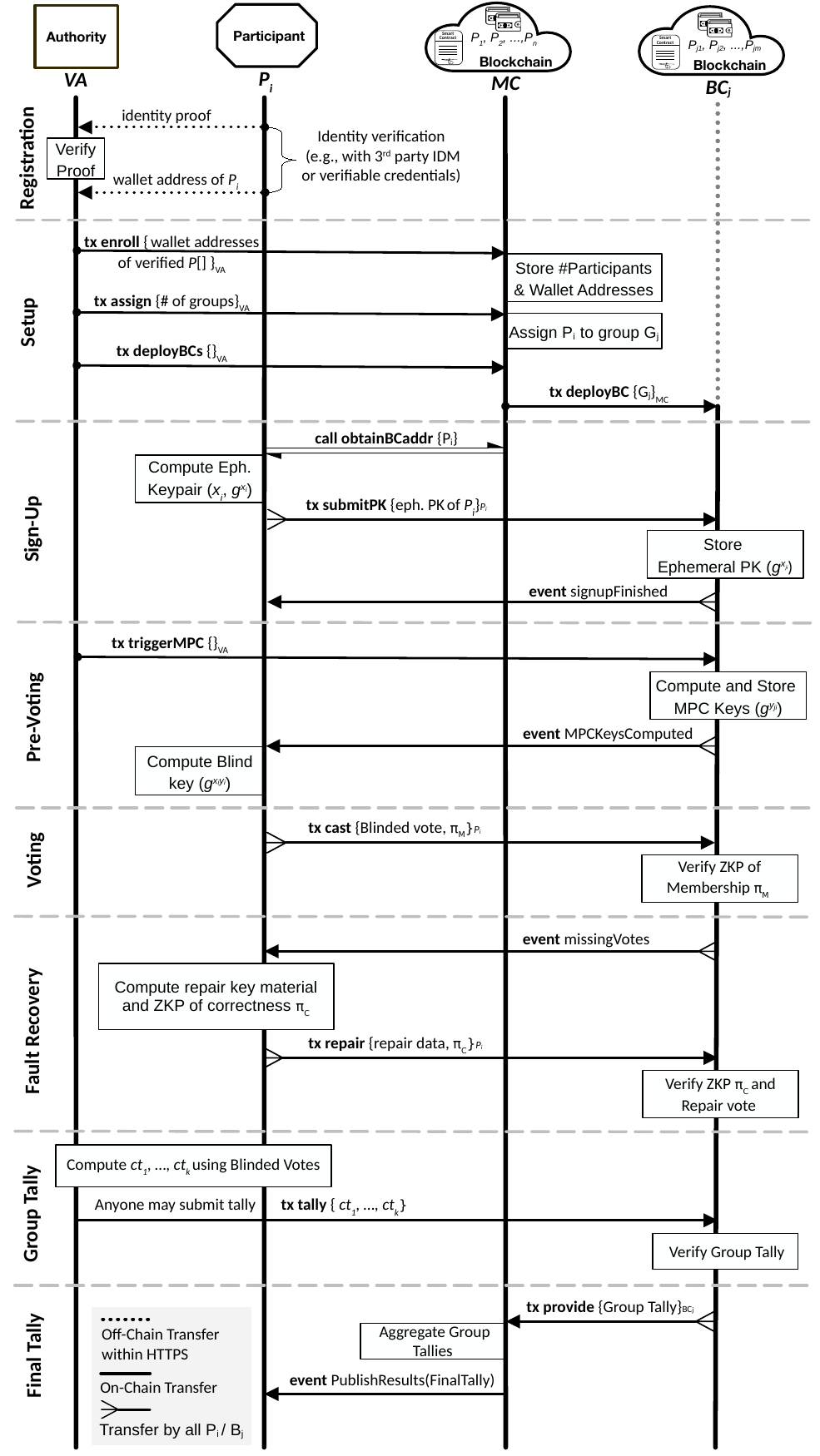}
	\caption{Overview of \name protocol.
	}
	\label{fig:operation-scheme-bw}
\end{figure}

\paragraph{\textbf{Phase 1 (Setup).}}
\label{ssec:setup}
First, $\mathbb{VA}$ deploys $\mathbb{MC}$ to the blockchain.
Then, $\mathbb{VA}$ enrolls the wallet addresses of all registered participants to $\mathbb{MC}$ within a transaction.\footnote{Note that in practice this step utilizes transaction batching to cope with the limits of the blockchain platform (see \autoref{ssec:optim}).}
Once all the registered participants have been enrolled, $\mathbb{VA}$ triggers $\mathbb{MC}$ to pseudo-randomly distribute enrolled participants into groups whose size is pre-determined and ensures a certain degree of privacy. Note that distributed randomness protocols such as RoundHound~\cite{syta2017scalable} might be used for this purpose, however, in this work we assume a trusted randomness source that is agreed upon by all voters (e.g., a hash of a Bitcoin block).

In every group, the participants agree on the parameters of the voting.
Let $n$ be the number of participants in the group and $k$ the number of candidates.
We specify the parameters of voting as follows:

\begin{compactenum}[$\;$]
	\item[1)] a common generator $g\in \mathbb{F}_p^*$, where $p=2\cdot  q +1$, $q$ is a prime and $n < p - 1$.
	
	\item[2)] $k$ independent generators $\{f_1,...,f_k\}$  in $\mathbb{F}_p^*$ such that $f_i = g^{2^{(i - 1)m}}$, where $m$ is the smallest integer such that $2^m > n$.
\end{compactenum}

Then, $\mathbb{VA}$ deploys a booth contract $\mathbb{BC}$ for each group of participants with these previously agreed upon voting parameters.
$\mathbb{MC}$ stores a mapping between a participant's wallet address and the group she was assigned to.

\paragraph{\textbf{Phase 2 (Sign-Up).}}
\label{ssec:signup}
Eligible voters enrolled in the setup phase review the candidates and the voting parameters.
Each voter who intends to participate obtains the address of $\mathbb{BC}$ she was assigned to by $\mathbb{MC}$. 
From this point onward, each participant interacts only with her $\mathbb{BC}$ representing the group she is part of. 
Every participant $P_i$ creates her ephemeral key pair consisting of a private key $x_i \in_{R} \mathbb{F}^{\ast}_p$ and public key $g^{x_i}$.
The $P_i$ then sends her public key to $\mathbb{BC}$. 
By submitting an ephemeral public key, the participant commits to cast a vote later.  
Furthermore, participants are required to send a deposit within this transaction.
If the voter does not cast her vote or later does not participate in the potential fault recovery phase, she will be penalized by losing the deposit.
Voters who participate correctly retrieve their deposit at the end of the voting.

\paragraph{\textbf{Phase 3 (Pre-Voting).}}
\label{ssec:prevoting}
In this step, each $\mathbb{BC}$ computes synchronized multi-party computation (MPC) keys from the participants' ephemeral public keys submitted in the previous step. 
To achieve scalability, the MPC keys are computed independently in each $\mathbb{BC}$ over the set of ephemeral public keys within the group.
The MPC key for participant $P_i$ is computed as follows:
\begin{equation}
	\label{eqn:eqn1}
	g^{y_i}=\prod\limits_{j=1}^{i-1} g^{x_j}/\prod\limits_{j=i+1}^{n} g^{x_j},
\end{equation}
where $y_i=\sum_{j<i}x_j- \sum_{j>i}x_j$
and
$\sum_{i}x_i y_i=0$ (see Hao et al.~\cite{HaoRZ10} for the proof).
The computation of MPC keys is triggered by $\mathbb{VA}$ in each $\mathbb{BC}$.
After the computation, each participant obtains her MPC key from $\mathbb{BC}$ and proceeds to compute her ephemeral blinding key as $g^{x_i y_i}$ using her private key $x_i$.

\paragraph{\textbf{Phase 4 (Voting).}}
\label{ssec:voting}
Before participating in this phase of the protocol, each voter must create her blinded vote and a NIZK proof of its correctness.
The blinded vote of the participant $P_i$ is $B_i = g^{x_i y_i} f_j$, where $f_j \in {f_1, ..., f_k}$ represents her choice of a candidate.
The participant casts the blinded vote by sending it to $\mathbb{BC}$ in a transaction $cast(B_{i}, \pi_M)$, where $\pi_M$ is a 1-out-of-$k$ NIZK proof of set membership. 
This proof allows  $\mathbb{BC}$ to verify that the vote contains one of the candidate generators from ${f_1, ..., f_k}$ without revealing the voter's choice.
$\mathbb{BC}$ performs a check of the proof's correctness and accepts well-formed votes.
Construction and verification of the NIZK proof are depicted in \autoref{fig:ZKPMultiCandidate} (from BBB-Voting).
	
	\paragraph{\textbf{Phase~5~(Fault-Recovery).}}
	\label{ssec:faultrec}
	The use of synchronized MPC keys ensures that a vote cast by each voter contains the key material shared with all voters within the group.
	If some of the voters within a group stall during the voting phase, the tally cannot be computed from the remaining data.
	Therefore, we include a fault-recovery phase, where remaining voters provide $\mathbb{BC}$ with the key material they share with each stalling voter, enabling $\mathbb{BC}$ to repair their votes.
	In detail, for a stalling voter $P_j$ and an active voter $P_i$ ($i \neq j$), the shared key material $g^{x_i x_j}$ consists of the stalling voter's ephemeral public key $g^{x_j}$ (previously published in $\mathbb{BC}$) and the active voter's ephemeral private key $x_i$. 
	The active voters send the shared key material to $\mathbb{BC}$ along with a NIZK proof depicted in \autoref{fig:ZKPdiffie} (from BBB-Voting).
	The NIZK proof allows  $\mathbb{BC}$ to verify that the shared key material provided by the voter corresponds to the ephemeral public keys $g^{x_i}$ and $g^{x_j}$.
	
	Suppose some of the previously active voters become inactive during the fault-recovery phase (i.e., do not provide the shared key material needed to repair their votes).
	In that case, the fault-recovery phase can be repeated to exclude these voters.
	Note that this phase takes place in groups where all the voters who committed to vote during the sign-up phase have cast their votes.
	
	\paragraph{\textbf{Phase~6~(Booth Tallies).}}
	\label{ssec:boothtally}
	At first, the tally has to be computed for each group separately.
	Computation of the result is not performed by  $\mathbb{BC}$ itself.
	Instead, $\mathbb{VA}$ (or any participant) obtains the blinded votes from  $\mathbb{BC}$, computes the tally, and then sends the result back to $\mathbb{BC}$, which verifies whether a provided tally~fits
	\begin{equation}
		\label{eqn:sbvote:eq3}
		\prod\limits_{i=1}^{n}B_i=
		\prod\limits_{i=1}^{n}g^{x_iy_i}f_j=
		g^{\sum_{i}x_i y_i}f_j=
		{f_1}^{ct_1}{f_2}^{ct_2}...{f_k}^{ct_k},
	\end{equation}
	where $ct_j \in  {ct_1, ..., ct_k}$ denotes the vote count for each candidate.

	%
	\paragraph{\textbf{Phase~7~(Final Tally).}}
	\label{ssec:finaltally}
	Once $\mathbb{BC}$ obtains a correctly computed tally, it sends it to $\mathbb{MC}$. 
	$\mathbb{MC}$ collects and summarizes the partial tallies from individual booths and announces the final tally once all booths have provided their results. 
	The participants can also review the partial results from already processed booths without waiting for the final tally since the booth tallies are processed independently.
	
	\subsection{Design Choices and Optimizations}
	\label{ssec:optim}
	We introduce several specific features of \name, which allow us to achieve the scalability and privacy properties.
	
	\paragraph{\textbf{Storage of Voters' Addresses.}}
	If we were to store the voters' wallet addresses in the booth contracts, it would cause high storage overhead and thus high costs.
	However, we proposed to store these addresses only in $\mathbb{MC}$, while booth contracts can only query $\mathbb{MC}$ whenever they require these addresses (i.e., when they verify whether a voter belongs to the booth's group).
	As a result, this eliminates the costs of transactions when deploying booth contracts, and moreover saving the blockchain storage space.
	
	\paragraph{\textbf{Elimination of Bottlenecks.}}
	The main focus of our proposed approach is to eliminate the bottlenecks that limit the number of voters and thus the size of the voting groups. 
	In particular, passing the necessary data within a single transaction could potentially exceed the block gas limit.
	
	The scalability of the Setup phase is straightforward to resolve since it does not involve any transient integrity violation checks (excluding duplicity checks).
	In all these cases, $\mathbb{VA}$ splits the data into multiple independent transactions.
	Similarly, each active voter can send the key material required to repair her vote in several batches in the Fault-Recovery phase, allowing the system to recover from an arbitrary number of stalling participants.
	
	In contrast to the Setup and Fault-Recovery phases, batching in the Pre-Voting phase is not trivial since it requires transient preservation of integrity between consecutive batches of the particular voting group.
	Therefore, we designed a custom batching mechanism, which eliminates this bottleneck while also optimizing the cost of the MPC computation.
	
	\begin{algorithm}[t]
		\scriptsize
		\caption{Pre-computation of right side values from \autoref{eqn:eqn1}.}
		\label{alg:precomp}
		\flushleft \textbf{Inputs:}	
		\begin{compactitem}
			\item $n$: $\#$ of voters 
			\item $mpc\_batch:$ batch size for MPC computation
			\item $voterPKs:$ array of voters' ephemeral public keys 
		\end{compactitem}
		\textbf{Outputs:}
		\begin{compactitem}
			\item $right\_markers:$ pre-computed right side values
		\end{compactitem}
		\hrulefill

			$right\_tmp \gets 0$ \\
			\If{$n \;\mathrm{mod}\; mpc\_batch \neq 0$}{
				$right\_markers.$push($right\_tmp$) \\
			}
			\For{$i \gets 0$ to $n$}{
				\If{$n \;\mathrm{mod}\; mpc\_batch$ = $(i-1) \;\mathrm{mod}\; mpc\_batch$}{
					$right\_markers.$push($right\_tmp$) \\
				}
				$right\_tmp \gets right\_tmp * voterPKs[n - i]$ \\
			}

	\end{algorithm}
	
	\paragraph{\textbf{MPC Batching and Optimization.}}
	If computed independently for each participant, the computation of MPC keys leads to a high number of overlapping multiplications.
	Therefore, we optimize this step by dividing the computation into two parts, respecting both sides of the expression in \autoref{eqn:eqn1} and reusing accumulated values for each side.
	
	First, we pre-compute the right part (i.e., divisor) of \autoref{eqn:eqn1}, which consists of a product of ephemeral public keys of voters with a higher index than the current voter's one (i.e., $i$ in \autoref{eqn:eqn1}).
	The product is accumulated and saved in the contract's storage at regular intervals during a single iteration over all ephemeral public keys.
	The size of these intervals corresponds to the batch size chosen for the computation of the remaining (left side) of the equation.
	We refer to these saved values as \emph{right markers} (see \autoref{alg:precomp}).
	We only choose to save the right markers in the storage of $\mathbb{BC}$ instead of saving all accumulated values due to the high cost of storing data in the smart contract storage.
	Though the intermediate values between right markers have to be computed again later, they are only kept in memory (not persistent between consecutive function calls).
	Therefore, they do not significantly impact the cost of the computation.

	The second part of the computation is processed in batches.
	First, the right-side values for all voters within the current batch are obtained using the pre-computed right marker corresponding to this batch (see lines 1--4 of \autoref{alg:mpc}).
	Then,  the left part of \autoref{eqn:eqn1} is computed for each voter within the batch, followed by evaluating the entire equation to obtain the MPC key (lines 6--8 of \autoref{alg:mpc}).
	This left-side value is not discarded; therefore, computing the left side for the next voter's MPC key only requires single multiplication.
	The last dividend value in the current batch is saved in the contract's storage to allow its reuse for the next batch.
	
	\begin{algorithm}[t]
		\scriptsize
		\caption{Computation of a batch of MPC keys.}
		\label{alg:mpc}
		\flushleft \textbf{Inputs:}
		\begin{compactitem}
			\item $voterPKs:$ array of voters' ephemeral public keys 
			\item $mpc\_batch:$ batch size for MPC computation
			\item $start, end:$ start and end index of the current batch
			\item $right\_marker:$ pre-computed right side value for the first index in a batch 
			\item $act\_left:$ left side value from the previous batch 
		\end{compactitem}
		\textbf{Outputs:}
		\begin{compactitem}
			\item $act\_left:$ left side value at the last index of the current batch
			\item $mpc\_keys:$ array of MPC keys for the current batch 
		\end{compactitem}
		
		\medskip
		Compute right side values for the batch: \\
		\hrulefill

			$right\_tab[mpc\_batch - 1] \gets right\_marker$ \\
			\For{$i \gets 0$ to $mpc\_batch$}{
				$j \gets mpc\_batch - i$ \\
				$right\_tab[j-1]$ $\gets$ $right\_tab[j]$ $*$ $voterPKs[i-1]$ \\
			}

		Compute the current batch of MPC keys:
	
			\setcounter{ALG@line}{5}
			\For{$i \gets start$ to $end$}{
				$act\_left \gets act\_left * voterPKs[i-1]$ \\
			}
			$mpc\_keys[i] \gets act\_left \div right\_tab[i] \;\mathrm{mod}\; mpc\_batch$ \\	

	\end{algorithm}

\subsection{Evaluation}\label{sec:sbvote:evaluation}

To evaluate the scalability of \name, we created the proof of concept implementation that builds on BBB-Voting~\cite{homoliak2023bbb}. 
We used the Truffle framework and Solidity programming language to implement the smart contract part and Javascript for the client API of all other components.
We also utilized the Witnet library~\cite{witnet-ecc-solidity} for on-chain elliptic curve operations on the standardized \emph{Secp256k1} curve~\cite{sec2002}. 
Although Solidity was primarily intended for Ethereum and its Ethereum Virtual Machine (EVM), we have not selected Ethereum for our evaluation due to its high operational costs and low transactional throughput, which is contrary to our goal of improving scalability.
However, there are many other smart contract platforms supporting Solidity and EVM, out of which we selected Gnosis\footnote{\url{https://developers.gnosischain.com}} and Harmony\footnote{\url{https://www.harmony.one}} due to their low costs and high throughput.

Throughout our evaluation, we considered the following parameters of the chosen platforms: 30M block gas limit with 5 second block creation time on Gnosis and 80M block gas limit with 2 second block creation time on Harmony.
%

\paragraph{\textbf{MPC Batch Size.}}
The MPC keys in the Pre-Voting phase are computed in batches (see \autoref{ssec:optim}).
In detail, there is a pre-computed value available for the first voter in each batch.
Using a small batch size imposes many transactions and high execution costs due to utilizing fewer pre-computed values. 
In contrast, using a large batch size requires more expensive pre-computation and storage allocation, which results in a trade-off.
This trade-off is illustrated in \autoref{fig:sbvote:mpc_batch}, depicting how the batch size affects the cost of the computation per voter.
We can see that the best value for our setup is 150 voters per batch.

\begin{figure}[t]
	\centering
	\includegraphics[width=0.5\columnwidth]{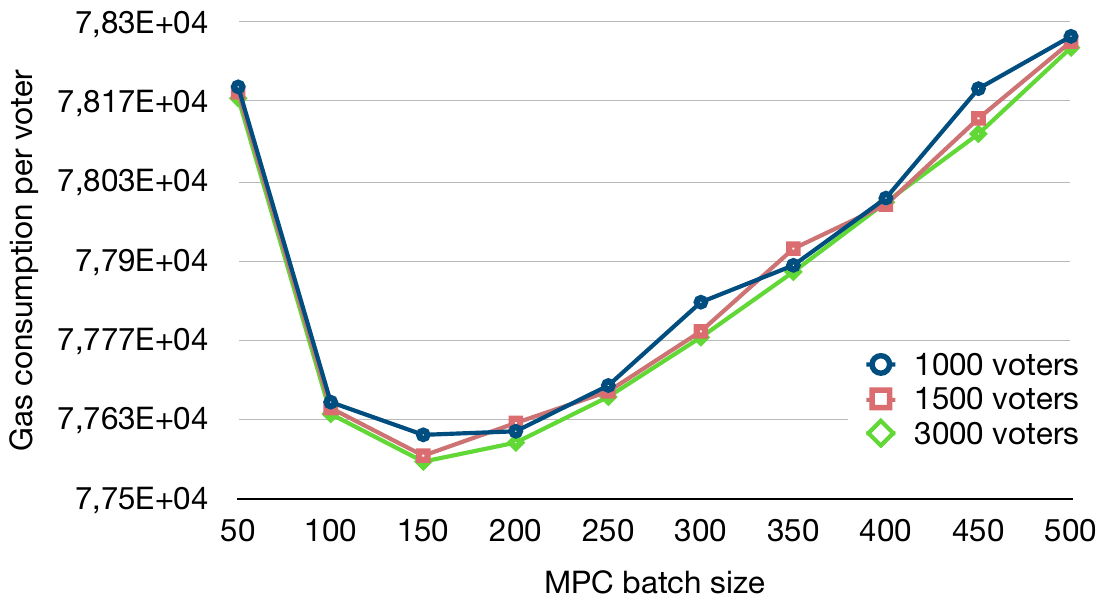}

	\caption{Per voter cost of the MPC key computation w.r.t. the batch size.}
	\label{fig:sbvote:mpc_batch}
\end{figure}
\begin{figure}[b]
	\centering
		\includegraphics[width=0.55\columnwidth]{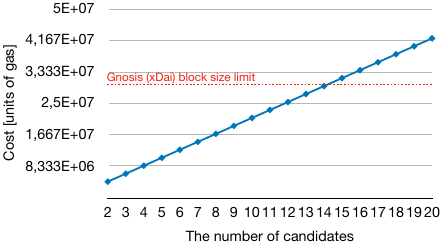}
		\caption{The cost of vote casting w.r.t. the number of candidates.} 
		\label{fig:sbvote:vote-cast}
\end{figure}

\paragraph{\textbf{The Number of Candidates.}}
The number of candidates our voting system can accommodate remains limited.
This is mainly caused by the block gas limit of a particular platform.
In detail, we can only run voting with a candidate set small enough so that the vote-casting transaction does not exceed the underlying platform's block gas limit.
Such transaction must be accompanied by a NIZK proof of set membership (i.e., proof that the voter's encrypted choice belongs to the set of candidates), and the size of the candidate set determines its execution complexity.
\autoref{fig:sbvote:vote-cast} illustrates this dependency.
Our experiments show that the proposed system can accommodate up to 38 and 14 candidates on Harmony and Gnosis, respectively.

\paragraph{\textbf{The Total Number of Participants.}}
The time period over which the voters can cast their ballots typically lasts only several days in realistic elections.
The platform's throughput over a restricted time period and the high cost of the vote-casting transactions result in a trade-off between the number of voters and the number of candidates.
We evaluated the limitations of the proposed voting protocol on both Harmony and Gnosis, as shown in \autoref{fig:cand-vote} and \autoref{fig:max_voters}.
Note that in these examples, we considered only the most expensive phase of the protocol (i.e., voting phase) to be time-restricted.

We determined that with two candidates, the proposed system can accommodate 1.5M voters over a 2-day voting period and up to 3.8M voters over a 5-day voting period on the Harmony blockchain.
On the other side of the trade-off, with the maximum number of 38 candidates on Harmony, maximally 216K voters can participate within a 5-day period.

\begin{figure}[t]
	\centering
	\includegraphics[width=0.6\columnwidth]{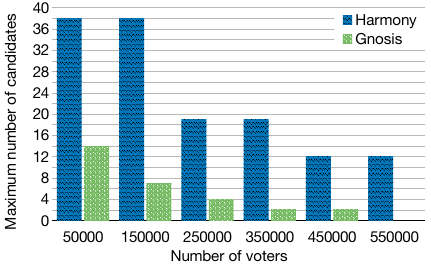}
	\caption{The maximum number of candidates our approach can process during a fixed 5-day voting interval, assuming various numbers of voting participants.}
	\label{fig:cand-vote}
\end{figure}
\begin{figure*}[t]
	\subfloat[Gnosis.\label{fig:max_vot_xdai}]{	
		\includegraphics[width=0.48\columnwidth]{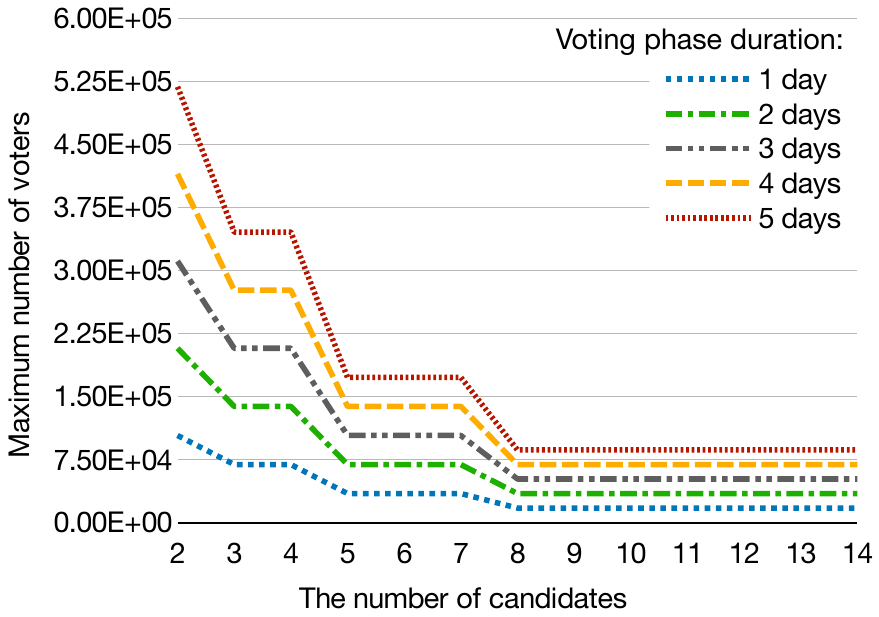}
	}\hfill
	\subfloat[Harmony.\label{fig:max_vot_harmony}]{	
		\includegraphics[width=0.48\columnwidth]{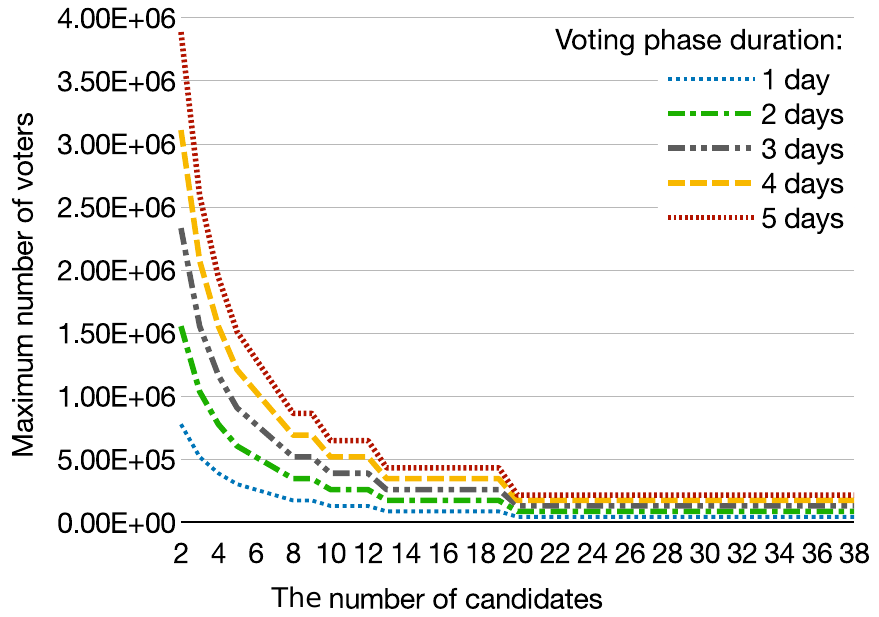}	
	}
	\vspace{0.3cm}
	\caption{The maximum number of voters that our approach can accommodate w.r.t. the number of candidates.}
	\label{fig:max_voters}
\end{figure*}

\subsection{Security Analysis and Discussion}
\label{sec:sbvote:discussion}
We discuss the properties and scenarios affecting the security and privacy of \name.

\paragraph{\textbf{Privacy.}}
Within each voting group, \name maintains perfect ballot secrecy.
The adversary, as defined in \autoref{ssec:sys-model}, cannot reveal a participant's vote through a collusion of all remaining participants since adversary can control at most $n-2$ participants.
The privacy of votes can be violated only if all participants in a voting group vote for the same candidate.
However, this is a natural property of voting protocols, which output the tally rather than only the winning candidate.
\name mitigates this problem by implementing transaction batching, 
which allows the authority to maintain a sufficiently large size of the voting groups to lower the probability of a unanimous vote within the groups.
This probability is further decreased in \name by the smart-contract-based pseudo-random assignment of participants to the groups.
We refer the reader to the work of Ullrich \cite{ullrich2017risk} that addresses the issue of unanimous voting and the probability of its occurrence.

\paragraph{\textbf{Deanonymization \& Linking Addresses.}}
In common block\-chains, the network-level adversary might be able to link the participant's address with her IP address. 
Such an adversary can also intercept the participant's blinded vote; however, she cannot extract the vote choice due to the privacy-preserving feature of our voting protocol.
Therefore, even if the adversary were to link the IP address to the participant's identity, the only information she could obtain is whether the participant has voted.
Nevertheless, to prevent the linking of addresses, participants can use VPNs or anonymization services such as Tor.

\paragraph{\textbf{Re-Voting.}}
It is important to ensure that no re-voting is possible, which is to avoid any inference about the final vote of a participant in the case she would reuse her ephemeral blinding key to change her vote during the voting stage.
Such a re-voting logic can be easily enforced by the smart contract, while the user interface of the participant should also not allow re-voting.
Also, note that ephemeral keys are one-time keys and thus are intended to use only within one instance of e-voting protocol to ensure the security and privacy of the protocol.
If a participant were to vote in a different instance of e-voting, she would generate new ephemeral keys.

\paragraph{\textbf{Forks in Blockchain.}}
Blockchains do not guarantee immediate immutability due to possible forks. This differentiates blockchains from public bulletin boards, as defined in~\cite{Kiayias2002}.
However, since our protocol does not contain any two-phase commitment scheme with revealed secrets, its security is not impacted by accidental or malicious forks. 
Temporary forks also do not impact the voting stage since the same votes can be resubmitted by client interfaces.

\paragraph{\textbf{Self-Tallying Property.}}
The self-tallying property holds within each voting group since the correctness of obtained tallies can be verified by anybody.
Consequently, this property holds for the whole voting protocol since the main contract aggregates the booth tallies of the groups in a verifiable fashion.


\paragraph{\textbf{Verifiability.}}
\name achieves both individual and universal verifiability.
By querying the booth contract, each voter can verify her vote has been recorded. 
Each voter (and any interested party) can verify the booth tally since it satisfies the self-tallying property, i.e., the \autoref{eqn:sbvote:eq3} would not hold should any vote be left out.
Any party can verify the final tally aggregated in the main contact by querying all the booth contracts to obtain individual booth tallies.

\paragraph{\textbf{Platform-Dependent Limitations.}}
Although our system itself does not limit the number of participants, the required transactions are computationally intensive, which results in high gas consumption.
Therefore, large-scale voting using our system might be too demanding on the underlying smart contract platform.
As a potential solution, public permissioned blockchains dedicated to e-voting might be utilized. 

\paragraph{\textbf{Adversary Controlling Multiple Participants in the Fault Recovery.}}
One issue that needs to be addressed in the fault recovery is the adversary controlling multiple participants and letting them stall one by one in each fault recovery round.
Even though the fault recovery mechanism will eventually finish with no new stalling participants, such behavior might increase the costs paid by remaining participants who are required to submit counter-party shares in each round of the protocol.
For this reason, similar to the voting stage, we require the fault recovery stage to penalize stalling participants by losing the deposit they put into the smart contract at the beginning of our protocol.
On the other hand, the adversary can cause a delay in the voting protocol within a particular booth. 
However, it does not impact other booths. 
To further disincentivize the adversary from such a behavior, the fault-recovery might require additional deposits that could be increased in each round, while all deposits could be redeemed at the tally stage.

\paragraph{\textbf{Tally computation.}}
Tallying the results in individual booths requires an exhaustive search for a solution of \autoref{eqn:sbvote:eq3} with $\binom{n+k-1}{k-1}$ possible combinations~\cite{HaoRZ10}, where $n$ is the number of votes and $t$ is the number of candidates.
Therefore, the authority should select the size of the voting groups accordingly to the budget and available computational resources (see \autoref{table:bruteforcetime} of BBB-Voting for the evaluation).

\section{Always on Voting}\label{sec:evoting-aov}

\paragraph{Verifiable Delay Function (VDF).}
Given a time delay $t$, a VDF  must satisfy the following conditions: 
for any input $x$, anyone equipped with commercial hardware can find $y$ = VDF($x, t$) in $t$ sequential steps, but an adversary with $p$ parallel processing units must not distinguish $y$ from a random number in significantly fewer steps (see also \autoref{sec:background:vdf}).
For our purposes, the value of $t$ is fixed once it is determined. Therefore, we use VDF($x$) instead of VDF($x,t$) in the remaining text.

\subsection{System Model}
\label{ssec:systemmodel}
Our model has the following main actors and components:
$\romannumeral
1$) A \textit{participant} ($P$) who partakes in governance by casting a vote for her choice or candidate.
$\romannumeral
2$) \textit{Election Authority} (EA) is responsible for validating the eligibility of participants to vote in elections, registering them, and shifting between the phases of the voting.
$\romannumeral
3$)
A \textit{smart contract} (SC) collects the votes, acts as a verifier of zero-knowledge proofs, enforces the rules of the voting and verifies the tallies of votes.
$\romannumeral
4$)
\textit{Bitcoin Puzzle Oracle} (BPO) provides an off-chain data feed from the Bitcoin network and supplies the requested Bitcoin block header (BH) when it is available on the Bitcoin network.
$\romannumeral
5$)
A \textit{VDF prover} is any benign party in the voting ecosystem who computes the output of VDF and supplies proof of its correctness to SC.

\subsubsection{Adversary Model}
\label{ssec:advmodel}

There are two types of adversaries: the network adversary
$Adv_{net}$ and a Bitcoin mining adversary $Adv_{min}$.
Both adversaries are static and have bounded computing power, i.e., they are unable to break used cryptographic primitives under their standard security assumptions.
$Adv_{net}$ is a passive listener of communication entering the blockchain network but cannot block it.
Her objective is to derive statistical inferences determining the voting patterns of participants (including the voting intervals in which they voted).
%
$Adv_{min}$ can mine on the Bitcoin blockchain.
	Her goal is to find a solution to the Bitcoin PoW puzzle that also triggers the end of the current voting interval, thereby influencing the end time of epoch. 
	Our voting framework uses a function of the Bitcoin block header (BH) inclusive of its  PoW solution $s$, i.e., $f(BH)$ to trigger the end of the current voting interval. 
	Such a  manipulation would potentially enable $Adv_{min}$ to prematurely finish the current interval and start the next one.
%
Finally, we assume that $EA$ verifies identities honestly and supply addresses of only verified participants to $SC$.

\subsection{Design Goals}
\label{ssec:designgoals}

The AoV framework has the following main design goals. 
\begin{compactenum}
	\item \textbf{Repeated voting epochs}:  Participants are allowed to continuously vote and change elected candidates or policies without waiting for the next election. Participants are permitted to privately change their vote at any point in time, while the effect of their change is considered rightful at the end of each epoch.
	The duration of such epochs is shorter than the time between the two main elections.
	\item \textbf{Randomized time epochs}: The end of each epoch is randomized and made unpredictable.
	In contrast to fixed-length time epochs, the proposed randomized time epochs are used to thwart the peak-end-effect.
	\item \textbf{Plug \& play voting protocols}: The AoV framework
	is designed to ``plug \& play'' new or existing voting
	protocols. As a result, AoV inherits the properties of
	the underlying protocol chosen.
		However, in the interest of vote confidentiality on a blockchain,  we recommend protocols providing 
		secret ballots whose correctness can be publicly verified by $SC$ without leaking any information, e.g., \cite{yu2018platform,ICBC:MRSS19,killerprovotum,Baudron2001}. 
		Also, due to the repetitive nature of AoV, e-voting protocols with expensive on-chain computations and required fault recovery (due to stalling participants)  may be less appealing but still acceptable with some limitations, e.g.,~\cite{McCorrySH17,Li2020,DBLP:conf/fc/SeifelnasrGY20,venugopalan2021bbbvoting,stanvcikova2022sbvote}. 
		\item \textbf{Privacy of participants}:
		Revoting by a participant $P$ may enable $Adv_{net}$ to link her blockchain wallet addresses  or link the IP address of $P$ across her multiple voting transactions. Therefore, it is important to achieve maximum voter privacy (anonymity) in the presence of adversary $Adv_{net}$  (see \autoref{ssec:Anonid}).
		\item \textbf{Privacy implications of booth sharding}: 
		The AoV framework supports booth sharding to distribute voting overheads, streamline operations, and respect hierarchical structure of elections. 
		In detail, instead of having a single booth smart contract for all participants, they are split into a number of smaller booth contracts. 
		For multi-booth elections with limited participants per booth and a high candidate win probability, there is a small chance that all participants in a booth voted for the same candidate. Since the $EA$ knows the identity of all the participants in the booth, it can trivially determine whom they voted for. 
		We discuss the effects of booth size and candidate winning probability to prevent such incidents in \autoref{ssec:shardinganalysis}.
\end{compactenum}

\subsection{Proposed Approach}
\label{sec:framework}

Always-on-Voting (AoV) is a framework for a blockchain-based e-voting, in which voting does not end when the votes are tallied and the winners are announced.
Instead, participants can continue voting for their previous vote choice or change their vote. 
A possible outcome of such repetitive voting is transitioning from a previous winning candidate to a new winner.
To achieve this, the whole time interval between two regularly scheduled elections is unpredictably divided into several intervals, denoted as voting epochs.
Participants may change their vote anytime before the end of a voting epoch (i.e., before a tally of the epoch is computed); however, they do not know beforehand when the end occurs.
Any vote choice that transitioned into the supermajority threshold of votes is declared as the new winner of the election,
and it remains a winning choice until another vote choice reaches a supermajority threshold.

\subsubsection{Underlying Voting Protocol}\label{sec:underlying-voting-protocol}
AoV provides the option to plug \& play any suitable e-voting protocol. 
To provide the baseline security and privacy of votes (with on-chain verifiability), we assume the voting protocol plugged into AoV allows participants to blind or encrypt their votes whose correctness is verified on-chain by $SC$. 
However, AoV does not deal with other features supported by the plugged-in voting protocol (such as end-to-end verifiability~\cite{Jonker13}, coercion-resistance~\cite{yu2018platform}, receipt-freeness~\cite{Kiayias2002}, and fairness~\cite{Kiayias2002}). 

\begin{figure*}
	\centering
	\includegraphics[width=1.0\textwidth]{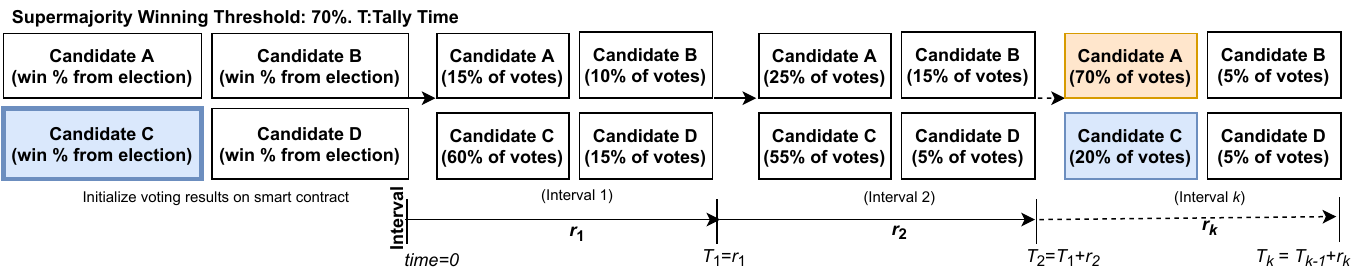}
	\caption{The time between two regular elections is divided into the fixed number $ft$ of intervals (a.k.a., epochs). First, the ratios of votes for all vote choices (i.e., candidates) are initialized from the last election. Next, repeated voting within $k$ epochs results in a winning vote choice transition (from C to A). The new winner A is declared when she obtains a supermajority of total votes (i.e., 70\%) at interval $k$; $k\leq ft$ (see \autoref{sec:framework}).
	Note that $r_1,\ldots,~r_k$ are  randomized times that determine the length of the intervals $1,...,k$. The tally is computed at the end of each interval.}
	\label{fig:transition-epoch}
\end{figure*}

\subsubsection{Example of Operation}
\autoref{fig:transition-epoch} illustrates a scenario with 4 candidates $A$-$D$, 
where \textit{C} is the present winner of the election.
For example, the supermajority threshold of 70\% votes is set for future winnings, which is a tunable parameter that may be suitably tailored to the situation.
All candidates are initialized to their winning percentages of obtained votes from the last election. 
Over time, the individual tally is observed to shift as the supermajority of participants decided to change their vote in favor of another candidate by voting in the epoch intervals.
Through \textit{k} intervals, the winner-ship is seen to transition from candidate \textit{C} to \textit{A}.
At the $\textit{k}^{th}$ interval, \textit{A} obtains the 70\% threshold of votes and is declared as the new winner.
Note that the supermajority is required only in the voting epochs between two regularly scheduled elections.
The regular elections are also executed in AoV, and they repeat every $M$ months/years, while requiring only a majority of votes (i.e., $>$50\%) to declare a winner. 
Hence, in contrast to existing electoral systems, we only propose changes between regularly scheduled elections and enable new candidates to be added or removed.


\myparagraph{Justification for Supermajority}
A supermajority of 70\% was chosen (see Appendix of \cite{venugopalan2023always} for background) to help the incumbent carry out reforms without the risk of losing when there is still sufficient support from participants. 
On the other hand, the main purpose of this threshold is to block (or repeal) policies that are unpopular or negatively affecting a vast majority of participants.
Additionally, we aim to avoid the quorum paradox (see Appendix of \cite{venugopalan2023always}) 
by setting a minimum participation requirement of 70\% from the just concluded main election.

\begin{figure}
	\centering
	\includegraphics[width=0.50\columnwidth]{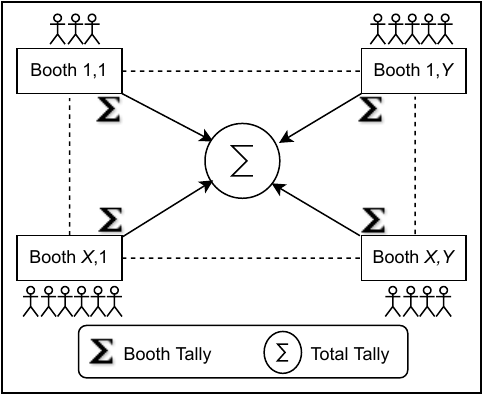}
	\caption{When the tally computation is triggered, each booth computes the sum of all votes cast at the booth (referred to as booth tally). Each booth tally is further summed up to determine the total tally.
		Pictorially, the booths are numbered 1 to $X$ along the rows and 1 to $Y$ along the columns. There are a total of $X\cdot Y$ booths.
	}
	\label{fig:booth}
\end{figure}

\subsubsection{Overview of AoV Phases}
\label{ssec:votingphases}
Once the setup phase (that ensures participants agree upon all system parameters) is completed, electronic voting frameworks typically consist of three phases: 
(1) a registration phase to verify voter credentials and add them to the voting system,  
(2) a voting phase, in which participants cast their vote via a secret ballot, and 
(3) a tally phase, where the total votes for each candidate are counted and revealed to participants.
The voting protocol plugged-in with the AoV framework may contain additional phases, but we omit them here for brevity.

\begin{figure}[t]
	\centering
	\includegraphics[width=0.6\columnwidth]{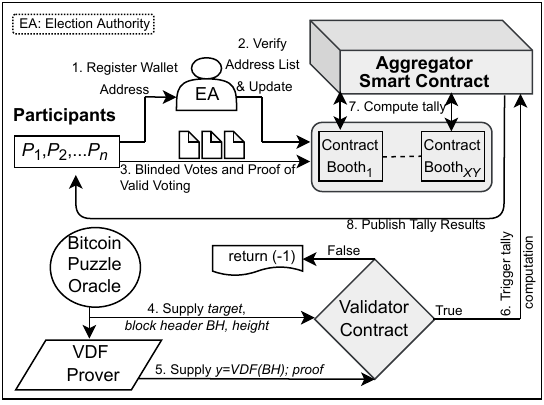}
	\caption{Interaction among participants ($P$s), election authority ($EA$), smart contracts, the Oracle, and VDF prover. 
		(1) Registering wallet addresses of participants and (2) their identity verification are made by the $EA$.
		(3) Participants send a blinded vote and its zero-knowledge proof of correctness to their assigned booth contract. 	
		The booth contract verifies the validity of the vote.
		(4) The Bitcoin Puzzle Oracle (BPO) provides the latest Bitcoin block header (BH) and 
		(5) VDF prover sends a proof of sequential work with $y$ (the output of VDF(BH)) to the validator contract. 
		(6) The validator contract finishes the epoch and shifts the state of the elections to the tally upon meeting the required conditions.
		(7) The aggregator contract is responsible for totaling individual booth tallies and (8) publicly announcing the total tally. The on-chain components of AoV are depicted in gray.}
	\label{fig:model}
\end{figure}
The architecture of AoV is shown in \autoref{fig:model}.
In AoV, participants (in step 1) register their wallet address with the EA, who then (in step 2) verifies and updates it on the booth smart contract\footnote{Participants are randomly grouped and assigned to booths  $\in\{1,2,...,X\cdot Y\}$ (see \autoref{fig:booth}), represented by a booth smart contract.}.
This is followed by the voting phase (in step 3), where participants publicly cast their secret ballots (i.e, not revealing the vote choice nor identity).
The BPO (step 4)  supplies the validator contract and VDF prover with the $target$, recent
Bitcoin block header $BH$ and its block height. 
The VDF prover\footnote{A VDF prover is any benign user in the voting ecosystem with commercial hardware to evaluate the input of VDF, i.e., $y~=~$VDF($BH$) and supply a proof $\pi$. 
} (in step 5) computes and submits   $VDF(BH)$ and a  proof of sequential work ($\pi$) to the validator contract. 
The validator contract (in step 6) verifies the VDF proof and checks whether the supplied nonce $s$ (included in the block header) is a valid solution to the Bitcoin PoW puzzle of the supplied header. 
If both verifications pass, 
the validator contract finalizes the epoch and triggers the tally computation for the epoch. 
Otherwise, it waits for the next block header submission from the BPO and the proof of sequential work from the VDF prover.
When the tally computation is triggered, each booth contract $\{1,2,...,X\cdot Y\}$, sums up all its local vote counts and sends them to the aggregator contract (step 7). 
Then, the aggregator contract totals the votes from each booth contract and publishes the final tally (step 8).
	In AoV, the $EA$ is authorized to register/remove participants and candidates in a future interval. 
	Nevertheless, candidates can also be managed by other means, and AoV does not mandate how it should be done.
	When there are no other changes in the next interval, revoting repeats with step~2 and ends with step~8.
%
From the initialization of AoV until the next regular elections, the validator smart contract accepts all future Bitcoin block headers.
The new block headers (as part of their blocks) arriving every 10 minutes on average are appended to the Bitcoin blockchain.
The BPO is responsible for timely supplying\footnote{
To respect the finality of the Bitcoin network, we assume that BPO supplies only the block headers that contain at least 6 confirmations on top of them. As a consequence, the	 probability that such a confirmed block will be reverted is negligible. Note that this does not influence the chances of $ADV_{min}$ to succeed since she is already  ``delayed'' by VDF in finding multiple PoW solutions at the same height; therefore, she prefers to work on top of the chain with her new attempts. }
each new block header to the VDF prover and validator contract.
The VDF prover computes the VDF on each of those block headers after they are supplied.

\subsubsection{Calculating the Epoch Tally Time}
\label{ssec:vrf}
Due to concerns that Bitcoin nonces are a weak entropy source, additional steps are taken to make it cryptographically secure (see details in  \autoref{ssec:rndnonce}).
Our notion of randomness relies on Bitcoin Proof-of-Work to generate valid nonces.\footnote{If the nonce overflows, a parameter called \textit{extraNonce}  (part of the coinbase transaction) is used to provide miners with the extra entropy needed to solve the PoW puzzle.}
The validator contract awaits future block headers yet to be mined on the Bitcoin network.
When new Bitcoin block headers arrive, they are sent to the validator contract and the VDF prover via the BPO.
The VDF ensures that a mining adversary cannot find more than one valid nonce to the block at a given height and test if the nonce is favorable within 10 minutes.
The VDF is computed with the block header at the input by the VDF prover, who then submits the VDF output and the proof of sequential work to the validator contract.
The choice of VDF depends on its security properties, speed of verification, and a size of the proof~\cite{Pietrzak2019}.  
Let  $BH$ be the Bitcoin block header.
Once VDF prover computes $y=VDF(BH)$, a small proof ($\pi$) is used to trivially verify its correctness using $VDF\_Verify(y,\pi)$. Wesolowski’s construction~\cite{Wesolowski2020} is known for its fast verification and a short proof:
Let $TL$ be the number of sequential computations. 
Prover claims 
$$y~=~BH^{2^{TL}}$$ 
and computes a proof $$\pi= BH^{\lfloor \frac{2^{TL}}{B} \rfloor} ,$$ where $B$~=~Blake256($BH~||~y~||~TL$) hash.  
Verifier checks whether $$\pi^{B}\cdot BH^{2^{TL}mod\; B}\stackrel{?}{=} y.$$
Since we employ VDF, $Adv_{min}$ does not know the value of $y$ before evaluating the VDF and is forced to wait for a given amount of time to see if the output is in her favor (before trying again).
However, since Bitcoin mining is a lottery, other miners can solve the puzzle and append a block by propagating the solution to the Bitcoin network, rendering any withheld or attempted solution by the adversary that was not published useless.

\paragraph{Interactions of BPO, VDF Prover, and Validator.}
\label{ssec:interactions}
Let $TotalTime$ be the time in minutes between 2 regular elections.
The BPO ({see step 4 in \autoref{fig:model}}) feeds the block header $BH$ of every future Bitcoin block (when it is available) to the validator contract and VDF prover.
Further, BPO provides validator contrast also with the value of $target$ when it changes; i.e., every 2016 blocks.

Upon obtaining data from BPO, the VDF prover computes VDF output  \begin{equation}\label{eqn_vdfprover}
y~=~VDF(BH)
\end{equation}
with the VDF proof $\pi$ and sends them to the validator contract ({see step 5 in \autoref{fig:model}}).
Next, the validator contract verifies the following conditions:
\begin{equation}\label{eqn_vdfverifier}
VDF\_Verify (y,\pi) \stackrel{?}{=} True,
\end{equation}
\begin{equation}\label{eqn_0}
SHA256(BH)<target.
\end{equation}
The first verification checks whether the VDF output $y$ and supplied proof $\pi$ (i.e., \autoref{eqn_vdfverifier}) correspond to the BPO-supplied block header $BH$. 
%
The second verification (i.e., \autoref{eqn_0}) checks whether the nonce received from BPO\footnote{We note that the BPO may be replaced by a quorum to improve decentralization. The validator contract will then accept the input from BPO only when $2/3$ (and more) of the quorum is in agreement.}
is a valid solution to the Bitcoin PoW puzzle.
Once both checks pass, the validator contract proceeds to compute 
\begin{equation}\label{eqn_1}
a=SHA(y),
\end{equation}
where SHA(.) is SHA-X-256\footnote{X denotes a suitable hash function such as SHA-3, and 256 is the output length in bits.}  hash.
The goal of \autoref{eqn_1} is to consolidate the entropy by passing it through a compression function that acts as a randomness extractor (see \autoref{ssec:rndnonce}).
Using $a$, the validator contract computes

\begin{equation}\label{eqn_2}
b= a\; (mod\; BHsInInterval),
\end{equation}

where the expected number of block headers are
\begin{equation}\label{BHInInterval}
BHsInInterval ~=~ \frac{IntervalTime}{BlockTime}.
\end{equation}


and the time interval is found as
\begin{equation}\label{eqn_3}
IntervalTime ~=~ \frac{TotalTime}{ft}.
\end{equation}

\noindent
As seen in \autoref{fig:transition-epoch}, $ft$ is the number of intervals (epochs) that the total time ($Total\-Time$) between 2 regular elections is divided into. $BlockTime$ is the average time of block generation (i.e., 10 minutes in Bitcoin).
The computation of tally for the current interval is triggered when the output of the validator contract is $True$ ({see step 6 in \autoref{fig:model}}):
\begin{equation}\label{eqn_4}
VC_{output} = \begin{cases} \mbox{True,} & \mbox{if } b = 0 \\ \mbox{False,} & \mbox{otherwise}. \end{cases} 
\end{equation}

\noindent
\paragraph{\textbf{Example.}}
Let $TotalTime$ = 4 years = $525600\cdot 4$ minutes and $ft=8$; then $IntervalTime$ $= (525600\cdot 4)/(8)= 262800$ minutes $\approx 182.5$ days and $BHsInInterval$ = $262800 / 10$ = 26280 blocks. 
Therefore, the BPO will send on average 26280 block headers ($BH$ values) to the validator contract within assumed 182.5 days long epoch (assuming 10 minutes block creation interval), i.e., $1/8$ of the total time.
We expect the tally will be triggered on average once in every 182.5 days because of the Poisson probability distribution of this event. 
Therefore, $ft$ expresses the expected number of epochs, while $ft$ might differ across the regular elections iterations.

\subsubsection{Anonymizing Identity}
\label{ssec:Anonid}
In our scheme, we employ wallet addresses to keep track of authorized participants.
{The map between participant $P$ and her wallet, recorded by the $EA$ is to prevent sybil attack (preventing any unauthorized person from voting) and double voting. For this reason, $EA$ is trusted to keep this mapping private. Only voters corresponding to white-listed wallets by $EA$ are allowed to vote and everyone else is blocked from voting on the smart contract. The wallet address corresponds to a unique random string associated with the voter. Its knowledge provides no additional information to the $EA$, since the $EA$ knows voters’ identities.} 

	The AoV framework permits participants to change their vote at any time. 
	The effect of the change is manifested at the end of the epoch when the tally is computed.  
	However, $Adv_{net}$ can observe the vote transactions on $SC$ even though the vote choice remains confidential (preserving the privacy of votes). 
	A participant might vote in one interval and then vote again in a future interval. 
	$Adv_{net}$ cannot distinguish whether the participant voted again for the same candidate or changed her vote to a different
	candidate due to assumed confidentiality-preserving properties of plugged voting protocols (see \autoref{sec:underlying-voting-protocol}) 
	However, both votes may be mapped to the same participant's blockchain wallet address if utilized naively, hence $Adv_{net}$ can determine how many times a participant voted.

%
%

To break the map between participant $P$ and her blockchain wallet address, we use the idea presented in type 2 deterministic wallets\footnote{See \url{https://en.bitcoin.it/wiki/BIP_0032}.}.
The objective is to synchronize a practically unlimited number of wallet public keys PKs (one per vote) between $EA$ and $P$ such that this PK list can be regenerated only by these two parties, while the corresponding private keys SKs can be computed only by $P$.
As mentioned in \autoref{ssec:advmodel}, $EA$ is assumed to verify identities honestly, and it supplies their corresponding wallet addresses to $SC$. 
The wallet address is generated as a function $f$ of the elliptic curve public key. 
Once the public key is available, it is straightforward to compute the corresponding wallet address.
Let $BP$ be the base point on the elliptic curve.
Further, let $PK$ be the blockchain wallet public key corresponding to a private key $SK$.
Here, $SK$ is chosen as a random positive integer whose size is bounded to the order of $BP$ on the chosen elliptic curve modulo a prime number. 
Note that \autoref{eqn_anon01} -- \autoref{eqn_anon05}  are computed off-chain.
As an illustration, let the first PK be computed as
\begin{equation}\label{eqn_anon01}
PK_{0} = SK_{0}\cdot BP,
\end{equation}
and the next PK be
\begin{equation}\label{eqn_anon02}
PK_{1} = PK_{0}+SK_{1}\cdot BP.
\end{equation}
From \autoref{eqn_anon01} and~\autoref{eqn_anon02}, we observe that
\begin{equation}\label{eqn_anon03}
(SK_{0}+SK_{1})\cdot BP = PK_{1}.
\end{equation}
The following steps ensue:
\begin{compactenum}[i)]
\item During the identity verification, $P$ sends to $EA$:
\textit{a}) wallet public key $PK_{0}$, \textit{b}) a random shared secret key $hk$, and \textit{c}) parameters $(g,p)$, where $g$ is a randomly chosen generator in $F_p^{*}$ {(i.e., a prime field)}  and $p$ is a large prime.
The wallet address is a public function of the wallet public key.
Hence, $EA$ computes $W_{0}=f(PK_0)$ and stores it.

\item The private key of $P$ at any future voting epoch $e=\{1,2,3,...,2^{128}-1\}$  is generated by $P$ as   
\begin{equation}\label{eqn_anon04}
	SK_{e}= SK_{0} + HMAC_{hk}(g^e),
\end{equation}
where $g^e\in F_p^{*}$ is the output of pseudo-random number generator (PRNG) in epoch $e$, HMAC$(.)$ is HMAC-X-256 using shared secret key ${hk}$ between $EA$ and $P$, which is unknown to $Adv_{net}$ and serves for stopping her from mapping $P$'s wallet addresses.

\item The corresponding $PK$ of $P$ for epoch $e$ is   

\begin{equation}\label{eqn_anon05}
	PK_{e}= PK_{0} + HMAC_{hk} (g^e)\cdot BP.
\end{equation}
\end{compactenum} 
$EA$ and $P$ can compute $PK_e$ but $SK_e$ is held only by $P$.
This effectively separates PKs from their SKs and, at the same time, maps it to $P$'s first wallet public key, i.e., $PK_{0}$.   
At any voting iteration $e$, the public key $PK_{e}$, and the corresponding wallet address can be computed by both $EA$ and $P$.
Since the shared secret $hk$ used with HMAC is known only to $EA$ and $P$, no third party, including $Adv_{net}$, is able to compute any future PKs. 
Hence, for a sequence of wallet addresses of $P$  given by $W_{e}$ = $f(PK_{e})$, the map between the wallet address and $P$ is broken for all other parties other than $EA$ and $P$. 

\paragraph{\textbf{Batching the Requests.}}
It is important to note that if a participant wishes to change her vote within the same voting interval,
{she should submit the request with the new wallet address to $EA$ who will approve the request in batches (aggregating multiple such requests) in order to improve the resistance against mapping of former wallet addresses to new ones.}
In detail, the $EA$ will mark all former addresses as \textit{invalid} and approve the new ones within a single transaction.
%
%
{In the extreme case, when a batch contains only one participant (who wanted to change her vote), $Adv_{net}$ can map the two wallets.
Therefore, such a participant should value her wallet privacy and vote only once during the next interval (or alternatively change her address again and be a part a bigger batch).
}
%
Using VPN/dVPNs may further limit the $Adv_{net}$'s ability to map participant IP addresses.

\begin{algorithm}[t]
	\DontPrintSemicolon
	
	\scriptsize
	\SetKwProg{Fn}{Function}{:}{\KwRet}
	\SetKwFunction{VDFADD}{VdfAdd}
	
	\SetKwProg{Fn}{Def}{:}{}
	\Fn{\VDFADD{$y, \pi, blockheight$}}{	
		writeState($``vdfadd"~||~ blockheight, y~||~\pi$)	
	}
	\caption{VDF Add.}\label{alg:vdfdeposit}
\end{algorithm}

\begin{algorithm}[t]
	\DontPrintSemicolon
	\scriptsize
	
	\SetKwProg{Fn}{Function}{:}{\KwRet}
	\SetKwFunction{BPOAdd}{BpoAdd}
	
	\SetKwProg{Fn}{Def}{:}{}
	\Fn{\BPOAdd{$Target, BH, blockheight$}}{	
		writeState($``bpoadd"~||~ blockheight,BH~||~Target$)\\	
		writeState($``blockheightStored",blockheight$)\\
		writeState($``blockheader"~||~blockheight,BH$)
	}
	\caption{BPO Add.}\label{alg:bpodeposit}
\end{algorithm}

\subsubsection{Functionality}
\label{ssec:pseudocode}

The high-level functionality of the AoV framework and its smart contracts is shown in \autoref{alg:framework}, and the trigger mechanism is presented in \autoref{alg:trigger}.
\autoref{alg:framework} comprises of 5 main functions --- \textit{setup, registration, voting, tally} and \textit{revote}. 
The system parameters agreed upon are added by $EA$ using the \textit{setup} function of the smart contract. 
The $EA$ is also responsible for adding the list of valid participants' wallet addresses to the contract through the  \textit{registration} function.
The \textit{voting} function is supplied by a participant's wallet address, blinded vote, and its proof of correctness. 
This information is signed with $P_i$'s private key and sent to the contract. 
The \textit{voting} function carries out the necessary verifications and adds her vote. The participant wallet address is set to ``voted'' to disallow its reuse. 
Before invoking the \textit{tally} function, the VDF prover and BPO store their respective data to the contract (see \autoref{alg:vdfdeposit} and \autoref{alg:bpodeposit}). 
Next, the \textit{tally} function is called by $EA$ or any authorized participant. 
The \textit{tally} function carries out two main tasks. First, it checks whether the condition to trigger the interval tally is satisfied (see \autoref{alg:trigger}). The second task (when triggered) is to tally the votes and return the results.
When a participant wishes to vote again, she sends her next wallet address (synchronized with $EA$) to the \textit{revote} function.
The $EA$ will verify the new wallet address offline and call the \textit{registration} function to set the new address to valid (preferably in batches as mentioned in \autoref{ssec:Anonid}). 
Next, a participant may call the function \textit{Voting} and vote using her new wallet address.

\begin{algorithm}[t]
	
	\DontPrintSemicolon
	\scriptsize
	
	\SetKwProg{Fn}{Function}{:}{\KwRet}
	\SetKwFunction{FTrigger}{VerifyTrigger}
	
	\SetKwProg{Fn}{Def}{:}{}
	\Fn{\FTrigger{$y, \pi, T, BH, params\_struct$}}{	
		$b$ = -$1$\\
		$tt$ = readState($params\_struct.totaltime$)\\ 
		$ft$ = readState($params\_struct.ft$)\\
		$seed$ = readState($params\_struct.key$)\\  
		$IntervalTime$ = $tt/ft$\\ 
		$BHsInInterval$ = $IntervalTime/10$\\
		
		\If{$SHA256(BH)<T$ } 
		{
			\If{ $Verify\_VDF(y,\pi)==True$}
			{
				$a = SHA(y)$\;
				$b = a (mod\; BHsInInterval)$\\			
			}
		}
		\If{b==0}
		{
			\KwRet $True$\;
		}
		\Else{\KwRet $False$\;}
	}

	\caption{Trigger mechanism.}	\label{alg:trigger}
\end{algorithm}

\begin{algorithm*}
	\label{alg:framework}
	
	\scriptsize
	\SetKwFunction{FMain}{Main}
	\SetKwFunction{FSetup}{Setup}
	\SetKwFunction{FReg}{Registration}
	\SetKwFunction{FVoting}{Voting}
	\SetKwFunction{FTally}{Tally}
	\SetKwFunction{FRevote}{Revote}
	
	\KwInput{Set1: $\forall$ participants $P_i$, wallet addresses $(WA_{ij})$, blinded vote $BV_{ij}$ (by $P_i$ for her $j^{th}$ voting occurence,  $j=0,1,2,3,...$) \& zero knowledge proof of vote correctness ($ZKP_{ij}$), $booth_{no}$. Set2: system parameters $init\_params$, BTC blockheader $BH$, $VDF(BH)$, proof $\pi$, BTC target $T$, BTC blockheight.}
	\KwOutput{Total tally of votes in the interval.}
	
	\medskip
	\SetKwProg{Fn}{Function}{:}{}
	\Fn{\FSetup{$init\_params$}}
	{
		writeState($params\_struct, init\_params$) \tcp{add system parameters as key-value pairs into $params\_struct$.}
	}
	\smallskip
	
	\SetKwProg{Fn}{Function}{:}{}
	\Fn{\FReg{$msg1=WA_{ij}, msg2=valid\_flag, EA\_signed\_msg$}}{
		$msg = msg1~||~msg2$ \tcp{concatenate message parts.}
		$EA\_{pubkey}$= readState($params\_struct.EA\_public\_key$) \tcp{get $EA$ public key.}
		
		\If{VerifySig$(msg,EA\_signed\_msg, EA\_{pubkey})=True$}
		{
			\If{$valid\_{flag}==True$}
			{
				writeState($WA_{ij},``valid"$) \tcp{set wallet address to valid.}
			}
			\Else
			{
				writeState($WA_{ij},``invalid"$) \tcp{set wallet address to invalid.}
				
			}
		}
	}
	\smallskip

	\SetKwProg{Fn}{Function}{:}{}
	\Fn{\FVoting{$msg1=WA_{ij},msg2=BV_{ij},msg3=ZKP_{ij},P_i\_signed\_msg$}}{
		$msg = msg1~||~msg2~||~msg3$\\
		$wallet\_status$= readState($WA_{ij}$)\\
		$P_i\_{pubkey}$= readState($params\_struct.P_i\_public\_key$)\\
		$sig\_flag$ = VerifySig($msg$,$P_i\_signed\_msg$,$P_i\_{pubkey}$)\\
		$zkp\_flag$ = VerifyZKP($BV_{ij},ZKP_{ij}$)\\
		
		\If{($sig\_flag$ and $zkp\_flag$) == $True$ and $wallet\_status==``valid"$ }
		{
			writeState($``vote"~||~WA_{ij},BV_{ij}$) \tcp{The latest wallet address of $P_i$ is mapped to her private vote. The key  in (key,value) is prefixed with `vote' tag to identify valid votes w.r.t. wallet addresses.}
			writeState($WA_{ij},``voted"$) \tcp{set $WA$ to voted \& prevent voting from that address again.}

		}
	}
	\smallskip

	\SetKwProg{Fn}{Function}{:}{}
	\Fn{\FTally{$blockheight$}}{
		$total\_tally=-1$\\
		$stored\_blockheight$ = readState($``blockheightStored"$)\tcp{\textit{blockheightStored} is from \autoref{alg:bpodeposit}.}
		$BH=$readState($``blockheader"~||~blockheight$)\tcp{\textit{blockheight} is the argument passed to the function Tally.}
		$y,\pi$ = readState($``vdfadd"~||~blockheight$)\tcp{\textit{vdfadd} is read from \autoref{alg:vdfdeposit}.}
		$BH,Target$ = readState($``bpoadd"~||~blockheight$)\\
		$trigger\_flag$=VerifyTrigger($y, \pi, Target,BH, params\_struct$)\tcp{ Call Algorithm~\autoref{alg:trigger}.}
		\If{($stored\_blockheight$ == $blockheight$) and $trigger\_flag$ == True}
		{
			
			$total\_tally = \sum\limits_{no=1}^{X\cdot Y} local\_tally(booth_{no})$\\
			
		}
		\KwRet $total\_{tally}$		
	}
	\smallskip
	
	\SetKwProg{Fn}{Function}{:}{}
	\Fn{\FRevote{$msg=WA_{ij}, P_i\_signed\_msg$}}{
		$P_i\_{pubkey}$= readState($params\_struct.P_i\_public\_key$) \tcp{get $P_i$ public key}
		
		\If{VerifySig$(msg,P_i\_signed\_msg, P_i\_{pubkey})=True$}
		{
			writeState($WA_{ij},``pending"$) \tcp{set wallet address to pending verification by $EA$.}
		}
		\tcp{Next, $EA$  calls $Registration()$, where it sets a new wallet address of $P_i$ to valid.}
		\tcp{Further, $P_i$ calls $Voting()$ to re-vote using the new wallet address.}
		
	}

	\caption{Always on Voting Framework.}	\label{alg:framework}
\end{algorithm*}

\subsection{Security Analysis}\label{sec:aov:analysis}

\subsubsection{Mining Adversary}
\label{sec:miningadv}
	The goal of  $Adv_{min}$ is to find a valid nonce $s$ that solves the Bitcoin puzzle such that $b$ in \autoref{eqn_2} is 0. When these two conditions are met, and the new epoch is about to start, the validator contract triggers the tally computation of votes.
	We set the difficulty for the benign VDF prover (with commercial hardware) to take 100 minutes\footnote{ We consider $A_{max}=10$, i.e., what is solved by a benign VDF prover in 10 units of time, while in the case of $Adv_{min}$ it is in 1 unit.} to solve VDF($BH$). Based on $A_{max}$ limit, we assume $Adv_{min}$ to take at least 10 minutes to solve the VDF. As a result, $Adv_{min}$ is restricted to a maximum of 1 try (considering 10 minutes as an average Bitcoin block creation time), excluding the Proof-of-Work required to solve the Bitcoin mining puzzle.  
	However, since a Bitcoin block header is generated on average once every 10 minutes and the benign VDF prover is occupied for 100 minutes, the question is -- how many VDF provers are required to prevent the block headers from queuing up?
	We can see in \autoref{tab:vdftimeline} that VDF prover 1 runs a task for time 0-100 minutes, and she picks up the next task to run for time 100-199 minutes. Similarly, all other provers pick up the next task after completing the present one.
	Hence, 10 VDF provers are sufficient to prevent block headers from queuing up because $A_{max}=10$.
	On the other hand, a benign VDF prover might reduce $A_{max}$ of VDF computation by using specialized hardware instead of commercial hardware (depending on the cost-to-benefit ratio).  
	However, we emphasize that the VDF can be computed only after solving the PoW mining puzzle, which is prohibitively expensive. Moreover, the puzzle difficulty increases proportionally to the mining power of the Bitcoin network. 
	Hence, the proposed serial combination of solving the Bitcoin mining puzzle followed by the computation of VDF output improves the aggregate security against $Adv_{min}$ from choosing a favorable nonce.
	The estimated requirement of $A_{max}=10$ might be further increased as more studies to efficiently solve VDFs on ASICs are carried out.
	However, if $A_{max}$ will increase in the future, our solution can cope with it by employing more VDF provers. 
\begin{table}[t]
	\centering	
	\scriptsize
	\begin{tabular}{ p{1cm} p{1.9cm} | p{1cm} p{1.9cm}  }

		\toprule
		VDF Prover& Time (minutes) & VDF Prover&Time (minutes)\\
		\midrule
		1   & 0-100    &  1   & 100-199\\
		2   & 10-110  &  2   & 110-209\\
		3   & 20-120  &  3   & 120-219\\
		4   & 30-130  &  4   & 130-229\\
		5   & 40-140  &  5   & 140-239\\
		6   & 50-150  &  6   & 150-249\\
		7   & 60-160  &  7   & 160-259\\
		8   & 70-170  &  8   & 170-269\\
		9   & 80-180  &  9   & 180-279\\
		10  & 90-190  & 10   & 190-289\\    
		\bottomrule
	\end{tabular}
	\vspace{0.2cm}
	\caption{Scheduling 10 VDF provers without queuing. Note the VDF computations on a VDF prover machine are not parallelized. It is the scheduling alone that is in parallel. The start time is based on the job arrival time at the VDF prover, where it will run for 100 minutes. Once completed, it is ready to take on the next job. In column 2, the start times are 10 minutes apart and correspond to the average BTC interblock (job) arrival time. The largest idle time in column 1 is for VDF Prover 10 at 90 minutes, waiting for the job to start. Beyond this, all VDF prover machines are continuously occupied since a new job is available to start immediately after the current job ends. }	\label{tab:vdftimeline}
\end{table}

	\subsubsection{Implications of VDF Prover Synchronization and Optimizing Frequency of Supplied Block Headers}
	\label{ssec:bh}
	Several VDF provers are  synchronised to supply the VDF proofs to the validator contract in sequence.
	However, there are no adverse effects when the proofs are generated and supplied out of sequence.
	The validator smart contract stores the latest \textit{block height} for which the VDF proof was last  accepted. It only allows  proof verification for stored \textit{block height+1} on the contract and any out-of-order proofs have to be resent. 
	Once the  order  is corrected, a handful of VDF proofs may appear in quick succession at the validator contract. 
	However, the tally for the interval is only triggered when $VC_{output}$ in \autoref{eqn_4} is \textit{True}.

	In terms of gas consumption, it can be costly to process every single Bitcoin block header (supplied to the VDF prover and the validator contract by the BPO).
	We suggest optimizing this by choosing a coarser time granularity of the block header supply, independent of the Bitcoin block interval (e.g., every x-th block).
	We modify the example from \autoref{ssec:interactions} by considering the processing of every 100$^{th}$ Bitcoin block header.\footnote{Note that this would need another condition to be met, i.e., the block height in BH (mod 100) should be equal to 0.}
	$TotalTime$ = 4 years = $525600\cdot 4$ minutes and the total number of intervals $ft=8$.
	Then, $IntervalTime = (525600\cdot 4)/(8) = 262800$ minutes and $BHsInInterval = 262800/$  $(10\cdot 100) = 262.8$. 
	On average, the oracle will send 262.8 block headers ($BH$ values) to the validator contract within 182.5 days instead of the 26280 block headers required in the original example.
	In this case, we only need 1 VDF prover instead of 10, and it provides similar security guarantees as before.

\begin{figure*}[t]
	\centering
	\subfloat[][Booth  with 30 participants, winning \ensuremath{p=0.9}.]{
		\includegraphics[width=0.48\textwidth]{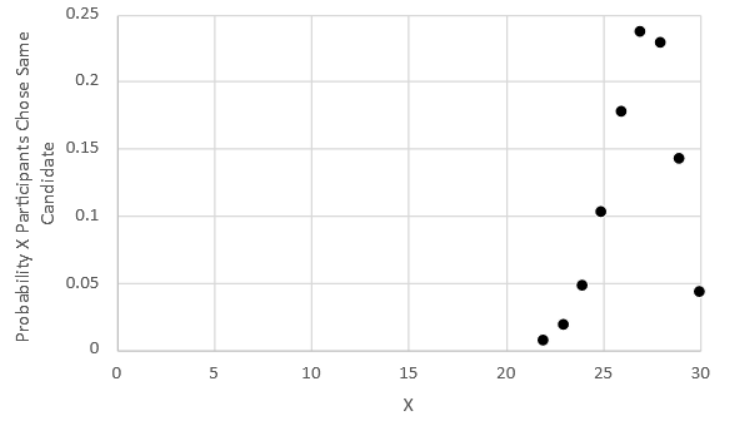}
	}
	\subfloat[][Booth  with 100 participants, winning \ensuremath{p=0.9}.]{		
		\includegraphics[width=0.48\textwidth]{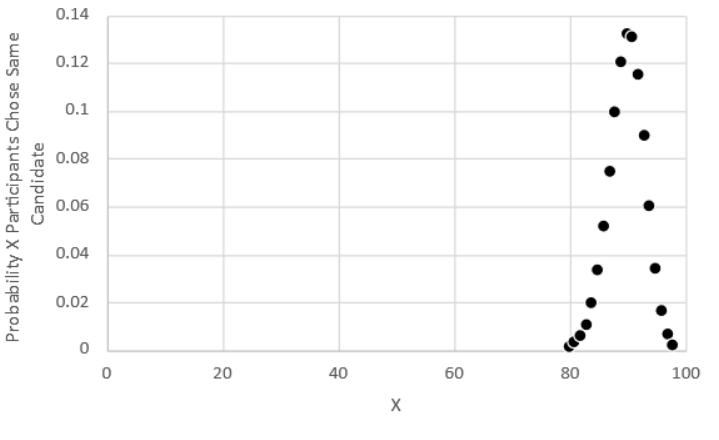}		
	}
	
	\vspace{0.3cm}
	\caption{Binomial probability distribution function of X booth participants voting for their favorite candidate whose winning probability is $p$. 
	}
	\label{fig:binomialpdf}
\end{figure*}

\subsubsection{Privacy Implications of Booth Sharding}
\label{ssec:shardinganalysis}

In this section, we look at the implications of booth sharding. 
Further, we recommend booth sizes to protect participant votes from being revealed to the $EA$ (in the case when all participants in a booth voted for the same candidate).
The map between a participant's wallet and her vote transaction is broken for all parties (including $Adv_{net}$) but the $EA$ (see \autoref{ssec:Anonid}). 
The $EA$ is aware of the participants' current wallet address used for voting, and hence it has an advantage over $Adv_{net}$ in statistical inference attacks.
Certain scenarios revealing vote choice are possible under some circumstances. 
In particular, if all participants in a booth voted for the same candidate or the winning probability of one candidate is much higher than the others.  
%
%

To demonstrate it, in \autoref{fig:binomialpdf} we provide the probability that $X$ participants in a booth voted for the same candidate, depending on the candidate winning probability $p$. 
%
\autoref{fig:binomialpdf}(a) represents a booth with 30 participants and candidate winning probability $p=0.9$. The probability that all participants voted for the same candidate is 
$P(30, X=30,p=0.9)\stackrel{\sim}{=} 0.0423$.
\autoref{fig:binomialpdf}(b) represents a booth with  100 participants and candidate winning probability $p=0.9$. The probability that all participants voted for the same candidate is 
$P(100, X = 100,p=0.9)\stackrel{\sim}{=} 0.00003$, which demonstrates that number of participants in a booth, influences $p$ in indirect proportion, favoring the booths with higher sizes.

Even though the probabilities in the booth with 100 participants are very low, this may not be sufficient depending on the total number of participants.
For example, consider the elections with 1 million participants.
First, let number of participants in a booth be 30 and the number of booths $M$ $ = \lceil(1~000~000/30)\rceil =33334$.
The number of booths where all participants likely voted for the same candidate is $0.0423\cdot33334 \stackrel{\sim}{=} 1410$.
For booths with 100 participants each and $M = \lceil(1~000~000/100)\rceil = 10000$, the number of booths where all participants likely voted for the same candidate is reduced to $0.00003\cdot10000 \stackrel{\sim}{=} 0.3$. 
Therefore, a suitable number of participants per booth should be determined based on the extreme estimations of the tally results and the total number of voters. 

\subsubsection{Randomness of Bitcoin Nonces \& AoV Entropy}\label{ssec:rndnonce}
We decided to utilize a single public source of randomness instead of a distributed randomness due to the low computation cost and synchronization complexity.
Bonneau et al.~\cite{BonneauCG15} showed that if the underlying hash function used to solve the Bitcoin PoW puzzle is secure, then each block in the canonical chain
has a computational min-entropy of at least $d$ bits, representing the mining difficulty.
I.e., $d$ consecutive 0 bits must appear in the hash of the block header.\footnote{At the time of writing, $d\approx$ 76.}
Hence, $\lfloor \frac{d}{2} \rfloor$ near-uniform bits can be securely extracted.
Nevertheless, empirical evaluation has shown that Bitcoin nonces have visible white spaces (non-uniformity) in its scatter-plot~\cite{Bitmex2019}.
A possible explanation is that some miners are presetting some of the bits in the 32-bit $nonce$ field and using the $extraNonce$ to solve the PoW puzzle.
We use the entire block header as the initial source of entropy instead of the 32-bit nonce alone to avoid such biases.
To reduce the probability of $Adv_{min}$ biasing the solution in her favor, the block header is passed through a verifiable delay function (see \autoref{eqn_vdfprover}). 
The output of VDF is hashed (see \autoref{eqn_1}) to consolidate the entropy.

\section{Contributing Papers}\label{sec:evoting-papers}
The papers that contributed to this research direction are enumerated in the following, while highlighted papers are attached to this thesis in their original form.

\begingroup
\let\clearpage\relax

\renewcommand\bibname{}

\endgroup

%% file: sec/logging.tex

In this chapter, we present our contributions to the area of secure logging using blockchain, which belongs to the application layer of our security reference architecture (see \autoref{chapter:sra}) and its category of data provenance.
In particular, this chapter is focused on the security, privacy, and scalability, of blockchain-based secure logging, which we further extend to the use case of a global Central Bank Digital Currency (CBDC).
Note that we base this chapter on not yet published papers \cite{homoliak2020aquareum} \cite{homoliak2023cbdc} (see also \autoref{sec:logging-papers}), and therefore we consider this chapter as additional to the thesis.

\renewcommand{\name}{Aquareum\xspace}
First, we present \name \cite{homoliak2020aquareum}, a novel framework for centralized ledgers
removing their main limitations.  By combining a trusted execution
environment with a public blockchain platform, \name provides publicly
verifiable, non-equivocating, censorship-evident, private, and high-performance
ledgers.  
\name ledgers are integrated with a Turing-complete virtual machine,
allowing arbitrary transaction processing logics, including tokens or
client-specified smart contracts (see details in \autoref{sec:logging-aquareum}).  

Next, we aim at interoperability for the environment of Central Bank Digital Currency (CBDC) containing multiple instances of centralized ledgers (based on Aquareum) that either represent central banks of more countries or retail banks (with a single central bank) of a single country.
In detail, we present CBDC-AquaSphere \cite{homoliak2023cbdc} a practical blockchain interoperability protocol that integrates important features of Digital Euro Association manifesto \cite{cbdc-manifesto}, such as strong value proposition for the end users, the highest degree of privacy, and interoperability. 
On top of the above-mentioned features, our work also provides proof-of-censorship and atomicity (see details in \autoref{sec:logging-cbdc}).

\section{Aquareum}\label{sec:logging-aquareum}
We depict overview of components that \name consists of in \autoref{fig:overview-details}. 
In the following, we elaborate on the system model and principles of our approach.
For background related to integrity preserving data structures and trusted computing, we refer the reader to 
\autoref{sec:background:integrity-structures} and \autoref{sec:background:tee}.

\subsubsection{Notation}
By $\{msg\}_\mathbb{U}$, we denote the message $msg$ digitally signed by $\mathbb{U}$, and by $msg.\sigma$ we refer to a signature;
$h(.)$ stands for a cryptographic hash function;
$\|$ is the string concatenation;
$\%$ represents modulo operation over integers; 
$\Sigma_{p}. \{KeyGen, Verify, Sign\}$ represents a signature (and encryption) scheme of the platform $p$, where $p \in \{pb, tee\}$ (i.e., public blockchain platform and trusted execution environment platform);
and $SK_\mathbb{U}^{p}$, $PK_\mathbb{U}^{p}$ is the private/public key-pair of $\mathbb{U}$, under $\Sigma_{p}$.
Then, we use $\pi^{s}$ for denoting proofs of various data structures $s \in \{mk, mem, inc\}$: 
$\pi^{mk}$ denotes the inclusion proof in the Merkle tree, 
$\pi^{mem}$ and $\pi^{inc}$ denote the membership proof and the incremental proof in the history tree, respectively.

\begin{figure*}[t]
	\begin{center}
		\includegraphics[width=0.95\textwidth]{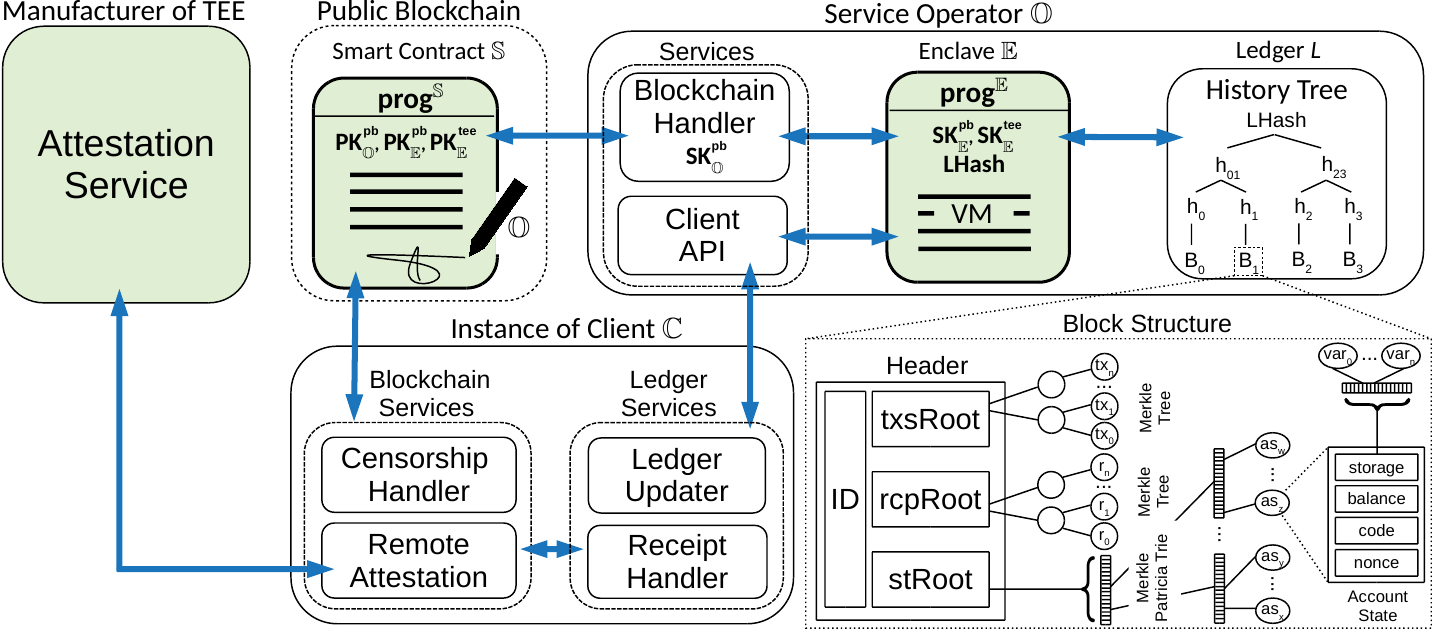} 
		\caption{\name components. Trusted components are depicted in green.}\label{fig:overview-details}		
	\end{center}	
\end{figure*}

\subsection{System Model} \label{sec:overview}
In \name, an \textit{operator} is an entity that maintains and manages a ledger containing chronologically sorted transactions. 
\textit{Clients} interact with the ledger by sending requests, such as queries and transactions to be handled.  
We assume that all involved parties can interact with a blockchain platform supporting smart contracts (e.g., Ethereum). 
Next, we assume that the operator has access to a TEE platform (e.g., Intel SGX).
Finally, we assume that the operator can be malicious and her goals are as follows:
\begin{compactitem}
	\item \textbf{Violation of the ledger's integrity} by creating its
	internal inconsistent state -- e.g., via inserting two conflicting
	transactions or by removing/modifying existing transactions.
	\item \textbf{Equivocation of the ledger} by presenting at least two inconsistent views of the ledger to (at least) two distinct clients who would accept such views as valid.
	\item \textbf{Censorship of client queries} without leaving any audit trails evincing the censorship occurrence.
\end{compactitem}
Next, we assume that the adversary cannot undermine the cryptographic primitives used, 
the underlying blockchain platform, and the TEE platform deployed.

\subsubsection{Desired Properties}\label{sec:desired-properties}
We target the following security properties for \name ledgers:
\begin{compactitem}
	\item[\textit{\textbf{Verifiability:}}] 
	clients should be able to obtain easily verifiable evidence that the
	ledger they interact with is internally \textit{correct} and \textit{consistent}. 
	In particular, it means that none of the previously inserted transaction was neither modified nor deleted, and there are no conflicting transactions.  
	Traditionally, the verifiability is achieved by replicating the ledger (like in
	blockchains) or by trusted auditors who download the full copy of the ledger and sequentially validate it. 
	However, this property should be provided even if the
	operator does not wish to share the full database with third parties.        
	Besides, the system should be \textit{self-auditable}, such that any client can easily verify (and prove to others) that some transaction is included in the ledger, and she can prove the state of the ledger at the given point in time.
	
	\item[\textit{\textbf{Non-Equivocation:}}]
	the system should protect from forking attacks and thus guarantee that no
	concurrent (equivocating) versions of the ledger exist at any point in
	time.  The consequence of this property is that whenever a client
	interacts with the ledger or relies on the ledger's logged artifacts, the client
	is ensured that other clients have ledger views consistent with her view.

	\item[\textit{\textbf{Censorship Evidence:}}] 
	preventing censorship in a centralized system is particularly
	challenging, as its operator can simply pretend unavailability in order to censor undesired queries or transactions.  
	Therefore, this property requires that  whenever the operator censors client's requests, the client can do a resolution of an arbitrary (i.e., censored) request publicly.  
	We emphasize that proving censorship is
	a non-trivial task since it is difficult to distinguish ``pretended''
	unavailability from ``genuine'' one. 
	Censorship evidence enables clients to
	enforce potential service-level agreements with the operator, either by
	a legal dispute or by automated rules encoded in smart contracts.
\end{compactitem}
Besides those properties, we intend the system to provide \textit{privacy}
(keeping the clients' communication confidential), \textit{efficiency} and
\textit{high performance}, not introducing any significant overhead,
\textit{deployability} with today's technologies and infrastructures, as well as
\textit{flexibility} enabling various applications and scenarios.

\begin{figure}[t]
	\centering
	\includegraphics[width=0.55\textwidth]{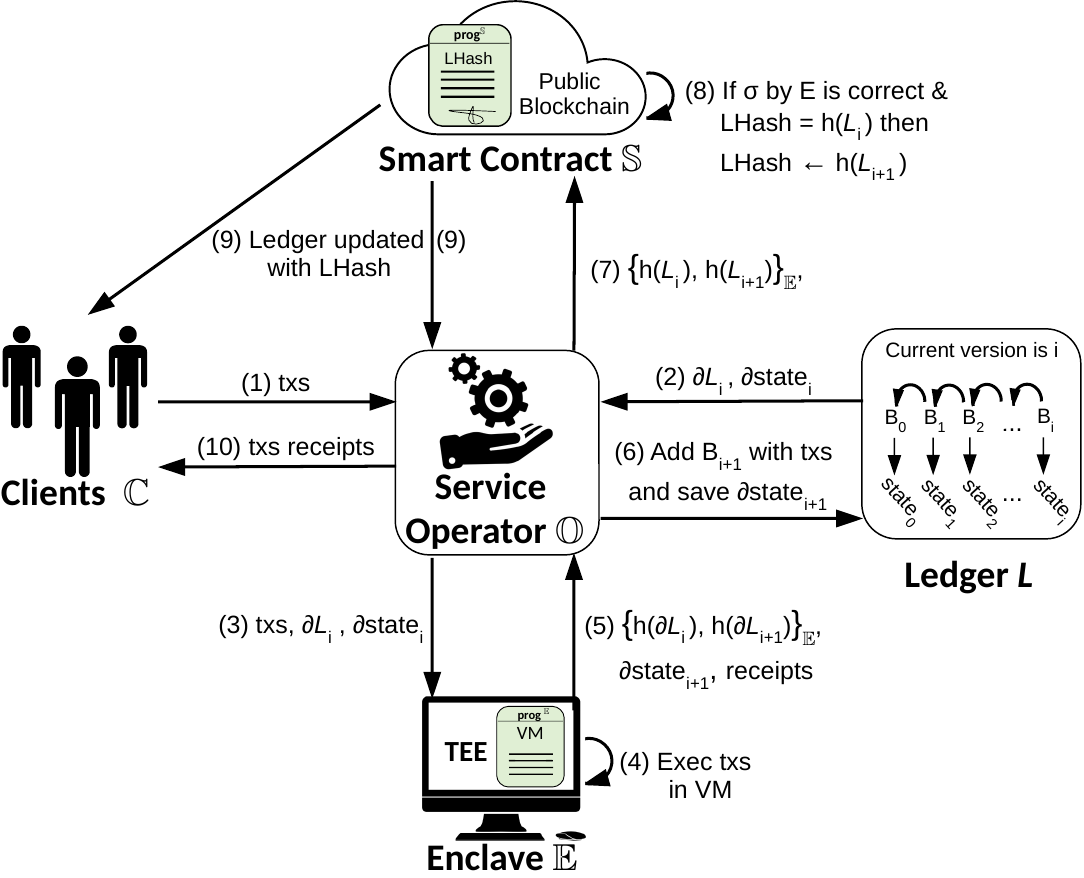} 
	\caption{Operation procedure of \name ledger.}
	\label{fig:overview-normal-op}
	
\end{figure}

\subsection{High-Level Overview}\label{sec:high-overview}

\name ledger is initialized by an operator ($\mathbb{O}$) who creates an internal ledger ($L$) that will store all transactions processed and the state that they render.  
Initially, $L$ contains an empty transaction set and a null state.
During the initialization, $\mathbb{O}$ creates a TEE enclave ($\mathbb{E}$) whose role is to execute updates of $L$ and verify consistency of $L$ before each update.
Initialization of $\mathbb{E}$ involves the generation of two public private key pairs -- one for the signature scheme of TEE (i.e.,  $PK_{\mathbb{E}}^{tee}, SK_{\mathbb{E}}^{tee}$) and one for the signature scheme of the public blockchain (i.e., $PK_{\mathbb{E}}^{pb}, SK_{\mathbb{E}}^{pb}$).\footnote{Note that neither of the private keys ever leaves $\mathbb{E}$.}
The code of $\mathbb{E}$ is public (see \autoref{alg:enclave-VM} and \autoref{alg:enclave-VM-cens}), and it can be remotely attested with the TEE infrastructure by any client.

\begin{algorithm}[!h] 
	\caption{The program $prog^{\mathbb{S}}$ of the smart contract $\mathbb{S}$ }\label{alg:log-smart-contract}
	\scriptsize
	
	\SetKwProg{func}{function}{}{}

	\smallskip
	$\triangleright$ \textsc{Declaration of types and constants:}\\		
	\hspace{1em} \textbf{CensInfo} \{ $etx, equery, status, edata$ \},  \\
	\hspace{1em} $msg$: a current transaction that called $\mathbb{S}$,  \\
	
	\smallskip
	$\triangleright$ \textsc{Declaration of functions:}
	
	\func{$Init$($PK_{\mathbb{E}}^{pb}, PK_{\mathbb{E}}^{tee}, PK_{\mathbb{O}} $) \textbf{public} }{
		$PK_{\mathbb{E}}^{tee}[].add(PK_{\mathbb{E}}^{tee})$; \Comment{PK of enclave $\mathbb{E}$ under $\Sigma_{tee}$.} \\ 
		$PK_{\mathbb{E}}^{pb}[].add(PK_{\mathbb{E}}^{pb})$; \Comment{PK of enclave $\mathbb{E}$ under $\Sigma_{pb}$.} \\
		$PK_{\mathbb{O}}^{pb} \leftarrow PK_{\mathbb{O}}$; \Comment{PK of operator $\mathbb{O}$ under $\Sigma_{pb}$.} \\
		$LRoot_{pb} \leftarrow \perp$; \Comment{The most recent root hash of $L$ synchronized with $\mathbb{S}$.} \\ 
		$censReqs \leftarrow []$; \Comment{Request that $\mathbb{C}$s wants to resolve publicly.} \\
	}
	
	\func{$PostLRoot$($root_A, root_B, \sigma$) \textbf{public} }{
		\Comment{Verify whether a state transition was made within $\mathbb{E}$. \hfill} \\
		\textbf{assert} $\Sigma_{pb}.verify((\sigma, PK_{\mathbb{E}}^{pb}[\text{-}1]), (root_A, root_B))$;  \\			
		
		\Comment{Verify whether a version transition extends the last one. \hfill} \\
		\If{$LRoot_{pb} = root_A$}{
			$LRoot_{pb} \leftarrow root_{B}$; \Comment{Do a version transition of L.} \\
		} 				
	}
	
	\func{$ReplaceEnc$($PKN_{\mathbb{E}}^{pb}, PKN_{\mathbb{E}}^{tee}, r_{A}, r_{B}, \sigma, \sigma_{msg}$) \textbf{public} }{
		\Comment{Called	 by $\mathbb{O}$ in the case of enclave failure.\hfill}\\
		
		\textbf{assert} $\Sigma_{pb}.verify((\sigma_{msg}, PK_{\mathbb{O}}^{pb}), msg)$;  \Comment{Avoiding MiTM attack.} \\
		$PostLRoot(r_{A}, r_{B}, \sigma)$ ; \Comment{Do a version transition.} \\

		$PK_{\mathbb{E}}^{tee}.add(PKN_{\mathbb{E}}^{tee})$; \Comment{Upon change, $\mathbb{C}s$ make remote attestation.} \\ 
		$PK_{\mathbb{E}}^{pb}.add(PKN_{\mathbb{E}}^{pb})$; \\	
	}
	
	\func{$SubmitCensTx$($etx, \sigma_{msg}$) \textbf{public} }{
		\Comment{Called by $\mathbb{C}$ in the case her TX is censored.\hfill \hfill \hfill}\\		
		accessControl($\sigma_{msg}, msg.PK_{\mathbb{C}}^{pb}$); \\

		$censReqs$.add(\textbf{CensInfo}($etx, \perp, \perp, \perp$)); \\			
		
	}
	\smallskip
	
	\func{$ResolveCensTx(idx_{req}, status, \sigma$) \textbf{public} }{
		\Comment{Called by $\mathbb{O}$ to prove that $\mathbb{C}$'s TX was processed.\hfill \hfill \hfill}\\
		
		\textbf{assert} $idx_{req} < |censReqs|$;\\
		$r \leftarrow censReqs[idx_{req}]$; \\
		
		\textbf{assert} $\Sigma_{pb}.verify((\sigma, PK_{\mathbb{E}}^{pb}[\text{-}1]), ~(h(r.etx), status))$; 	\\ 
		$r.status \leftarrow status$;\\
	}	
	
	\func{$SubmitCensQry$($equery, \sigma_{msg}$) \textbf{public} }{
		\Comment{Called by $\mathbb{C}$ in the case its read query is censored.\hfill \hfill}\\		
		accessControl($\sigma_{msg}, msg.PK_{\mathbb{C}}^{pb}$); \\

		$censReqs$.add(\textbf{CensInfo}($\perp, equery, \perp, \perp$)); \\		
		
	}
	\smallskip
	
	\func{$ResolveCensQry(idx_{req}, status, edata, \sigma$) \textbf{public} }{
		\Comment{Called by $\mathbb{O}$ as a response to the $\mathbb{C}$'s censored  read query.\hfill \hfill \hfill}\\		
		\textbf{assert} $idx_{req} < |censReqs|$;\\
		$r \leftarrow censReqs[idx_{req}]$; \\
		
		\textbf{assert} $\Sigma_{pb}.verify((\sigma, PK_{\mathbb{E}}^{pb}[\text{-}1]), (h(r.equery), status, h(edata)))$; 	\\ 
		$r.\{edata \leftarrow edata, status \leftarrow status\}$;\\		
	}

	\vspace{-0.1cm}
\end{algorithm}

\begin{algorithm}[t] 
	\scriptsize
	\SetKwProg{func}{function}{}{}
	
	$\triangleright$ \textsc{Declaration of types and functions:}\\
	\hspace{1em} \textbf{Header} \{ $ID$, $txsRoot$, $rcpRoot$, $stRoot$\}; \\
	
	\hspace{1em} $\#(r) \rightarrow v$: denotes the version $v$ of $L$ having  $LRoot$ $=$ $r$,\\		
	
	$\triangleright$ \textsc{Variables of TEE:} \\
	\hspace{1em} $SK_{\mathbb{E}}^{tee}, PK_{\mathbb{E}}^{tee}$: keypair of $\mathbb{E}$ under $\Sigma_{tee}$,\\	
	
	\hspace{1em} $SK_{\mathbb{E}}^{pb}, PK_{\mathbb{E}}^{pb}$: keypair of $\mathbb{E}$ under $\Sigma_{pb}$,\\
	\hspace{1em} $hdr_{last} \leftarrow \perp$: the last header created by $\mathbb{E}$,\\		
	\hspace{1em} $LRoot_{pb} \leftarrow \perp$: the last root of $L$ flushed to PB,\\
	\hspace{1em} $LRoot_{cur} \leftarrow \perp$: the root of $L \cup blks_p$ (not flushed to PB),\\	
	\hspace{1em} $ID_{cur} \leftarrow 1$: the current version of $L$ (not flushed to PB),\\

	\smallskip
	$\triangleright$ \textsc{Declaration of functions:}
	
	\func{$Init$() \textbf{public}} {
		($SK_{\mathbb{E}}^{pb}$, $PK_{\mathbb{E}}^{pb}$)$ \leftarrow\Sigma_{pb}.Keygen()$;\\ 
		($SK_{\mathbb{E}}^{tee}$, $PK_{\mathbb{E}}^{tee}$)$ \leftarrow\Sigma_{tee}.Keygen()$;\\
		
		\textbf{Output}($PK_{\mathbb{E}}^{tee}, PK_{\mathbb{E}}^{pb}$); \\
	}					
	\smallskip
	
	\func{$Exec$($txs[], \partial st^{old}, ~\pi^{inc}_{next}, LRoot_{tmp}$) \textbf{public}}{
		
		\textbf{assert} $\partial st^{old}.root = hdr_{last}.stRoot$; \\
		
		$\partial st^{new}, rcps, txs_{er} \leftarrow ~~processTxs(txs, ~\partial st^{old}, ~\pi^{inc}_{next}, ~LRoot_{tmp})$;\\
		
		$\sigma \leftarrow \Sigma_{pb}.sign(SK_{\mathbb{E}}^{pb}, (LRoot_{pb}, LRoot_{cur}))$; \\
		\textbf{Output}($LRoot_{pb}, LRoot_{cur}, \partial st^{new}, hdr_{last}, rcps$, $txs_{er}$, $\sigma$); \\	
		
	}
	\smallskip
	
	\func{$Flush$() \textbf{public}}{									
		$LRoot_{pb} \leftarrow LRoot_{cur}$; \Comment{Shift the version of $L$ synchronized with PB.} \\					
	}
	\smallskip
	
	\func{$processTxs$($txs[], \partial st^{old}, ~\pi^{inc}_{next}, ~LRoot_{tmp}$) \textbf{private}}{
		
		$\partial st^{new}, rcps[], txs_{er} \leftarrow$ runVM($txs$, $\partial st^{old}$); \Comment{Run $txs$ in VM.} \\
		$txs \leftarrow txs \setminus txs_{er}$; \Comment{Filter out parsing errors/wrong signatures.} \\
		
		$hdr \leftarrow ~~\mathbf{Header}(ID_{cur}, MkRoot(txs), MkRoot(rcps), \partial st^{new}.root))$;\\ 
		$hdr_{last} \leftarrow  hdr$; \\
		$ID_{cur} \leftarrow  ID_{cur} + 1$; \\		
		$LRoot_{cur} \leftarrow newLRoot(hdr, ~\pi^{inc}_{next}, ~LRoot_{tmp})$; \\
		\textbf{return} $\partial st^{new}$, $rcps$, $txs_{er}$; \\
	}
	\smallskip
	
	\func{$newLRoot(hdr, ~\pi^{inc}_{next}, ~LRoot_{tmp})$ \textbf{private}}{
		\Comment{A modification of the incr. proof. template to contain $hdr$ \hfill}\\
		\textbf{assert} $\#(LRoot_{cur}) + 1 = \#(LRoot_{tmp})$; \Comment{1 block $\Delta$.} \\
		\textbf{assert} $\pi^{inc}_{next}.Verify(LRoot_{cur},~LRoot_{tmp})$;\\
		
		$\pi^{inc}_{next}[\text{-}1] \leftarrow h(hdr)$; \\
		
		\textbf{return} $deriveNewRoot(\pi^{inc}_{next})$; \\
	}
	
	\caption{The program $prog^{\mathbb{E}}$ of enclave $\mathbb{E}$}
	\label{alg:enclave-VM}
	\vspace{-0.1cm}
\end{algorithm}

Next, $\mathbb{O}$ generates her public-private key pair (i.e., $PK_{\mathbb{O}}, SK_{\mathbb{O}}$) and deploys a special smart contract $\mathbb{S}$ (see \autoref{alg:log-smart-contract}) initialized with the empty $L$ represented by its hash $LHash$, the operator's public key $PK_{\mathbb{O}}$, and both enclave public keys $PK_{\mathbb{E}}^{tee}$ and $PK_{\mathbb{E}}^{pb}$.
After the deployment of $\mathbb{S}$, an instance of $L$ is uniquely identified by the address of~$\mathbb{S}$.
A client ($\mathbb{C}$) wishing to interact with $L$ obtains the address of $\mathbb{S}$ and performs the remote attestation of $\mathbb{E}$ using the $PK_{\mathbb{E}}^{tee}$.

\begin{algorithm}[!h] 
	\scriptsize
	\SetKwProg{func}{function}{}{}

	\func{$Decrypt(edata)$ \textbf{public}}{
		$data \leftarrow \Sigma_{pb}.Decrypt(SK_\mathbb{E}^{pb}, edata)$; \\
		
		\textbf{Output}($data$); \\
	}
	\smallskip
	
	\func{$SignTx(etx, \pi^{mk}_{tx}, hdr, \pi^{mem}_{hdr})$ \textbf{public}}{
		\Comment{Resolution of a censored write tx. \hfill \hfill \hfill \hfill} \\
		
		$tx \leftarrow \Sigma_{pb}.Decrypt(SK_\mathbb{E}^{pb}, etx)$; \\
		\If{$ERROR = parse(tx)$}{					
			$status$ = PARSING\_ERROR;\\
			
		} \ElseIf{$ERROR = \Sigma_{pb}.Verify((tx.\sigma, tx.PK_{\mathbb{C}}^{pb}), tx) $}{
			$status$ = SIGNATURE\_ERROR;\\	
			
		} \Else{		
			\Comment{Verify proofs binding TX to header and header to $L$.\hfill \hfill}\\
			\textbf{assert} $\pi_{tx}^{mk}$.Verify($tx, hdr.txsRoot$);\\
			
			\textbf{assert} $\pi_{hdr}^{mem}$.Verify($hdr.ID, hdr, LRoot_{pb}$);\\					
			$status \leftarrow$  INCLUDED;\\
		}

		\Comment{TX was processed, so $\mathbb{E}$ can issue a proof.}\hfill \hfill \hfill \\		
		$\sigma \leftarrow \Sigma_{pb}.sign(SK_{\mathbb{E}}^{pb},~ (h(etx), status))$;  \\
		\textbf{Output}($\sigma$, $status$); 				
	}		
	\smallskip	
	
	\func{$SignQryTx(equery, blk, \pi^{mem}_{hdr})$ \textbf{public}}{
		\Comment{Resolution of a censored read tx query. \hfill \hfill \hfill \hfill} \\
		
		$\ldots, id_{tx}, id_{blk}, PK_{\mathbb{C}}^{pb} \leftarrow parse(Decrypt(equery))$; \\		
		\If{$id_{blk} > \#(LRoot_{pb})$}{
			$status \leftarrow$  BLK\_NOT\_FOUND, $~~edata \leftarrow$  $\perp$;\\
		}\Else{  
			
			\textbf{assert} $\pi_{hdr}^{mem}$.Verify($blk.hdr.ID, blk.hdr, LRoot_{pb}$);\\
			
			\textbf{assert} VerifyBlock(blk); \Comment{Full check of block consistency.} \\				
			
			$tx \leftarrow$ findTx$(id_{tx}, blk.txs)$; \\
			\If{$\perp ~=~ tx$}{
				$status \leftarrow$  TX\_NOT\_FOUND, $~~edata \leftarrow$  $\perp$;\\
			}\Else{
				$status \leftarrow$  OK, $~~edata \leftarrow  \Sigma_{pb}.Encrypt(PK_\mathbb{C}^{pb}, tx)$;\\
			}
		}
		
		$\sigma \leftarrow \Sigma_{pb}.sign(SK_{\mathbb{E}}^{pb},~ (h(equery), status, edata))$;  \\	
		\textbf{Output}($\sigma$, $status$, $edata$); 	
	}		
	
	\func{$SignQryAS(equery, as, \pi^{mpt}_{as})$ \textbf{public}}{
		\Comment{Resolution of a censored read account state query. \hfill \hfill \hfill \hfill} \\
		
		$\ldots, id_{as}, PK_{\mathbb{C}}^{pb} \leftarrow parse(Decrypt(equery))$; \\				
		\If{$\perp ~=~ as$}{
			\textbf{assert} $\pi^{mpt}_{as}.VerifyNeg(id_{as}, LRoot_{cur})$; \\
			$status \leftarrow$  NOT\_FOUND, $~~edata \leftarrow$  $\perp$;\\
		}\Else{
			\textbf{assert} $\pi^{mpt}_{as}.Verify(id_{as}, LRoot_{cur})$; \\
			$status \leftarrow$  OK, $~~edata \leftarrow  \Sigma_{pb}.Encrypt(PK_\mathbb{C}^{pb}, as)$;\\
		}
		
		$\sigma \leftarrow \Sigma_{pb}.sign(SK_{\mathbb{E}}^{pb},~ (h(equery), status, h(edata)))$;  \\	
		\textbf{Output}($\sigma$, $status$, $edata$); 	
	}		
	
	\caption{Censorship resolution in $\mathbb{E}$ (part of $prog^{\mathbb{E}})$.}
	\label{alg:enclave-VM-cens}
\end{algorithm}

Whenever $\mathbb{C}$ sends a transaction to $\mathbb{O}$ (see \autoref{fig:overview-normal-op}),  $\mathbb{E}$ validates whether it is authentic and non-conflicting; and if so, $\mathbb{E}$ updates $L$ with the transaction, yielding the new version of $L$. 
The $\mathbb{C}$ is responded with a \textit{receipt} and \textit{``a version transition of $L$''}, both signed  by $\mathbb{E}$, which prove that the transaction was processed successfully and is included in the new version of $L$. 
For efficiency reasons, transactions are processed in batches that are referred to as \textit{blocks}.  
In detail, $\mathbb{O}$  starts the update procedure of $L$ (see \autoref{fig:overview-normal-op}) as follows:
\begin{compactenum}[a)]
	\item $\mathbb{O}$ sends all received transactions since the
	previous update to $\mathbb{E}$, together with the current partial state of $L$ and a small subset of $L$'s data $\partial L_i$, such that $h(\partial L_i) = h(L_i)$, which is required to validate $L$'s consistency and perform its incremental extension.
	
	\item $\mathbb{E}$ validates and executes the transactions in its virtual machine, updates the current partial state and partial data of $L$, and finally creates a blockchain transaction\footnote{Note that $\{h(\partial L_i), h(\partial L_{i+1})\}_\mathbb{E} = \{h(L_i), h(L_{i+1})\}_\mathbb{E}$ } $\{h(\partial L_i),$ $h(\partial L_{i+1})\}_\mathbb{E}$ signed by $SK_{\mathbb{E}}^{pb}$, which represents a version transition of the ledger from version $i$ to its new version $i+1$, also referred to as the \textbf{version transition pair}.
	
	\item The blockchain transaction with version transition pair is returned to $\mathbb{O}$, who sends this transaction to $\mathbb{S}$. 
	
	\item $\mathbb{S}$ accepts the second item of the version transition pair iff it is signed by $SK_{\mathbb{E}}^{pb}$ and the current hash stored by $\mathbb{S}$ (i.e., $LHash$) is equal to the first item of the pair.  
\end{compactenum}

\smallskip\noindent
After the update of $L$ is finished, clients with receipts can verify that their transactions were processed by $\mathbb{E}$ (see details in \cite{homoliak2020aquareum}).  
The update procedure ensures that the new version of $L$ is:
(1) \textbf{internally correct} since it was executed by trusted code of $\mathbb{E}$, 
(2) a \textbf{consistent} extension of the previous version -- relying on trusted code of $\mathbb{E}$ and a witnessed version transition by $\mathbb{S}$, and 
(3) \textbf{non-equivocating} since $\mathbb{S}$ stores only hash of a single version of $L$ (i.e., $LHash$) at any point in time.  

Whenever $\mathbb{C}$ suspects that her transactions or read queries are censored, $\mathbb{C}$ might report censorship via $\mathbb{S}$. 
To do so, $\mathbb{C}$ encrypts her request with $PK_{\mathbb{E}}^{pb}$ and publishes it on the blockchain. 
$\mathbb{O}$ noticing a new request is obligated to pass the request to $\mathbb{E}$, which will process the request and reply with an encrypted response (by $PK_{\mathbb{C}}^{pb}$) that is processed by $\mathbb{S}$.  
If a pending request at $\mathbb{S}$ is not handled by $\mathbb{O}$, it is public evidence that $\mathbb{O}$ censors the request. 
We do not specify how can $\mathbb{C}$ use such a proof; it could be shown in a legal dispute or $\mathbb{S}$ itself could have an automated deposit-based punishments rules.

\subsubsection{Design Consideration}\label{sec:design-considerations}
We might design $L$ as an append-only chain (as in blockchains), but such a design would bring a high overhead on clients who want to verify that a particular block belongs to $L$.
During the verification, clients would have to download the headers of all blocks between the head of $L$ and the block in the query, resulting into linear space \& time complexity.
In contrast, when a history tree (see \autoref{sec:background-historyT}) is utilized for integrity preservation of $L$, the presence of any block in $L$ can be verified with logarithmic space and time complexity.

\subsubsection{Terminated and Failed Enclave}\label{sec:failed-enclave}
During the execution of $prog^\mathbb{E}$, $\mathbb{E}$ stores its secrets and state objects in a sealed file, which is updated and stored on the hard drive of $\mathbb{O}$ with each new block created.
Hence, if $\mathbb{E}$ terminates due a temporary reason, such as a power outage or intentional command by $\mathbb{O}$, it can be initialized again by $\mathbb{O}$ who provides $\mathbb{E}$ with the sealed file; this file is used to recover its protected state objects. 

However, if $\mathbb{E}$ experiences a permanent hardware failure of TEE, the sealed file cannot be decrypted on other TEE platforms.
Therefore, we propose a simple mechanism that deals with this situation under the assumption that $\mathbb{O}$ is the only allowed entity that can replace the platform of $\mathbb{E}$.
In detail, $\mathbb{O}$ first snapshots the header $hdr_{sync}$ of the last block that was synchronized with $\mathbb{S}$ as well as all blocks $blks_{unsync}$ of $L$ that were not synchronized with $\mathbb{S}$.
Then, $\mathbb{O}$ restores $L$ and her  internal state objects into the version $\#(LRoot_{pb})$.
After the restoration of $L$, $\mathbb{O}$ calls the function $ReInit()$ of $\mathbb{E}$ (see \autoref{alg:enclave-VM-failure}) with $hdr_{sync}$, $blks_{unsync}$, and $LRoot_{pb}$ as the arguments.
In this function, $\mathbb{E}$ first generates its public/private key-pair $SK_\mathbb{E}^{pb}, PK_\mathbb{E}^{pb}$, and then stores the passed header as $hdr_{last}$ and copies the passed root hash into $LRoot_{cur}$ and $LRoot_{pb}$.
Then, $\mathbb{E}$ iterates over all passed unprocessed blocks and their transactions $txs$, which are executed within VM of $\mathbb{E}$.
Before the processing of $txs$ of each passed block, $\mathbb{E}$ calls the unprotected code of $\mathbb{O}$ to obtain the current partial state $\partial st^{old}$ of $L$ and incremental proof template (see details in \cite{homoliak2020aquareum}) that serves for extending $L$ within $\mathbb{E}$.
However, these unprotected calls are always verified within $\mathbb{E}$ and malicious $\mathbb{O}$ cannot misuse them.
In detail, $\mathbb{E}$ verifies $\partial st^{old}$ obtained from $\mathbb{O}$ against the root hash of the state stored in the last header $hdr_{last}$ of $\mathbb{E}$, while the incremental proof template is also verified against $LRoot_{cur}$ in the function $newLRoot()$ of $\mathbb{E}$.

\begin{algorithm}[b] 
	
	\scriptsize
	\SetKwProg{func}{function}{}{}
	
	\func{$ReInit$($LRoot_{old}$, $prevBlks[], hdr_{last}$) \textbf{public}} {
		
		($SK_{\mathbb{E}}^{pb}$, $PK_{\mathbb{E}}^{pb}$)  $ \leftarrow\Sigma_{pb}.Keygen()$;\\	
		$hdr_{last} \leftarrow hdr_{last}$; \\						
		$LRoot_{cur}  \leftarrow LRoot_{old}, LRoot_{pb} \leftarrow LRoot_{old}$;\\			
		\smallskip
		
		\For{$\{ b: prevBlks \}$}{	
			$\pi^{inc}_{next}, ~LRoot_{tmp} \leftarrow prog^{\mathbb{O}}.nextIncProof()$; \\
			$\partial st^{old} \leftarrow prog^{\mathbb{O}}.getPartialState(b.txs)$;\\
			\textbf{assert} $\partial st^{old}.root = hdr_{last}.stRoot$; \\
			
			$\ldots \leftarrow ~processTxs(b.txs, \partial st^{old}, \pi^{inc}_{next}, ~LRoot_{tmp})$;\\
			
			$LRoot_{ret} \leftarrow prog^{\mathbb{O}}.runVM(b.txs)$; \Comment{Run VM at $\mathbb{O}$.} \\
			\textbf{assert} $LRoot_{cur} = LRoot_{ret}$; \Comment{$\mathbb{E}$ and $\mathbb{O}$  are at the same point.}\\
		}
		
		$\sigma \leftarrow \Sigma_{pb}.sign(SK_{\mathbb{E}}^{pb}, (LRoot_{pb}, LRoot_{cur}))$; \\
		
		\textbf{Output}($LRoot_{pb}, LRoot_{cur}, \sigma, PK^{\mathbb{E}}_{pb}, PK^{\mathbb{E}}_{tee}$); 	\\
		
	}					
	\smallskip

	\caption{Reinitialization of a failed $\mathbb{E}$ (part of $prog^{\mathbb{E}}$).}
	\label{alg:enclave-VM-failure}
	\vspace{-0.2cm}
\end{algorithm}

Next, $\mathbb{E}$ processes $txs$ of a block, extends $L$, and then it calls the unprotected code of $\mathbb{O}$ again, but this time to process $txs$ of the current block by $\mathbb{O}$, and thus getting the same version and state of $L$ in both $\mathbb{E}$ and $\mathbb{O}$. 
Note that any adversarial effect of this unprotected call is eliminated by the checks made after the former two unprotected calls.
When all passed blocks are processed, $\mathbb{E}$ signs the version transition pair $\langle LRoot_{pb}, LRoot_{cur} \rangle$ and returns it to $\mathbb{O}$, together with the new public keys of $\mathbb{E}$.
$\mathbb{O}$ creates a blockchain transaction that calls the function $ReplaceEnc()$ of $\mathbb{S}$ with data from $\mathbb{E}$ passed in the arguments.
In $ReplaceEnc()$, $\mathbb{S}$ first verifies whether the signature of the transaction was made by $\mathbb{O}$ to avoid MiTM attacks on this functionality.
Then, $\mathbb{S}$ calls its function $PostLRoot()$ with the signed version transition pair in the arguments.
Upon the success, the current root hash of $L$ is updated and $\mathbb{S}$ replaces the stored $\mathbb{E}$'s PKs by PKs passed in parameters.
Finally, $\mathbb{E}$ informs $\mathbb{C}s$ by an event containing new PKs of $\mathbb{E}$, and $\mathbb{C}$s perform the remote attestation of $prog^{\mathbb{E}}$ using the new key $PK^{\mathbb{E}}_{tee}$ and the attestation service.
We refer the reader to Appendix of \cite{homoliak2020aquareum} for the relevant pseudo-code of~$\mathbb{O}$.

\subsection{Implementation}
\label{sec:aquareum:implementation}

We have made a proof-of-concept implementation of \name, where we utilized Intel SGX and C++ for instantiation of $\mathbb{E}$, while $\mathbb{S}$ was built on top of Ethereum and Solidity.
Although \name can be integrated with various VMs running within $\mathbb{E}$, we selected EVM since it provides a Turing-complete execution environment and it is widely adopted in the community of decentralized applications.
In detail, we utilized OpenEnclave SDK \cite{open-enclave} and a minimalistic EVM, called eEVM \cite{eEVM-Microsoft}.
However, eEVM is designed with the standard C++ map for storing the full state of $L$, which lacks efficient integrity-oriented operations.
Moreover, eEVM assumes the unlimited size of $\mathbb{E}$ for storing the full state, while the size of $\mathbb{E}$ in SGX is constrained to $\sim$100 MB.
This might work with enabled swapping but the performance of $\mathbb{E}$ would be significantly deteriorated with a large full state.
Due to these limitations, we replaced eEVM's full state handling by Merkle-Patricia Trie from Aleth \cite{aleth}, which we customized to support operations with the partial state. 
$\mathbb{O}$ and $\mathbb{C}$ were also implemented in C++. 

Our implementation enables the creation and interaction of simple accounts as well as the deployment and execution of smart contracts written in Solidity. 
We verified the code of $\mathbb{S}$ by static/dynamic analysis tools Mythril \cite{mythrill}, Slither \cite{slither}, and ContractGuard \cite{ContractGuard-fuzzer}; none of them detected any vulnerabilities.
The source code of our implementation will be made available upon publication.

\subsection{Performance Evaluation}
All our experiments were performed on commodity laptop with Intel i7-10510U CPU supporting SGX v1, and they were aimed at reproducing realistic conditions -- i.e., they included all operations and verifications described in \cite{homoliak2020aquareum}, such as verification of recoverable ECDSA signatures, aggregation of transactions by Merkle tree, integrity verification of partial state, etc.
We evaluated the performance of \name in terms of transaction throughput per second, where we distinguished transactions with native payments (see \autoref{fig:performance-enc-payments}) and transactions with ERC20 smart contract calls (see \autoref{fig:performance-enc-erc}).
All measurements were repeated 100 times, and we depict the mean and standard deviation in the graphs.

\begin{figure}[t]
	\centering	
	
	\subfloat[\label{fig:pay-TB} Turbo Boost enabled]{
		\hspace{-0.2cm}
		\includegraphics[width=0.4\textwidth]{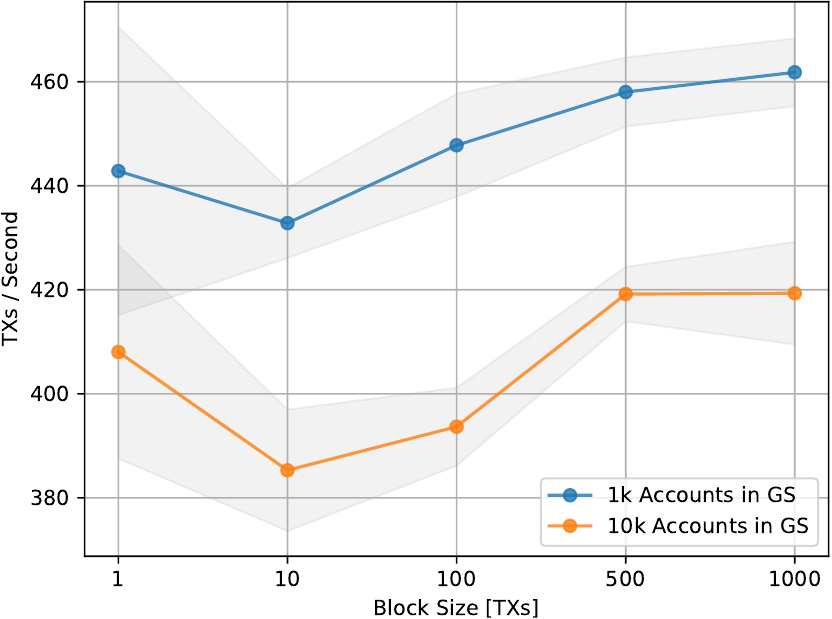} 
	}
	\subfloat[\label{fig:pay-noTB} Turbo Boost disabled]{
		\includegraphics[width=0.4\textwidth]{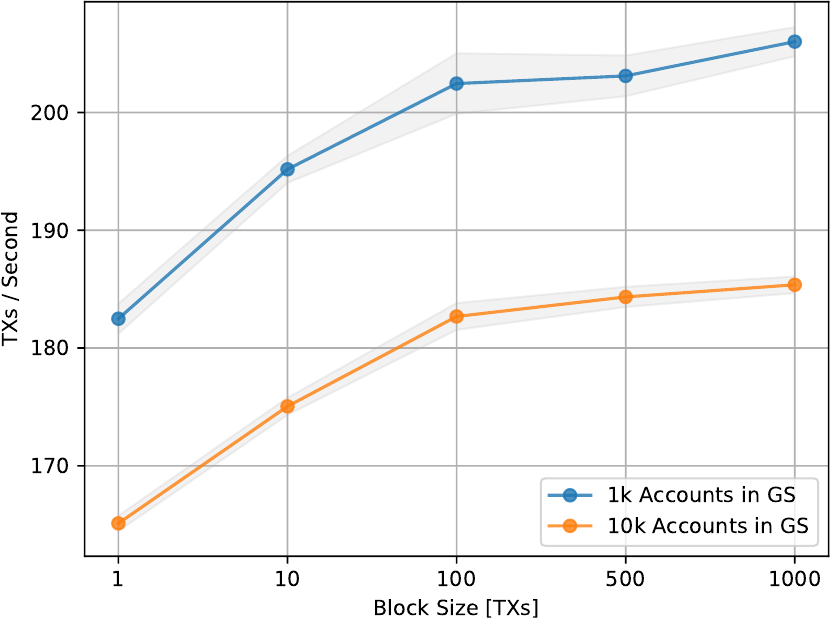} 
	}
	
	\vspace{0.2cm}
	\caption{Performance of \name for native payments.}
	\label{fig:performance-enc-payments}
\end{figure}
\begin{figure}[t]
	\centering	
	
	\subfloat[\label{fig:erc-TB} Turbo Boost enabled]{
		\hspace{-0.2cm}
		\includegraphics[width=0.4\textwidth]{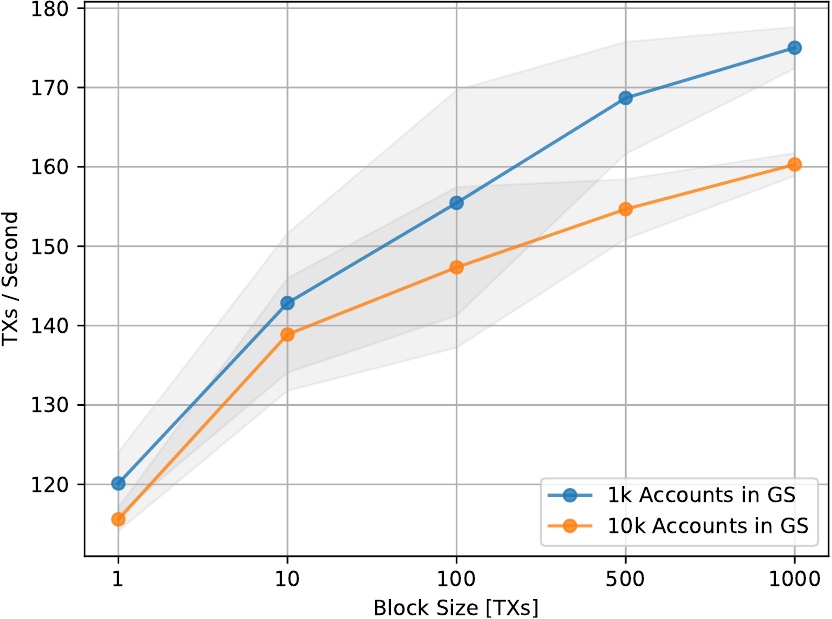} 
	}
	\subfloat[\label{fig:erc-noTB} Turbo Boost disabled]{
		\includegraphics[width=0.4\textwidth]{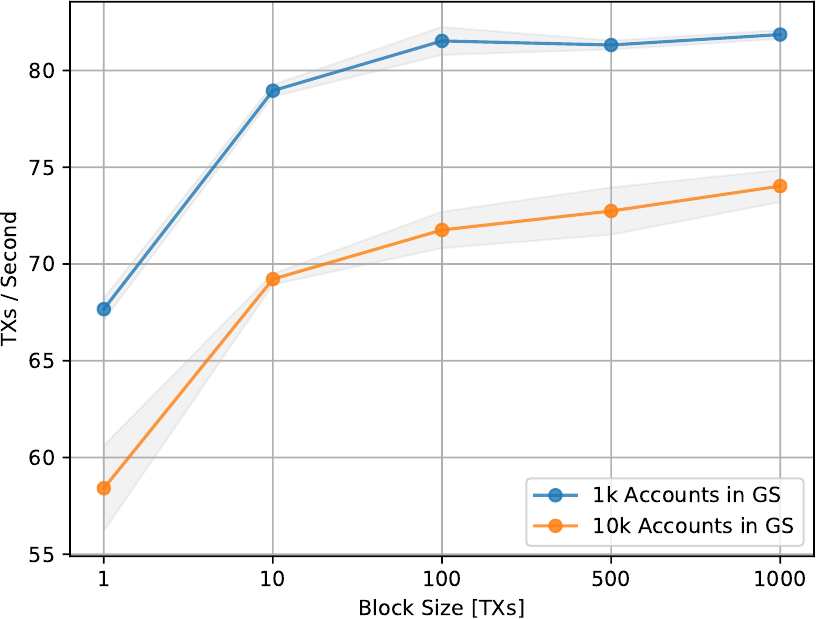} 
	}
	
	\vspace{0.2cm}
	\caption{Performance of \name for ERC20 smart contract calls.}
	\label{fig:performance-enc-erc}
\end{figure}

\paragraph{\textbf{A Size of the Full State.}}
The performance of \name is dependent on a size of data that is copied from $\mathbb{O}$ to $\mathbb{E}$ upon call of $Exec()$.
The most significant portion of the copied data is a partial state, which depends on the height of the MPT storing the full state. 
Therefore, we repeated our measurements with two different full states, one containing $1k$ accounts and another one containing $10k$ accounts.
In the case of native payments, the full state with 10k accounts caused a decrease of throughput by 7.8\%-12.1\% (with enabled TB) in contrast to the full state with 1k accounts.
In the case of smart contract calls, the performance deterioration was in the range 2.8\%-8.4\% (with enabled TB).

\begin{figure}[t]
	\centering	
	
	\subfloat[\label{fig:cens-tx-submit} Submit TX]{
		\hspace{-0.2cm}
		\includegraphics[width=0.415\textwidth]{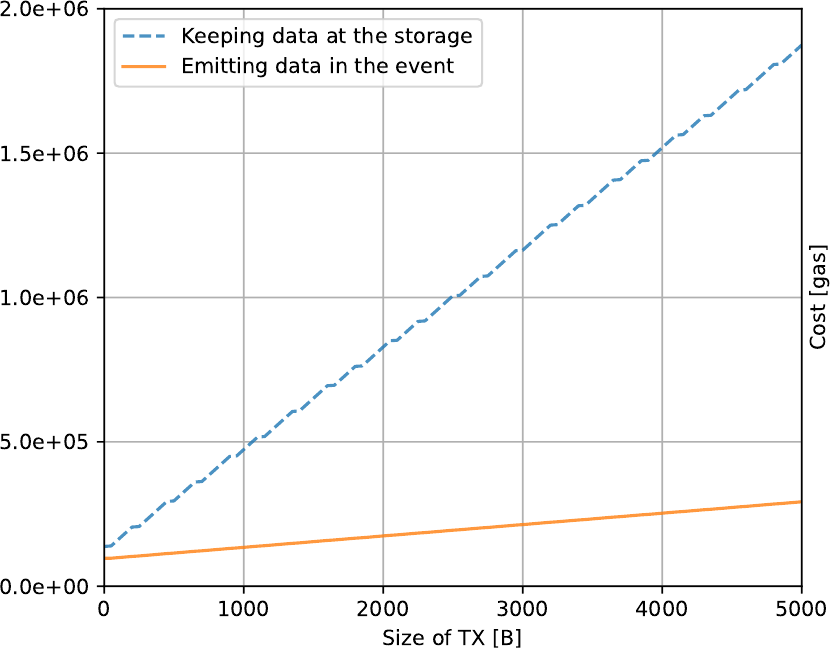} 
	}	
	\hspace{0.05cm}
	\subfloat[\label{fig:cens-tx-resolve} Resolve TX]{
		\hspace{-0.2cm}
		\includegraphics[width=0.4\textwidth]{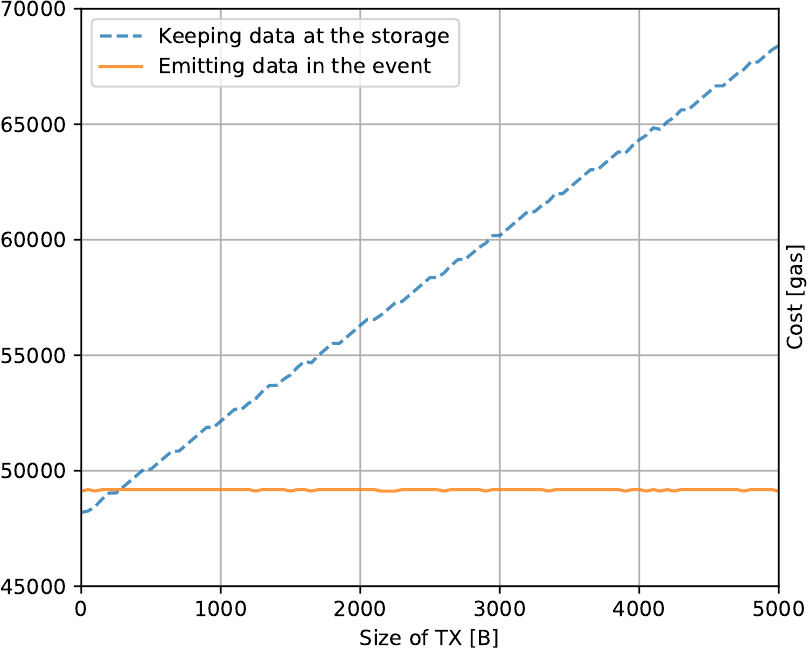} 
	}
	
	\vspace{0.4cm}
	\subfloat[\label{fig:cens-qry-submit} Submit Query (Get TX)]{
		\includegraphics[width=0.4\textwidth]{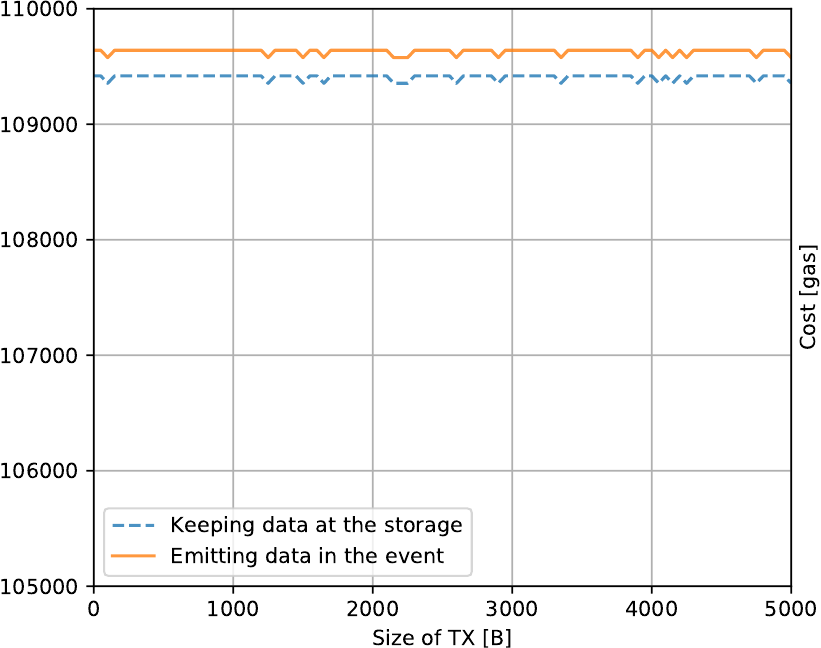} 
	}
	\subfloat[\label{fig:cens-qry-resolve} Resolve Query (Get TX)]{
		\includegraphics[width=0.4\textwidth]{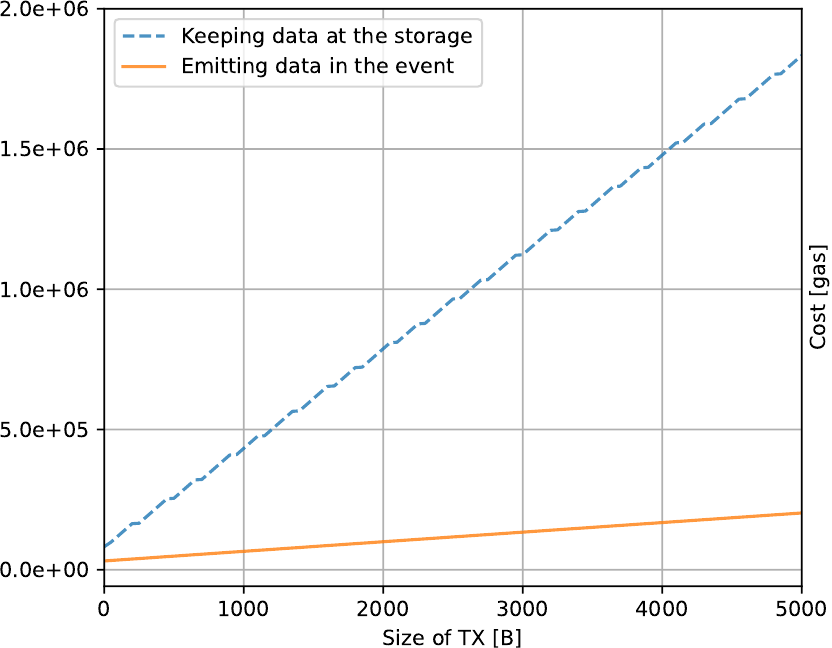} 
	}
	
	\vspace{0.2cm}
	\caption{Costs for resolution of censored transactions and queries.}
	\label{fig:-submit-cens-queries}
\end{figure}

\paragraph{\textbf{Block Size \& Turbo Boost.}}
In each experiment, we varied the block size in terms of the number of transactions in the block.
Initially, we performed measurements with enabled Turbo Boost (see \autoref{fig:pay-TB} and \autoref{fig:erc-TB}), where we witnessed a high throughput and its high variability. 
For smart contract calls (see \autoref{fig:erc-TB}), the throughput increased with the size of the block modified from 1 to 1000 by $45.7\%$ and $38.7\%$ for a full state with 1k and 10k accounts, respectively.  
However, in the case of native payments the improvement was only $4.3\%$ and $2.8\%$, while the throughput was not increased monotonically with the block size.
Therefore, we experimentally disabled Turbo Boost (see \autoref{fig:pay-noTB}) and observed the monotonic increase of throughput with increased block size, where the improvement achieved was $11.41\%$ and $12.26\%$ for a full state with 1k and 10k accounts, respectively.
For completeness, we also disabled Trubo Boost in the case of smart contract calls (see \autoref{fig:erc-noTB}), where the performance improvement was $20.9\%$ and $26.7\%$ for both full states under consideration.

\subsubsection{Analysis of Costs}\label{sec:cost-analysis}
Besides the operational cost resulting from running the centralized infrastructure,
\name imposes costs for interaction with the public blockchain with $\mathbb{S}$  deployed.
The deployment cost of $\mathbb{S}$ is $1.51M$ of gas and the cost of most frequent operation -- syncing $L$ with $\mathbb{S}$ (i.e., $PostLRoot()$) -- is $33k$ of gas, which is only $33\%$ higher than the cost of a standard Ethereum transaction.\footnote{This cost is low since we leverage the native signature scheme of the blockchain $\Sigma_{pb}$.}
For example, if $L$ is synced with $\mathbb{S}$ every 5 minutes, $\mathbb{O}$'s monthly expenses for this operation would be $285M$ of gas, while in the case of syncing every minute, monthly expenses would be $1,425M$ of gas.

\paragraph{\textbf{Censorship Resolution.}}\label{sec:exps-censorship}
Our mechanism for censorship resolution imposes costs on $\mathbb{C}$s submitting requests as well as for $\mathbb{O}$ resolving these requests.
The cost of submitting a censored request is mainly dependent on the size of the request/response and whether $\mathbb{S}$ keeps data of a request/response in the storage (i.e., an expensive option) or whether it just emits an asynchronous event with the data (i.e., a cheap option).
We measured the costs of both options and the results are depicted in \autoref{fig:-submit-cens-queries}.
Nevertheless, for practical usage, only the option with event emitting is feasible (see solid lines in \autoref{fig:-submit-cens-queries}).

\autoref{fig:cens-tx-submit} and \autoref{fig:cens-tx-resolve} depict the resolution of a censored transaction, which is more expensive for $\mathbb{C}$ than for $\mathbb{O}$, who resolves each censored transaction with constant cost $49$k of gas (see \autoref{fig:cens-tx-resolve}).
On the other hand, the resolution of censored queries is more expensive for $\mathbb{O}$ since she has to deliver a response with data to $\mathbb{S}$ (see \autoref{fig:cens-qry-resolve}), while $\mathbb{C}$ submits only a short query, e.g., get a transaction (see \autoref{fig:cens-qry-submit}).

\subsection{Security Analysis and Discussion}
\label{sec:aquareum:analysis}
In this section, we demonstrate resilience of \name against adversarial actions that the malicious operator $\mathcal{A}$ can perform to violate the desired properties (see \autoref{sec:desired-properties}). 

\begin{theorem}\label{theorem:correctness}
	(Correctness) $\mathcal{A}$ is unable to modify the full state of $L$ in a way that does not respect the semantics of VM deployed in $\mathbb{E}$.
\end{theorem}
\begin{proof1}
	The update of the $L$'s state is performed exclusively in $\mathbb{E}$.	
	Since $\mathbb{E}$ contains trusted code that is publicly known and remotely attested by $\mathbb{C}$s, $\mathcal{A}$ cannot tamper with this code.
\end{proof1}

\begin{theorem}\label{theorem:consistency}
	(Consistency) $\mathcal{A}$ is unable to extend $L$ and modify the past records of $L$.
\end{theorem}
\begin{proof1}
	All extensions of $L$ are performed within trusted code of $\mathbb{E}$, while utilizing the history tree \cite{crosby2009efficient} as a tamper evident data structure, which enables us to make only such incremental extensions of $L$ that are consistent with $L$'s past. 
\end{proof1}

\begin{theorem}
	(Verifiability) $\mathcal{A}$ is unable to unnoticeably modify or delete a transaction $tx$ that was previously inserted to $L$ using $\Pi_{N}$, if sync with $\mathbb{S}$ was executed anytime afterward.
\end{theorem}

\begin{proof1}
	Since $tx$ was correctly executed (\autoref{theorem:correctness}) as a part of the block $b_i$ in a trusted code of $\mathbb{E}$, $\mathbb{E}$ produced a signed version transition pair $\{h(L_{i-1}), h(L_i)\}_\mathbb{E}$ of $L$ from the version $i-1$ to the new version $i$ that corresponds to $L$ with $b_i$ included.
	$\mathcal{A}$ could either sync $L$ with $\mathbb{S}$ immediately after $b_i$ was appended or she could do it $n$ versions later.
	In the first case, $\mathcal{A}$ published $\{h(L_{i-1}), h(L_i)\}_\mathbb{E}$ to $\mathbb{S}$, which updated its current version of $L$ to $i$ by storing $h(L_i)$ into $LRoot_{pb}$. 
	In the second case, $n$ blocks were appended to $L$, obtaining its $(i+n)$th version. 
	$\mathbb{E}$ executed all transactions from versions $(i+1),\ldots, (i+n)$  of $L$, while preserving correctness (\autoref{theorem:correctness}) and consistency (\autoref{theorem:consistency}). 
	Then $\mathbb{E}$ generated a version transition pair $\{h(L_{i-1}), h(L_{i+n})\}_\mathbb{E}$ and $\mathcal{A}$  posted  it to $\mathbb{S}$, where 
	the current version of $L$ was updated to $i+n$ by storing $h(L_{i+n})$ into $LRoot_{pb}$.
	When any $\mathbb{C}$ requests $tx$ and its proofs from $\mathcal{A}$ with regard to publicly visible $LRoot_{pb}$, she might obtain a modified $tx'$ with a valid membership proof $\pi^{mem}_{hdr_i}$ of the block $b_i$ but an invalid Merkle proof  $\pi^{mk}_{tx'}$, which cannot be forged. \hfill$\Box$
	
	In the case of $tx$ deletion, $\mathcal{A}$ provides $\mathbb{C}$ with the tampered full block $b_i'$ (maliciously excluding $tx$) whose membership proof $\pi^{mem}_{hdr_i'}$ is invalid -- it cannot be forged. 	
\end{proof1}

\begin{theorem}\label{theorem:non-equivocation}
	(Non-Equivocation) Assuming $L$ synced with $\mathbb{S}$: $\mathcal{A}$ is unable to provide two distinct $\mathbb{C}$s with two distinct valid views on $L$. 
\end{theorem}
\begin{proof1}
	Since $L$ is regularly synced with publicly visible $\mathbb{S}$, and $\mathbb{S}$ stores only a single current version of $L$ (i.e., $LRoot_{pb}$), all $\mathbb{C}s$ share the same view on $L$.
\end{proof1}

\begin{theorem}\label{theorem:censorhip}
	(Censorship Evidence) $\mathcal{A}$ is unable to censor any request (transaction or query) from $\mathbb{C}$ while staying unnoticeable.
\end{theorem}
\begin{proof1}
	If $\mathbb{C}$'s request is censored, $\mathbb{C}$ asks for a resolution of the request through public $\mathbb{S}$.
	$\mathcal{A}$ observing the request might either ignore it and leave the proof of censoring at $\mathbb{S}$ or she might submit the request to $\mathbb{E}$ and obtain an enclave signed proof witnessing that a request was processed -- this proof is submitted to $\mathbb{S}$, whereby publicly resolving the request.
\end{proof1}

\subsubsection{Other Properties and Implications}

\paragraph{\textbf{Privacy VS Performance.}}
\name provides privacy of data submitted to $\mathbb{S}$ during the censorship resolution since the requests and responses are encrypted.
However, \name does not provide privacy against $\mathbb{O}$ who has the read access to $L$.
Although \name could be designed with the support of full privacy, a disadvantage of such an approach would be the performance drop caused by the decryption of  requested data from $L$ upon every $\mathbb{C}$'s read query, requiring a call of $\mathbb{E}$.
In contrast, with partial-privacy, $\mathbb{O}$ is able to respond queries of $\mathbb{C}$s without touching $\mathbb{E}$.

\paragraph{\textbf{Access Control at $\mathbb{S}$.}}\label{sec:cens-write-access}
$\mathbb{C}$s interact with $\mathbb{S}$ only through functions for submission of censored requests.
Nevertheless, access to these functions must be regulated through an access control mechanism in order to avoid exhaustion (i.e., DoS) of this functionality by external entities. 
This can be performed with a simple access control requiring $\mathbb{C}$s to provide access tickets when calling the functions of $\mathbb{S}$.
An access ticket could be provisioned by $\mathbb{C}$ upon registration at $\mathbb{O}$, and it could contain $PK_{\mathbb{C}}^{pb}$ with a time expiration of the subscription, signed by $\mathbb{E}$.
Whenever $\mathbb{C}$ initiates a censored request, verification of an access ticket would be made by $\mathbb{S}$, due to which DoS of this functionality would not be possible.

\paragraph{\textbf{Security of TEE.}}
\name assumes that its TEE platform is secure.
However, previous research showed that this might not be the case in practical implementations of TEE, such as SGX that was vulnerable to memory corruption attacks \cite{biondo2018guard} as well as side channel attacks \cite{brasser2017dr,van2018foreshadow,Lipp2021Platypus,Murdock2019plundervolt}.
A number of software-based defense and mitigation techniques have been  proposed \cite{shih2017t,gruss2017strong,chen2017detecting,brasser2017dr,seo2017sgx} and some vulnerabilities were patched by Intel at the hardware level \cite{intel-sgx-response}.
Nevertheless, we note that \name is TEE-agnostic thus can be integrated with other TEEs such as ARM TrustZone or RISC-V architectures (using Keystone-enclave \cite{Keystone-enclave} or Sanctum \cite{costan2016sanctum}).

Another class of SGX vulnerabilities was presented by Cloosters et al. \cite{cloosters2020teerex} and involved incorrect application designs enabling arbitrary reads and writes of protected memory or work done by Borrello et al. which involves more serious microarchitectural flaws in chip design \cite{Borrello2022AEPIC}. 
%
Since the authors did not provide public with their tool (and moreover it does not support Open-enclave SDK), we did manual inspection of \name code and did not find any of the concerned vulnerabilities. 

\paragraph{\textbf{Security vs. Performance.}}
In SGX, performance is always traded for security, and our intention was to optimize the performance of \name while making custom security checks whenever possible instead of using expensive buffer allocation and copying to/from $\mathbb{E}$ by trusted runtime of SDK (\textit{trts}).
\begin{compactitem}
	\item \textbf{Output Parameters:}
	In detail, in the case of ECALL $Exec()$ function where $\mathbb{E}$ is provided with \texttt{[user\_check]} output buffers pointing to host memory, the strict location-checking is always made in $\mathbb{E}$ while assuming maximal size of the output buffer passed from the host (i.e., $oe\_is\_outside\_enclave(buf,~max)$).
	Moreover, the maximal size is always checked before any write to such output buffers.
	The concerned parameters of $Exec()$ are buffers for newly created and modified account state objects.
	\item \textbf{Input Parameters:}
	On the other hand, in the case of input parameters of $Exec()$, we utilize embedded buffering provided by trts of SDK since $\mathbb{E}$ has to check the integrity of input parameters before using them, otherwise Time-of-Check != Time-of-Use vulnerability \cite{cloosters2020teerex} might be possible.
	The concerned input parameters of $Exec()$ are transactions to process and their corresponding codes.
\end{compactitem}

\paragraph{\textbf{Time to Finality.}}
Many blockchain platforms suffer from accidental forks, which temporarily create parallel inconsistent blockchain views.
To mitigate this phenomenon, it is recommended to wait a certain number of block confirmations after a given block is created before considering it irreversible with overwhelming probability.
This waiting time (a.k.a., time to finality) influences the non-equivocation property of \name, and \name inherits it from the underlying blockchain platform.
Most blockchains have a long time to finality, e.g., $\sim$10mins in Bitcoin \cite{nakamoto2008bitcoin}, $\sim$3mins in Ethereum \cite{wood2014ethereum}, $\sim$2mins in Cardano \cite{kiayias2017ouroboros}.
However, some blockchains have a short time to finality, e.g., HoneyBadgerBFT \cite{miller2016honey}, Algorand \cite{gilad2017algorand}, and StrongChain \cite{strongchain}. 
The selection of the underlying blockchain platform is dependent on the requirements of the particular use case that \name is applied for.

\section{CBDC-AquaSphere}\label{sec:logging-cbdc}
For background related to blockchains, integrity-preserving data structures, atomic swap, and CBDC, we refer the reader to \autoref{chapter:background}.

\subsection{Problem Definition}\label{sec:cbdc:problem}
Our goal is to propose a CBDC approach that respects the features proposed in DEA manifesto~\cite{cbdc-manifesto} released in 2022, while on top of it, we assume other features that might bring more benefits and guarantees.
First, we start with a specification of the desired features related to a single instance of CBDC that we assume is operated by a single entity (further a bank or its operator) that maintains its ledger. 
Later, we describe desired features related to multiple instances of CBDC that co-exist in the ecosystem of wholesale and/or retail CBDC.\footnote{Note that we will propose two deployment scenarios, one for the wholesale environment and the second one for the retail environment of multiple retail banks interacting with a single central bank.}
In both cases, we assume that a central bank might not be a trusted entity.
All features that respect this assumption are marked with asterisk $^*$ and are considered as requirements for such an attacker model.

\subsubsection{Single Instance of CBDC}\label{sec:problem-features-single}
When assuming a basic building block of CBDC -- a single bank's CBDC working in an isolated environment from the other banks -- we specify the desired features of CBDC as follows:
\begin{compactenum}
	

	\item[\textbf{Correctness of Operation Execution$^*$:}] 
	The clients who are involved in a monetary operation (such as a transfer) should be guaranteed with a correct execution of their operation.
	
	\item[\textbf{Integrity$^*$:}] 
	The effect of all executed operations made over the client accounts should be irreversible, and no ``quiet'' tampering of the data by a bank should be possible.
	Also, no conflicting transactions can be (executed and) stored by the CBDC instance in its ledger. 
	
	\item[\textbf{Verifiability$^*$:}] 
	This feature extends integrity and enables the clients of CBDC to obtain easily verifiable evidence that the ledger they interact with is internally correct and consistent.         
	In particular, it means that none of the previously inserted transactions was neither modified nor deleted. 
	
	\item[\textbf{Non-Equivocation$^*$:}]
	From the perspective of the client's security, the bank should not be able to present at least two inconsistent views on its ledger to (at least) two distinct clients who would accept such views as valid.
	
	\item[\textbf{Censorship Evidence$^*$:}]
	The bank should not be able to censor a client's request without leaving any public audit trails proving the censorship occurrence.
	
	\item[\textbf{Transparent Token Issuance$^*$:}] Every CBDC-issued token should be publicly visible (and thus audit-able) to ensure that a bank is not secretly creating token value ``out-of-nothing,'' and thus causing uncontrolled inflation.
	The transparency also holds for burning of existing tokens. 
	
	
	\item[\textbf{High Performance:}] A CBDC instance should be capable of processing a huge number of transactions per second since it is intended for daily usage by thousands to millions of people.
	
	\item[\textbf{Privacy:}]
	All transfers between clients as well as information about the clients of CBDC should remain private for the public and all other clients that are not involved in particular transfers.
	However, a bank can access this kind of information and potentially provide it to legal bodies, if requested.
\end{compactenum}

\subsubsection{Multiple Instances of CBDC}
In the case of  multiple CBDC instances that can co-exist in a common environment, we extend the features described in the previous listing by features that are all requirements:
\begin{compactenum}
	
	\item[\textbf{Interoperability$^*$:}] 
	As a necessary prerequisite for co-existence of multiple CBDC instances, we require them to be mutually interoperable, which means that tokens issued by one bank can be transferred to any other bank.
	For simplicity, we assume that all the CBDC instances are using the unit token of the same value within its ecosystem.\footnote{On the other hand, conversions of disparate CBDC-backed tokens would be possible by following trusted oracles or oracle networks.}
	At the hearth of interoperability lies atomicity of supported operations.
	Atomic interoperability, however, requires means for accountable coping with censorship and recovery from stalling. 
	We specify these features in the following.
	
	\begin{compactenum}
		\item[\textbf{Atomicity$^*$:}] 
		Any operation (e.g., transfer) between two interoperable CBDC instances must be either executed completely or not executed at all. 
		As a consequence,
		no new tokens can be created out-of-nothing and no tokens can be lost in an inter-bank operation.
		Note that even if this would be possible, the state of both involved instances of CBDC would remain internally consistent; therefore, consistency of particular instances (\autoref{sec:problem-features-single}) is not a sufficient feature to ensure atomicity within multiple interoperable CBDC instances.
		This requirement is especially important due to trustless assumption about particular banks, who might act in their benefits even for the cost of imposing the extreme inflation to the whole system.\footnote{For example, if atomicity is not enforced, one bank might send the tokens to another bank, while not decreasing its supply due to pretended operation abortion.}
		
		\item[\textbf{Inter-CBDC Censorship Evidence$^*$:}] 
		Having multiple instances of CBDC enables a different way of censorship, where one CBDC (and its clients) might be censored within some inter-CBDC operation with another CBDC instance, precluding them to finish the operation.
		Therefore, there should exist a means how to accountably detect this kind of censorship as well.
		
		\item[\textbf{Inter-CBDC Censorship Recovery$^*$:}] 
		If the permanent censorship happens and is indisputably proven, it must not impact other instances of CBDC, including the ones that the inter-CBDC operations are undergoing.
		Therefore, the interoperable CBDC environment should provide a means to recover from inter-CBDC censorship of unfinished operations.
	\end{compactenum}

	\item[\textbf{Identity Management of CBDC Instances$^*$:}] 
	Since we assume that CBDC in\-stan\-ces are trustless, in theory, there might emerge a fake CBDC instance, pretentding to act as a valid one.
	To avoid this kind of situation, it is important for the ecosystem of wholesale CBDC to manage identities of particular valid CBDC instances in a secure manner.
\end{compactenum}

\subsubsection{Adversary Model}
The attacker can be represented by the operator of a bank or the client of a bank, and her intention is to break functionalities that are provided by the features described above. 
%
Next, we assume that the adversary cannot undermine the cryptographic primitives used, the blockchain platform, and the TEE platform deployed. 

\subsection{Proposed Approach}\label{sec:cbdc:design}



\begin{figure}[b]
	\centering
		\vspace{-0.3cm}
	\includegraphics[width=0.55\columnwidth]{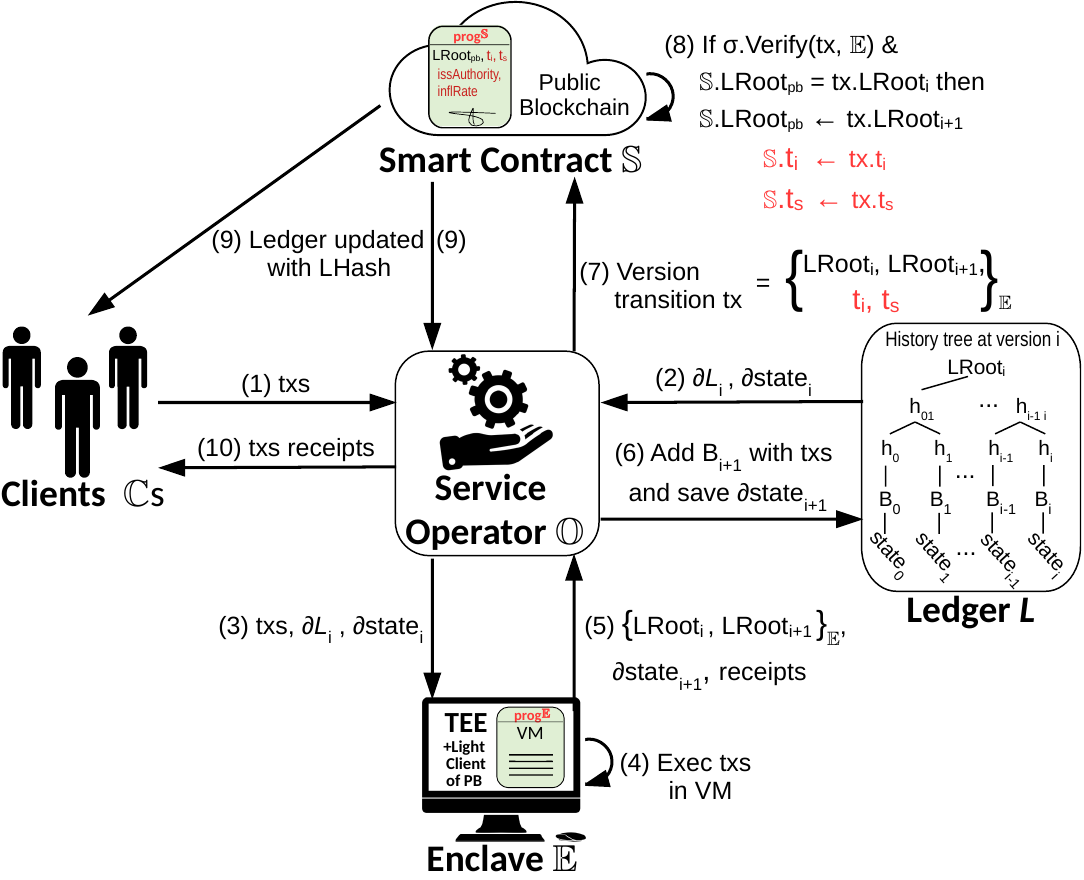}
	\caption{Architecture of Aquareum with our modifications in red.}
	\label{fig:architecture-aquareum}
	\vspace{-0.3cm}
\end{figure}

We propose a holistic approach for the ecosystem of wholesale and/or retail CBDC, which aims at meeting the features described in \autoref{sec:cbdc:problem}. 
To accomplish these features, we leverage interesting properties stemming from a combination of a public blockchain (with smart contract platform) and TEE. 
Such a combination was proposed for various purposes in related work, out of which the use case of generic centralized ledger Aquareum~\cite{homoliak2020aquareum} is most convenient to build on.
Therefore, we utilize Aquareum as a building block for a single instance of CBDC, 
and we make a few CBDC-specific modifications to it, enhancing its transparency and functionality.
Our modifications are outlined in \autoref{fig:architecture-aquareum} by red color, while the details of them (especially changes in programs of smart contract and enclave) will be described in this section.
First, we start by a description of a single CBDC instance and then we extend it to a fully interoperable environment consisting of multiple CBDC instances.

Note that we focus solely on the transfer of tokens operation within the context of CBDC interoperability.
However, our approach could be extended to different operations, involving inter-CBDC smart contract invocations.
Also, note that to distinguish between smart contracts on a public blockchains and smart contracts running in TEE, we will denote latter as \textbf{micro contracts} (or $\mu$-contracts).
Similarly, we denote transactions sent to TEE as \textbf{micro transactions} (or $\mu$-transactions) and blocks created in the ledger of CBDC instance as \textbf{micro blocks} (or $\mu$-blocks).

\subsubsection{\textbf{A CBDC Instance}}
Alike in Aquareum, the primary entity of each CBDC instance is its operator $\mathbb{O}$ (i.e., a bank), who is responsible for (1) maintaining the ledger $L$, 
(2) running the TEE enclave $\mathbb{E}$, 
(3) synchronization of the $L$'s snapshot to a public blockchain with smart contract $\mathbb{IPSC}$ ($\mathbb{I}$ntegrity $\mathbb{P}$reserving $\mathbb{S}$mart $\mathbb{C}$ontract), 
(4) resolving censorship requests, and 
(5) a communication with clients $\mathbb{C}$s.

\paragraph{\textbf{Token Issuance.}}
On top of Aquareum's $\mathbb{S}$, our $\mathbb{IPSC}$ contains snapshotting of the total issued tokens $t_i$ by the current CBDC instance and the total supply $t_s$ available at the instance for the purpose of transparency in token issuance (and potentially even burning).
Therefore, we extend the $\mathbb{E}$-signed version transition pair periodically submitted to $\mathbb{IPSC}$ by these two fields that are relayed to $\mathbb{IPSC}$ upon snapshotting $L$ (see red text in \autoref{fig:architecture-aquareum}).
Notice that $t_i = t_s$ in the case of a single instance since the environment of the instance is isolated.
\begin{compactitem}
	\item {\textbf{An Inflation Bound.}}
	Although snapshotting the total tokens in circulation is useful for the transparency of token issuance, $\mathbb{O}$ might still hyper-inflate the CBDC instance.
	Therefore, we require $\mathbb{O}$ to guarantee a maximal inflation rate $i_r$ per year, which can be enforced by $\mathbb{IPSC}$ as well as $\mathbb{E}$ since the code of both is publicly visible and attestable.
	The $i_r$ should be adjusted to a constant value by $\mathbb{O}$ at the initialization of $\mathbb{IPSC}$ and verified every time the new version of $L$ is posted to $\mathbb{IPSC}$; 
	in the case of not meeting the constrain, the new version would not be accepted at $\mathbb{IPSC}$.
	However, another possible option is that the majority vote of $\mathbb{C}$s can change $i_r$ even after initialization.
	Besides, $\mathbb{E}$ also enforces $i_r$ on $t_i$ and does not allow $\mathbb{O}$ to issue yearly more tokens than defined by $i_r$. 
	Nevertheless, we put the inflation rate logic also into $\mathbb{IPSC}$ for the purpose of transparency.
	
\end{compactitem}

\paragraph{\textbf{Initialization.}}
First, $\mathbb{E}$ with program $prog^{\mathbb{E}}$ (see  \autoref{alg:enclave-VM} of Appendix) generates and stores two key pairs, one under $\Sigma_{pb}$ (i.e., $SK_{\mathbb{E}}^{pb}$, $PK_{\mathbb{E}}^{pb}$) and one under $\Sigma_{tee}$ (i.e.,  $SK_{\mathbb{E}}^{tee}$, $PK_{\mathbb{E}}^{tee}$).
Then, $\mathbb{O}$ generates one key pair under $\Sigma_{pb}$ (i.e.,  $SK_{\mathbb{O}}^{pb}$, $PK_{\mathbb{O}}^{pb}$), which is then  used as the sender of a transaction deploying $\mathbb{IPSC}$ with program $prog^{\mathbb{IPSC}}$ (see \autoref{alg:IPSC-smart-contract} of Appendix) at public blockchain with parameters $PK_{\mathbb{E}}^{pb}$, $PK_{\mathbb{E}}^{tee}$, $PK_{\mathbb{O}}^{pb}$, $t_i$, and $i_r$.
Then, $\mathbb{IPSC}$ stores the keys in parameters, sets the initial version of $L$ by putting $LRoot_{pb} \leftarrow ~\perp$, and sets the initial total issued tokens and the total supply, both to $t_i$.\footnote{Among these parameters, a constructor of $\mathbb{IPSC}$ also accepts the indication whether an instance is allowed to issue tokens. This is, however, implicit for the single instance, while restrictions are reasonable in the case of multiple instances.}

\begin{compactitem}
\item {\textbf{Client Registration.}}
A client $\mathbb{C}$ registers with $\mathbb{O}$, who performs know your customer (KYC) checks and submits her public key $PK_{pb}^{\mathbb{C}}$ to $\mathbb{E}$.
Then, $\mathbb{E}$ outputs an execution receipt about the successful registration of $\mathbb{C}$ as well as her access ticket $t^{\mathbb{C}}$ that will serve for potential communication with $\mathbb{IPSC}$ and its purpose is to avoid spamming $\mathbb{IPSC}$ by invalid requests. 
In detail, $t^{\mathbb{C}}$ is the $\mathbb{E}$-signed tuple that contains $PK_{pb}^{\mathbb{C}}$ and optionally other fields such as the account expiration timestamp. 
Next, $\mathbb{C}$ verifies whether her registration (proved by the receipt) was already snapshotted by $\mathbb{O}$ at $\mathbb{IPSC}$.
\end{compactitem}

\begin{algorithm} 
	\scriptsize 
	\caption{The program $prog^{\mathbb{IPSC}}$ of $\mathbb{IPSC}$ with our modifications in red (as opposed to Aquareum).}\label{alg:IPSC-smart-contract}
	
	\SetKwProg{func}{function}{}{}
	
	\smallskip
	$\triangleright$ \textsc{Declaration of types and constants:}\\		
	\hspace{1em} \textbf{CensInfo} \{ $\mu$-$etx, \mu$-$equery, status, edata$ \},  \\
	\hspace{1em} $msg$: a current transaction that called $\mathbb{IPSC}$,  \\
	
	\smallskip
	$\triangleright$ \textsc{Declaration of functions:}
	
	\func{$Init$($PK_{\mathbb{E}}^{pb}, PK_{\mathbb{E}}^{tee}, PK_{\mathbb{O}}, ~\diff{\_i_r},~ \diff{[ia \leftarrow \textbf{T}]} $) \textbf{public} }{
		$PK_{\mathbb{E}}^{tee}[].add(PK_{\mathbb{E}}^{tee})$; \Comment{PK of enclave $\mathbb{E}$ under $\Sigma_{tee}$.} \\ 
		$PK_{\mathbb{E}}^{pb}[].add(PK_{\mathbb{E}}^{pb})$; \Comment{PK of enclave $\mathbb{E}$ under $\Sigma_{pb}$.} \\
		$PK_{\mathbb{O}}^{pb} \leftarrow PK_{\mathbb{O}}$; \Comment{PK of operator $\mathbb{O}$ under $\Sigma_{pb}$.} \\
		$LRoot_{pb} \leftarrow \perp$; \Comment{The most recent root hash of $L$ synchronized with $\mathbb{IPSC}$.} \\ 
		$censReqs \leftarrow []$; \Comment{Request that $\mathbb{C}$s wants to resolve publicly.} \\
		\diff{$t_s \leftarrow 0$}; \Comment{The total supply of the instance.}\\
		\diff{$t_i \leftarrow 0$}; \Comment{The total issued tokens by the instance.}\\
		\diff{\textbf{const} $issueAuthority \leftarrow ia$}; \Comment{Token issuance  capability of the instance.} \\ 
		\diff{\textbf{const} $i_r \leftarrow \_i_r$}; \Comment{Max. yearly inflation of the instance.} \\ 
		\diff{\textbf{const} $createdAt \leftarrow timestamp()$}; \Comment{The timestamp of creation a CBDC instance.} \\ 
	}
	
	\func{$snapshotLedger$($root_A, root_B, \diff{\_t_i, \_t_s,}~ \sigma$) \textbf{public} }{
		\Comment{Verify whether msg was signed by $\mathbb{E}$. \hfill} \\
		\textbf{assert} $\Sigma_{pb}.verify((\sigma, PK_{\mathbb{E}}^{pb}[\text{-}1]), (root_A, root_B, \diff{\_t_i, \_t_s}))$;  \\			
		
		\Comment{Snapshot issued tokens and total supply. \hfill} \\
		\diff{
			\If{$issueAuthority$}{
				\textbf{assert} $\_\_meetsInflationRate(\_t_i)$;  \Comment{The code is trivial, and we omit it.} \\
				$t_i \leftarrow \_t_i$;			\\
			}
			\Else{
				\diff{\textbf{assert} $t_i = \_t_i$}; \\ 
			}
		}
		
		\Comment{Verify whether a version transition extends the last one. \hfill} \\
		\If{$LRoot_{pb} = root_A$}{
			$LRoot_{pb} \leftarrow root_{B}$; \Comment{Do a version transition of $L$.} \\
		}	
		
	}
		
	\func{$SubmitCensTx$($\mu$-$etx, \sigma_{msg}$) \textbf{public} }{
		\Comment{Called by $\mathbb{C}$ in the case her $\mu$-tx is censored. $\mathbb{C}$ encrypts it by $PK^{tee}_{\mathbb{E}}$. \hfill \hfill}\\		
		accessControl($\sigma_{msg}, msg.PK_{\mathbb{C}}^{pb}$); \\
		
		$censReqs$.add(\textbf{CensInfo}($\mu$-$etx, \perp, \perp, \perp$)); \\			
		
	}
	
	\func{$ResolveCensTx(idx_{req}, status, \sigma$) \textbf{public} }{
		\Comment{Called by $\mathbb{O}$ to prove that $\mathbb{C}$'s $\mu$-tx was processed.\hfill \hfill \hfill}\\
		
		\textbf{assert} $idx_{req} < |censReqs|$;\\
		$r \leftarrow censReqs[idx_{req}]$; \\
		
		\textbf{assert} $\Sigma_{pb}.verify((\sigma, PK_{\mathbb{E}}^{pb}[\text{-}1]), ~(h(r.\mu$-$etx), status))$; 	\\ 
		$r.status \leftarrow status$;\\
	}	
	
	\func{$SubmitCensQry$($\mu$-$equery, \sigma_{msg}$) \textbf{public} }{
		\Comment{Called by $\mathbb{C}$ in the case its read query is censored. $\mathbb{C}$ encrypts it by $PK^{tee}_{\mathbb{E}}$.}\\		
		accessControl($msg$, $\sigma_{msg}, msg.PK_{\mathbb{C}}^{pb}$); \\
		
		$censReqs$.add(\textbf{CensInfo}($\perp, \mu$-$equery, \perp, \perp$)); \\		
		
	}
	
	\func{$ResolveCensQry(idx_{req}, status, edata, \sigma$) \textbf{public} }{
		\Comment{Called by $\mathbb{O}$ as a response to the $\mathbb{C}$'s censored  read query.\hfill \hfill \hfill}\\		
		\textbf{assert} $idx_{req} < |censReqs|$;\\
		$r \leftarrow censReqs[idx_{req}]$; \\
		
		\textbf{assert} $\Sigma_{pb}.verify((\sigma, PK_{\mathbb{E}}^{pb}[\text{-}1]), (h(r.\mu$-$equery), status, h(edata)))$; 	\\ 
		$r.\{edata \leftarrow edata, status \leftarrow status\}$;\\		
	}
	
	\medskip
	
	\func{$ReplaceEnc$($PKN_{\mathbb{E}}^{pb}, PKN_{\mathbb{E}}^{tee}, r_{A}, r_{B}, ~\diff{\_t_i, \_t_s,}~  \sigma, \sigma_{msg}$) \textbf{public} }{
		\Comment{Called	 by $\mathbb{O}$ in the case of enclave failure.\hfill}\\
		
		\textbf{assert} $\Sigma_{pb}.verify((\sigma_{msg}, PK_{\mathbb{O}}^{pb}), msg)$;  \Comment{Avoiding MiTM attack.} \\
		$snapshotLedger(r_{A}, r_{B}, ~\diff{\_t_i, \_t_s},~ \sigma)$ ; \Comment{Do a version transition.} \\

		$PK_{\mathbb{E}}^{tee}.add(PKN_{\mathbb{E}}^{tee})$; \Comment{Upon change, $\mathbb{C}s$ make remote attestation.} \\ 
		$PK_{\mathbb{E}}^{pb}.add(PKN_{\mathbb{E}}^{pb})$; \\	
	}
\end{algorithm}

\paragraph{\textbf{Normal Operation.}}
$\mathbb{C}$s send \textit{$\mu$-transactions} (writing to $L$) and \textit{queries} (reading from $L$) to $\mathbb{O}$, who validates them and relays them to $\mathbb{E}$, which processes them within its virtual machine (Aquareum uses eEVM~\cite{eEVM-Microsoft}).
Therefore, $L$ and its state are modified in a trusted code of $\mathbb{E}$, creating a new version of $L$, which is represented by the root hash $LRoot$ of the history tree.
Note that program $prog^{\mathbb{E}}$ is public and can be remotely attested by $\mathbb{C}$s (or anybody).
$\mathbb{O}$ is responsible for a periodic synchronization of the most recent root hash $LRoot_{cur}$ (i.e., snapshotting the current version of $L$ ) to $\mathbb{IPSC}$, running on a public blockchain $PB$. 
Besides, $\mathbb{C}$s use this smart contract to resolve censored transactions and queries, while preserving the privacy of data. 

\begin{figure*}[t]
	\centering
	\includegraphics[width=0.85\textwidth]{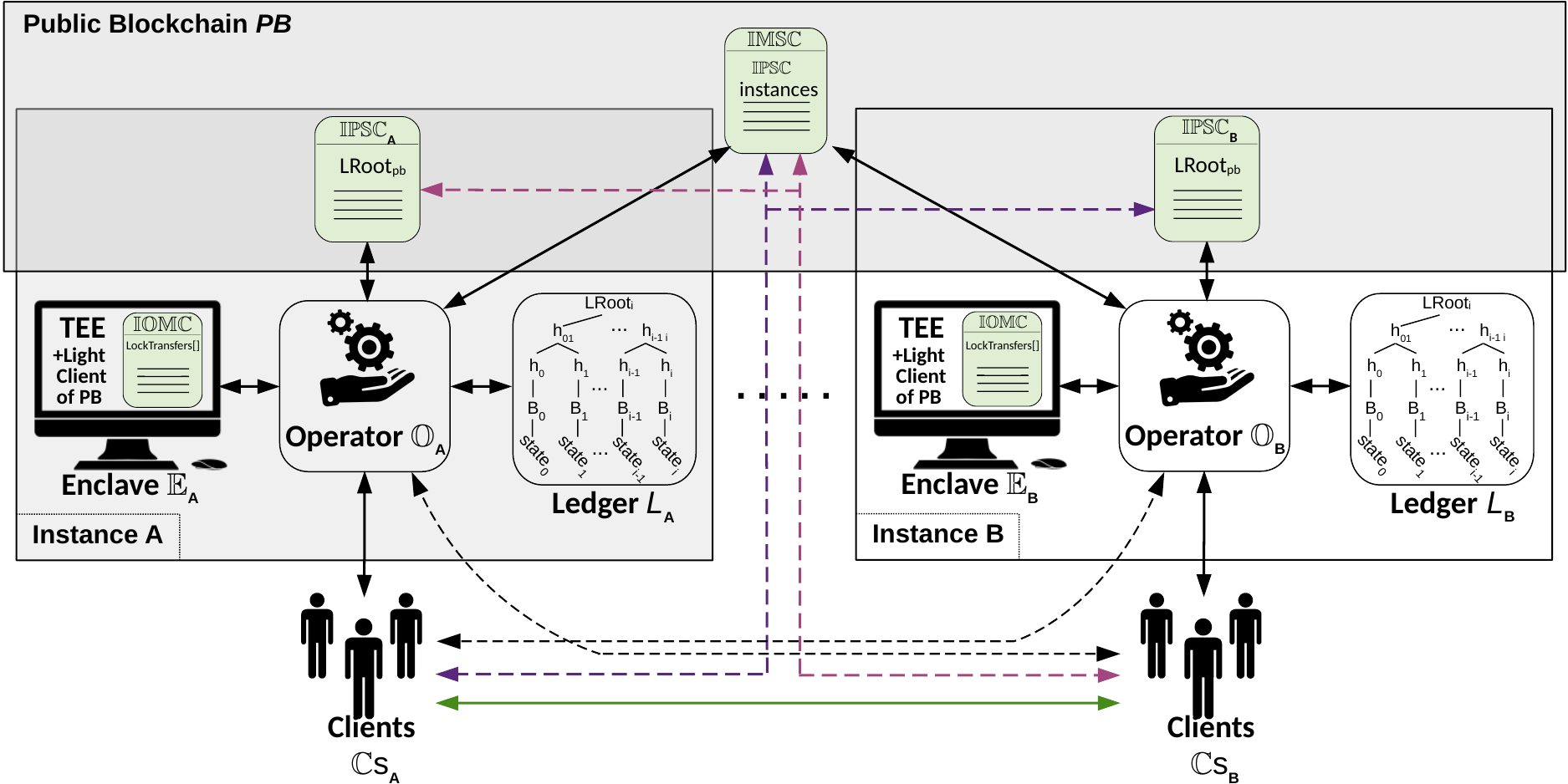}
	\caption{Overview of our CBDC architecture supporting interoperability among multiple CBDC instances (i.e., banks). The schema depicts two instances, where each of them has its own centralized ledger $L$ modified in a secure way through TEE of $\mathbb{E}$, while its integrity is ensured by periodic integrity snapshots to the integrity preserving smart contract ($\mathbb{IPSC}$) in a public blockchain $PB$. Each CBDC instance is registered in the identity management smart contract $\mathbb{IMSC}$ of a public blockchain, serving as a global registry of bank instances. A client who makes an inter-bank transfer communicates with her bank and the counter-party bank utilizing interoperability micro contracts ($\mathbb{IOMC}$), running in the TEE. Any censored request of a client is resolved by $\mathbb{IPSC}$ of a particular bank and can be initiated by its client or a counter-party client. 
	}
	\label{fig:architecture}
\end{figure*}

\paragraph{\textbf{Censorship Resolution.}}\label{sec:approch-single-instance-cens}
$\mathbb{O}$ might potentially censor some \textit{write} transactions or \textit{read} queries of $\mathbb{C}s$.
However, these can be resolved by Aquareum's mechanism as follows. 
%
If $\mathbb{C}$'s $\mu$-transaction $\mu$-tx is censored by $\mathbb{O}$, $\mathbb{C}$ first creates $PK_{\mathbb{E}}^{tee}$-encrypted $\mu$-$etx$  (to ensure privacy in $PB$), and then she creates and signs a transaction containing $\mathbb{C}'s$ access ticket $t^{\mathbb{C}}$ and $\mu$-$etx$.
$\mathbb{C}$ sends this transaction to $\mathbb{IPSC}$, which verifies $t^{\mathbb{C}}$ and stores $\mu$-$etx$, which is now visible to $\mathbb{O}$ and the public.
Therefore, $\mathbb{O}$ might relay $\mu$-$etx$ to $\mathbb{E}$ for processing and then provide $\mathbb{E}$-signed execution receipt to $\mathbb{IPSC}$ that publicly resolves this censorship request.
On the other hand, if $\mathbb{O}$ were not to do it, $\mathbb{IPSC}$ would contain an indisputable proof of censorship by $\mathbb{O}$ on a client $\mathbb{C}$.

\subsubsection{\textbf{Multiple CBDC Instances}}
The conceptual model of our interoperable CBDC architecture is depicted in \autoref{fig:architecture}.
It consists of multiple CBDC instances (i.e., at least two), whose $\mathbb{C}s$ communicate in three different ways: (1) directly with each other, (2) in the instance-to-instance fashion through the infrastructure of their $\mathbb{O}$ as well as counterpart's $\mathbb{O}$, (3) through $PB$ with $\mathbb{IPSC}$ of both $\mathbb{O}$s and a global registry $\mathbb{IMSC}$ managing identities of instances. 

For simplified description, in the following we assume the transfer operation where a local CBDC instance in \autoref{fig:architecture} is A (i.e., the sender of tokens) and the external one is B (i.e., the receiver of tokens).
To ensure interoperability, we require a communication channel of local clients $\mathbb{C}s_A$ to external clients $\mathbb{C}s_{B}$ (the green arrow), the local operator $\mathbb{O}_{A}$ (the black arrow), and the external operator $\mathbb{O}_{B}$  (the black dashed arrow).
In our interoperability protocol $\Pi^{T}$ (described later in \autoref{sec:our-protocol}), external $\mathbb{C}s_{B}$ use the channel with the local operator $\mathbb{O}_{A}$ only for obtaining incremental proofs of $L_A$'s history tree to verify inclusion of some $\mu$-transactions in $L_A$.
However, there might arise a situation in which $\mathbb{O}_{A}$ might censor such queries, therefore, we need to address it by another communication channel -- i.e., the public blockchain $PB$.

\begin{compactitem}
\item {\textbf{Censorship of External Clients.}}
We allow external clients $\mathbb{C}s_{B}$ to use the same means of censorship resolution as internal clients of a single CBDC instance (see \autoref{sec:approch-single-instance-cens}). 
To request a resolution of a censored query, the external $\mathbb{C}_B$ uses the access ticket $t^{\mathbb{C}_B}$ at $\mathbb{IPSC}_A$, which is issued by $\mathbb{E}_A$ in the first phase of $\Pi^{T}$. 

\item {\textbf{Identification of Client Accounts}}
To uniquely identify $\mathbb{C}$'s account at a particular CBDC instance, first it is necessary to specify the globally unique identifier of the CBDC instance. 
The best candidate is the blockchain address of the $\mathbb{IPSC}$ in $PB$ since it is publicly visible and unique in $PB$ (and we denote it by $\mathbb{IPSC}$). 
Then, the identification of $\mathbb{C}$'s relevant account is a pair $\mathbb{C}^{ID} = \{PK_{pb}^{\mathbb{C}}||~ \mathbb{IPSC}\}$.
Note that $\mathbb{C}$ might use the same $PK_{pb}^{\mathbb{C}}$ for the registration at multiple CBDC instances (i.e., equivalent of having accounts in multiple banks); however, to preserve better privacy, making linkage of $\mathbb{C}$'s instances more difficult, we recommend $\mathbb{C}s$ to have dedicated key pair for each instance. 

\end{compactitem}


\paragraph{\textbf{Identity Management of CBDC Instances.}}\label{sec:identity-management-of-CBDC}
To manage identities of all CBDC instances in the system, we need a global registry of their identifiers ($\mathbb{IPSC}$ addresses).
For this purpose, we use the $\mathbb{I}$dentity $\mathbb{M}$anagement $\mathbb{S}$mart $\mathbb{C}$ontract ($\mathbb{IMSC}$) deployed in $PB$ (see program $prog^{\mathbb{IMSC}}$ in \autoref{alg:imsc}).
We propose $\mathbb{IMSC}$ to be managed in either  decentralized or centralized fashion, depending on the \textbf{deployment scenario}: 
\begin{compactitem}
\item {\textbf{Decentralized Scheme}}
In the decentralized scheme, the enrollment of a new CBDC instance must be approved by a majority vote of the already existing instances.
This might be convenient for interconnecting  central banks from various countries/regions.
The enrollment requires creating a request entry at $\mathbb{IMSC}$ (i.e., $newJoin\-Request()$) by a new instance specifying the address of its $\mathbb{IPSC}_{new}$ and $PK_{PB}^{\mathbb{O}_{new}}$.
Then, the request has to be approved by voting of existing instances. 
Prior to voting (i.e., $approveJoinRequest()$), the existing instances should first verify a new instance by certain legal processes as well as by technical means: 
do the remote attestation of $prog^\mathbb{E}_{new}$, verify the inflation rate $i_r$ and the initial value of total issued tokens $t_i$ in $\mathbb{IPSC}$, etc.
Removing of the existing instance also requires the majority of all instances, who should verify legal conditions prior to voting.

\item {\textbf{Centralized Scheme}} 
So far, we were assuming that CBDC instances are equal, which might be convenient for interconnection of central banks from different countries.
However, from the single-country point-of-view, there usually exist only one central bank, which might not be interested in decentralization of its competences (e.g., issuing tokens, setting inflation rate) among multiple commercial banks. 
We respect this and enable our approach to be utilized for such a use case, while the necessary changes are made to $\mathbb{IMSC}_c$ (see \autoref{alg:imsc-centralized}), allowing to have only one CBDC authority that can add or delete instances of (commercial) banks, upon their verification (as outlined above).
The new instances can be adjusted even with token issuance capability and constraints on inflation, which is enforced within the code of $\mathbb{E}$ as well as $\mathbb{IPSC}$.
\end{compactitem}

\begin{algorithm}[t] 
	\scriptsize 
	\caption{$prog^{\mathbb{IMSC}}_d$ of decentralized $\mathbb{IMSC}$ }\label{alg:imsc}
	
	\SetKwProg{func}{function}{}{}
	
	\smallskip
	$\triangleright$ \textsc{Declaration of types and variables:}\\
	\hspace{1em} $msg$: a current transaction that called $\mathbb{IMSC}$,  \\
	\hspace{1em} struct \textbf{InstanceInfo} \{ \\
	\begin{scriptsize}
		\hspace{3em} $operator$ : $PK_{\mathbb{O}}^{PB}$ of the instance's $\mathbb{O}$,\\
		\hspace{3em} $isApproved$: admission status of the instance,\\
		\hspace{3em} $approvals \leftarrow []$ : $\mathbb{O}$s who have approved the instance creation (or deletion),\\	
	\end{scriptsize}
	\hspace{1em} \} \\
	
	\hspace{1em} $instances[]$: a mapping of $\mathbb{IPSC}$ to \textbf{InstanceInfo}, \\
	\smallskip
	$\triangleright$ \textsc{Declaration of functions:}\\
	\func{$Init$($\mathbb{IPSC}s[], \mathbb{O}s[]$) \textbf{public} \Comment{Initial instances are implicitly approved.} }{
		\textbf{assert} $|\mathbb{IPSC}s| = |\mathbb{O}s|$ ; \\
		
		\For{$i \leftarrow 0;\ i \le |\mathbb{O}s|;\ i \leftarrow i + 1$}{ 
			
			$instances[\mathbb{IPSC}s[i]] \leftarrow \textbf{InstanceInfo}(\mathbb{O}s[i], True, [])$; \\
		}
	}
	
	\smallskip
	\func{$newJoinRequest$($\mathbb{IPSC}$) \textbf{public} }{
		\textbf{assert} $instances[\mathbb{IPSC}] = ~\perp$;  \Comment{The instance must not exist yet.}\\
		$instances[\mathbb{IPSC}] \leftarrow \textbf{InstanceInfo}(msg.sender, False, [])$; \\
	}
	\smallskip
	\func{$approveJoinRequest$($\mathbb{IPSC}_{my}, \mathbb{IPSC}_{new}$) \textbf{public} }{
		\textbf{assert} $instances[\mathbb{IPSC}_{my}].operator = msg.sender$; \Comment{Sender's check.}\\
		\textbf{assert} $instances[\mathbb{IPSC}_{my}].isApproved$; \Comment{The sending $\mathbb{O}$ has valid instance.}\\
		\textbf{assert} $!instances[\mathbb{IPSC}_{new}].isApproved$; \Comment{The new instance is not approved.}\\
		
		$r \leftarrow instances[\mathbb{IPSC}_{new}]$; \\
		$r.approvals[msg.sender] \leftarrow True$; \Comment{The sender acknowledges the request.}\\
		\If{$|r.approvals| > \lfloor |instances| / 2 \rfloor$}{
			$r.isApproved \leftarrow True$; \Comment{Majority vote applies.}\\
			$r.approvals \leftarrow []$; \Comment{Switch this field for a potential deletion.} \\
		} 
	}
	\smallskip
	\func{$approveDelete$($\mathbb{IPSC}_{my}, \mathbb{IPSC}_{del}$) \textbf{public} }{
		\textbf{assert} $instances[\mathbb{IPSC}_{my}].operator = msg.sender$; \Comment{Sender's check.}\\
		\textbf{assert} $instances[\mathbb{IPSC}_{my}].isApproved$; \Comment{The sending $\mathbb{O}$ has valid instance.}\\
		\textbf{assert} $instances[\mathbb{IPSC}_{del}].isApproved$; \Comment{An instance to delete must be approved.}\\
		
		$r \leftarrow instances[\mathbb{IPSC}_{del}]$; \\
		$r.approvals[msg.sender] \leftarrow True$; \Comment{The sender acknowledges the request.}\\
		\If{$|r.approvals| > \lfloor |instances| / 2 \rfloor$}{
			\textbf{delete} $r$;
		} 
	}
	
	\smallskip
	
	\vspace{-0.1cm}
\end{algorithm}

\begin{algorithm}[t] 
	\scriptsize 
	\caption{$prog^{\mathbb{IMSC}}_c$ of centralized $\mathbb{IMSC}$ }\label{alg:imsc-centralized}
	
	\SetKwProg{func}{function}{}{}
	
	\smallskip
	$\triangleright$ \textsc{Declaration of types and variables:}\\
	\hspace{1em} $msg$: a current transaction that called $\mathbb{IMSC}$,  \\
	\hspace{1em} $authority$: $\mathbb{IPSC}$ of the authority bank, \\
	\hspace{1em} $authority^\mathbb{O}$: $PK^\mathbb{O}_{pb}$ of $\mathbb{O}$ at authority bank, \\
	\hspace{1em} $instances[]$: a mapping of $\mathbb{IPSC}$ to $PK_{\mathbb{O}}^{PB}$, \\
	\smallskip
	$\triangleright$ \textsc{Declaration of functions:}\\
	\func{$Init$($\mathbb{IPSC}$) \textbf{public} \Comment{Initial instances are implicitly approved.} }{				
		$authority^\mathbb{O} \leftarrow msg.sender$; \\
		$authority \leftarrow \mathbb{IPSC}$; \\
	}
	
	\smallskip
	\func{$addInstance$($\mathbb{IPSC}_{new}, ~\mathbb{O}_{new}$) \textbf{public} }{
		\textbf{assert} $msg.sender =  authority^\mathbb{O}$ ; \Comment{Only the authority can add instances.}\\
		\textbf{assert} $instances[\mathbb{IPSC}_{new}] = ~\perp$;  \Comment{The instance must not exist yet.}\\
		$instances[\mathbb{IPSC}_{new}] \leftarrow \mathbb{O}_{new}$; \\
	}
	\smallskip
	
	\func{$delInstance$($\mathbb{IPSC}_{del}$) \textbf{public} }{
		\textbf{assert} $msg.sender =  authority^\mathbb{O}$ ; \Comment{Only the authority can delete instances.}\\
		\textbf{delete} $instances[\mathbb{IPSC}_{del}]$; \\
	}
	\smallskip
	\vspace{-0.1cm}
\end{algorithm}

\paragraph{\textbf{Token Issuance.}}
With multiple CBDC instances, $\mathbb{C}$s and the public can obtain the total value of issued tokens in the ecosystem of CBDC and compare it to the total value of token supply of all instances.
Nevertheless, assuming only two instances A and B, the value of $t_s$ snapshotted by $\mathbb{IPSC}_A$ might not reflect the recently executed transfers to instance B that might have already made the snapshot of its actual $L_B$ version to $\mathbb{IPSC}_B$, accounting for the transfers.
As a consequence, given a set of instances, the value of the aggregated $t_s$ should always be greater or equal than the corresponding sum of $t_i$:
\begin{eqnarray}
	t_i^A + t_i^B  &\leq& t_s^A + t_s^B. 
\end{eqnarray}
We can generalize it for $N$ instances known by $\mathbb{IMSC}$ as follows:
\begin{eqnarray}\label{eq:token-sum-multi}
	\sum_{\forall X ~\in~ \mathbb{IMSC}} t_i^X   &\leq& \sum_{\forall X ~\in~ \mathbb{IMSC}} t_s^X. 
\end{eqnarray}

\paragraph{\textbf{Inflation Rate.}}
In contrast to a single CBDC instance, multiple independent instances must provide certain guarantees about inflation not only to their clients, but also to each other.
For this purpose, the parameter inflation rate $i_r$ is adjusted to a constant value in the initialization of $\mathbb{IPSC}$ and checked before the instance is approved at $\mathbb{IMSC}$.

If one would like to enable the update of $i_r$ at CBDC instances, a majority vote at $\mathbb{IMSC}$ on a new value could be utilized (or just the vote of authority in the case of centralized scenario). 
Nevertheless, to support even fairer properties, $\mathbb{C}$s of a particular instance might vote on the value of $i_r$ upon its acceptance by $\mathbb{IOMC}$ and before it is propagated to $\mathbb{IPSC}$ of an instance.
Then, based on the new value of $\mathbb{IPSC}.i_r$, $\mathbb{E}.i_r$  can be adjusted as well (i.e., upon the validation by the light client of $\mathbb{E}$). 
However, the application of such a mechanism might depend on the use case, and we state it only as a possible option that can be enabled in our approach.  

\begin{algorithm}[t] 
	\caption{$prog^{\mathbb{IOMC}^{S}}$ of sending $\mathbb{IOMC}^{S}$ }\label{alg:iomc-send}
	\scriptsize
	\SetKwProg{func}{function}{}{}
	
	\smallskip
	$\triangleright$ \textsc{Declaration of types and variables:}\\
	
	\hspace{1em} $\mathbb{E}$, \Comment{The reference to $\mathbb{E}_A$ of sending party. } \\
	\hspace{1em} $msg$, \Comment{The current $\mu$-transaction that called $\mathbb{IOMC}^S$.}  \\
	\hspace{1em} struct \textbf{LockedTransfer} \{ \\
	\hspace{3em} $sender$, \Comment{Sending client $\mathbb{C}_A$.} \\
	\hspace{3em} $receiver$, \Comment{Receiving client $\mathbb{C}_B$.} \\
	\hspace{3em} $receiver\mathbb{IPSC}$, \Comment{The $\mathbb{IPSC}$ contract address of the receiver's instance.}\\
	\hspace{3em} $amount$, \Comment{Amount of tokens sent.}\\
	\hspace{3em} $hashlock$,  \Comment{Hash of the secret of the sending $\mathbb{C}_A$.}\\
	\hspace{3em} $timelock$, \Comment{A timestamp defining the end of validity of the transfer.}\\
	\hspace{3em} $isCompleted$, \Comment{Indicates whether the transfer has been completed.}\\
	\hspace{3em} $isReverted$, \Comment{Indicates whether the transfer has been canceled.}\\ 
	\hspace{1em} \}, \\

	\hspace{1em} $transfers \leftarrow  []$, \Comment{Initiated outgoing transfers (i.e., LockedTransfer).}  \\
	\hspace{1em}\textbf{const} $timeout^{HTLC} \leftarrow 24h$, \Comment{Set the time lock for e.g., 24 hours.}   \\
	
	\smallskip
	$\triangleright$ \textsc{Declaration of functions:}\\
		%
	
	\func{$sendInit$($receiver, receiver\mathbb{IPSC}, hashlock$) \textbf{public payable} }{
		\textbf{assert} $msg.value > 0$; \Comment{Checks the amount of tokens.} \\
		$timelock \leftarrow timestamp.now() + timeout^{HTLC}$; \\
		$t \leftarrow \textbf{LockedTransfer}(msg.sender, receiver, receiver\mathbb{IPSC}, $\\
		\hspace{1em} $msg.value$, $hashlock, timelock, False, False)$; \Comment{A new receiving transfer.} \\
		$transfers.append(t)$; \\
		
		\textbf{Output} $("sendInitialized", transferID \leftarrow |transfers| - 1))$; \\
	}
	\smallskip
	\func{$sendCommit$($transferID, secret, extTransferID$) \textbf{public} }{
		\textbf{assert} $transfers[transferID] \neq \perp$;  \Comment{Check the existence of locked transfer.}\\
		$t \leftarrow transfers[transferID]$; \\
		\textbf{assert} $t.hashlock = h(secret)$;  \Comment{Check the secret.}\\
		\textbf{assert} $!t.isCompleted ~\wedge~ !t.isReverted$;  \Comment{Test if the transfer is still pending.}\\
		$t.isCompleted \leftarrow True$; \\
		\textbf{burn} t.amount; \Comment{Burn tokens.}\\
		$\mathbb{E}.t_s \leftarrow \mathbb{E}.t_s - t.amount$; \Comment{Decrease the total supply of the instance.} \\
		\textbf{Output} $("sendCommitted",transferID,$ $extTransferID,$ $t.receiver,$ $t.receiver\mathbb{IPSC},$ $t.amount)$;\\
	}
	\smallskip
	\func{$sendRevert$($transferID$) \textbf{public} }{
		\textbf{assert} $transfers[transferID] \neq \perp$;  \Comment{Check the existence of locked transfer.}\\
		$t \leftarrow transfers[transferID]$; \\
		\textbf{assert} $!t.isCompleted ~\wedge~ !t.isReverted$;  \Comment{Test the transfer is still pending.}\\
		\textbf{assert} $t.timelock \leq timestamp.now()$;  \Comment{Check the HTLC expiration.}\\
		$transfer(t.amount, t.sender)$; \Comment{Returning tokens back to the sender.}\\
		$t.isReverted \leftarrow True$; \\
		\textbf{Output}$("sendReverted", transferID)$;\\
	}
	\smallskip
	\vspace{-0.1cm}
\end{algorithm}

\paragraph{\textbf{Interoperability.}}
The interoperability logic itself is provided by our protocol $\Pi^{T}$ that utilizes $\mathbb{I}$nter$\mathbb{O}$perability $\mathbb{M}$icro  $\mathbb{C}$on\-tracts $\mathbb{IOMC}^S$ and $\mathbb{IOMC}^R$, which serve for sending and receiving tokens, respectively.
Therefore, in the context of $\mathbb{E}$-isolated environment these $\mu$-contracts allow to mint and burn tokens, reflecting the changes in $t_s$ after sending or receiving tokens between CBDC instances.
Both $\mu$-contracts are deployed in $\mathbb{E}$ by each $\mathbb{O}$ as soon as the instance is created, while $\mathbb{E}$ records their addresses that can be obtained and attested by $\mathbb{C}$s.
We briefly review these contracts in the following, while they detailed usage will be demonstrated in \autoref{sec:our-protocol}. 

\begin{compactitem}

\item {\textbf{The Sending $\mathbb{IOMC}^S$}}
The sending $\mathbb{IOMC}^S$ (see \autoref{alg:iomc-send}) is based on Hash Time LoCks (HTLC), thus upon initialization of transfer by $hashlock$ provided by $\mathbb{C}_A$ (i.e., $hashlock \leftarrow h(secret)$) and calling $sendInit(hash\-lock, \ldots)$,  $\mathbb{IOMC}^S$ locks transferred tokens for the timeout required to complete the transfer by $send\-Commit\-(secret, \ldots)$. 
If tokens are not successfully transferred to the recipient of the external instance during the timeout, they can be recovered by the sender (i.e., $sendRevert()$).\footnote{Note that setting a short timeout might prevent the completion of the protocol.}
If tokens were sent successfully from $\mathbb{C}_A$ to  $\mathbb{C}_B$, then instance A burns them within $sendCommit()$ of $\mathbb{IOMC}^S$ and deducts them from $t_s$.
Note that deducting $t_s$ is a special operation that cannot be executed within standard $\mu$-contracts, but $\mathbb{IOMC}$ contracts are exceptions and can access some variables of~$\mathbb{E}$.

\item {\textbf{The Receiving $\mathbb{IOMC}^R$}}
The receiving  $\mathbb{IOMC}^R$ (see \autoref{alg:iomc-recv}) is based on Hashlocks (referred to as HLC) and works pairwise with sending $\mathbb{IOMC}^S$ to facilitate four phases of our interoperable transfer protocol $\Pi_{T}$ (described below).
After calling $\mathbb{IOMC^S}.sendInit()$, incoming initiated transfer is recorded at $\mathbb{IOMC}^R$ by $receiveInit(hash\-lock, \ldots)$.
Similarly, after executing token deduction at instance A (i.e., $\mathbb{IOMC^S}.send\-Commit\-(secret, \ldots)$), incoming transfer is executed at $\mathbb{IOMC}^R$ by $receiveCommit(secret, \ldots)$ that mints tokens to $\mathbb{C}_B$ and increases $t_s$.
Similar to $\mathbb{IOMC}^S$, minting tokens and increasing $t_s$ are special operations requiring access to $\mathbb{E}$, which is exceptional for $\mathbb{IOMC}$.
The overview of $\Pi_T$ is depicted in \autoref{fig:protocol-simplified}.
\end{compactitem}

\begin{algorithm}[t] 
	\caption{$prog^{\mathbb{IOMC}^{R}}$ of receiving $\mathbb{IOMC}^{R}$}\label{alg:iomc-recv}
	\scriptsize
	\SetKwProg{func}{function}{}{}
	
	\smallskip
	$\triangleright$ \textsc{Declaration of types and variables:}\\
	\hspace{1em} $\mathbb{E}$, \Comment{The reference to $\mathbb{E}_B$ of receiving party.} \\
	
	\hspace{1em} struct \textbf{LockedTransfer} \{ \\
	\begin{scriptsize}
		\hspace{3em} $sender$, \Comment{Sending client $\mathbb{C}_A$.} \\
		\hspace{3em} $sender\mathbb{IPSC}$, \Comment{The IPSC contract address of the sender's instance.}\\
		\hspace{3em} $receiver$,\Comment{Receiving client $\mathbb{C}_B$.}  \\
		\hspace{3em} $amount$, \Comment{Amount of transferred tokens.}\\
		\hspace{3em} $hashlock$, \Comment{Hash of the secret of the sending $\mathbb{C}_A$.}\\
		\hspace{3em} $isCompleted$, \Comment{Indicates whether the transfer has been completed.}\\    	
	\end{scriptsize}
	\hspace{1em}	\}, \\
	
	\hspace{1em} $transfers \leftarrow  []$, \Comment{Initiated incoming transfers (i.e., LockedTransfer).}  \\	
	\smallskip
	$\triangleright$ \textsc{Declaration of functions:}
	
	\smallskip
	\func{$receiveInit$($sender, sender\mathbb{IPSC}, hashlock, amount$) \textbf{public} }{
		\textbf{assert} $amount > 0$;  \\
		$t \leftarrow \textbf{LockedTransfer}(sender, sender\mathbb{IPSC}, msg.sender, amount,$\\
		\hspace{6em} $hashlock, False)$; \Comment{Make a new receiving transfer entry.} \\
		$transfers.append(t)$; \\
		\textbf{Output}$("receiveInitialized", transferID \leftarrow |transfers| - 1)$;\\
	}
	\smallskip
	\func{$receiveCommit$($transferID, secret$) \textbf{public} }{
		\textbf{assert} $transfers[transferID] ~\neq~ \perp$;  \Comment{Check the existence of transfer entry.}\\
		$t \leftarrow transfers[transferID]$; \\
		\textbf{assert} $t.hashlock = h(secret)$;  \Comment{Check the secret.}\\
		\textbf{assert} $!t.isCompleted$;  \Comment{Check whether the transfer is pending.}\\
		
		$\mathbb{E}.mint(this, t.amount)$; \Comment{Call $\mathbb{E}$ to mint tokens on $\mathbb{IOMC}_R$.} \\
		$\mathbb{E}.t_s \leftarrow \mathbb{E}.t_s + t.amount$ ; \Comment{Increase the total supply of the instance.} \\
		$transfer(t.amount, t.receiver)$; \Comment{Credit tokens to the recipient.}\\
		$t.isCompleted \leftarrow True$; \\        
		\textbf{Output}$("receiveCommited", transferID)$;\\
		
	}
	\smallskip
	\vspace{-0.1cm}
\end{algorithm}

\subsubsection{Interoperable Transfer Protocol $\mathbf{\Pi^T}$}\label{sec:our-protocol}
In this section we outline our instance-to-instance interoperable transfer protocol $\Pi^{T}$ for inter-CBDC transfer operation, which is inspired by the atomic swap protocol (see \autoref{sec:atomicswap}), but in contrast to the exchange-oriented approach of atomic swap, $\Pi^{T}$ focuses only on one-way atomic transfer between instances of the custodial environment of CBDC, where four parties are involved in each transfer -- a sending $\mathbb{C}_A$ and $\mathbb{O}_A$ versus a receiving $\mathbb{C}_B$ and $\mathbb{O}_B$.
The goal of $\Pi^{T}$ is to eliminate any dishonest behavior by $\mathbb{C}s$ or $\mathbb{O}$s that would incur token duplication or the loss of tokens.
\begin{figure}[hb]
	\centering
	\includegraphics[width=0.6\columnwidth]{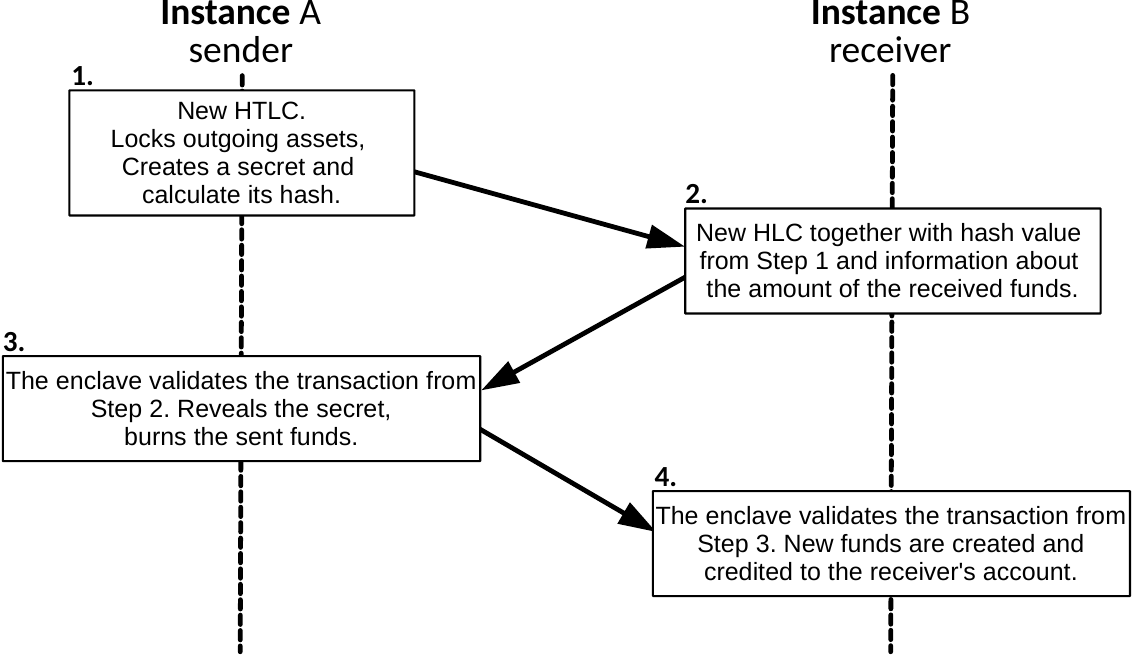}
	\caption{Overview of the protocol $\Pi^{T}$, consisting of 4 phases. 
	}
	\label{fig:protocol-simplified}
\end{figure}




	\begin{figure*}[th]
		\includegraphics[width=0.95\textwidth]{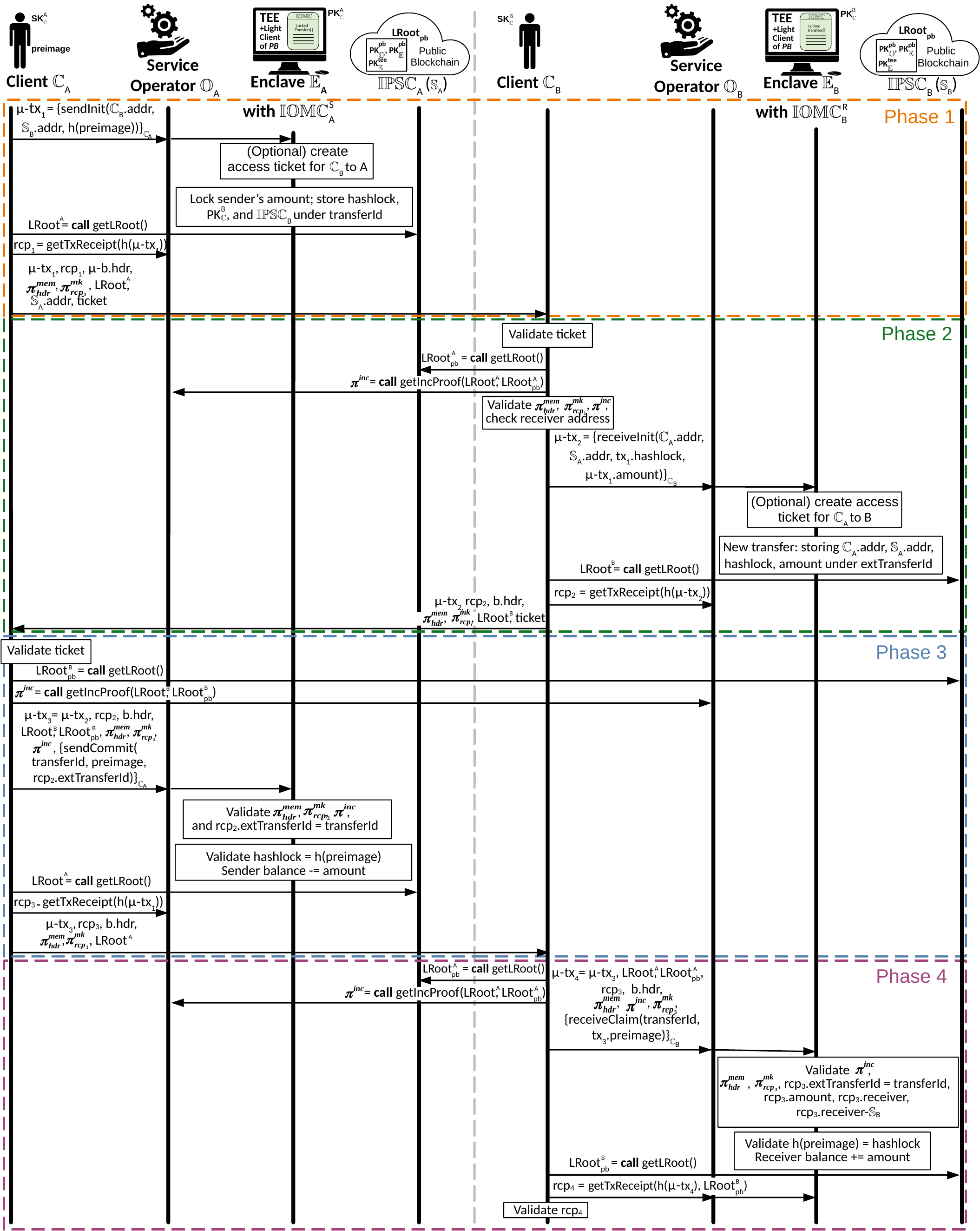}
		\caption{The details of the proposed interoperability protocol $\Pi^{T}$.}
		\label{fig:protocol-full}
	\end{figure*}
	\clearpage



To execute $\Pi^{T}$ it is necessary to inter-connect $\mathbb{E}$s of two instances involved in a transfer.
However, $\mathbb{E}$ does not allow direct communication with the outside world, and therefore it is necessary to use an intermediary.
One solution is to involve $\mathbb{O}$s but they might be overwhelmed with other activities, updating the ledger by executing $\mu$-transactions, and moreover, they might not have direct incentives to execute inter-CBDC transfers.
Therefore, we argue that in contrast to the above option, involving $\mathbb{C}s$ as intermediaries has two advantages: 
(1) elimination of the synchronous communication overhead on $\mathbb{O}$s and 
(2) enabling $\mathbb{C}$s to have a transparent view about the status of the transfer and take action if required.
%
The details of the protocol are depicted in \autoref{fig:protocol-full}.

\paragraph{Phase 1 -- Client $\mathbb{C}_A$ Initiates the Protocol.}
\label{design:phase1}
The client $\mathbb{C}_A$ creates a $\mu$-tx$_1$ with the amount being sent, which invokes the \texttt{sendInit()} of $\mathbb{IOMC}_A$ with  
arguments containing the address of the external client $\mathbb{C}_B$, the address of $\mathbb{IPSC}_B$ (denoted as $\mathbb{S}_B$ in \autoref{fig:protocol-full} for brevity), and the hash of the secret that is created by $\mathbb{C}_A$.
$\mathbb{C}_A$ sends signed $\mu$-tx$_1$ to $\mathbb{O}_A$ who forwards it to the $\mathbb{E}_A$.
Before executing the $\mu$-tx$_1$, $\mathbb{E}_A$ ensures that the external recipient (i.e., $\mathbb{C}_B$) has the access ticket already issued and valid, enabling her to post censorship resolution requests to $\mathbb{IPSC}_A$ (if needed).
The access ticket should be valid for at least the entire period defined by the HTLC of $\mathbb{IOMC}_A$.
%
In the next step, a $\mu$-tx$_1$ is executed by $\mathbb{E}_A$, creating a new transfer record with $transferId$ in $\mathbb{IOMC}_A$.
During the execution, $\mathbb{C}_A$'s tokens are transferred (and thus locked) to the $\mathbb{IOMC}_A$'s address.
%
$\mathbb{C}_A$ waits until the new version of $L_A$ is snapshotted to $\mathbb{IPSC}_A$, and then obtains $LRoot^{A}$ from it.
Then $\mathbb{C}_A$ asks $\mathbb{O}_A$ for the execution receipt $rcp_1$ of $\mu$-tx$_1$ that also contains a set of proofs ($\pi_{hdr}^{mem}$, $\pi_{rcp_1}^{mk}$) and the header of the $\mu$-block that includes  $\mu$-tx$_1$. 
In detail, $\pi_{hdr}^{mem}$ is the inclusion proof of the $\mu$-block \textit{b} in the current version of $L_A$; $\pi_{rcp_1}^{mk}$ is the Merkle proof proving that $rcp_1$ is included in \textit{b} (while $rcp_1$ proves that $\mu$-tx$_1$ was executed correctly).
The mentioned proofs and the receipt are provided to $\mathbb{C}_B$, who verifies that $\mu$-tx$_1$ was executed and included in the $L_A$'s version that is already snapshotted to $\mathbb{IPSC}_A$,  thus irreversible (see below).


\paragraph{Phase 2 -- $\mathbb{C}_B$ Initiates Receive.}
\label{design:phase2}
First, $\mathbb{C}_B$ validates an access ticket to $\mathbb{IPSC}_A$ using the enclave $\mathbb{E}_A$'s public key accessible in that smart contract.
Next, $\mathbb{C}_B$ obtains the root hash $LRoot_{pb}^{A}$ of $L_A$ to ensure that $\mathbb{C}_B$'s received state has been already published in $\mathbb{IPSC}_A$, and thus contains $\mu$-tx$_1$.
After obtaining $LRoot_{pb}^{A}$, $\mathbb{C}_B$ forwards it along with the root $LRoot^{A}$ obtained from $\mathbb{C}_A$ to $\mathbb{O}_A$, who creates an incremental proof $\pi^{inc}$ of $\langle LRoot^{A}, LRoot_{pb}^{A}\rangle$.
%
Once the proof $\pi^{inc}$ has been obtained and validated, the protocol can proceed to validate the remaining proofs sent by the client $\mathbb{C}_A$ along with verifying that the receiving address belongs to $\mathbb{C}_B$.
Next, $\mathbb{C}_B$ creates $\mu$-$tx_2$, invoking the method \texttt{receiveInit()} with the arguments: the address of $\mathbb{C}_A$ obtained from $\mu$-$tx_1$,\footnote{Note that we assume that the address is extractable from the signature.} the address $\mathbb{IPSC}_A.addr$ of $\mathbb{C}_A$'s instance, the hash value of the secret, and the amount of crypto-tokens being sent.
$\mathbb{C}_B$ sends $\mu$-$tx_2$ to $\mathbb{O}_B$, who forwards it to $\mathbb{E}_B$ for processing.
During processing of $\mu$-$tx_2$, $\mathbb{E}_B$ determines whether the external client (from its point of view -- i.e., $\mathbb{C}_A$) has an access ticket issued with a sufficiently long validity period; if not, one is created.
Subsequently, $\mathbb{E}_B$ creates a new record in $\mathbb{IOMC}_B$ with $extTransferId$.
Afterward, $\mathbb{C}_B$ retrieves the $LRoot^{B}$ from $L_B$ and requests the execution receipt  $rcp_2$ from $\mathbb{O}_B$, acknowledging that the $\mu$-$tx_2$ has been executed. 
Finally, $\mathbb{C}_B$ sends a message $\mathbb{C}_A$ with $\mu$-$tx_2$ and cryptographic proofs $\pi_{hdr}^{mem}$, $\pi_{rcp_2}^{mk}$, the execution receipt of $\mu$-$tx_2$, the block header $b$ in which the $\mu$-$tx_2$ was included, $LRoot^B$ (i.e., the root value of $L_B$ after $\mu$-$tx_2$ was executed), and the valid client access ticket for $\mathbb{C}_A$.

\paragraph{Phase 3 -- Confirmation of Transfer by $\mathbb{C}_A$.}
\label{design:phase3}
First, $\mathbb{C}_A$ validates the received access ticket to $\mathbb{IPSC}_B$. 
Next, $\mathbb{C}_A$ obtains the snapshotted root hash $LRoot_{pb}^{B}$ of $L_B$ from $\mathbb{IPSC}_B$.
As in the previous phases, it is necessary to verify that the version of $L_B$ that includes $\mu$-$tx_2$ is represented by $LRoot_{pb}^B$ (thus is irreversible). 
Next, both root hashes ($LRoot^B$ and $LRoot_{pb}^B$) are sent to the external operator $\mathbb{O}_B$, which produces the incremental proof $\pi^{inc}$ from them.
%
Next, $\mathbb{C}_A$ creates $\mu$-$tx_3$ that consists of invoking the \texttt{sendCommit()} method at $\mathbb{E}_A$ with the arguments containing the published secret (i.e., $preimage$) and the record identifier of the transfer at local instance (i.e., $transferId$) as well as the external one (i.e., $extTransferId$).
Along with the invocation of \texttt{sendCommit()}, $\mu$-$tx_3$ also wraps $\pi^{inc}$ with its versions ($LRoot_{pb}^{B}$ and $LRoot^{B}$), $\mu$-$tx_2$, its execution receipt $rcp_2$ with its Merkle proof $\pi_{rcp_2}^{mk}$, $b.hdr$ -- the header of the block that included $\mu$-$tx_2$, and its membership proof $\pi^{mem}_{hdr}$ of $L_B$.
Next, $\mathbb{C}_A$ sends $\mu$-$tx_3$ to $\mathbb{E}_A$ through $\mathbb{O}_A$.
During the execution of $\mu$-$tx_3$, $\mathbb{E}_A$ validates the provided proofs and the equality of transfer IDs from both sides of the protocol.
Note that to verify $\pi^{mem}_{hdr}$, $\mathbb{E}_A$ uses its light client to $L_B$.
$\mathbb{E}_A$ then validates whether $\mathbb{C}_A$'s provided secret corresponds to the hashlock recorded in the 1st phase of the protocol, and if so, it burns the sent balance of the transfer.

Next, $\mathbb{C}_A$ waits until the new version of $L_A$ is snapshotted to $\mathbb{IPSC}_A$, and then obtains $LRoot^{A}$ from it.
Then $\mathbb{C}_A$ asks $\mathbb{O}_A$ for the execution receipt $rcp_3$ of $\mu$-tx$_3$ that also contains a set of proofs ($\pi_{hdr}^{mem}$, $\pi_{rcp_3}^{mk}$) and the header of the $\mu$-block that includes  $\mu$-tx$_3$. 
The proofs have the same interpretation as in the end of the 1st phase. 
The mentioned proofs and the receipt are provided to $\mathbb{C}_B$, who verifies that $\mu$-tx$_1$ was executed and included in the $L_A$'s version that is already snapshotted to $\mathbb{IPSC}_A$,  thus irreversible.

\paragraph{Phase 4 -- Acceptance of Tokens by $\mathbb{C}_B$.}
\label{design:phase4}
After receiving a message from client $\mathbb{C}_A$, the client $\mathbb{C}_B$ obtains $LRoot_{pb}^A$ from $\mathbb{IPSC}_A$ and then requests the incremental proof between versions $\langle LRoot^{A}, LRoot_{pb}^{A}\rangle$ from $\mathbb{O}_A$.
Then, $\mathbb{C}_B$ creates $\mu$-$tx_4$ invoking the \texttt{receiveClaim()} function at $\mathbb{E}_B$ with $transferId$ and the disclosed secret by $\mathbb{C}_A$ as the arguments.
Moreover, $\mu$-$tx_4$ contains remaining items received from $\mathbb{C}_A$. 
Then, $\mu$-$tx_4$ is sent to $\mathbb{O}_B$, who forwards it to $\mathbb{E}_B$.
During the execution of $\mu$-$tx_4$, $\mathbb{E}_B$ verifies the provided proofs, the equality of transfer IDs from both sides of the protocol, the amount being sent, and the receiver of the transfer (i.e., $\mathbb{C}_B$ || $\mathbb{IPSC}_B$).
Note that to verify $\pi^{mem}_{hdr}$, $\mathbb{E}_B$ uses its light client to $L_A$.
$\mathbb{E}_A$ then validates whether $\mathbb{C}_A$'s provided secret corresponds to the hashlock recorded in the 2nd phase of the protocol, and if so, it mints the sent balance of the transfer on the receiver's account $\mathbb{C}_B$.
Finally, $\mathbb{C}_B$ verifies that $\mu$-$tx_4$ was snapshotted at $\mathbb{IPSC}_B$, thus is irreversible. 
In detail, first $\mathbb{C}_B$ obtains $LRoot_{pb}^B$ from $\mathbb{IPSC}_B$ and then asks $\mathbb{O}_B$ to provide her with the execution receipt $rcp_4$ of $\mu$-$tx_4$ in the version of $L_B$ that is equal or newer than $LRoot_{pb}^B$. 
Then, $\mathbb{C}_B$ verifies $rcp_4$, which completes the protocol.

\subsection{Evaluation}
We used \textit{Ganache}\footnote{\url{https://github.com/trufflesuite/ganache-cli}} and \textit{Truffle},\footnote{\url{https://github.com/trufflesuite/truffle}} to develop $\mathbb{IOMC}$, $\mathbb{IPSC}$, and $\mathbb{IMSC}$ contracts. 
In addition, using the \textit{Pexpect}\footnote{\url{https://github.com/pexpect/pexpect}} tool, we tested the intercommunication of the implemented components and validated the correctness of the implemented interoperability protocol.
The tool enabled the parallel execution and control of numerous programs (in this case, multiple Aquareum instances and client programs) to check the correctness of the expected output.

The computational cost of executing the operations defined in $\mathbb{IOMC}$ and $\mathbb{IMSC}^X$ contracts is presented in \autoref{table:gas-iomc-send}, \autoref{table:gas-iomc-recv}, and \autoref{table:gas-imsc}.\footnote{Note that we do not provide the gas measurements for $\mathbb{IPSC}$ since these are almost the same as in Aquareum~\cite{homoliak2020aquareum}.}
We optimized our implementation to minimize the storage requirements of smart contract platform.
On the other hand, it is important to highlight that $\mathbb{IOMC}^X$ $\mu$-contracts are executed on a private ledger corresponding to the instance of CBDC, where the cost of gas is minimal or negligible as compared to a public blockchain.
Other experiments  are the subject of our future work.

\begin{table}[t]
	\centering
	\scriptsize
	\begin{tabular}
		{  r  c  c  c  c  } 
		\toprule
		\textbf{Function} & \vbox{\hbox{{\rotatebox[origin=c]{45}{constructor}}}\vspace{2pt}} & \vbox{\hbox{{\rotatebox[origin=c]{45}{sendInitialize}}}\vspace{2pt}} & \vbox{\hbox{{\rotatebox[origin=c]{45}{sendCommit}}}\vspace{2pt}} & \vbox{\hbox{{\rotatebox[origin=c]{45}{sendRevert}}}\vspace{2pt}}  \\
		\midrule
		\textbf{Deployment} & 901 509 & 160 698 & 64 629 & 60 923 \\
		\textbf{Execution} & 653 689 & 134 498 & 42 717 & 39 523 \\
		\bottomrule
	\end{tabular}
	\caption{The cost of deployment and invocation of functions in the sending $\mathbb{IOMC}^S$ $\mu$-contract in gas units (CBDC private ledger).}
	\label{table:gas-iomc-send}
\end{table}

\begin{table}[t]
	\centering
	\scriptsize
	\renewcommand{\arraystretch}{1.5}
	\begin{tabular}
			{  r  c  c  c  c  } 
			\toprule
			\textbf{Function} & \vbox{\hbox{{\rotatebox[origin]{45}{constructor}}}\vspace{2pt}} & \vbox{\hbox{{\rotatebox[origin]{45}{receiveInit}}}\vspace{2pt}} & \vbox{\hbox{{\rotatebox[origin]{45}{receiveClaim}}}\vspace{2pt}} & \vbox{\hbox{{\rotatebox[origin]{45}{fund}}}\vspace{2pt}}  \\
			\midrule
			\textbf{Deployment} & 716 330 & 139 218 & 61 245 & 23 168 \\
			\textbf{Execution} & 509 366 & 112 762 & 39 653 & 1 896 \\
			\bottomrule
		\end{tabular}
		\caption{The cost of deployment  and invocation of functions in the receiving $\mathbb{IOMC}^R$ $\mu$-contract in units of gas (CBDC private ledger).}
		\label{table:gas-iomc-recv}
	\end{table}
	
	\begin{table}[t]
		\centering
		\scriptsize
		\begin{tabular}
			{  r  c  c  c  c  } 
			\toprule
			\textbf{Function} & \vbox{\hbox{{\rotatebox[origin=c]{45}{constructor}}}\vspace{2pt}} & \vbox{\hbox{{\rotatebox[origin=c]{45}{newJoinRequest}}}\vspace{2pt}} & \vbox{\hbox{{\rotatebox[origin=c]{45}{approveRequest}}}\vspace{2pt}} & \vbox{\hbox{{\rotatebox[origin=c]{45}{isApproved}}}\vspace{2pt}}  \\
			\midrule
			\textbf{Deployment} & 830 074 & 48 629 & 69 642 & 0 \\
			\hline
			\textbf{Execution} & 567 838 & 25 949 & 45 554 & 0 \\
			\bottomrule
		\end{tabular}
		\caption{The invocation cost of functions in $\mathbb{IMSC}$ smart contract in units of gas (Ethereum public blockchain).}
		\label{table:gas-imsc}
	\end{table}

\subsection{Security Analysis}\label{sec:cbdc:secanalysis}
In this section, we analyze our approach in terms of security-oriented features and requirements specified in \autoref{sec:cbdc:problem}.
In particular, we focus on resilience analysis of our approach against adversarial actions that the malicious CBDC instance (i.e., its operator $\mathcal{O}$) or malicious client (i.e., $\mathcal{C}$) can perform to violate the security requirements.


\subsubsection{Single Instance of CBDC}
For properties common with Aquareum (i.e., correctness of operation execution, integrity, verifiability, non-equivocation, censorship evidence,  and privacy) as well as for security analysis of TEE and time to finality, we refer the reader to \autoref{sec:aquareum:analysis}.

\begin{theorem}\label{theorem:token-issuance}
	(Transparent Token Issuance)  $\mathcal{O}$ is unable to issue or burn any tokens without leaving a publicly visible evidence.
\end{theorem}
\begin{proof1}
	All issued tokens of CBDC are publicly visible at $\mathbb{IPSC}$ since each transaction posting a new version transition pair also contains $\mathbb{E}$-signed information about the current total issued tokens $t_i$ and total supply of the instance $t_s$,\footnote{Note that in the case of single CBDC instance $t_i = t_s$} while $t_i$ was updated within the trusted code of  $\mathbb{E}$.
	The information about $t_i$ is updated at $\mathbb{IPSC}$ along with the new version of $L$.
	Note that the history of changes in total issued tokens $t_i$ can be parsed from all transactions updating version of $L$ published by $\mathcal{O}$ to $PB$.
\end{proof1}

\subsubsection{Multiple Instances of CBDC}
In the following, we assume two CBDC instances A and~B.

\begin{theorem}\label{theorem:token-multi-inter}
	(Atomic Interoperability I) Neither $\mathcal{O}_A$ (operating $A$) nor $\mathcal{O}_B$ (operating $B$) is unable to steal any tokens during the inter-bank CBDC transfer. 
	%
\end{theorem}
\begin{proof1}
	Atomic interoperability is ensured in our approach by adaptation of atomic swap protocol for all inter-bank transfers, which enables us to preserve the wholesale environment of CBDC in a consistent state (respecting \autoref{eq:token-sum-multi}).
	In detail, the transferred tokens from CBDC instance $A$ to instance $B$ are not credited to $B$ until $A$ does not provide the indisputable proof that tokens were deducted from a relevant $A$'s account.
	This proof confirms irreversible inclusion of $tx_3$ (i.e., $\mathbb{E}_A.sendCommit()$ that deducts account of $A$'s client) in $A$'s ledger and it is verified in 4th stage of our protocol by the trusted code of $\mathbb{E}_B$.
	
	In the case that $\mathcal{O}_A$ would like to present $B$ with integrity snapshot of $L_A$ that was not synced to  $\mathbb{IPSC}_A$ yet, B will not accept it since the 4th phase of our protocol requires $\mathcal{O}_B$ to fetch the recent $\mathbb{IPSC}_A.LRoot_{pb}$ and verify its consistency with A-provided $LRoot$ as well as inclusion proof in $PB$; all executed/verified within trusted code of $\mathbb{E}_B$.
\end{proof1}

\begin{theorem}\label{theorem:token-multi-inter2}
	(Atomic Interoperability II) Colluding clients $\mathcal{C}_A$ and $\mathcal{C}_B$  of two CBDC instances cannot steal any tokens form the system during the transfer operation of our protocol.
\end{theorem}
\begin{proof1}
	If the first two phases of our protocol have been executed, $\mathcal{C}_A$ might potentially reveal the $preimage$ to $\mathcal{C}_B$ without running the 3rd phase with the intention to credit the tokens at $B$ while deduction at $A$ had not been executed yet.
	However, this is prevented since the trusted code of $\mathbb{E}_B$ verifies that the deduction was performed at $A$ before crediting the tokens to $\mathcal{C}_B$ -- as described in  \autoref{theorem:token-multi-inter}.
\end{proof1}

\begin{theorem}\label{theorem:multi-censorhip1}
	(Inter-CBDC Censorship Evidence) $\mathcal{O}_A$ is unable to unnoticeably censor any request (transaction or query) from $\mathbb{C}_B$.
\end{theorem}
\begin{proof1}
	If $\mathbb{C}_B$'s request is censored by $\mathcal{O}_A$, $\mathbb{C}_B$ can ask for a resolution of the request through public $\mathbb{IPSC}_A$ since $\mathbb{C}_B$ already has the access ticket to instance $A$.
	The access ticket is signed by $\mathbb{E}_A$ and thus can be verified at $\mathbb{IPSC}_A$.
	Hence, the censorship resolution/evidence is the same as in \autoref{theorem:censorhip} of a single CBDC instance.
\end{proof1}

\begin{theorem}\label{theorem:multi-censorship2}
	(Inter-CBDC Censorship Recovery) A permanent inter-CBDC censorship by  $\mathcal{O}_A$ does not cause an inconsistent state or permanently frozen funds of undergoing transfer operations at any other CBDC instance -- all initiated and not finished transfer operations can be recovered from.
	
\end{theorem}
\begin{proof1}
	If $\mathcal{O}_A$ were to censor $\mathbb{C}_B$ in the 2nd phase of our protocol, no changes at ledger $L_B$ would be made.
	If $\mathcal{O}_A$ were to censor $\mathbb{C}_B$ in the 4th phase of our protocol, $L_B$ would contain an initiated transfer entry, which has not any impact on the consistency of the ledger since it does not contain any locked tokens.~\hfill$\Box$
	
	
	\smallskip \noindent
	If $\mathcal{O}_B$ were to censor $\mathbb{C}_A$ in the 3rd phase of our protocol, $A$ would contain some frozen funds of the initiated transfer.
	However, these funds can be recovered back to $\mathbb{C}_A$ upon a recovery call of $\mathbb{E}_A$ after a recovery timeout has passed.
	Note that after tokens of $\mathbb{C}_A$ have been recovered and synced to $\mathbb{IPSC}_A$ in $PB$, it is not possible to finish the 4th stage of our protocol since it requires providing the proof that tokens were deducted at $A$ and such a proof cannot be constructed anymore.
	The same holds in the situation where the sync to  $\mathbb{IPSC}_A$ at $PB$ has not been made yet -- after recovery of tokens, $\mathbb{E}_A$ does not allow to deduct the same tokens due to its correct execution (see \autoref{theorem:correctness}). 
\end{proof1}

\begin{theorem}\label{theorem:multi-identity-1}
	(Identity Management of CBDC Instances I) A new (potentially fake) CBDC instance cannot enter the ecosystem of wholesale CBDC upon its decision.
\end{theorem}
\begin{proof1}
	To extend the list of valid CBDC instances (stored in IMSC contract), the majority vote of all existing CBDC instances must be achieved through public voting on IMSC.
\end{proof1}

\begin{theorem}\label{theorem:multi-identity-2}
	(Identity Management of CBDC Instances II) Any CBDC instance (that e.g., does not respect certain rules for issuance of tokens) might be removed from the ecosystem of CBDC by majority vote.
\end{theorem}
\begin{proof1}
	A publicly visible voting about removal of a CBDC instance from the ecosystem is realized by IMSC contract that resides in $PB$, while each existing instance has a single vote. 
\end{proof1}


%
%

\section{Contributing Papers}\label{sec:logging-papers}
The papers that contributed to this research direction are enumerated in the following, while highlighted papers are attached to this thesis in their original form. 
Note that these papers were not yet published nor accepted as of writing.

\begingroup
\let\clearpage\relax

\renewcommand\bibname{}

\endgroup

%% file: sec/conclusion.tex
In this work, we presented the summary of our research and its contributions to the standardization of vulnerability/threat analysis and modeling in blockchains as well as particular areas in blockchains' consensus protocols, cryptocurrency wallets,  electronic voting, and secure logging with the focus on security and/or privacy aspects.   

In detail, first, we introduced the security reference architecture for blockchains that adopts a stacked model, describing the nature and hierarchy of various security and privacy aspects.
Then, we proposed a blockchain-specific version of the threat-risk assessment standard ISO/IEC 15408 by embedding the stacked model into this standard. 
Next, we investigated a few attacks on Proof-of-Work consensus protocols such as selfish mining attacks, greedy transaction selection attacks, and undercutting attacks -- for all of them we proposed mitigation techniques. 
Then, we dealt with cryptocurrency wallets, where we described our proposed classification of authentication schemes and proposed SmartOTPs, two-factor authentication for smart contract wallets based on One-Time Passwords.  
Next, we focused on electronic voting using blockchains as an instance of a public bulletin board, and we described our proposals BBB-Voting and SBvote as well as the Always-on-Voting framework for repetitive voting.
Finally, we dealt with secure logging, where we presented Aquareum, a centralized ledger based on blockchain and trusted computing. 
At a follow-up stage in secure logging direction, we built on top of Aquareum and proposed CBDC-AquaSphere, an interoperability protocol for central bank digital currencies.

In our future work, we plan to focus on several directions, such as secure and efficient Proof-of-Stake protocols based on Direct Acyclic Graphs, simulations of selfish mining and similar incentive attacks (with various numbers of attackers) on consensus protocols designed to mitigate selfish mining, optimized on-chain verification in e-voting protocols (including relatively new concepts of partial-tally-hiding), interoperable execution of smart contracts across CBDC instances, optimization of Merkle-Patricia tries to enable secure and consistent parallel processing and thus improvement in processing throughput.

%% file: ms.bbl
\vspace{-7em}

\begin{thebibliography}{10}{}
	\itemsep 0.08em
	
	\bibitem[HVHS19]{sra1}
	Ivan Homoliak, Sarad Venugopalan, Qingze Hum, and Pawel Szalachowski.
	\newblock A security reference architecture for blockchains.
	\newblock In {\em 2019 IEEE International Conference on Blockchain
		(Blockchain)}, pages 390--397. IEEE, 2019.
	
	\bibitem[HVR{+}20]{sra2}
	\textbf{Ivan Homoliak, Sarad Venugopalan, Dani{\"e}l Reijsbergen, Qingze Hum, Richard
		Schumi, and Pawel Szalachowski.
		\newblock The security reference architecture for blockchains: Toward a
		standardized model for studying vulnerabilities, threats, and defenses.
		\newblock {\em IEEE Communications Surveys \& Tutorials}, 23(1):341--390, 2020.}
	
	
\end{thebibliography}

\vspace{-7em}

\begin{thebibliography}{10}{}
	\itemsep 0.08em
	
\bibitem[SRHS19]{cons1}
\textbf{Pawel Szalachowski, Dani{\"{e}}l Reijsbergen, Ivan Homoliak, and Siwei Sun.
\newblock Strongchain: Transparent and collaborative proof-of-work consensus.
\newblock In {\em 28th {USENIX} Security Symposium, {USENIX} Security 2019,
	Santa Clara, CA, USA, August 14-16, 2019.}, pages 819--836, 2019.	
}

\bibitem[PBH{+}23]{cons2}
\textbf{Martin Pere{\v{s}}{\'\i}ni, Federico~Matteo Ben{\v{c}}i{\'c}, Martin Hrub{\'y},
Kamil Malinka, and Ivan Homoliak.
\newblock Incentive attacks on DAG-based blockchains with random transaction
selection.
\newblock In {\em {IEEE} International Conference on Blockchain, Blockchain
	2023, Hainan, China, December 17-21, 2023}. {IEEE}.}


\bibitem[BHS23]{cons3}
\textbf{Rastislav Budinsk{\'y},  and Ivana Stan{\v{c}}{\'\i}kov{\'a} and Ivan Homoliak.
\newblock Fee-redistribution smart contracts for transaction-fee-based regime
of blockchains with the longest chain rule.
\newblock In {\em {IEEE} International Conference on Blockchain, Blockchain
	2023, Hainan, China, December 17-21, 2023}. {IEEE}, 2023.}

\bibitem[PHMH24]{cons4}
Martin Pere{\v{s}}{\'\i}ni, Tom{\'a}{\v{s}} Hladk{\'y}, Kamil Malinka, and Ivan
Homoliak.
\newblock Dag-Sword: A simulator of large-scale network topologies for
dag-oriented proof-of-work blockchains.
\newblock In {\em Hawaii International Conference on System Sciences (HICSS),
	Hawaii, USA, January 3-6, 2024.} {IEEE}, 2024.
	
\end{thebibliography}

\vspace{-7em}

\begin{thebibliography}{10}{}
	\itemsep 0.08em
	

	\bibitem[HBH{+}20a]{wallets1}
\textbf{	Ivan Homoliak, Dominik Breitenbacher, Ondrej Hujnak, Pieter Hartel, Alexander
	Binder, and Pawel Szalachowski.
	\newblock {SmartOTPs}: An air-gapped 2-factor authentication for smart-contract
	wallets.
	\newblock In {\em Proceedings of the 2nd {ACM} Conference on Advances in
		Financial Technologies, {AFT} 2020, New York, NY, USA, October 21-23, 2020}.
	{ACM}, 2020.}
	
	\bibitem[HBH{+}20b]{wallets2}
	Ivan Homoliak, Dominik Breitenbacher, Ondrej Hujnak, Pieter Hartel, Alexander
	Binder, and Pawel Szalachowski.
	\newblock {SmartOTPs: An Air-Gapped 2-Factor Authentication for Smart-Contract
		Wallets {(Extended Version)}}.
	\newblock {\em arXiv preprint arXiv:1812.03598}, 2020.	
	
	
\end{thebibliography}

\vspace{-7em}

\begin{thebibliography}{10}{}
	\itemsep 0.08em

	
\bibitem[HLS23]{homoliak2023bbb}
\textbf{Ivan Homoliak, Zengpeng Li, and Pawel Szalachowski.
\newblock {BBB-Voting: 1-out-of-k blockchain-based boardroom voting}.
\newblock In {\em {IEEE} International Conference on Blockchain, Blockchain
	2023, Hainan, China, December 17-21, 2023}. {IEEE}, 2023.
}

\bibitem[SH23]{stanvcikova2023sbvote}
\textbf{Ivana Stan{\v{c}}{\'\i}kov{\'a} and Ivan Homoliak.
\newblock SBvote: Scalable self-tallying blockchain-based voting.
\newblock In {\em Proceedings of the 38th ACM/SIGAPP Symposium on Applied Computing}, pages 203--211, 2023.}

\bibitem[VSH23]{venugopalan2023always}
\textbf{Sarad Venugopalan, Ivana Stan{\v{c}}{\'\i}kov{\'a}, and Ivan Homoliak.
\newblock Always on Voting: A framework for repetitive voting on the
blockchain.
\newblock {\em IEEE Transactions on Emerging Topics in Computing}, 2023.}
	
	
		
\end{thebibliography}

\vspace{-7em}

\begin{thebibliography}{10}{}
	\itemsep 0.08em
	

\bibitem[HS20]{logging1}
\textbf{Ivan Homoliak and Pawel Szalachowski.
\newblock Aquareum: A centralized ledger enhanced with blockchain and trusted computing.
\newblock {\em arXiv preprint arXiv:2005.13339}, 2020.	
}	
	
\bibitem[HPH{+}23]{logging2}
\textbf{Ivan Homoliak, Martin Pere{\v{s}}{\'\i}ni, Patrik Holop, Jakub Handzu{\v{s}},
and Fran Casino.
\newblock CBDC-Aquasphere: Interoperable central bank digital currency built on trusted computing and blockchain.
\newblock {\em arXiv preprint arXiv:2305.16893}, 2023.	}
	
\end{thebibliography}
